\documentclass[aps,prb,twocolumn,superscriptaddress,floatfix,showpacs,10pt,letterpaper]{revtex4-1}

\usepackage{float}

\allowdisplaybreaks[1]

\pdfoutput=1

\usepackage{graphicx,amsmath,amsfonts,amsthm,amssymb,bm}
\usepackage{color}
\usepackage[all,cmtip]{xy}
\usepackage{relsize}

\newcommand{\mce}[1]{\mathrm{UMTC}_{/#1}}
\newcommand{\aute}{\mathrm{Aut}}

\theoremstyle{definition}
\newtheorem{thm}{Theorem}
\newtheorem{cor}[thm]{Corollary}

\newtheorem{conj}{Conjecture}
\newtheorem{dfn}{Definition}
\newtheorem{rmk}{Remark}
\newtheorem{lem}{Lemma}
\newtheorem{exa}{Example}
\newcommand{\mext}{\mathcal{M}_{ext}}
\newcommand{\ov}{\overline}
\newcommand{\Hilb}{\mathrm{Vec}}
\newcommand{\Fun}{\mathrm{Fun}}

\begin{document}
\begin{titlepage}

  \title{Classification of 2+1D topological orders and SPT orders\\
for bosonic and fermionic systems with on-site symmetries}

\author{Tian Lan} 
\affiliation{Perimeter Institute for Theoretical Physics, Waterloo, Ontario N2L 2Y5, Canada} 
\affiliation{Department of Physics and Astronomy,
  University of Waterloo, Waterloo, Ontario N2L 3G1, Canada}

\author{Liang Kong} 
\affiliation{Department of Mathematics and Statistics,
University of New Hampshire, Durham, NH 03824, USA}
\affiliation{Center of Mathematical Sciences and Applications, Harvard University, Cambridge, MA 02138}
%\affiliation{Institute for Advanced Study, Tsinghua University, Beijing 100084, China} 

\author{Xiao-Gang Wen} 
\affiliation{Department of Physics, Massachusetts Institute of Technology, Cambridge, Massachusetts 02139, USA}
\affiliation{Perimeter Institute for Theoretical Physics, Waterloo, Ontario N2L 2Y5, Canada}

\begin{abstract} 
Gapped quantum liquids (GQL) include both topologically ordered states (with
long range entanglement) and symmetry protected topological (SPT) states (with
short range entanglement).  In this paper, we propose a classification of 2+1D
GQL for both bosonic and fermionic systems: 2+1D bosonic/fermionic GQLs with
finite on-site symmetry are classified by non-degenerate unitary braided fusion
categories over a symmetric fusion category (SFC) $\cal E$, abbreviated as
UMTC$_{/\cal E}$, together with their modular extensions and total chiral
central charges.  The SFC $\cal E$ is $\text{Rep}(G)$ for bosonic symmetry $G$,
or $\text{sRep}(G^f)$ for fermionic symmetry $G^f$.  As a special case of the
above result, we find that the modular extensions of $\text{Rep}(G)$ classify
the 2+1D bosonic SPT states of symmetry $G$, while the $c=0$ modular extensions
of $\text{sRep}(G^f)$ classify the 2+1D fermionic SPT states of symmetry $G^f$.
Many fermionic SPT states are studied based on the constructions from
free-fermion models. But it is not clear if free-fermion constructions can produce
all fermionic SPT states. Our classification does not have such a drawback. 
We show that, for interacting 2+1D fermionic systems, there are
exactly 16 superconducting phases with no symmetry and no fractional
excitations (up to $E_8$ bosonic quantum Hall states). Also, there are exactly
8 $Z_2\times Z_2^f$-SPT phases, 2 $Z_8^f$-SPT phases, and so on.  Besides, we
show that two topological orders with identical bulk excitations and central
charge always differ by the stacking of the SPT states of the same symmetry.

%To support our
%conjecture, we investigate the stacking operation of GQLs. The corresponding
%mathematical construction of stacking operation for UMTC$_{\cal E}$'s as well
%as their modular extensions is introduced.  It agrees with the physical
%pictures for invertible bosonic GQLs, and gives some nontrivial predictions for
%non-invertible bosonic or fermionic GQLs.  We also propose a second way to
%classify 2+1D GQLs, with a restriction that the symmetry group is abelian or
%simple.  The second way is based on the data $(\tilde N^{ab}_c,\tilde s_a;
%N^{ij}_k,s_i; {\cal N}^{IJ}_K,{\cal S}_I;c)$ (up to some permutations of the
%indices) plus the conditions on those data. By numerically solving those
%conditions, we computed the list of 2+1D bosonic/fermionic GQLs for simple
%symmetries.

\end{abstract}


\maketitle

\end{titlepage}

{\small \setcounter{tocdepth}{1} \tableofcontents }
\section{Introduction}

Topological order\cite{W8987,W9039,WN9077} is a new kind of order beyond the
symmetry breaking orders\cite{L3726} in gapped quantum systems.  Topological
orders are patterns of \emph{long-range entanglement}\cite{CGW1038} in
\emph{gapped quantum liquids} (GQL)\cite{ZW1490}.  Based on the unitary modular
tensor category (UMTC) theory for non-abelian
statistics\cite{MS8977,BakK01,K062}, in \Ref{KW1458,W150605768}, it is proposed
that 2+1D bosonic topological orders are classified by $\{\text{UMTC}\}\times
\{\text{iTO}_B\}$, where $\{\text{UMTC}\}$ is the set of UMTCs and
$\{\text{iTO}_B\}$ is the set of invertible topological orders
(iTO)\cite{KW1458,F1478} for 2+1D boson systems.  In fact $\{\text{iTO}_B\}=\Z$
which is generated by the $E_8$ bosonic quantum Hall (QH) state, and a table of
UMTCs was obtained in \Ref{RSW0777,W150605768}.  Thus, we have a table (and a
classification) of 2+1D bosonic topological orders.

In a recent work\cite{LW150704673}, we show that 2+1D fermionic topological
orders are classified by $\{\mce{\sRp(Z_2^f)}\}\times \{\text{iTO}_F\}$,
where $\{\mce{\sRp(Z_2^f)}\}$ is the set of non-degenerate unitary
braided fusion categories (UBFC) over the symmetric fusion category (SFC)
$\sRp(Z_2^f)$ (see Definition \ref{defMTC}). We also require
$\mce{\sRp(Z_2^f)}$s to have modular extensions.
$\{\text{iTO}_F\}$ is the set of invertible topological orders for 2+1D fermion
systems.  In fact $\{\text{iTO}_F\}=\Z$ which is generated by the $p+\ii p$
superconductor. In \Ref{LW150704673} we computed the table for
$\mce{\sRp(Z_2^f)}$s, and  obtained a table (and a classification) of
2+1D fermionic topological orders.

In \Ref{LW150704673}, we also point out the importance of modular extensions.
If a $\mce{\sRp(Z_2^f)}$ does not have a modular extension, it means
that the fermion-number-parity symmetry is not on-site (\ie
anomalous\cite{W1313}).  On the other hand, if a $\mce{\sRp(Z_2^f)}$
does have modular extensions, then the $\mce{\sRp(Z_2^f)}$ is
realizable by a lattice model of fermions. In this case, a given $\mce{\sRp(Z_2^f)}$ may have several
modular extensions. We found that different modular extensions of
$\mce{\sRp(Z_2^f)}$ contain information of iTO$_F$s.

Our result on fermionic topological orders can be easily generalized to
describe bosonic/fermionic topological orders with symmetry.  This will be the
main topic of this paper. (Some of the results are announced in
\Ref{LW150704673}).  In this paper, we will consider symmetric GQL phases for
2+1D bosonic/fermionic systems.  The notion of GQL was defined in \Ref{ZW1490}.
The symmetry group of GQL is $G$ (for bosonic systems) or $G^f$ (for fermionic
systems).  If a symmetric GQL has long-range entanglement (as defined in
\Ref{CGW1038,ZW1490}), it corresponds to a symmetry enriched topological
(SET) order\cite{CGW1038}.  If a symmetric GQL has short-range entanglement, it
corresponds to a symmetry protected trivial (SPT) order [which is also
known as symmetry protected topological (SPT)
order]\cite{GW0931,PBT1225,CLW1141,CGL1314,CGL1204}.

In this paper, we are going to show that, 2+1D symmetric GQLs are classified by
$\mce{\cE}$ plus their modular extensions and chiral central charge.  In other
words, GQLs are labeled by triples $(\cC,\cM,c)$, where $\cC$ is a
$\mce{\cE}$, $\cM$ a modular extension of $\cC$, and $c$ the chiral central
charge of the edge state. (To be more precise, a modular extension of $\cC$,
$\cM$, is a UMTC with a fully faithful embedding $\cC \to \cM$.  In particular,
even if the UMTC $\cM$ is fixed, different embeddings correspond to
different modular extensions.) Here the SFC $\cE$ is given by $\cE=\Rp(G)$ for
bosonic cases, or $\cE=\sRp(G^f)$ for fermionic cases.  In yet another way to
phrase our result: we find that the structure $\cE \hookrightarrow \cC
\hookrightarrow \cM$ classifies
the 2+1D GQLs with symmetry $\cE$, where $\hookrightarrow$ represents the
embeddings and
$\cen{\cE}{\cM}=\cC$ (see Definition \ref{cendef}).

As a special case of the above result, we find that bosonic 2+1D SPT phase with
symmetry $G$ are classified by the modular extensions of $\Rp(G)$, while
fermionic 2+1D SPT phase with symmetry $G^f$ are classified by the modular
extensions of $\sRp(G^f)$ that have central charge $c=0$.

We like to mention that \Ref{BBC1440} has classified bosonic GQLs with symmetry
$G$, using $G$-crossed UMTCs. This paper uses a different approach so that we
can classify both bosonic and fermionic  GQLs with symmetry.  We also like to
mention that there is a mathematical companion \Ref{LW160205936} of this paper, where
 one can find detailed proof and explanations for related mathematical results.

The paper is organized as the following.  In Section \ref{GQLsymm}, we review
the notion of topological order and introduce category theory as a theory of
quasiparticle excitations in a GQL.  We will introduce a categorical way to
view the symmetry as well.  In Section \ref{inv}, we discuss invertible GQLs
and their classification based on modular extensions.  In Sections \ref{clGQL} and
\ref{clGQL2}, we generalize the above results and propose a classification of
all GQLs.  Section \ref{stack} investigates the stacking operation from
physical and mathematical points of view.  Section \ref{howto} describes how to
numerically calculate the modular extensions and Section \ref{examples}
discusses some simple examples.  For people with physics background, one way to
read this paper is to start with the Sections \ref{GQLsymm} and \ref{clGQL2},
and then go to Section \ref{examples} for the examples.

\section{Gapped quantum liquids, topological order and symmetry}
\label{GQLsymm}

\subsection{The finite on-site symmetry and symmetric fusion category}

In this paper, we consider physical systems with an on-site symmetry described
by a finite group $G$.  For fermionic systems, we further require that $G$
contains a central $Z_2$ fermion-number-parity subgroup.  More precisely,
fermionic symmetry group is a pair $(G,f)$, where $G$ is a finite group, $f\neq
1$ is an element of $G$ satisfying $f^2=1,fg=gf,\forall g\in G$. We denote the
pair $(G,f)$ as $G^f$.

There is another way to view the on-site symmetries, which is nicer because
bosonic and fermionic symmetries can be formulated in the same manner.
Consider a bosonic/fermionic product state $|\psi\ket$ that does not break the
symmetry $G$: $U_g|\psi\ket=|\psi\ket,\ g\in G$. Then the new way to view the
symmetry is to use the properties of the excitations above the product state to
encode the information of the symmetry $G$.

The product state contain only local excitations that can be created by acting
local operators $O$ on the ground state $O|\psi\ket$.  For any group action
$U_g$, $U_g O|\psi\ket=U_g O U_g^\dag U_g|\psi\ket=U_g O U_g^\dag |\psi\ket$ is
an excited state with the same energy as $O|\psi\ket$. Since we assume the
symmetry to be on-site, $U_g OU_g^\dag$ is also a local operator.  Therefore,
$U_g OU_g^\dag|\psi\ket$ and $O|\psi\ket$ correspond to the degenerate local
excitations. We see that  local excitations ``locally'' carry group
representations.  In other words, different types of local excitations are
labeled by irreducible representations of the symmetry group. 

By looking at how the local excitations (more precisely, their group
representations) fuse and braid with each other, we arrive at the mathematical
structure called symmetric fusion categories (SFC).  By definition a SFC is a
braided fusion category where all the objects (the excitations) have trivial
mutal statistics (\ie centralize each other, see next section). 
A SFC is automatically a unitary braided fusion category. 

In fact, there are only two kinds of SFCs: one is representation category of
$G$: $\Rp(G)$, with the usual braiding (all representations are bosonic); the
other is $\sRp(G^f)$ where an irreducible representation is bosonic if $f$ is
represented trivially ($+1$), and fermionic if $f$ is represented
non-trivially($-1$).

It turns out SFC (or the fusion and braiding properties of the local
excitations) fully characterize the symmetry group.  Therefore, it is
equivalent to say finite on-site symmetry is given by a SFC $\cE$. By Tannaka
duality $\cE$ gives rises to a unique finite group $G$ and by checking the
braiding in $\cE$ we know whether it is bosonic or fermionic.  This is the new
way, the categorical way, to view the symmetry.  Such a categorical view of
bosonic/fermionic symmetry allows us to obtain a classification of symmetric
topological/SPT orders.

\subsection{Categorical description of topological excitations with symmetry}

In symmetric GQLs with topological order (\ie with long range entanglement),
there can be particle-like excitations with local energy density, but they
cannot be created by local operators. They are known as (non-trivial)
topological excitations. Topological excitations do not necessarily carry group
representations. Nevertheless, we can still study how they fuse and braid with
each other; so we have a unitary braided fusion category (UBFC) to describe the
particle-like excitations. To proceed, we need the following key definition on
``centralizers.''
\begin{dfn}
  The objects $X,Y$ in a UBFC $\cC$ are said to \emph{centralize} (mutually
  local to) each other if 
  \begin{align}
   c_{Y,X}\circ c_{X,Y}=\mathrm{id}_{X\otimes Y},
  \end{align}
  where $c_{X,Y}: X\otimes Y\cong Y\otimes X$ is the braiding in $\cC$.
\end{dfn}

  Physically, we say that $X$ and $Y$ have trivial mutual statistics.

\begin{dfn}
\label{cendef}
  Given a subcategory $\cD\subset \cC$, its \emph{centralizer}
  $\cen{\cD}{\cC}$ in $\cC$
  is the full subcategory of objects in $\cC$ that centralize all the objects in
  $\cD$.  
\end{dfn}

We may roughly view a category as a ``set'' of particle-like excitations.
So the centralizer $\cen{\cD}{\cC}$ is the ``subset'' of particles in $\cC$
that have trivial mutual statistics with all the particles in $\cD$.

\begin{dfn}
\label{defMTC}
  A UBFC $\cE$ is a \emph{symmetric} fusion category if $\cen{\cE}{\cE}=\cE$.
A UBFC $\cC$ with a fully faithful embedding $\cE\inj \cen{\cC}{\cC}$ is called
a UBFC over $\cE$. Moreover, $\cC$ is called a non-degenerate UBFC over $\cE$, or
$\mce{\cE}$, if $\cen{\cC}{\cC}=\cE$.
\end{dfn}

\begin{dfn} \label{def:hom-bfce}
Two UBFCs over $\cE$, $\cC$ and $\cC'$ are equivalent if there is a unitary braided
equivalence $F:\cC\to\cC'$ such that it preserves the embeddings, i.e.,
the following diagram commute.
  \newdir^{ (}{{}*!/-5pt/@^{(}}
\begin{align}
\label{Ceq}
\xymatrix{
  \cE\ar@{^{ (}->}[r]\ar@{=}[d]&\cC\ar[d]^{F}
  \\
  \cE\ar@{^{ (}->}[r]&\cC'
}
\end{align}
We denote the category of unitary braided auto-equivalences of $\cC$ by $\mathcal{A}\mathrm{ut}(\cC)$ and its underlining group by $\aute(\cC)$. 
\end{dfn}

%To rephrase the above definitions in physical terms, we note that a category is
%just a set of particles (together with their braiding and fusion properties).  Symmetric fusion category is a set of particles that
%have trivial mutual statistics with each other. A UBFC over $\cE$ is a set of
%particles that contains a subset $\cE$ of particles with trivial mutual
%statistics with all the particles.  If no particle
%outside of $\cE$ can have trivial mutual statistics with all the particles,
%then the UBFC is called non-degenerate. 
We recover the usual definition of
UMTC when $\cE$ is trivial, \ie the category of Hilbert spaces, denoted by
$\Hilb=\Rp(\{1\})$.  In this case the subscript is omitted.

Physically, a UBFC $\cC$ is the collection of all bulk topological excitations
plus their fusion and braiding data. Requiring $\cC$ to be a 
$\mce{\cE}$ means: (1) the set of local excitations, $\cE$ (which is the set of
all the irreducible representations of the symmetry group), is included in
$\cC$ as a subcategory; (2) $\cC$ is anomaly-free, \ie all the topological
excitations (the ones not in $\cE$) can be detected by mutual
braiding\cite{KW1458}. In other words, every topological excitation must have
non-trivial mutual statistics with some excitations.  Those excitations that
cannot be detected by mutual braiding (i.e., $\cen{\cC}{\cC}$) are exactly the
local excitations in $\cE$. Moreover, we want the symmetry to be on-site
(gaugeable), which requires the existence of modular extensions (see Definition \ref{mextdef}). Such an understanding leads to the following
conjecture:
\begin{conj}
  Bulk topological excitations of topological orders with symmetry $\cE$ are
classified by $\mce{\cE}$'s that have modular extensions.
%The symmetry is on-site (\ie anomaly free) iff the $\mce{\cE}$ has modular extensions (see Definition \ref{mextdef}).  The symmetry is non-on-site (\ie anomalous) if the $\mce{\cE}$ has no modular extensions.
\end{conj}
\noindent

We like to remark that $\mce{\cE}$'s fail to classify
topological orders.  This is because two different topologically ordered phases
may have bulk topological excitations with the same non-abelian statistics (\ie
described by the same $\mce{\cE}$).  However, $\mce{\cE}$'s, with modular
extensions, do classify topological orders up to invertible ones.  See next
section for details.  The relation between anomaly and modular extension will
also be discussed later.

\section{Invertible GQLs and modular extension}
\label{inv}

\subsection{Invertible GQLs}

There exist non-trivial topological ordered states that have only trivial
topological excitations in the bulk (but non-trivial edge states). They are
``invertible'' under the stacking operation\cite{KW1458,F1478}
(see Section \ref{stack} for details). More generally,
we define  
\begin{dfn}
A GQL is invertible if its bulk topological excitations are all trivial
(\ie can all be created by local operators).
\end{dfn}
Consider some invertible GQLs with the same symmetry $\cE$.  The bulk
excitations of those invertible GQLs are the same which are described by the
same SFC $\cE$. Now the question is: How to distinguish those invertible GQLs?

First, we believe that invertible bosonic topological orders with no symmetry
are generated by the $E_8$ QH state (with central charge $c=8$) via
time-reversal and stacking, and form a $\Z$ group. Stacking with an $E_8$ QH state only
changes the central charge by $8$, and does not change the bulk excitations or
the symmetry. So the only data we need to know to determine the invertible
bosonic topological order with no symmetry is the central charge $c$.  The
story is parallel for invertible fermionic topological orders with no symmetry,
which are believed to be generated by the $p+\ii p$ superconductor state with
central charge $c=1/2$.

Second, invertible bosonic GQLs with symmetry are generated by bosonic SPT
states and invertible bosonic topological orders (\ie $E_8$ states) via
stacking.  We know that the bosonic  SPT states with symmetry $G$ are
classified by the 3-cocycles in $H^3[G,U(1)]$.  Therefore, bosonic invertible
GQLs with symmetry $G$ are classified by $H^3[G,U(1)]\times \Z$ (where $\Z$
corresponds to layers of $E_8$ states).

However, this result and this point of view is not natural to generalize to
fermionic cases or non-invertible GQLs. Thus, we introduce an equivalent point
of view, which can cover boson, fermion,  and non-invertible GQLs in the same
fashion.

\subsection{Modular extension}

First, we introduce the notion of modular extension of a $\mce{\cE}$:
\begin{dfn}
\label{mextdef}
  Given a $\mce{\cE}$ $\cC$, its \emph{modular extension} is a UMTC
  $\cM$, together with a fully faithful embedding $\iota_\cM:
  \cC\hookrightarrow\cM$, such that $\cen{\cE}{\cM}=\cC$, equivalently
  $\dim(\cM)=\dim(\cC)\dim(\cE)$.

  Two modular extensions $\cM$ and $\cM'$ are equivalent if
  there is an equivalence between the UMTCs ${F:\cM\to\cM'}$ that preserves the
  embeddings, i.e., the following diagram commute.
  \newdir^{ (}{{}*!/-5pt/@^{(}}
\begin{align}
\label{MEeq}
\xymatrix{
  \cC\ar@{^{ (}->}[r]\ar@{=}[d]&\cM\ar[d]^{F}\\
  \cC\ar@{^{ (}->}[r]&\cM'
}
\end{align}
  We
  denote the set of equivalent classes of modular extensions of $\cC$ by $\mext(\cC)$.
\end{dfn}
\begin{rmk}
  Since the total quantum dimension of modular extensions of a given $\cC$ is
fixed, there are only finitely many different modular extensions, due to
\Ref{BNRW13}. In principle we can always perform a finite search to exhaust all
the modular extensions.
\end{rmk}

Remember that $\cC$ describes the particle-like excitations in our topological
state. Some of those excitations are local that have trivial mutual statistics
with all other excitations.  Those local excitation form $\cE \subset \cC$.
The modular extension $\cM$ of $\cC$ is obtained as adding particles that have
non-trivial mutual statistics with the local excitations in $\cE$, so that
every particle in $\cM$ will always have non-trivial mutual statistics with
some particles in $\cM$.  Since the particles in $\cE$ carry ``charges'' (\ie
the irreducible representations of $G$), the added particles correspond to
``flux'' (\ie the symmetry twists of $G$).  So the modular extension correspond
to gauging\cite{LG1209} the on-site symmetry $G$. Since we can use the gauged
symmetry to detect SPT orders\cite{HW1339}, we like to propose the following
conjecture
\begin{conj}\label{invclassB}
  Invertible bosonic GQLs with symmetry $\cE=\Rp(G)$ are classified by
$(\cM,c)$ where $\cM$ is a  modular extension of $\cE$ and $c=0$ mod 8.
\end{conj}

\subsection{Classify 2+1D bosonic SPT states}

Invertible bosonic GQLs described by $(\cM,c)$ include both bosonic SPT states
and bosonic topological orders. Among those, $(\cM,c=0)$ classify bosonic SPT
states.  In other words:
\begin{cor}
2+1D bosonic SPT states with symmetry $G$ are classified by the modular
extensions of $\Rp(G)$ (which always have $c=0$).
\end{cor} 

In \Ref{CLW1141,CGL1314,CGL1204}, it was shown that 2+1D bosonic SPT states are
classified by $H^3[G,U(1)]$.  Such a result agrees with our conjecture, due to
the following theorem, which follows immediately from results in \Ref{dgno2007}. 

\begin{thm} \label{thm:1to1-repG}
The modular extensions of $\Rp(G)$ 1-to-1 correspond to 3-cocycles in $H^3[G,U(1)]$. The central charge of these modular extensions are $c=0$ mod 8.
\end{thm}

\begin{rmk}
In Sec.\,\ref{sec:mext-repG}, we give more detailed explanation of the 1-to-1 correspondence in Theorem\,\ref{thm:1to1-repG}. Moreover, we will prove a stronger result in Theorem\,\ref{thm:spt}. It turns out that the set $\cM_{ext}(\Rp(G))$ of modular extensions of $\Rp(G)$ is naturally equipped with a physical stacking operation such that $\cM_{ext}(\Rp(G))$ forms an abelian group, which is isomorphic to the group $H^3[G,U(1)]$. 
\end{rmk}

\begin{rmk}
$c/8$ determines the number of layers of the $E_8$ QH states, which is the
topological order part of invertible bosonic symmetric GQLs.  In other words
\begin{align}
&\ \ \ \
\{ \text{invertible bosonic symmetric GQLs} \} 
\nonumber\\
&=
\{ \text{bosonic SPT states} \}\times \{ \text{layers of $E_8$ states} \}.
\end{align}
\end{rmk}

\subsection{Classify 2+1D fermionic SPT states}

The above approach also apply to fermionic case.  Note that, the invertible
fermionic GQLs with symmetry $G^f$
have bulk excitations described by SFC $\cE=\sRp(G^f)$.
So we would like to conjecture that
\begin{conj}\label{invclassF}
  Invertible fermionic GQLs with symmetry $G^f$  are classified by
$(\cM,c)$, where $\cM$ is a  modular extension of $\cE=\sRp(G^f)$,
and $c$ is the central charge determining the layers of $\nu=8$ IQH states.
\end{conj}
\begin{rmk}
Note that, the central charge $c$ mod 8 is determined by $\cM$, while
$ (c - \text{mod}(c,8))/8$ determines the number of layers of the
$\nu=8$ IQH states.
\end{rmk}
\begin{rmk}
Invertible fermionic symmetric GQLs include both fermion SPT states and
fermionic topological orders. $(\cM,c)$ with $c=0$ classify fermionic SPT
states.  
\end{rmk} 
In other words, 
\begin{cor}
2+1D fermionic SPT states with symmetry $G$ are classified by the  $c=0$ modular extensions of
$\sRp(G^f)$. 
\end{cor} 
\begin{rmk}
Unlike the bosonic case, in general
\begin{align}
&\ \ \ \
\{ \text{invertible fermionic symmetric GQLs} \} 
\\
& \neq
\{ \text{fermionic SPT states} \}\times \{ \text{layers of $p+\ii p$ states} \}.
\nonumber 
\end{align}
\end{rmk}

When there is no symmetry, the invertible fermionic GQLs become the invertible
fermionic topological order, which have bulk excitations described by
$\cE=\sRp(Z_2^f)$.  $\sRp(Z_2^f)$ has 16  modular extensions, with central
charges $c=n/2, n=0,1,2,\dots,15$.  There is only one modular extension with
$c=0$, which correspond to trivial product state. Thus there is no non-trivial
fermionic SPT state when there is no symmetry, as expected.

The modular extensions with $c=n/2$ correspond to invertible fermionic
topological order formed by $n$ layers of $p+\ii p$ states.
Since the modular extensions can only determine $c$ mod 8,
in order for the above picture to be consistent, we need to show the following
\begin{thm}\label{pE8}
The stacking of 16 layers $c=1/2$ $p+\ii p$ states is equivalent to a $\nu=8$
IQH state, which is in turn equivalent to a $E_8$ bosonic QH state stacked with
a trivial fermionic product state.
\end{thm}
\begin{proof}
First, two layers of $p+\ii p$ states is equal to one layer of $\nu=1$ IQH
state.  Thus, 16 layers $c=1/2$ $p+\ii p$ states is equivalent to a $\nu=8$ IQH
state.  To show $\nu=8$ IQH state is equivalent to $E_8$ bosonic QH state
stacked with a trivial fermionic product state, we note that the  $\nu=8$ IQH
state is described by $K$-matrix $K_{\nu=8}=I_{8\times 8}$ which is a 8-by-8
identity matrix. While the $E_8$ bosonic QH state stacked with a trivial
fermionic product state is described by $K$-matrix $K_{E_8\boxtimes
\cF_0}=K_{E_8}\oplus \bpm 1 &0 \\ 0& -1 \epm $, where $K_{E_8}$ is the matrix
that describe the $E_8$ root lattice. We also know that two odd\footnote{An odd
matrix is a symmetric integer matrix with at least one of its diagonal elements
being odd.} $K$-matrices $K_1$ and $K_2$ describe the same fermionic topological
order if after direct summing with proper number of $\bpm 1 &0 \\ 0& -1 \epm
$'s:
\begin{align}
K_1'&=K_1\oplus  \bpm 1 &0 \\ 0& -1 \epm \oplus \cdots
\nonumber\\
K_2'&=K_2\oplus  \bpm 1 &0 \\ 0& -1 \epm \oplus \cdots,
\end{align}
$K_1'$ and $K_2'$ become equivalent, \ie
\begin{align}
 K_1' = U K_2' U^T,\ \ \ \ U \in SL(N,\Z).
\end{align}
Notice that $K_{\nu=8}\oplus \bpm 1 &0 \\ 0& -1 \epm$ and $K_{E_8\boxtimes
\cF_0}$ have the same determinant $-1$ and the same signature. Using the result
that odd matrices with $\pm 1$ determinants are equivalent if they have
the same signature, we find that $K_{\nu=8}\oplus \bpm 1 &0 \\ 0& -1 \epm$
and $K_{E_8\boxtimes \cF_0}$ are equivalent. Therefore $\nu=8$ IQH state is
equivalent to $E_8$ bosonic QH state stacked with a trivial fermionic product
state.
\end{proof}

\section{A full classification of 2+1D GQLs with symmetry}
\label{clGQL}

We have seen that all invertible GQLs with symmetry $G$ (or $G^f$) have the
same kind of bulk excitations, described by $\Rp(G)$ (or $\sRp(G^f)$).  To
classify distinct invertible GQLs that shared the same kind of bulk
excitations, we need to compute the modular extensions of $\Rp(G)$ (or
$\sRp(G^f)$).  This result can be generalized to non-invertible topological
orders.

In general, the bulk excitations of a 2+1D bosonic/fermionic SET are described
by a $\mce{\cE}$ $\cC$.  However, there can be many distinct SET orders that
have the same kind of bulk excitations described by the same $\cC$.  To
classify distinct invertible SET orders that shared the same kind of bulk
excitations $\cC$, we need to compute the modular extensions of $\cC$.  This
leads to the following
\begin{conj}
\label{classSET}
  2+1D GQLs with symmetry $\cE$ (\ie the 2+1D SET orders) are classified by $(\cC,\cM,c)$,
  where $\cC$ is a $\mce{\cE}$ describing the bulk topological excitations,
  $\cM$ is a  modular extension of $\cC$ describing the edge state up to
  $E_8$ states, and $c$ is the central charge determining the layers of $E_8$
  states.
\end{conj}
%Note that $c$ mod 8 is determined by $\cM$.  ${[c - (c \text{ mod } 8)]/8}$
%determines the number of layers of $\nu=8$ IQH states.

Let $\cM$ be a modular extension of a $\mce{\cE}$ $\cC$. We note that
all the simple objects (particles) in $\cC$ are contained in $\cM$ as
simple objects.  Assume that the particle labels of $\cM$ are
$\{i,j,\dots, x, y,\dots\}$, where $i,j,\cdots $ correspond to the particles in
$\cC$ and $x,y,\cdots $ the additional particles (not in $\cC$).  Physically,
the additional particles $x,y,\cdots $ correspond to the symmetry twists of the
on-site symmetry\cite{W1447}.  The  modular extension $\cM$ describes the
fusion and the braiding of original particles $i,j,\cdots $ with the symmetry
twists.  In other words, the modular extension $\cM$ is the resulting
topological order after we gauge the on-site symmetry\cite{LG1209}.

Now, it is clear that the existence of modular extension is closely related to
the on-site symmetry (\ie anomaly-free symmetry) which is gaugable (\ie allows
symmetry twists).  For non-on-site symmetry (\ie anomalous
symmetry\cite{W1313}), the  modular extension does not exist since the symmetry
is not gaugable (\ie does not allow symmetry twists). We also have
\begin{conj}
\label{classSETA}
  2+1D GQLs with anomalous symmetry\cite{W1313} $\cE$ are
classified by $\mce{\cE}$'s that have no modular extensions.
\end{conj}

It is also important to clarify the equivalence relation between the triples
$(\cC,\cM,c)$. Two triples $(\cC,\cM,c)$ and $(\cC',\cM',c')$ are equivalent
if: (1) $c=c'$; (2) there exists braided equivalences $F_\cC:\cC\to\cC'$ and
$F_\cM:\cM\to\cM'$ such that all the embeddings are preserved, i.e., the
following diagram commutes.
  \newdir^{ (}{{}*!/-5pt/@^{(}}
\begin{align}
\label{TOeq}
\xymatrix{
  \cE\ar@{^{ (}->}[r]\ar@{=}[d]&\cC\ar@{^{ (}->}[r]\ar[d]^{F_\cC}
  &\cM\ar[d]^{F_\cM}\\
  \cE\ar@{^{ (}->}[r]&\cC'\ar@{^{ (}->}[r]&\cM'
}
\end{align}
The equivalence classes will be in one-to-one
correspondence with GQLs (\ie SET orders and SPT orders).

Note that the group of the automorphisms of a $\mce{\cE}$ $\cC$, denoted by
$\aute(\cC)$ (recall Definition\,\ref{def:hom-bfce}), naturally acts on the
modular extensions $\mext(\cC)$ by changing the embeddings, i.e. $F\in\aute(\cC)$ acts as follows: 
$$
(\cC\hookrightarrow\cM)\mapsto (\cC\xrightarrow{F}\cC\hookrightarrow \cM)
$$ 
For a fixed $\cC$, the above equivalence relation amounts to say that GQLs with bulk excitations described by a fixed $\cC$ are in one-to-one correspondence with the quotient
$\mext(\cC)/\aute(\cC)$ plus a central charge $c$. When $\cC=\cE$, the GQLs
with bulk excitations described by $\cE$ and central charge $c=0$ are SPT phases. In this case, the group $\aute(\cE)$, where $\cE$ is viewed as the trivial $\mce{\cE}$, is trivial. Thus, SPT phases are classified by the modular extensions of
$\cE$ with $c=0$.

\section{Another description of 2+1D GQLs with symmetry}
\label{clGQL2}

Although the above result has a nice mathematical structure, it is hard to
implement numerically to produce a table of GQLs.  To fix this problem, we
propose a different description of 2+1D GQLs.  The second description is
motivated by a conjecture that the fusion and the spins of the particles,
$(\cN^{IJ}_K,\cS_I)$, completely characterize a UMTC. We conjecture that
\begin{conj}
\label{NsNsNs}
The data $( \tilde N^{ab}_c,\tilde s_a; N^{ij}_k,s_i; \cN^{IJ}_K,\cS_I;c)$,  up
to some equivalence relations, gives a one-to-one classification of
2+1D GQLs with symmetry $G$ (for boson) or $G^f$ (for fermion), with a
restriction that the symmetry group can be fully characterized by the fusion
ring of its irreducible representations.  The data $( \tilde N^{ab}_c,\tilde
s_a; N^{ij}_k,s_i; \cN^{IJ}_K,\cS_I;c)$ satisfies the conditions described in
Appendix \ref{cnds} (see \Ref{W150605768} for UMTCs).  
\end{conj}

Here $( \tilde N^{ab}_c,\tilde s_a; N^{ij}_k,s_i; \cN^{IJ}_K,\cS_I;c)$
is closely related to $(\cE;\cC;\cM;c)$ discussed above.  The data $(\tilde
N^{ab}_c,\tilde s_a)$ describes the symmetry (\ie the SFC $\cE$):
$a=1,\cdots,\tilde N$ label the irreducible representations and $\tilde
N^{ab}_c$ are the fusion coefficients of irreducible representations.  $\tilde
s_a =0$ or $1/2$ depending on if the fermion-number-parity transformation $f$
is represented trivially or non-trivially in the representation $a$.  The data
$(N^{ij}_k,s_i)$ describes fusion and the spins of the bulk particles
$i=1,\cdots,N$ in the GQL. The data $(N^{ij}_k,s_i)$ contains $(\tilde
N^{ab}_c,\tilde s_a)$ as a subset, where $a$ is identified with the first
$\tilde N$ particles of the GQL.  The data $(\cN^{IJ}_K,\cS_I)$  describes
fusion and the spins of a UMTC, and it includes $(N^{ij}_k,s_i)$ as a subset,
where $i$ is identified with the first $N$ particles of the UMTC.  Also among
all the particles in UMTC, only the first $N$ (\ie $I=1,\cdots,N$) have trivial
mutual statistics with first $\tilde N$ particles (\ie $I=1,\cdots,\tilde N$).
Last, $c$ is the chiral central charge of the edge state.

If the data $( \tilde N^{ab}_c,\tilde s_a; N^{ij}_k,s_i)$ fully characterized
the $\mce{\cE}$, then the Conjecture \ref{NsNsNs} would be equivalent to the
Conjecture \ref{classSET}.
However, for non-modular tensor category, $( \tilde
N^{ab}_c,\tilde s_a; N^{ij}_k,s_i)$ fails to to fully characterize a
$\mce{\cE}$. In other words, there are different $\mce{\cE}$'s that have the
same data $( \tilde N^{ab}_c,\tilde s_a; N^{ij}_k,s_i)$.  We need to include
the extra data, such as the $F$-tensor and the $R$-tensor, to fully
characterize the $\mce{\cE}$.

In Appendix \ref{SETtbl}, we list the data $( \tilde N^{ab}_c,\tilde s_a;
N^{ij}_k,s_i)$ that satisfy the conditions in Appendix \ref{cnds} (without the
modular extension condition) in many tables. Those tables include all the
$\mce{\cE}$'s (up to certain total quantum dimensions), but the tables are not
perfect: (1) some entries in the tables may be fake and do not correspond to
any $\mce{\cE}$ (for the conditions are only necessary); (2) some entries in
the tables may correspond to more then one $\mce{\cE}$ (since $( \tilde
N^{ab}_c,\tilde s_a; N^{ij}_k,s_i)$ does not fully characterize a $\mce{\cE}$).

We then continue to compute $(\cN^{IJ}_K,\cS_I;c)$, the modular extensions of
$( \tilde N^{ab}_c,\tilde s_a; N^{ij}_k,s_i)$.  We find that the modular
extensions can fix the imperfectness mentioned above.  First, we find that the
fake entries do not have modular extensions, and are ruled out.  Second, as we
will show in Section \ref{stack}, all $\mce{\cE}$'s have the same numbers of
modular extensions (if they exist); therefore, the entry that corresponds to
more $\mce{\cE}$'s has more modular extensions. The modular extensions can tell
us which entries correspond to multiple $\mce{\cE}$'s.  This leads to the
conjecture that the full data $( \tilde N^{ab}_c,\tilde s_a; N^{ij}_k,s_i;
\cN^{IJ}_K,\cS_I;c)$ gives rise to an one-to-one classification of 2+1D GQLs, and allows us to calculate the
tables of 2+1D GQLs, which include 2+1D SET states and 2+1D SPT states.  Those
are given in Section \ref{examples}.

As for the equivalence relation, we only need to consider
$(\cN^{IJ}_K,\cS_I;c)$, since the data  $(\tilde N^{ab}_c,\tilde s_a;
N^{ij}_k,s_i)$ is included in $(\cN^{IJ}_K,\cS_I;c)$.  Two such data
$(\cN^{IJ}_K,\cS_I;c)$ and $(\bar \cN^{IJ}_K,\bar \cS_I;\bar c)$ are called
equivalent if $c=\bar c$, and $(\cN^{IJ}_K,\cS_I)$ and $(\bar \cN^{IJ}_K,\bar
\cS_I)$ are related by two permutations of indices in the range $N_{\cM} \geq I
> N$ and in the range $N\geq I > \tilde N$, where $N_{\cM}$ is the range of
$I$.  Such an equivalence relation corresponds to the one in eqn. (\ref{TOeq})
and will be called the TO-equivalence relation.  We use the TO-equivalence
relation to count the number of GQL phases (\ie the number of SET orders and
SPT orders).

We can also define another equivalence relation, called ME-equivalence
relation: we say $(\cN^{IJ}_K,\cS_I;c)$ and $(\bar \cN^{IJ}_K,\bar \cS_I;\bar
c)$ to be ME-equivalent if $c=\bar c$ and they only differ by a permutation
of indices in range $I > N$.  The ME-equivalence relation is closely related to
the one defined in eqn. (\ref{MEeq}). We use the ME-equivalence relation to
count the number of modular extensions of a \emph{fixed} $\cC$.  

Last, let us explain the restriction on the symmetry group.  In the Conjecture
\ref{NsNsNs}, we try to use the fusion $\tilde N^{ab}_c$ of the irreducible
representations to characterize the symmetry group.  However, it is known that
certain different groups may have identical fusion ring for their irreducible
representations.  So we need to restrict the symmetry group to be the group
that can be fully characterized by its fusion ring.  Those groups include
simple groups and abelian groups\cite{Yuan16}.  If we do not impose such a
restriction, then the Conjecture \ref{NsNsNs} give rise to GQLs with a given
symmetry fusion ring, instead of a given symmetry group.

\section{The stacking operation of GQLs}
\label{stack}

\subsection{Stacking operation}

Consider two GQLs $\cC_1$ and $\cC_2$. If we stack them together (without
introducing interactions between them), we obtain another GQL, which is denoted
by $\cC_1\boxtimes \cC_2$.  The stacking operation $\boxtimes$ makes the set of
GQLs into a monoid.  $ \boxtimes $ does not makes the set of GQLs into a group,
because in general, a GQL $\cC$ may not have an inverse under $\boxtimes$.  \ie
there is no GQL $\cD$ such that $\cC\boxtimes \cD$ becomes a trivial product
state.  This is because when a GQL have non-trivial topological excitations,
stacking it with another GQL can never cancel out those topological
excitations.

When we are considering GQLs with symmetry $\cE$, the simple stacking
$\boxtimes$ will ``double'' the symmetry, leads to a GQL with symmetry
$\cE\bt\cE$ ($\Rp(G\times G)$ or $\sRp(G^f\times G^f)$). In general we allow
local interactions between the two layers to break some symmetry such that the
resulting system only has the original symmetry $\cE$ (In terms of the symmetry group, keep only the
subgroup $G\hookrightarrow G\times G$ with the diagonal embedding $g\mapsto
(g,g)$). This leads to the stacking between GQLs with symmetry $\cE$,
denoted by $\bt_\cE$. Similarly, $\bt_\cE$ makes GQLs with symmetry $\cE$ a
monoid, but in general not all GQLs are invertible.

However, if the bulk excitations of $\cC$ are all local (\ie all described by
SFC $\cE$), then $\cC$ will have an inverse under the stacking operation
$\boxtimes_\cE$, and this is why we call such GQL invertible.  Those invertible
GQLs include invertible topological orders and SPT states.

\subsection{The group structure of bosonic SPT states}
\label{grpSPT}

We have proposed that 2+1D SPT states are classified by $c=0$ modular
extensions of the SFC $\cE$ that describes the symmetry.  Since SPT states are
invertible, they form a group under the stacking operation $ \boxtimes_\cE $.
This implies that the modular extensions of the SFC should also form a group
under the  stacking operation.  So checking if the  modular extensions of the
SFC have a group structure is a way to find support for our conjecture.

However, in this section, we will first discuss such stacking
operation and group structure from a physical point of view.
We will only consider bosonic SPT states.

It has been proposed that the bosonic SPT states are described by group
cohomology $\cH^{d+1}[G,U(1)]$\cite{CLW1141,CGL1314,CGL1204}.  However, it has
not been shown that those bosonic SPT states form a group under stacking
operation. Here we will fill this gap.  An ideal bosonic SPT state of symmetry
$G$ in $d+1$D is described the following path integral
\begin{align}
 Z =\sum_{\{g_i\}} \prod_{\{i,j,\cdots \}} \nu_{d+1}(g_i,g_j,\cdots )
\end{align}
where $\nu_{d+1}(g_i,g_j,\cdots )$ is a function $G^{d+1} \to U(1)$, which is a
cocycle $\nu_{d+1}\in \cH^{d+1}[G,U(1)]$. Here the space-time is a complex whose
vertices are labeled by $i,j,\cdots $, and $\prod_{\{i,j,\cdots \}}$ is the
product over all the simplices of the space-time complex.
Also $\sum_{\{g_i\}}$ is a sum over all $g_i$ on each vertex.

Now consider the stacking of two SPT states described by cocycle
$\nu_{d+1}'$ and 
$\nu_{d+1}''$: 
\begin{align}
 Z =\sum_{\{g'_i,g''_i\}} \prod_{\{i,j,\cdots \}} 
\nu'_{d+1}(g'_i,g'_j,\cdots )
\nu''_{d+1}(g''_i,g''_j,\cdots ) .
\end{align}
Such a stacked state has a symmetry
$G \times G$ and is a $G \times G$ SPT state.

Now let us add a term to break the $G \times G$-symmetry to $G$-symmetry
and consider
\begin{align}
\label{ZU}
 Z =\sum_{\{g'_i,g''_i\}}  \prod_{\{i,j,\cdots \}} &
\nu'_{d+1}(g'_i,g'_j,\cdots )
\nu''_{d+1}(g''_i,g''_j,\cdots ) \times 
\nonumber\\
 \prod_i & \ee^{-U|g_i'-g_i''|^2}
,
\end{align}
where $|g'-g''|$ is an invariant distance between group elements.  As we change
$U=0$ to $U=+ \infty $, the stacked system changes into the system for an ideal
SPT state described by the cocycle
$\nu_{d+1}(g_i,g_j,\cdots)=\nu_{d+1}'(g_i,g_j,\cdots )
\nu_{d+1}''(g_i,g_j,\cdots )$.  If such a deformation does not cause any phase
transition, then we can show that the stacking of a $\nu_{d+1}'$-SPT state with
a $\nu_{d+1}''$-SPT state give rise to a $\nu_{d+1}=\nu_{d+1}'\nu_{d+1}''$-SPT
state.  Thus, the key to show the stacking operation to give rise to the group
structure for the SPT states, is to show the theory \eqn{ZU} has no phase
transition as we change $U=0$ to $U= +\infty $.

To show there is no phase transition, we put the system on a closed space-time
with no boundary, say $S^{d+1}$.  In this case, $\prod_{\{i,j,\cdots \}} 
\nu'_{d+1}(g'_i,g'_j,\cdots ) \nu''_{d+1}(g''_i,g''_j,\cdots )=1$, since
$\nu_{d+1}'$ and $\nu_{d+1}''$ are cocycles.
Thus the path integral \eq{ZU} is reduced to
\begin{align}
 Z =\sum_{\{g'_i,g''_i\}}  \prod_i  \ee^{-U|g_i'-g_i''|^2} 
= \Big(|G| \sum_g \ee^{-U|1-g|^2}\Big)^{N_v} ,
\end{align}
where $N_v$ is the number of vertices and $|G|$ the order of the symmetry
group.  We see that the free energy density
\begin{align}
 f = -\lim_{N_v\to\infty}\ln Z/N_v
\end{align}
is a smooth function of $U$ for $U\in [0, \infty )$. There is indeed no phase
transition.

The above result is highly non trivial from a categorical point of view.
Consider two 2+1D bosonic SPT states described by two modular extensions $\cM'$
and $\cM''$ of $\Rp(G)$.  The natural tensor product $\cM' \boxtimes  \cM''$ is
not a modular extension of $\Rp(G)$,  but a modular extension of $\Rp(G)
\boxtimes \Rp(G)=\Rp(G \times G)$.  So, $\cM' \boxtimes  \cM''$ describes a $G
\times G$-SPT state.  According to the above discussion, we need to break the
$G \times G$-symmetry down to the $G$-symmetry to obtain the $G$-SPT state.
Such a symmetry breaking process correspond to the so call ``anyon
condensation'' in category theory. We will discuss such anyon condensation
later.  The stacking operation $\bt_\cE$, with such a symmetry breaking process
included, is the correct stacking operation that maintains the symmetry $G$.

\subsection{Mathematical construction of the stacking operation}

We have conjectured that a 2+1D topological order with symmetry $\cE$ is
classified by $(\cC,\cM_\cC,c)$, where $\cC$ is a $\mce{\cE}$ , $\cM_\cC$ is a
modular extension of $\cC$, and $c$ is the central charge.  If we have another
topological order of the same symmetry $\cE$ described by $(\cC',\cM_{\cC'},c')$,
stacking $(\cC,\cM_\cC,c)$ and $(\cC',\cM_{\cC'},c')$ should give a third
topological order described by similar data $(\cC'',\cM_{\cC''},c'')$:
\begin{align}
 (\cC,\cM_\cC,c) \boxtimes_\cE (\cC',\cM_{\cC'},c') =
(\cC'',\cM_{\cC''},c'')
\end{align}

In this section, we will show that such a stacking operation can be defined
mathematically. This is an evidence supporting our Conjecture \ref{classSET}.
We like to point out that a special case of the above result for
$\cC=\cC'=\cC''=\cE=\Rp(G)$ was discussed in section \ref{grpSPT}.

To define $\boxtimes_\cE$ mathematically, first, we like to introduce
\begin{dfn}\label{alg}
  A \emph{condensable algebra} in a UBFC $\cC$ is a
  triple $(A,m,\eta)$, $A\in\cC$,
  $m:A\ot A\to A$, $\eta:\one\to A$ satisfying
  \begin{itemize}
    \item Associative: $m(\id_A\ot m)=m(m\ot \id_A)$
    \item Unit: $m(\eta\ot\id_A)=m(\id_A\ot\eta)=\id_A$
    \item Isometric: $m m^\dag=\id_A$
    \item Connected: $\Hom(\one,A)=\C$
    \item Commutative: $m c_{A,A}=m$
  \end{itemize}
\end{dfn}
Physically, such an condensable algebra $A$ is a composite self-bosonic anyon
satisfies additional conditions such that one can condense $A$ to obtain
another topological phase.

\begin{dfn}
  A (left) \emph{module} over a condensable algebra $(A,m,\eta)$ in $\cC$ is a
  pair $(X,\rho)$, $X\in\cC$, $\rho:A\ot X\to X$ satisfying
  \begin{gather}
    \rho(\id_A\ot\rho)=\rho(m\ot \id_M),\nonumber\\
    \rho(\eta\ot\id_M)=\id_M.
  \end{gather}
  It is further a \emph{local} module if
  \begin{align*}
    \rho c_{M,A} c_{A,M}=\rho.
  \end{align*}
\end{dfn}
  We denote the category of left $A$ modules by $\cC_A$.
  A left module $(X,\rho)$ is turned into a right module via the braiding,
  $(X,\rho c_{X,A})$ or $(X,\rho c_{A,X}^{-1})$, and thus an $A$-$A$ bimodule.
  The relative tensor functor $\ot_A$ of bimodules then turns $\cC_A$ into a fusion category.
  (This is known as $\alpha$-induction in subfactor context.)
  In general there can be two monoidal structures on $\cC_A$, since there are
  two ways to turn a left module into a bimodule (usually we pick one for
  definiteness
  when considering $\cC_A$ as a fusion category).
  The two monoidal structures coincide for the fusion subcategory $\cC_A^0$ of
  local $A$ modules. Moreover, $\cC_A^0$ inherited the braiding from $\cC$ and
  is also a UBFC. The local modules are nothing but the anyons in the
  topological phases after condensing $A$.
\begin{lem}[DMNO\cite{dmno}]
  \[\dim(\cC_A)=\frac{\dim(\cC)}{\dim(A)}.\]
  If $\cC$ is a UMTC, then so is $\cC_A^0$, and
  \[\dim(\cC_A^0)=\frac{\dim(\cC)}{\dim(A)^2}.\]
\end{lem}
  A non-commutative algebra $A$ is also of interest. We have the left center
  $A_l$ of $A$, the maximal subalgebra such that $m c_{A_l,A}=m$, and the right
  center $A_r$, the maximal subalgebra such that $m c_{A,A_r}=m$. $A_l$ and
  $A_r$ are commutative subalgebras, thus condensable.
\begin{thm}[FFRS\cite{FFRS03}]
  There is a canonical equivalence between the categories of local modules
  over the left and right centers, $\cC_{A_l}^0=\cC_{A_r}^0$.
\end{thm}
\begin{dfn}
  The Drinfeld center $Z(\cA)$ of a monoidal category $\cA$ is a monoidal category with
  objects as pairs $(X\in\cA,b_{X,-})$, where $b_{X,-}: X\ot -\to -\ot X$ are
  half-braidings that satisfy similar conditions as braidings. Morphisms and
  the tensor product are naturally defined.
\end{dfn}
  $Z(\cA)$ is a braided monoidal category. There is a forgetful tensor functor
  $for_\cA:Z(\cA)\to \cA$, $(X,b_{X,-})\mapsto X$ that forgets the half-braidings.
\begin{thm}[M{\"u}ger\cite{Mue01}]
  $Z(\cA)$ is a UMTC if $\cA$ is a fusion category and
  $\dim(Z(\cA))=\dim(\cA)^2$.
\end{thm}
\begin{dfn}
  Let $\cC$ be a braided fusion category and $\cA$ a fusion category, a tensor
  functor $F:\cC\to \cA$ is called a central functor if it factorizes through
  $Z(\cA)$, i.e., there exists a braided tensor functor $F':\cC\to Z(\cA)$ such
  that $F=F'for_\cA$.
\end{dfn}

\begin{lem}
  [DMNO\cite{dmno}]
  Let $F:\cC\to\cA$ be a central functor, and $R:\cA\to\cC$ the right adjoint
functor of $F$.
Then the object $A=R(\one) \in\cC$ has a canonical structure of condensable algebra.
$\cC_A$ is monoidally equivalent to the image of
$F$, i.e. the smallest fusion subcategory of $\cA$ containing $F(\cC)$.
\end{lem}

\begin{exa}
  If $\cC$ is a UBFC, it is naturally embedded into
  $Z(\cC)$, so is $\ov\cC$. Therefore, $\cC\bt\ov\cC\hookrightarrow Z(\cC)$.
  Compose this embedding with the forgetful functor $for_\cC:Z(\cC)\to\cC$ we
  get a central functor
\begin{align*}
  \cC\bt\ov\cC &\to \cC\\
  X\bt Y&\mapsto X\ot Y.
\end{align*}
Let $R$ be its right adjoint functor, we obtain a condensable algebra
$L_\cC:=R(\one)\cong \oplus_i ( i\bt \bar i) \in \cC\bt\ov\cC$ ($\bar i$
denotes the dual object, or anti-particle of $i$) and $\cC=(
\cC\bt\ov\cC)_{L_\cC}$, $\dim(L_\cC)=\dim(\cC)$.
In particular, for a symmetric category $\cE$, $L_\cE$ is a condensable algebra
in $\cE\bt\cE$, and $\cE=(\cE\bt\cE)_{L_\cE}=(\cE\bt\cE)_{L_\cE}^0$ for $\cE$
is symmetric, all $L_\cE$-modules are local.
Condensing $L_\cE$ is nothing but breaking the symmetry from $\cE\bt\cE$ to
$\cE$.
\end{exa}

Now, we are ready to define the stacking operation for $\mce{\cE}$'s as well
as their  modular extensions.
\begin{dfn}\label{stacking}
  Let $\cC,\cD$ be $\mce{\cE}$'s, and $\cM_\cC,\cM_\cD$ their 
  modular extensions. The stacking is defined by:
  \begin{align*}
    \cC\bt_\cE\cD:=(\cC\bt\cD)^0_{L_\cE},\quad \cM_\cC\bt_\cE \cM_\cD:=(\cM_\cC\bt
    \cM_\cD)_{L_\cE}^0
  \end{align*}
\end{dfn}
Note that in Ref.~\onlinecite{DNO11}, the tensor product $\bt_\cE$ for
$\mce{\cE}$'s is defined as $(\cC\bt\cD)_{L_\cE}$. For $\mce{\cE}$'s the two
definitions coincide $(\cC\bt\cD)^0_{L_\cE}=(\cC\bt\cD)_{L_\cE}$, for $L_\cE$ lies in
the centralizer of $\cC\bt\cD$ which is $\cE\bt\cE$. But for the  modular extensions we have to take
the unusual definition above.

\begin{thm}
  $\cC\bt_\cE\cD$ is a $\mce{\cE}$, and
  $\cM_\cC\bt_\cE\cM_\cD$ is a  modular extension of $\cC\bt_\cE\cD$.
\end{thm}
\begin{proof}
  The embeddings
$\cE=(\cE\bt\cE)_{L_\cE}^0\hookrightarrow (\cC\bt\cD)^0_{L_\cE}=\cC\bt_\cE\cD
\hookrightarrow (\cM_\cC\bt\cM_\cD)^0_{L_\cE}=\cM_\cC\bt_\cE\cM_\cD$
are obvious.
So $\cC\bt_\cE\cD$ is a UBFC over $\cE$. Also
\begin{align*}
  \dim(\cC\bt_\cE\cD)=\frac{\dim(\cC\bt\cD)}{\dim(L_\cE)}
  =\frac{\dim(\cC)\dim(\cD)}{\dim(\cE)},
\end{align*}
and $\cM_\cC\bt_\cE\cM_\cD$ is a UMTC,
\begin{align*}
  \dim(\cM_\cC\bt_\cE\cM_\cD)
  =\frac{\dim(\cM_\cC\bt\cM_\cD)}{\dim(L_\cE)^2}=\dim(\cC)\dim(\cD).
\end{align*}
  Thus, $\cM_\cC\bt_\cE\cM_\cD$ is a
  modular extension of $\cC\bt_\cE\cD$.
\end{proof}

  Take $\cD=\cE$. Note that $\cC\bt_\cE\cE=\cC$. Therefore,  for any
  modular extension $\cM_\cE$ of $\cE$, $\cM_\cC\bt_\cE\cM_\cE$ is still a
  modular extension of $\cC$. In the following we want to show the inverse,
  that one can extract the ``difference'', a modular extension of $\cE$,
  between two modular extensions of $\cC$.

\begin{lem}\label{Lag}
  We have $(\cC\bt\ov\cC)_{L_\cC}^0=\cen{\cC}{\cC}$.
\end{lem}
\begin{proof}
  $(\cC\bt\ov\cC)_{L_\cC}$ is equivalent to $\cC$ (as a fusion
category). Moreover, for $X\in\cC$ the equivalence gives the free module $
L_\cC\ot(X\bt
\one )\cong L_\cC\ot(\one\bt X)$. $L_\cC\ot(X\bt\one )$ is a local $L_\cC$ module if and only
if $X\bt \one$ centralize $L_\cC$. This is the same as $X\in
\cen{\cC}{\cC}$. Therefore we have $(\cC\bt \ov\cC)_{L_\cC}^0=\cen{\cC}{\cC}$.
\end{proof}

\begin{thm}\label{main}
 let $\cM$ and $\cM'$ be two modular extensions of the $\mce{\cE}$ $\cC$. There
 exists a unique $\cK\in\mext(\cE)$ such that $\cK\bt_\cE\cM=\cM'$. Such $\cK$
 is given by
   \begin{align*}
     \cK=(\cM'\bt \ov\cM)_{L_\cC}^0.
   \end{align*}
\end{thm}
\begin{proof}
  $\cK$ is a modular extension of $\cE$. This follows
  Lemma \ref{Lag}, that $\cE=\cen{\cC}{\cC}=(\cC\bt\ov \cC)^0_{L_\cC}$ is a full
  subcategory of $\cK$. $\cK$ is a UMTC by construction, and
  $\dim(\cK)=\frac{\dim(\cM)\dim(\cM')}{\dim(L_\cC)^2}=\dim(\cE)^2$.

  To show that $\cK=(\cM'\bt\ov\cM)_{L_\cC}$ satisfies
  $\cM'=\cK\bt_\cE\cM$, note that
  $\cM'=\cM'\bt\Hilb=\cM'\bt(\ov\cM\bt\cM)_{L_{\ov\cM}}^0$. It suffies that
  \begin{align*}
    (\cM'\bt\ov\cM\bt\cM)_{\one\bt
      L_{\ov\cM}}^0=[(\cM'\bt\ov\cM)_{L_\cC}^0\bt
      \cM]_{L_\cE}^0\\
      =(\cM'\bt\ov\cM\bt\cM)_{(L_\cC\bt\one)\ot(\one\bt
      L_\cE)}^0.
  \end{align*}
  This follows that $\one\bt L_{\ov\cM} $ and $(L_\cC\bt\one)\ot(\one\bt
  L_\cE)$ are left and right centers of the algebra $(L_\cC\bt\one)\ot(\one\bt
  L_{\ov\cM})$.

  If $\cM'=\cK\bt_\cE\cM=(\cK\bt\cM)_{L_\cE}^0$,
  then
  \begin{align*}
    \cK= (\cK\bt\cM\bt\ov\cM)_{\one\bt
    L_\cM}^0=
    (\cK\bt\cM\bt\ov\cM)_{(L_\cE\bt\one)\ot(\one\bt 
    L_\cC)}^0\\
    =[(\cK\bt_\cE\cM)\bt\ov\cM]_{L_\cC}^0
    =(\cM'\bt\ov\cM)_{L_\cC}^0.
  \end{align*}
  It is similar here that $\one\bt L_{\cM} $ and
  $(L_\cE\bt\one)\ot(\one\bt
  L_\cC)$ are the left and right centers of the algebra
  $(L_\cE\bt\one)\ot(\one\bt
  L_{\cM})$. This proves the uniqueness of $\cK$.

\end{proof}

Let us list several consequences of Theorem \ref{main}.
\begin{thm}\label{hegroup}
  $\mext(\cE)$ forms an finite abelian group.
 \end{thm}
 \begin{proof}
   Firstly, there exists at least one modular extension of a symmetric fusion
   category $\cE$,
   the Drinfeld center $Z(\cE)$. So the set $\mext(\cE)$ is not empty.
   The multiplication is given by the stacking $\bt_\cE$.
   It is easy to verify that the stacking $\bt_\cE$ for modular extensions
   is associative and commutative. To show that they form a group we only need
   to find out the identity and inverse.
   In this case $\cK=(\cM'\bt \ov \cM)^0_{L_\cE}=\cM'\bt_\cE\ov \cM$,
   Theorem \ref{main} becomes $\cM'\bt_\cE\ov\cM\bt_\cE\cM=\cM'$, for any
   modular extensions $\cM,\cM'$ of $\cE$.
   Thus, $\ov{\cM'}\bt_\cE
   \cM'=\ov{\cM'}\bt_\cE
   \cM'\bt_\cE\ov\cM\bt_\cE\cM
   =\ov\cM\bt_\cE\cM$, i.e. $\ov\cM\bt_\cE\cM$, is the same category
   for any extension $\cM$, which turns out to be $Z(\cE)$. It is exactly the identity element. It is then
   obvious that the inverse of $\cM$ is $\ov\cM$.
   The finiteness follows from \Ref{BNRW13}.
 \end{proof}

  \begin{exa}
    For bosonic case we find that $\mext(\Rp(G))=H^3(G,U(1))$, which is
    discussed in more detail in the next subsection. For fermionic case a
    general group cohomological classification is still lacking. We know some
    simple ones such as $\mext(\sRp(\Z_2^f))=\Z_{16}$, which agrees with
    Kitaev's
    16-fold way\cite{K062}.
  \end{exa}

 \begin{thm}\label{hetorsor}
   For a $\mce{\cE}$ $\cC$, if the modular extensions exist, $\mext(\cC)$ form
   a $\mext(\cE)$-torsor. In particular, $|\mext(\cC)|=|\mext(\cE)|$.
 \end{thm}
 \begin{proof}
   The action is given by the stacking $\bt_\cE$.
   For any two extensions $\cM,\cM'$, there is a unique extension $\cK$ of
   $\cE$, such that $\cM\bt_\cE\cK=\cM'$. To see $Z(\cE)$ acts trivially, note
   that $\cM'\bt_\cE Z(\cE)=\cM\bt_\cE \cK\bt_\cE Z(\cE)=\cM\bt_\cE\cK=\cM'$
   holds for any $\cM'$. Due to uniqueness we also know that only $\cZ_\cE$
   acts trivially. Thus, the action is free and transitive.
 \end{proof}
 This means that for any modular extension of $\cC$,
   stacking with a nontrivial modular extensions of $\cE$, one always obtains
   a different modular extension of $\cC$; on the other hand, starting with a
   particular modular extension of $\cC$, all the other modular extensions can
   be generated by staking modular extensions of $\cE$ (in other words, there
   is only on orbit). However, in general, there is no preferred choice of the
   starting modular extension, unless $\cC$ is the form $\cC_0\bt \cE$ where
   $\cC_0$ is a UMTC.

\subsection{Modular extensions of $\Rp(G)$} \label{sec:mext-repG}

\begin{figure}[tb]
$$
%\raisebox{-3.5em}{
\setlength{\unitlength}{.5pt}
\begin{picture}(100, 185)
   \put(-80,){\scalebox{1}{\includegraphics{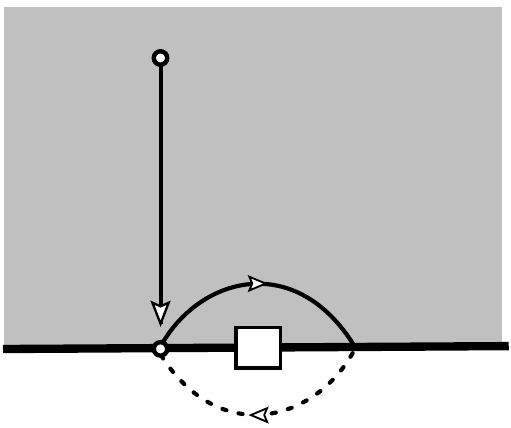}}}
   \put(0,0){
     \put(-18,-40){
     %\put(51, 148)     {  $x $}
     \put(20, 265)     { $e\in \Rp(G) \subset \cM$}
     \put(-30, 170)    { $-\otimes A$}
     \put(75, 81)     { $x$}
     \put(75,30)        { $\gamma_2$}
     \put(75,132)        { $\gamma_1$}
     \put(100, 210)    {$\cM$}
     \put(230, 81)    {$\cM_A$}
     \put(160,45)     {$\Hilb$}
     \put(0, 60)        {$F(e)$}
     }\setlength{\unitlength}{1pt}}
  \end{picture}
  %}
$$
\caption{Consider a physical situation in which the excitations in the $2+1$D
  bulk are given by a modular extension $\cM$ of $\Rp(G)$, and those on the
  gapped boundary by the UFC $\cM_A$. Consider a simple particle $e\in \Rp(G)$
  in the bulk moving toward the boundary. The bulk-to-boundary map is given by
  the central functor $-\otimes A: \cM \to \cM_A$, which restricted to $\Rp(G)$
  is nothing but the forgetful functor $F:\Rp(G) \to \Hilb$. Let $x$ be a
  simple excitation in $\cM_A$ sitting next to $F(e)$. We move $F(e)$ along the
  semicircle $\gamma_1$ (defined by the half-braiding), then move along the
  semicircle $\gamma_2$ (defined by the symmetric braiding in the trivial phase
  $\mathrm{Vec}$).}
\label{fig:G-grading}
\end{figure}

We set $\cE=\Rp(G)$ throughout this subsection. Let $(\cM,\iota_\cM)$ be a modular
extension of $\Rp(G)$. $\iota_\cM$ is the embedding
$\iota_\cM:\cE\hookrightarrow\cM$ that we need to consider explicitly in this
subsection. The algebra $A=\Fun(G)$ is a condensable algebra in $\Rp(G)$ and
also a condensable algebra in $\cM$. Moreover, $A$ is a Lagrangian algebra in
$\cM$ because $(\dim A)^2 = |G|^2=(\dim \Rp(G))^2 = \dim \cM$. Therefore, $\cM\simeq Z(\cM_A)$, where $\cM_A$ is the category of right $A$-modules in $\cM$. In other words, $\cM$ describes the bulk excitations in a 2+1D topological phase with a gapped boundary (see Fig.\,\ref{fig:G-grading}). Moreover, the fusion category $\cM_A$ is pointed and equipped with a canonical fully faithful $G$-grading\cite{dgno2007}, which means that 
$$
\cM_A=\oplus_{g\in G} (\cM_A)_g, \quad (\cM_A)_g\simeq \Hilb, \,\, \forall g\in G, 
$$ 
$$
\quad \mbox{and} \quad \otimes: (\cM_A)_g \boxtimes (\cM_A)_h \xrightarrow{\simeq} (\cM_A)_{gh}.
$$
Let us recall the construction of this $G$-grading. The physical meaning of
acquiring a $G$-grading on $\cM_A$ after condensing the algebra $A=\Fun(G)$ in
$\cM$ is depicted in Figure\,\ref{fig:G-grading}.  The process in
Figure\,\ref{fig:G-grading} defines the isomorphism $F(e) \otimes_A x
  \xrightarrow{z_{e,x}} x \otimes_A F(e) = F(e) \otimes_A x$,
  which further gives a monoidal automorphism $\phi(x)\in \mathrm{Aut}(F)=G$ of
the fiber functor $F: \Rp(G) \to \Hilb$.  

Since $\phi$ is an isomorphism, the associator of the monoidal category $\cM_A$ determines a unique
$\omega_{(\cM,\iota_M)}\in H^3(G, U(1))$ such that $\cM_A \simeq \Hilb_G^\omega$ as $G$-graded fusion categories. 

%Conversely, for any $\omega\in H^3(G, U(1))$, there is a canonical braided
%embedding $\iota_\omega: \Rp(G) \hookrightarrow Z(\Hilb_G^\omega)$ such that
%the composition $\Rp(G) \hookrightarrow Z(\Hilb_G^\omega) \to \Hilb_G^\omega$
%defines a symmetric fiber functor $\Rp(G) \to \Hilb\subset \Hilb_G^\omega$
%and the induced group isomorphism $\phi: G=O(\Hilb_G^\omega) \to G$ is the
%identity map, i.e. $\phi=\id_G$.

\begin{thm} \label{thm:spt}
The map $(\cM, \iota_\cM) \mapsto \omega_{(\cM, \iota_\cM)}$ defines a group isomorphism $\cM_{ext}(\Rp(G)) \simeq H^3(G, U(1))$. In particular, we have 
$$ 
(Z(\mathrm{Vec}_G^{\omega_1}),\iota_{\omega_1}) \boxtimes_\cE
(Z(\mathrm{Vec}_G^{\omega_2}),\iota_{\omega_2}) \simeq (Z(\mathrm{Vec}_G^{\omega_1+\omega_2}), \iota_{\omega_1+\omega_2}).
$$
\end{thm}

For the proof and more related details, see also \Ref{LW160205936}.

\subsection{Relation to numerical calculations}
In Section \ref{clGQL2} we proposed another way to characterise GQLs, using
the data $( \tilde N^{ab}_c,\tilde s_a; N^{ij}_k,s_i; \cN^{IJ}_K,\cS_I;c)$
which is more friendly in numerical calculations. We would like to investigate how
to calculate the stacking operation in terms of these data.

Assuming that $\cC$ and $\cC'$ can be characterized by data
$(N^{ij}_k,s_i)$ and $( N^{\prime ij}_k,s^\prime_i)$. Let $(
N^{\cD, ij}_k, s^\cD_i)$ be the data that characterizes the stacked $\mce{\cE}$ $\cD
=\cC\bt_\cE \cC'$.

To calculate $(N^{\cD, ij}_k, s^\cD_i)$, let us first construct
\begin{align}
 N^{ii',jj'}_{kk'} = N^{ij}_{k} N^{\prime i'j'}_{k'} 
,\ \ \ \ \
s_{ii'}=s_i+s'_{i'}.
\end{align}
Note that, the above data describes a $\mce{\cE\boxtimes \cE}$ $\cD' =\cC\bt
\cC'$ (\ie with centralizer $\cE\bt \cE$), which is not what we want.  We need
reduce centralizer from $\cE\bt \cE$ to $\cE$. This is the $G\times G$ to $G$
process and $\cC$-$\cC'$ coupling, or condensing the $L_\cE$
algebra, as discussed above

To do the $\cE\bt \cE$ to $\cE$ reduction (\ie to obtain the real stacking
operation $\bt_\cE$), we can introduce an equivalence relation. Noting that the
excitations in $\cD'=\cC\bt \cC'$ are labeled by $ii'=i\bt i'$, the  equivalence
relation is 
\begin{align}
  ii'\sim jj', \quad{\text{if }} ii'\ot L_\cE=jj'\ot L_\cE.
\end{align}
where $L_\cE=\oplus_a a\bar a, a\in\cE$. In the simple case of abelian groups,
where all the $a$'s are abelian particles, the equivalence relation reduces to
\begin{align}
 (a\ot i)i'\sim i(a\ot i'),\ \ \ \
\forall \ \ i\in \cC,\ i'\in \cC',\  a\in \cE.
\end{align}
Mathematically, this amounts to consider only the free local $L_\cE$ modules.
The equivalent classes $[ii']$ are then some composite anyons in
$\cD=\cC\bt_\cE\cC'$
\begin{align}
  [ii']=k\oplus l \oplus \cdots , \quad\text{ for some }k,l,\dots\in\cD.
\end{align}
In other words, they form a fusion sub ring of $\cD$.
Moreover, the spin of $ii'$ is the same as the direct summands
\begin{align}
  s_{ii'}=s_k^\cD=s_l^\cD=\cdots
\end{align}
Since it is
limited to a subset of data of $\mce{\cE}$'s, we can only give these necessary
conditions. However, as we already give a large list of GQLs in terms of these data,
they are usually enough to pick the resulting $\cC\bt_\cE\cC'$ from the list.

\section{How to calculate the modular extension of a $\mce{\cE}$}
\label{howto}

\subsection{A naive calculation}

How do we calculate the modular extension $\cM$ of $\mce{\cE}$ $\cC$ from
the data of $\cC$?  Actually, we do not know how to do that.  So here, we will
follow a closely related Conjecture \ref{NsNsNs}, and calculate instead
$(\cN^{IJ}_K,\cS_I,c)$ (that fully characterize $\cM$) from the data $(
\tilde N^{ab}_c,\tilde s_a; N^{ij}_k,s_i)$ (that partially characterize $\cC$).
In this section, we will describe such a calculation.

We note that all the simple objects (particles) in $\cC$ are contained in $\cM$
as simple objects, and $\cM$ may contain some extra simple objects.  Assume
that the particle labels of $\cM$ are $\{I,J,\cdots \}=\{i,j,\dots, x,
y,\dots\}$, where we use $i,j,\cdots $ to label the particles in $\cC$ and
$x,y,\cdots $ to label the additional particles (not in $\cC$).  Also let us
use $a,b,\cdots $ to label the simple objects in the centralizer of $\cC$:
$\cE=\cC_\cC^\text{cen}$.  Let $\cN^{IJ}_K$, $\cS_{I}$ be the fusion
coefficients and the spins for $\cM$, and $N^{ij}_k,\ s_i$ be the fusion
coefficients and the spins for $\cC$.  The idea is to find as many conditions on
$(\cN^{IJ}_K, \cS_{I})$ as possible, and use those conditions to solve for
$(\cN^{IJ}_K, \cS_{I})$.  Since the data $(\cN^{IJ}_K, \cS_{I})$ describe the
UMTC $\cM$, they should satisfy all the conditions discussed in
\Ref{W150605768}.  On the other hand, as a modular extension of $\cC$,
$(\cN^{IJ}_K, \cS_{I})$ also satisfy some additional conditions. Here, we will
discuss those additional conditions.

First, the modular extension $\cM$ has a fixed total quantum dimension.
\begin{align}
\label{dimMCE}
 \dim(\cM)=\dim(\cE)\dim(\cC).
\end{align}
In other words
\begin{align}
 \sum_{I\in \cM} d_I^2 = \sum_{a\in \cE} d_a^2 \sum_{i\in \cC} d_i^2.
\end{align}

Physically, the modular extension $\cM$ is obtained by ``gauging'' the
symmetry $\cE$ in $\cC$ (\ie adding the symmetry twists of $\cE$).  So the
additional particles $x,y,\cdots$ correspond to the symmetry twists.  Fusing an
original particle $i\in \cC$ to a symmetry twist $x\notin \cC$ still give us a
symmetry twist.  Thus
\begin{align}
 \cN^{ix}_j = \cN^{xi}_j = \cN^{ij}_x =0.
\end{align}

Therefore, $\cN_i$ for $i\in \cC$ is block diagonal: 
\begin{align}
\cN_i= N_i \oplus \hat N_i, 
\end{align}
where $( N_i)_{jk}=\cN^{ij}_{k}=N^{ij}_{k}$ and $(\hat N_i)_{x 
y}=\cN^{iy}_{x}$.  

If we pick a charge conjugation for the additional particles $x\mapsto
\bar x$, the conditions for fusion rules reduce to
\begin{align}
  \cN^{i x}_{y}=\cN^{x i}_{y}=\cN^{\bar{x} y}_{i}=\cN^{i \bar{y}}_{\bar{x}},
\nonumber \\
  \sum_{k\in\cC} N^{ij}_k \cN^{kx}_{y}= \sum_{z \notin\cC} \cN^{i z}_{ x} \cN^{j y}_{z}.
  \label{extN}
\end{align}
With a choice of charge conjugation, it is enough to construct (or search for)
the matrices $\hat N_i$ and $\cN^{xy}_z$ to determine all the extended fusion
rules $\cN^{IJ}_K$.  

Besides the general condition \eqref{extN}, there are also some simple
constraints on $\hat N_i$ that may speed up the numerical search.
Firstly, observe that \eqref{extN} is the same as
\begin{align}
\hat N_i \hat N_j = \sum_{k\in \cC} N^{ij}_k \hat N_k,
\end{align}
where $i,j,k \in \cF$.  This means that $\hat N_i$ satisfy the same fusion
algebra as $ N_i$, and $N^{ij}_k=\cN^{ij}_k$ is the structure constant;
therefore, the eigenvalues of $\hat N_i$ must be a subset of the eigenvalues
of $ N_i$. 

Secondly, since $\sum_{y\notin\cC} \cN^{i x}_{y}d_{y}= d_i d_{x}$, by
Perron-Frobenius theorem, we know that $d_i$ is the largest eigenvalue of
$\hat N_i$, with eigenvector $v, v_{x}=d_{x}$. ($d_i$ is also the largest
absolute values of the eigenvalues of $\hat N_i$.) Note that ${\hat N_{\bar
i} \hat N_{i}= \hat N_i \hat N_{\bar i},}$ ${ \hat N_{\bar i}=\hat
N_i^\dag}$.  Thus, $d_i^2$ is the largest eigenvalue of the positive
semi-definite Hermitian matrix $\hat N_{i}^\dag\hat N_{i}$. For any unit
vector $v$ we have $v^\dag \hat N_i^\dag \hat N_i v\leq d_i^2$, in
particular,
\begin{align}
  (\hat N_i^\dag \hat N_i)_{xx}=\sum_{y} (\cN^{ix}_{y})^2\leq
  d_i^2.
  \label{extentry}
\end{align}
The above result is very helpful to reduce the scope of numerical search.

Once we find the fusion rules, $\cN^{IJ}_K$, we can then use the rational
conditions and other conditions to determine the spins $\cS_I$ (for details, see
\Ref{W150605768}).  The set of data $(\cN^{IJ}_K,\cS_I)$ that satisfy all the
conditions give us the set of modular extensions.

The above proposed calculation for modular extensions is quite expensive.  If
the quantum dimensions of the particles in $\cC$ are all equal to 1: $d_i=1$,
then there is another much cheaper way to calculate the fusion coefficient
$\cN^{IJ}_K$ of the modular extension $\cM$.  Such an approach is explained in
Appendix \ref{FRgroup}.  We will also use such an approach in our calculation.

Last, we would like to mention that two sets of data $(\cN^{IJ}_K,\cS_I)$ and
$(\bar \cN^{IJ}_K,\bar \cS_I)$ describe the same modular extension of
$\cC$, if they only differ by a permutation of indices $x \in \cM$ but
$x \notin \cC$.  So some times, two sets of data $(\cN^{IJ}_K,\cS_I)$ and
$(\bar \cN^{IJ}_K,\bar \cS_I)$ can describe different modular extensions,
even through they describe the same UMTC.  (Two sets of data
$(\cN^{IJ}_K,\cS_I)$ and $(\bar \cN^{IJ}_K,\bar \cS_I)$ describe the same
UMTC, if they are only different by a permutation of indices $I \in \cM$.)

Why we use such a permutation in the calculation of modular extensions.  (which
is the ME-equivalence relation discussed before)?  This is because when we
considering modular extensions, the particle $x \in \cM$ but $x \notin \cC$
correspond to symmetry twists. They are extrinsic excitations that do not
appear in the finite energy spectrum of the Hamiltonian.  While the particle
$i\in \cC$ are intrinsic excitations that do appear in the finite energy
spectrum of the Hamiltonian.  So $x \notin \cC$ and $i\in \cC$ are physically
distinct and we do not allow permutations that mix them.  Also we should not
permute the particles $a\in \cE$, because they correspond to symmetries. We
should not mix, for example, the $Z_2$ symmetry of exchange layers and the
$Z_2$ symmetry of 180$^\circ$ spin rotation.

\subsection{The limitations of the naive calculation}

Since a $\mce{\cE}$ $\cC$ is not modular, the data $(\tilde N^{ab}_c,\tilde
s_a;N^{ij}_k,s_i)$ may not fully characterize $\cC$.  To fully characterize $\cC$, we need to use additional data, such
as the $F$-tensor and the $R$-tensor\cite{K062,W150605768}.  

In this paper, we will not use those additional data. As a result, the data
$(\tilde N^{ab}_c,\tilde s_a;N^{ij}_k,s_i)$ may correspond to several different
$\mce{\cE}$ $\cC$'s.  In other words, $(\tilde N^{ab}_c,\tilde
s_a;N^{ij}_k,s_i)$ is a one-to-many labeling of  $\mce{\cE}$'s.

So in our naive calculation, when we calculate the modular extensions of
$(\tilde N^{ab}_c,\tilde s_a;N^{ij}_k,s_i)$, we may actually calculate the
modular extension of several different $\cC$'s that are described by the same
data $(\tilde N^{ab}_c,\tilde s_a;N^{ij}_k,s_i)$.  But for $\mce{\cE}$'s that
can be fully characterized by the data $(\tilde N^{ab}_c,\tilde
s_a;N^{ij}_k,s_i)$, our calculation produce the modular extensions of a single
$\cC$.  For example, the naive calculation can obtain the correct modular
extensions of $\cC=\Rp(G)$ and $\cC=\sRp(G^f)$, when $G$ and $G^f$ are abelian
groups, or simple finite groups\cite{Yuan16}.

If the $(\tilde N^{ab}_c,\tilde s_a;N^{ij}_k,s_i)$ happen to describe two
different $\mce{\cE}$'s, we find that our naive calculation will produce the
modular extensions for both of $\mce{\cE}$'s (see Section \ref{Z2N5}).  So by
computing the modular extensions  of $(\tilde N^{ab}_c,\tilde s_a;N^{ij}_k,s_i)$,
we can tell if $(\tilde N^{ab}_c,\tilde s_a;N^{ij}_k,s_i)$ corresponds to none,
one, two, \etc $\mce{\cE}$'s.  This leads to the Conjecture \ref{NsNsNs} that
$(\tilde N^{ab}_c,\tilde s_a;N^{ij}_k,s_i,\cN^{IJ}_K,\cS_I;c)$ can fully and
one-to-one classify GQLs in 2+1D.

\section{Examples of 2+1D SET orders and SPT orders}
\label{examples}

\def\arraystretch{1.25} \setlength\tabcolsep{3pt}
\begin{table}[t] 
\caption{ The bottom two rows correspond to the two modular extensions of
$\Rp(Z_2)$ (denoted by $N_c^{|\Th|}=2^{\zeta^1_2}_0$).  Thus we have two
different trivial topological orders with $Z_2$ symmetry in 2+1D (\ie two $Z_2$
SPT states). 
} 
\label{mextZ2} 
\centering
\begin{tabular}{ |c|c|l|l|l| } 
\hline 
$N^{|\Th|}_{c}$ & $D^2$ & $d_1,d_2,\cdots$ & $s_1,s_2,\cdots$ & comment \\
 \hline 
$2^{\zeta_{2}^{1}}_{ 0}$ & $2$ & $1, 1$ & $0, 0$ & $\Rp(Z_2)$ \\
\hline
$4^{ B}_{ 0}$ & $4$ & $1, 1, 1, 1$ & $0, 0, 0, \frac{1}{2}$ & $Z_2$ gauge\\
$4^{ B}_{ 0}$ & $4$ & $1, 1, 1, 1$ & $0, 0, \frac{1}{4}, \frac{3}{4}$ & double semion\\
 \hline 
\end{tabular} 
\end{table}

\def\arraystretch{1.25} \setlength\tabcolsep{3pt}
\begin{table}[t] 
\caption{The two modular extensions of $N^{|\Th|}_{c}=3^{\zeta_{2}^{1}}_{ 2}$.
$3^{\zeta_{2}^{1}}_{ 2}$ has a centralizer $\Rp(Z_2)$.  Thus we have two
topological orders with $Z_2$ symmetry in 2+1D which has only one type of
spin-$1/3$ topological excitations.  We use $N^{|\Th|}_{c}$ to label
$\mce{\cE}$'s, where $\Theta ={D}^{-1}\sum_{i}\ee^{2\pi\ii s_i} d_i^2=
|\Th|\ee^{2\pi \ii c/8}$ and $D^2=\sum_id_i^2$.
} 
\label{mextZ2a} 
\centering
\begin{tabular}{ |c|c|l|l|l| } 
\hline 
$N^{|\Th|}_{c}$ & $D^2$ & $d_1,d_2,\cdots$ & $s_1,s_2,\cdots$ & comment \\
\hline 
$3^{\zeta_{2}^{1}}_{ 2}$ & $6$ & $1, 1, 2$ & $0, 0, \frac{1}{3}$ & 
\tiny $K=\begin{pmatrix}
 2 & -1 \\
 -1 & 2 \\
\end{pmatrix}
$
\\
\hline
$5^{ B}_{ 2}$ & $12$ & $1, 1, 2,\zeta_{4}^{1},\zeta_{4}^{1}=\sqrt{3}$ & $0, 0, \frac{1}{3}, \frac{1}{8}, \frac{5}{8}$ & $(A_1,4)$ \\
$5^{ B}_{ 2}$ & $12$ & $1, 1, 2,\zeta_{4}^{1},\zeta_{4}^{1}$ & $0, 0, \frac{1}{3}, \frac{3}{8}, \frac{7}{8}$ & \\
 \hline 
\end{tabular} 
\end{table}

In this section, we will discuss simple examples of $\mce{\cE}$ $\cC$'s, and
their modular extensions $\cM$. The triple $(\cC, \cM,c)$ describe a
topologically ordered or SPT phase.  A single $\mce{\cE}$ $\cC$ only describes
the set of bulk topological excitations, which correspond to topologically
ordered states up to invertible ones.

However, in this section we will not discuss examples of $\mce{\cE}$ $\cC$.
What we really do is to discuss examples of the solutions $(\tilde
N^{ab}_c,\tilde s_a;N^{ij}_k,s_i)$ (which are not really $\mce{\cE}$'s, but
closely related).  We will also discuss the modular extensions
$(\cN^{IJ}_K,\cS_I;c)$ of $(\tilde N^{ab}_c,\tilde s_a;N^{ij}_k,s_i)$.
$(\tilde N^{ab}_c,\tilde s_a;N^{ij}_k,s_i)$ will correspond to $\mce{\cE}$
$\cC$ if it has modular extensions $(\cN^{IJ}_K,\cS_I;c)$.  This allows us to
classify GQLs in terms of the data $(\tilde N^{ab}_c,\tilde
s_a;N^{ij}_k,s_i,\cN^{IJ}_K,\cS_I;c)$.

\subsection{$Z_2$ bosonic SPT states}

Tables \ref{SETZ2-34}, \ref{SETZ2-5}, and \ref{SETZ2-6} list the solutions
$(\tilde N^{ab}_c,\tilde s_a;N^{ij}_k,s_i)$ when  $(\tilde N^{ab}_c,\tilde
s_a)$ describes a SFC $\Rp(Z_2)$.  The table contains all $\mce{\Rp(Z_2)}$'s
but may contain extra fake entries.  Physically, they describe possible
sets of bulk excitations for $Z_2$-SET orders of bosonic systems.  The sets of
bulk excitations are listed by their quantum dimensions $d_i$ and spins $s_i$.

For example, let us consider the entry $N_c^{|\Th|}=2_0^{\zeta_2^1}$ in Table
\ref{SETZ2-34}.  Such an entry has a central charge $c=0$.  Also $N=2$, hence
the $Z_2$-SET state has two types of bulk excitations both with $d_i=1$ and
$s_i=0$.  Both types of excitations are local excitations; one is the trivial type
and the other carries an $Z_2$ charge.

The first question that we like to ask is that ``is such an entry a fake entry,
or it corresponds to some $Z_2$-symmetric GQL's?'' If it corresponds to
some $Z_2$-symmetric GQL's, how many distinct $Z_2$-symmetric GQL
phases that it corresponds to?  In other word, how many distinct
$Z_2$-symmetric GQL phases are there, that share the same set of bulk
topological excitations described by the entry $2_0^{\zeta_2^1}$?

Both questions can be answered by computing the modular extensions of
$2_0^{\zeta_2^1}$ (which is also denoted as $\Rp(Z_2)$).  We find that the
modular extensions exist, and thus $\Rp(Z_2)$ does correspond to some 
$Z_2$-symmetric GQL's.  In fact, one of the $Z_2$-symmetric GQL's is the
trivial product state with $Z_2$ symmetry.  Other $Z_2$-symmetric GQL's are
$Z_2$ SPT states.

After a numerical calculation, we find that there are only two different
modular extensions of $\Rp(Z_2)$ (see Table \ref{mextZ2}).  Thus there are two
distinct $Z_2$-symmetric GQL phases whose bulk excitations are described by the
$\Rp(Z_2)$.  The first one corresponds to the trivial product states whose
modular extension is the $Z_2$ gauge theory which has four types of particles
with $(d_i,s_i)=(1,0), (1,0),(1,0),(1,\frac12)$.  (Gauging the $Z_2$ symmetry
of the trivial product state gives rise to a $Z_2$ gauge theory.)  The second
one corresponds to the only non-trivial $Z_2$ bosonic SPT state in 2+1D, whose
modular extension is the double-semion theory which has four types of particles
with $(d_i,s_i)=(1,0), (1,0),(1,\frac14),(1,-\frac14)$. (Gauging the $Z_2$
symmetry of the $Z_2$-SPT state gives rise to a double-semion theory
\cite{LG1209}.)  So the $Z_2$-SPT phases are classified by $\Z_2$, reproducing
the group cohomology result\cite{CLW1141,CGL1314,CGL1204}.  In general, the
modular extensions of $\Rp(G)$ correspond to the bosonic SPT states in 2+1D
with symmetry $G$.

\subsection{$Z_2$-SET orders for bosonic systems}

\begin{table}[t] 
\caption{
The fusion rule of the $N_c^{|\Th|}=3_2^{\zeta_2^1}$ $Z_2$-SET order.
The particle $\textbf{1}$
carries the $Z_2$-charge $0$, and the particle $s$ carries the $Z_2$-charge
$1$.  From the table, we see that $\sigma\otimes\sigma=\textbf{1} \oplus s
\oplus \sigma$.
} 
\label{SET32} 
\centering
\begin{tabular}{ |c|ccc|}
 \hline 
 $s_i$ & $0$ & $ 0$ & $ \frac{1}{3}$\\
 $d_i$ & $1$ & $ 1$ & $ 2$\\
\hline
 $3^{\zeta_{2}^{1}}_{ 2}$ & $\textbf{1}$  & $s$  & $\sigma$ \\
\hline
$\textbf{1}$  & $ \textbf{1}$  & $ s$  & $ \sigma$  \\
$s$  & $ s$  & $ \textbf{1}$  & $ \sigma$  \\
$\sigma$  & $ \sigma$  & $ \sigma$  & $ \textbf{1} \oplus s \oplus \sigma$  \\
\hline
\end{tabular}
\end{table} 

\begin{table}[t] 
\caption{
The fusion rules of the two $N_c^{|\Th|}=4_1^{\zeta_2^1}$ $Z_2$ symmetry
enriched topological orders with identical $d_i$ and $s_i$. 
We see that one has a $Z_2\times Z_2$ fusion rule and
the other has a $Z_4$ fusion rule.
} 
\label{SETZ2-45} 
\centering
\begin{tabular}{ |c|cccc|}
 \hline 
 $s_i$ & $0$ & $ 0$ & $ \frac{1}{4}$ & $ \frac{1}{4}$\\
 $d_i$ & $1$ & $ 1$ & $ 1$ & $ 1$\\
\hline
 $4^{\zeta_{2}^{1}}_{ 1}$ & $\textbf{00}$  & $\textbf{01}$  & $\textbf{10}$  & $\textbf{11}$ \\
\hline
$\textbf{00}$  & $ \textbf{00}$  & $ \textbf{01}$  & $ \textbf{10}$  & $ \textbf{11}$  \\
$\textbf{01}$  & $ \textbf{01}$  & $ \textbf{00}$  & $ \textbf{11}$  & $ \textbf{10}$  \\
$\textbf{10}$  & $ \textbf{10}$  & $ \textbf{11}$  & $ \textbf{00}$  & $ \textbf{01}$  \\
$\textbf{11}$  & $ \textbf{11}$  & $ \textbf{10}$  & $ \textbf{01}$  & $ \textbf{00}$  \\
\hline
\end{tabular}
~~~~~~
\begin{tabular}{ |c|cccc|}
 \hline 
 $s_i$ & $0$ & $ 0$ & $ \frac{1}{4}$ & $ \frac{1}{4}$\\
 $d_i$ & $1$ & $ 1$ & $ 1$ & $ 1$\\
\hline
 $4^{\zeta_{2}^{1}}_{ 1}$ & $\textbf{0}$  & $\textbf{2}$  & $\textbf{1}$  & $\textbf{3}$ \\
\hline
$\textbf{0}$  & $ \textbf{0}$  & $ \textbf{2}$  & $ \textbf{1}$  & $ \textbf{3}$  \\
$\textbf{2}$  & $ \textbf{2}$  & $ \textbf{0}$  & $ \textbf{3}$  & $ \textbf{1}$  \\
$\textbf{1}$  & $ \textbf{1}$  & $ \textbf{3}$  & $ \textbf{2}$  & $ \textbf{0}$  \\
$\textbf{3}$  & $ \textbf{3}$  & $ \textbf{1}$  & $ \textbf{0}$  & $ \textbf{2}$  \\
\hline
\end{tabular}
\end{table} 

The entry $N_c^{|\Th|}=3_2^{\zeta_2^1}$ in Table \ref{SETZ2-34} corresponds to
more non-trivial $\mce{\Rp(Z_2)}$. It describes the bulk excitations of
$Z_2$-SET orders which has only one type of non-trivial topological
excitation(with quantum dimension $d=2$ and spin $s=1/3$, see Table
\ref{SET32}).  The other two types of excitations are local excitations with
$Z_2$-charge $0$ and $1$.  We find that $3_2^{\zeta_2^1}$ has modular
extensions and hence is not a fake entry.

To see how many SET orders that have such set of bulk excitations, we need to
compute how many modular extensions are there for $3_2^{\zeta_2^1}$.  We find
that $3_2^{\zeta_2^1}$ has two modular extensions (see Table \ref{mextZ2a}).
Thus there are two $Z_2$-SET orders with the above mentioned bulk excitations.
It is not an accident that the number of $Z_2$-SET orders with the same set of
bulk excitations is the same as the number of $Z_2$ SPT states.  This is
because the different $Z_2$-SET orders with a fixed set of bulk excitations are
generated by stacking with  $Z_2$ SPT states.

We would like to point out that for any $G$-SET state, if we break the
symmetry, the $G$-SET state will reduce to a topologically ordered state
described by a UMTC.  In fact,  the different $G$-SET states described by the
same $\mce{\cE}$ (\ie with the same set of bulk excitations) will reduce to the
same topologically ordered state (\ie the same UMTC).  In Appendix \ref{SB}, we
discussed such a symmetry breaking process and how to compute UMTC from
$\mce{\cE}$.  We found that the two $Z_2$-SET orders from $3_2^{\zeta_2^1}$
reduce to an abelian topological order described by a $K$-matrix $\bpm 2& -1\\
-1& 2 \epm$.  This is indicated by SB:$K=\bpm 2& -1\\ -1& 2 \epm$ in the
comment column of Table \ref{SETZ2-34}.  In other place, we use SB:$N^B_c$ or
SB:$N^F_c({a \atop b})$ to indicate the reduced topological order after the
symmetry breaking (for bosonic or fermionic cases).  (The topological orders
described by $N^B_c$ or $N^F_c({a \atop b})$ are given by the tables in
\Ref{W150605768} or \Ref{LW150704673}.)

As we have mentioned, there are two $Z_2$-SET orders with the same bulk
excitations.  But how to realize those $Z_2$-SET orders? We find that one of
the $Z_2$-SET orders is the double layer FQH state with $K$-matrix $\bpm 2 &
-1\\ -1 & 2\\ \epm$ (same as the reduced topological order after symmetry
breaking), where the $Z_2$ symmetry is the layer-exchange symmetry.  The
quasiparticles are labeled by the $l$-vectors $l=\bpm l_1\\l_2\epm$.  The two
non-trivial quasiparticles are given by 
\begin{align} 
l
&= \bpm 1 \\ 0\epm, \ \ \bpm 0 \\ 1\epm, \ \ 
\end{align} 
whose spins are all equal to $\frac13$.

Since the layer-exchange $Z_2$ symmetry exchanges $l_1$ and $l_2$, we see that
the two excitations $ \bpm 1 \\ 0\epm, \ \ \bpm 0 \\ 1\epm$ always have the
same energy. Despite the $Z_2$ symmetry has no 2-dim irreducible
representations, the above spin-1/3 topological excitations has an exact
two-fold degeneracy due to the $Z_2$ layer-exchange symmetry.   This effect is
an interplay between the long-range entanglement and the symmetry:
\emph{degeneracy in excitations may not always arise from high dimensional
irreducible representations of the symmetry.}

Such two degenerate excitations are viewed as one type of topological
excitations with quantum dimension $d=2$ (for the two-fold degeneracy) and spin
$s=\frac13$  (see Table \ref{SETZ2-34}).  The $Z_2$ symmetry twist in such a
double-layer state carry a non-abelian statistics with quantum dimension
$d=\sqrt{3}$.  In fact, there are two such $Z_2$ symmetry twists whose spin
differ by 1/2. % (\ie differ by a fermion).

The other $Z_2$-SET order can be viewed as the above
double layer FQH state $K=\bpm 2 & -1\\ -1 & 2\\ \epm$ stacked with a $Z_2$ SPT
state.

\def\arraystretch{1.25} \setlength\tabcolsep{3pt}
\begin{table}[t] 
\caption{ The four modular extensions of $N^{|\Th|}_{c}=5^{\zeta_{2}^{1}}_{ 0}$ with $Z_2\times Z_2$ fusion.
$5^{\zeta_{2}^{1}}_{ 0}$  has a centralizer $\Rp(Z_2)$. The first pair and the
second pair turns out to be equivalent.
} 
\label{mextZ2b} 
\centering
\begin{tabular}{ |c|c|l|l|l| } 
\hline 
$N^{|\Th|}_{c}$ & $D^2$ & $d_1,d_2,\cdots$ & $s_1,s_2,\cdots$ & comment \\
\hline 
$5^{\zeta_{2}^{1}}_{ 0}$ & $8$ & $1\times 4, 2$ & $0, 0, \frac{1}{2}, \frac{1}{2}, 0$ & \\
\hline
$9^{ B}_{ 0}$ & $16$ &\tiny $1\times 4, 2,\zeta_{2}^{1}\times 4$ &\tiny $0, 0,
\frac{1}{2}, \frac{1}{2}, 0, \frac{15}{16}, \frac{1}{16}, \frac{7}{16},
\frac{9}{16}$ & $3^{ B}_{-1/2}\boxtimes 3^{ B}_{ 1/2}$\\
$9^{ B}_{ 0}$ & $16$ &\tiny $1\times 4, 2,\zeta_{2}^{1}\times 4$ &\tiny $0, 0,
\frac{1}{2}, \frac{1}{2}, 0, \frac{3}{16}, \frac{13}{16}, \frac{11}{16},
\frac{5}{16}$ & $3^{ B}_{ 3/2}\boxtimes 3^{ B}_{-3/2}$\\
\hline
$9^{ B}_{ 0}$ & $16$ &\tiny $1\times 4, 2,\zeta_{2}^{1}\times 4$ &\tiny $0, 0,
\frac{1}{2}, \frac{1}{2}, 0, \frac{1}{16}, \frac{15}{16}, \frac{9}{16},
\frac{7}{16}$ & $3^{ B}_{ 1/2}\boxtimes 3^{ B}_{-1/2}$\\
$9^{ B}_{ 0}$ & $16$ &\tiny $1\times 4, 2,\zeta_{2}^{1}\times 4$ &\tiny $0, 0,
\frac{1}{2}, \frac{1}{2}, 0, \frac{13}{16}, \frac{3}{16}, \frac{5}{16},
\frac{11}{16}$ & $3^{ B}_{-3/2}\boxtimes 3^{ B}_{ 3/2}$\\
\hline 
\end{tabular} 
\end{table}

\def\arraystretch{1.25} \setlength\tabcolsep{3pt}
\begin{table}[t] 
\caption{ The four modular extensions of $N^{|\Th|}_{c}=5^{\zeta_{2}^{1}}_{ 1}$ with $Z_2\times Z_2$ fusion.
$5^{\zeta_{2}^{1}}_{ 1}$  has a centralizer $\Rp(Z_2)$.
} 
\label{mextZ2c} 
\centering
\begin{tabular}{ |c|c|l|l|l| } 
\hline 
$N^{|\Th|}_{c}$ & $D^2$ & $d_1,d_2,\cdots$ & $s_1,s_2,\cdots$ & comment \\
\hline 
$5^{\zeta_{2}^{1}}_{ 1}$ & $8$ & $1\times 4, 2$ & $0, 0, \frac{1}{2}, \frac{1}{2}, \frac{1}{8}$ & \\
\hline
$9^{ B}_{ 1}$ & $16$ &\tiny $1\times 4, 2,\zeta_{2}^{1}\times 4$ &\tiny $0, 0, \frac{1}{2}, \frac{1}{2}, \frac{1}{8}, \frac{1}{16}, \frac{1}{16}, \frac{9}{16}, \frac{9}{16}$ & $3^{ B}_{ 1/2}\boxtimes 3^{ B}_{ 1/2}$\\
$9^{ B}_{ 1}$ & $16$ &\tiny $1\times 4, 2,\zeta_{2}^{1}\times 4$ &\tiny $0, 0, \frac{1}{2}, \frac{1}{2}, \frac{1}{8}, \frac{13}{16}, \frac{13}{16}, \frac{5}{16}, \frac{5}{16}$ & $3^{ B}_{-3/2}\boxtimes 3^{ B}_{ 5/2}$\\
\hline
$9^{ B}_{ 1}$ & $16$ &\tiny $1\times 4, 2,\zeta_{2}^{1}\times 4$ &\tiny $0, 0, \frac{1}{2}, \frac{1}{2}, \frac{1}{8}, \frac{15}{16}, \frac{3}{16}, \frac{7}{16}, \frac{11}{16}$ & $3^{ B}_{-1/2}\boxtimes 3^{ B}_{ 3/2}$\\
$9^{ B}_{ 1}$ & $16$ &\tiny $1\times 4, 2,\zeta_{2}^{1}\times 4$ &\tiny $0, 0,
\frac{1}{2}, \frac{1}{2}, \frac{1}{8}, \frac{3}{16}, \frac{15}{16},
\frac{11}{16}, \frac{7}{16}$ & $3^{ B}_{ 3/2}\boxtimes 3^{ B}_{-1/2}$\\
\hline 
\end{tabular} 
\end{table}

\subsection{Two other $Z_2$-SET orders for bosonic
systems}

The fourth and fifth entries in Table \ref{SETZ2-34} describe the bulk
excitations of two other $Z_2$-SET  orders.  Those bulk excitations have
identical $s_i$ and $d_i$, but they have different fusion rules $N^{ij}_k$ (see
Table \ref{SETZ2-45}).  

Both entries have two modular extensions, and correspond to two SET orders.
Among the two SET orders for the $Z_2\times Z_2$ fusion rule,  one of them is
obtained by stacking a $Z_2$ \emph{neutral} $\nu=1/2$ Laughlin state with a
trivial $Z_2$ product state.  The other is obtained by stacking a $Z_2$ neutral
$\nu=1/2$ Laughlin state with a non-trivial $Z_2$ SPT state.  

The entry with $Z_4$ fusion rule also correspond to two SET orders.  They are
obtained by stacking a $Z_2$ \emph{charged} $\nu=1/2$ Laughlin state with a
trivial or a non-trivial $Z_2$ SPT state.  Here, \emph{charged} means that the
particles forming the $\nu=1/2$ Laughlin state carry $Z_2$-charge 1.  In this
case, the anyon in the $\nu=1/2$ Laughlin state carries a fractional
$Z_2$-charge $1/2$.  So the fusion of two such anyons give us a $Z_2$-charge 1
excitation instead of a trivial neutral excitation.  This leads to the $Z_4$
fusion rule.

\subsection{The rank $N=5$ $Z_2$-SET  orders for bosonic systems}
\label{Z2N5}

\def\arraystretch{1.25} \setlength\tabcolsep{3pt}
\begin{table}[t] 
\caption{
The first and the third entries in Table \ref{mextZ2b} have different
fusion rules, despite they have the same $(d_i,s_i)$.
} 
\label{frZ2_5_1} 
\centering
\begin{tabular}{ |c|ccccccccc|}
 \hline 
 $s_i$ & $0$ & $ 0$ & $ \frac{1}{2}$ & $ \frac{1}{2}$ & $ 0$ & $ \frac{1}{16}$ & $ \frac{7}{16}$ & $ \frac{9}{16}$ & $ \frac{15}{16}$\\
 $d_i$ & $1$ & $ 1$ & $ 1$ & $ 1$ & $ 2$ & $\zeta_{2}^{1}$ & $\zeta_{2}^{1}$ & $\zeta_{2}^{1}$ & $\zeta_{2}^{1}$\\
\hline
 $9^{ 1}_{ 0}$ & $\textbf{1}$  & $\textbf{2}$  & $\textbf{3}$  & $\textbf{4}$  & $\textbf{5}$  & $\textbf{6}$  & $\textbf{7}$  & $\textbf{8}$  & $\textbf{9}$ \\
\hline
$\textbf{1}$  & $ \textbf{1}$  & $ \textbf{2}$  & $ \textbf{3}$  & $ \textbf{4}$  & $ \textbf{5}$  & $ \textbf{6}$  & $ \textbf{7}$  & $ \textbf{8}$  & $ \textbf{9}$  \\
$\textbf{2}$  & $ \textbf{2}$  & $ \textbf{1}$  & $ \textbf{4}$  & $ \textbf{3}$  & $ \textbf{5}$  & $ \textbf{8}$  & $ \textbf{9}$  & $ \textbf{6}$  & $ \textbf{7}$  \\
$\textbf{3}$  & $ \textbf{3}$  & $ \textbf{4}$  & $ \textbf{1}$  & $ \textbf{2}$  & $ \textbf{5}$  & $ \textbf{8}$  & $ \textbf{7}$  & $ \textbf{6}$  & $ \textbf{9}$  \\
$\textbf{4}$  & $ \textbf{4}$  & $ \textbf{3}$  & $ \textbf{2}$  & $ \textbf{1}$  & $ \textbf{5}$  & $ \textbf{6}$  & $ \textbf{9}$  & $ \textbf{8}$  & $ \textbf{7}$  \\
$\textbf{5}$  & $ \textbf{5}$  & $ \textbf{5}$  & $ \textbf{5}$  & $ \textbf{5}$  & $ \textbf{1} \oplus \textbf{2} \oplus \textbf{3} \oplus \textbf{4}$  & $ \textbf{7} \oplus \textbf{9}$  & $ \textbf{6} \oplus \textbf{8}$  & $ \textbf{7} \oplus \textbf{9}$  & $ \textbf{6} \oplus \textbf{8}$  \\
$\textbf{6}$  & $ \textbf{6}$  & $ \textbf{8}$  & $ \textbf{8}$  & $ \textbf{6}$  & $ \textbf{7} \oplus \textbf{9}$  & $ \textbf{1} \oplus \textbf{4}$  & $ \textbf{5}$  & $ \textbf{2} \oplus \textbf{3}$  & $ \textbf{5}$  \\
$\textbf{7}$  & $ \textbf{7}$  & $ \textbf{9}$  & $ \textbf{7}$  & $ \textbf{9}$  & $ \textbf{6} \oplus \textbf{8}$  & $ \textbf{5}$  & $ \textbf{1} \oplus \textbf{3}$  & $ \textbf{5}$  & $ \textbf{2} \oplus \textbf{4}$  \\
$\textbf{8}$  & $ \textbf{8}$  & $ \textbf{6}$  & $ \textbf{6}$  & $ \textbf{8}$  & $ \textbf{7} \oplus \textbf{9}$  & $ \textbf{2} \oplus \textbf{3}$  & $ \textbf{5}$  & $ \textbf{1} \oplus \textbf{4}$  & $ \textbf{5}$  \\
$\textbf{9}$  & $ \textbf{9}$  & $ \textbf{7}$  & $ \textbf{9}$  & $ \textbf{7}$  & $ \textbf{6} \oplus \textbf{8}$  & $ \textbf{5}$  & $ \textbf{2} \oplus \textbf{4}$  & $ \textbf{5}$  & $ \textbf{1} \oplus \textbf{3}$  \\
\hline
\end{tabular}
\\[3mm]
\begin{tabular}{ |c|ccccccccc|}
 \hline 
 $s_i$ & $0$ & $ 0$ & $ \frac{1}{2}$ & $ \frac{1}{2}$ & $ 0$ & $ \frac{1}{16}$ & $ \frac{7}{16}$ & $ \frac{9}{16}$ & $ \frac{15}{16}$\\
 $d_i$ & $1$ & $ 1$ & $ 1$ & $ 1$ & $ 2$ & $\zeta_{2}^{1}$ & $\zeta_{2}^{1}$ & $\zeta_{2}^{1}$ & $\zeta_{2}^{1}$\\
\hline
 $9^{ 1}_{ 0}$ & $\textbf{1}$  & $\textbf{2}$  & $\textbf{3}$  & $\textbf{4}$  & $\textbf{5}$  & $\textbf{6}$  & $\textbf{7}$  & $\textbf{8}$  & $\textbf{9}$ \\
\hline
$\textbf{1}$  & $ \textbf{1}$  & $ \textbf{2}$  & $ \textbf{3}$  & $ \textbf{4}$  & $ \textbf{5}$  & $ \textbf{6}$  & $ \textbf{7}$  & $ \textbf{8}$  & $ \textbf{9}$  \\
$\textbf{2}$  & $ \textbf{2}$  & $ \textbf{1}$  & $ \textbf{4}$  & $ \textbf{3}$  & $ \textbf{5}$  & $ \textbf{8}$  & $ \textbf{9}$  & $ \textbf{6}$  & $ \textbf{7}$  \\
$\textbf{3}$  & $ \textbf{3}$  & $ \textbf{4}$  & $ \textbf{1}$  & $ \textbf{2}$  & $ \textbf{5}$  & $ \textbf{6}$  & $ \textbf{9}$  & $ \textbf{8}$  & $ \textbf{7}$  \\
$\textbf{4}$  & $ \textbf{4}$  & $ \textbf{3}$  & $ \textbf{2}$  & $ \textbf{1}$  & $ \textbf{5}$  & $ \textbf{8}$  & $ \textbf{7}$  & $ \textbf{6}$  & $ \textbf{9}$  \\
$\textbf{5}$  & $ \textbf{5}$  & $ \textbf{5}$  & $ \textbf{5}$  & $ \textbf{5}$  & $ \textbf{1} \oplus \textbf{2} \oplus \textbf{3} \oplus \textbf{4}$  & $ \textbf{7} \oplus \textbf{9}$  & $ \textbf{6} \oplus \textbf{8}$  & $ \textbf{7} \oplus \textbf{9}$  & $ \textbf{6} \oplus \textbf{8}$  \\
$\textbf{6}$  & $ \textbf{6}$  & $ \textbf{8}$  & $ \textbf{6}$  & $ \textbf{8}$  & $ \textbf{7} \oplus \textbf{9}$  & $ \textbf{1} \oplus \textbf{3}$  & $ \textbf{5}$  & $ \textbf{2} \oplus \textbf{4}$  & $ \textbf{5}$  \\
$\textbf{7}$  & $ \textbf{7}$  & $ \textbf{9}$  & $ \textbf{9}$  & $ \textbf{7}$  & $ \textbf{6} \oplus \textbf{8}$  & $ \textbf{5}$  & $ \textbf{1} \oplus \textbf{4}$  & $ \textbf{5}$  & $ \textbf{2} \oplus \textbf{3}$  \\
$\textbf{8}$  & $ \textbf{8}$  & $ \textbf{6}$  & $ \textbf{8}$  & $ \textbf{6}$  & $ \textbf{7} \oplus \textbf{9}$  & $ \textbf{2} \oplus \textbf{4}$  & $ \textbf{5}$  & $ \textbf{1} \oplus \textbf{3}$  & $ \textbf{5}$  \\
$\textbf{9}$  & $ \textbf{9}$  & $ \textbf{7}$  & $ \textbf{7}$  & $ \textbf{9}$  & $ \textbf{6} \oplus \textbf{8}$  & $ \textbf{5}$  & $ \textbf{2} \oplus \textbf{3}$  & $ \textbf{5}$  & $ \textbf{1} \oplus \textbf{4}$  \\
\hline
\end{tabular} 
\end{table}

\def\arraystretch{1.25} \setlength\tabcolsep{3pt}
\begin{table}[t] 
\caption{
The third and the fourth entries in Table \ref{mextZ2c} have different
fusion rules, despite they have the same $(d_i,s_i)$.
} 
\label{frZ2_5_3} 
\centering
\begin{tabular}{ |c|ccccccccc|}
 \hline 
 $s_i$ & $0$ & $ 0$ & $ \frac{1}{2}$ & $ \frac{1}{2}$ & $ \frac{1}{8}$ & $ \frac{3}{16}$ & $ \frac{7}{16}$ & $ \frac{11}{16}$ & $ \frac{15}{16}$\\
 $d_i$ & $1$ & $ 1$ & $ 1$ & $ 1$ & $ 2$ & $\zeta_{2}^{1}$ & $\zeta_{2}^{1}$ & $\zeta_{2}^{1}$ & $\zeta_{2}^{1}$\\
\hline
 $9^{ 1}_{ 1}$ & $\textbf{1}$  & $\textbf{2}$  & $\textbf{3}$  & $\textbf{4}$  & $\textbf{5}$  & $\textbf{6}$  & $\textbf{7}$  & $\textbf{8}$  & $\textbf{9}$ \\
\hline
$\textbf{1}$  & $ \textbf{1}$  & $ \textbf{2}$  & $ \textbf{3}$  & $ \textbf{4}$  & $ \textbf{5}$  & $ \textbf{6}$  & $ \textbf{7}$  & $ \textbf{8}$  & $ \textbf{9}$  \\
$\textbf{2}$  & $ \textbf{2}$  & $ \textbf{1}$  & $ \textbf{4}$  & $ \textbf{3}$  & $ \textbf{5}$  & $ \textbf{8}$  & $ \textbf{9}$  & $ \textbf{6}$  & $ \textbf{7}$  \\
$\textbf{3}$  & $ \textbf{3}$  & $ \textbf{4}$  & $ \textbf{1}$  & $ \textbf{2}$  & $ \textbf{5}$  & $ \textbf{8}$  & $ \textbf{7}$  & $ \textbf{6}$  & $ \textbf{9}$  \\
$\textbf{4}$  & $ \textbf{4}$  & $ \textbf{3}$  & $ \textbf{2}$  & $ \textbf{1}$  & $ \textbf{5}$  & $ \textbf{6}$  & $ \textbf{9}$  & $ \textbf{8}$  & $ \textbf{7}$  \\
$\textbf{5}$  & $ \textbf{5}$  & $ \textbf{5}$  & $ \textbf{5}$  & $ \textbf{5}$  & $ \textbf{1} \oplus \textbf{2} \oplus \textbf{3} \oplus \textbf{4}$  & $ \textbf{7} \oplus \textbf{9}$  & $ \textbf{6} \oplus \textbf{8}$  & $ \textbf{7} \oplus \textbf{9}$  & $ \textbf{6} \oplus \textbf{8}$  \\
$\textbf{6}$  & $ \textbf{6}$  & $ \textbf{8}$  & $ \textbf{8}$  & $ \textbf{6}$  & $ \textbf{7} \oplus \textbf{9}$  & $ \textbf{1} \oplus \textbf{4}$  & $ \textbf{5}$  & $ \textbf{2} \oplus \textbf{3}$  & $ \textbf{5}$  \\
$\textbf{7}$  & $ \textbf{7}$  & $ \textbf{9}$  & $ \textbf{7}$  & $ \textbf{9}$  & $ \textbf{6} \oplus \textbf{8}$  & $ \textbf{5}$  & $ \textbf{1} \oplus \textbf{3}$  & $ \textbf{5}$  & $ \textbf{2} \oplus \textbf{4}$  \\
$\textbf{8}$  & $ \textbf{8}$  & $ \textbf{6}$  & $ \textbf{6}$  & $ \textbf{8}$  & $ \textbf{7} \oplus \textbf{9}$  & $ \textbf{2} \oplus \textbf{3}$  & $ \textbf{5}$  & $ \textbf{1} \oplus \textbf{4}$  & $ \textbf{5}$  \\
$\textbf{9}$  & $ \textbf{9}$  & $ \textbf{7}$  & $ \textbf{9}$  & $ \textbf{7}$  & $ \textbf{6} \oplus \textbf{8}$  & $ \textbf{5}$  & $ \textbf{2} \oplus \textbf{4}$  & $ \textbf{5}$  & $ \textbf{1} \oplus \textbf{3}$  \\
\hline
\end{tabular}
\\[3mm]
\begin{tabular}{ |c|ccccccccc|}
 \hline 
 $s_i$ & $0$ & $ 0$ & $ \frac{1}{2}$ & $ \frac{1}{2}$ & $ \frac{1}{8}$ & $ \frac{3}{16}$ & $ \frac{7}{16}$ & $ \frac{11}{16}$ & $ \frac{15}{16}$\\
 $d_i$ & $1$ & $ 1$ & $ 1$ & $ 1$ & $ 2$ & $\zeta_{2}^{1}$ & $\zeta_{2}^{1}$ & $\zeta_{2}^{1}$ & $\zeta_{2}^{1}$\\
\hline
 $9^{ 1}_{ 1}$ & $\textbf{1}$  & $\textbf{2}$  & $\textbf{3}$  & $\textbf{4}$  & $\textbf{5}$  & $\textbf{6}$  & $\textbf{7}$  & $\textbf{8}$  & $\textbf{9}$ \\
\hline
$\textbf{1}$  & $ \textbf{1}$  & $ \textbf{2}$  & $ \textbf{3}$  & $ \textbf{4}$  & $ \textbf{5}$  & $ \textbf{6}$  & $ \textbf{7}$  & $ \textbf{8}$  & $ \textbf{9}$  \\
$\textbf{2}$  & $ \textbf{2}$  & $ \textbf{1}$  & $ \textbf{4}$  & $ \textbf{3}$  & $ \textbf{5}$  & $ \textbf{8}$  & $ \textbf{9}$  & $ \textbf{6}$  & $ \textbf{7}$  \\
$\textbf{3}$  & $ \textbf{3}$  & $ \textbf{4}$  & $ \textbf{1}$  & $ \textbf{2}$  & $ \textbf{5}$  & $ \textbf{6}$  & $ \textbf{9}$  & $ \textbf{8}$  & $ \textbf{7}$  \\
$\textbf{4}$  & $ \textbf{4}$  & $ \textbf{3}$  & $ \textbf{2}$  & $ \textbf{1}$  & $ \textbf{5}$  & $ \textbf{8}$  & $ \textbf{7}$  & $ \textbf{6}$  & $ \textbf{9}$  \\
$\textbf{5}$  & $ \textbf{5}$  & $ \textbf{5}$  & $ \textbf{5}$  & $ \textbf{5}$  & $ \textbf{1} \oplus \textbf{2} \oplus \textbf{3} \oplus \textbf{4}$  & $ \textbf{7} \oplus \textbf{9}$  & $ \textbf{6} \oplus \textbf{8}$  & $ \textbf{7} \oplus \textbf{9}$  & $ \textbf{6} \oplus \textbf{8}$  \\
$\textbf{6}$  & $ \textbf{6}$  & $ \textbf{8}$  & $ \textbf{6}$  & $ \textbf{8}$  & $ \textbf{7} \oplus \textbf{9}$  & $ \textbf{1} \oplus \textbf{3}$  & $ \textbf{5}$  & $ \textbf{2} \oplus \textbf{4}$  & $ \textbf{5}$  \\
$\textbf{7}$  & $ \textbf{7}$  & $ \textbf{9}$  & $ \textbf{9}$  & $ \textbf{7}$  & $ \textbf{6} \oplus \textbf{8}$  & $ \textbf{5}$  & $ \textbf{1} \oplus \textbf{4}$  & $ \textbf{5}$  & $ \textbf{2} \oplus \textbf{3}$  \\
$\textbf{8}$  & $ \textbf{8}$  & $ \textbf{6}$  & $ \textbf{8}$  & $ \textbf{6}$  & $ \textbf{7} \oplus \textbf{9}$  & $ \textbf{2} \oplus \textbf{4}$  & $ \textbf{5}$  & $ \textbf{1} \oplus \textbf{3}$  & $ \textbf{5}$  \\
$\textbf{9}$  & $ \textbf{9}$  & $ \textbf{7}$  & $ \textbf{7}$  & $ \textbf{9}$  & $ \textbf{6} \oplus \textbf{8}$  & $ \textbf{5}$  & $ \textbf{2} \oplus \textbf{3}$  & $ \textbf{5}$  & $ \textbf{1} \oplus \textbf{4}$  \\
\hline
\end{tabular}
\end{table}

\def\arraystretch{1.25} \setlength\tabcolsep{3pt}
\begin{table}[t] 
\caption{
The fusion rules of the first and the second entries in Table \ref{mextZ2c}.
} 
\label{frZ2_5_3a} 
\centering
\begin{tabular}{ |c|ccccccccc|}
 \hline 
 $s_i$ & $0$ & $ 0$ & $ \frac{1}{2}$ & $ \frac{1}{2}$ & $ \frac{1}{8}$ & $ \frac{1}{16}$ & $ \frac{1}{16}$ & $ \frac{9}{16}$ & $ \frac{9}{16}$\\
 $d_i$ & $1$ & $ 1$ & $ 1$ & $ 1$ & $ 2$ & $\zeta_{2}^{1}$ & $\zeta_{2}^{1}$ & $\zeta_{2}^{1}$ & $\zeta_{2}^{1}$\\
\hline
 $9^{ 1}_{ 1}$ & $\textbf{1}$  & $\textbf{2}$  & $\textbf{3}$  & $\textbf{4}$  & $\textbf{5}$  & $\textbf{6}$  & $\textbf{7}$  & $\textbf{8}$  & $\textbf{9}$ \\
\hline
$\textbf{1}$  & $ \textbf{1}$  & $ \textbf{2}$  & $ \textbf{3}$  & $ \textbf{4}$  & $ \textbf{5}$  & $ \textbf{6}$  & $ \textbf{7}$  & $ \textbf{8}$  & $ \textbf{9}$  \\
$\textbf{2}$  & $ \textbf{2}$  & $ \textbf{1}$  & $ \textbf{4}$  & $ \textbf{3}$  & $ \textbf{5}$  & $ \textbf{8}$  & $ \textbf{9}$  & $ \textbf{6}$  & $ \textbf{7}$  \\
$\textbf{3}$  & $ \textbf{3}$  & $ \textbf{4}$  & $ \textbf{1}$  & $ \textbf{2}$  & $ \textbf{5}$  & $ \textbf{8}$  & $ \textbf{7}$  & $ \textbf{6}$  & $ \textbf{9}$  \\
$\textbf{4}$  & $ \textbf{4}$  & $ \textbf{3}$  & $ \textbf{2}$  & $ \textbf{1}$  & $ \textbf{5}$  & $ \textbf{6}$  & $ \textbf{9}$  & $ \textbf{8}$  & $ \textbf{7}$  \\
$\textbf{5}$  & $ \textbf{5}$  & $ \textbf{5}$  & $ \textbf{5}$  & $ \textbf{5}$  & $ \textbf{1} \oplus \textbf{2} \oplus \textbf{3} \oplus \textbf{4}$  & $ \textbf{7} \oplus \textbf{9}$  & $ \textbf{6} \oplus \textbf{8}$  & $ \textbf{7} \oplus \textbf{9}$  & $ \textbf{6} \oplus \textbf{8}$  \\
$\textbf{6}$  & $ \textbf{6}$  & $ \textbf{8}$  & $ \textbf{8}$  & $ \textbf{6}$  & $ \textbf{7} \oplus \textbf{9}$  & $ \textbf{1} \oplus \textbf{4}$  & $ \textbf{5}$  & $ \textbf{2} \oplus \textbf{3}$  & $ \textbf{5}$  \\
$\textbf{7}$  & $ \textbf{7}$  & $ \textbf{9}$  & $ \textbf{7}$  & $ \textbf{9}$  & $ \textbf{6} \oplus \textbf{8}$  & $ \textbf{5}$  & $ \textbf{1} \oplus \textbf{3}$  & $ \textbf{5}$  & $ \textbf{2} \oplus \textbf{4}$  \\
$\textbf{8}$  & $ \textbf{8}$  & $ \textbf{6}$  & $ \textbf{6}$  & $ \textbf{8}$  & $ \textbf{7} \oplus \textbf{9}$  & $ \textbf{2} \oplus \textbf{3}$  & $ \textbf{5}$  & $ \textbf{1} \oplus \textbf{4}$  & $ \textbf{5}$  \\
$\textbf{9}$  & $ \textbf{9}$  & $ \textbf{7}$  & $ \textbf{9}$  & $ \textbf{7}$  & $ \textbf{6} \oplus \textbf{8}$  & $ \textbf{5}$  & $ \textbf{2} \oplus \textbf{4}$  & $ \textbf{5}$  & $ \textbf{1} \oplus \textbf{3}$  \\
\hline
\end{tabular}
\\[3mm]
\begin{tabular}{ |c|ccccccccc|}
 \hline 
 $s_i$ & $0$ & $ 0$ & $ \frac{1}{2}$ & $ \frac{1}{2}$ & $ \frac{1}{8}$ & $ \frac{5}{16}$ & $ \frac{5}{16}$ & $ \frac{13}{16}$ & $ \frac{13}{16}$\\
 $d_i$ & $1$ & $ 1$ & $ 1$ & $ 1$ & $ 2$ & $\zeta_{2}^{1}$ & $\zeta_{2}^{1}$ & $\zeta_{2}^{1}$ & $\zeta_{2}^{1}$\\
\hline
 $9^{ 1}_{ 1}$ & $\textbf{1}$  & $\textbf{2}$  & $\textbf{3}$  & $\textbf{4}$  & $\textbf{5}$  & $\textbf{6}$  & $\textbf{7}$  & $\textbf{8}$  & $\textbf{9}$ \\
\hline
$\textbf{1}$  & $ \textbf{1}$  & $ \textbf{2}$  & $ \textbf{3}$  & $ \textbf{4}$  & $ \textbf{5}$  & $ \textbf{6}$  & $ \textbf{7}$  & $ \textbf{8}$  & $ \textbf{9}$  \\
$\textbf{2}$  & $ \textbf{2}$  & $ \textbf{1}$  & $ \textbf{4}$  & $ \textbf{3}$  & $ \textbf{5}$  & $ \textbf{8}$  & $ \textbf{9}$  & $ \textbf{6}$  & $ \textbf{7}$  \\
$\textbf{3}$  & $ \textbf{3}$  & $ \textbf{4}$  & $ \textbf{1}$  & $ \textbf{2}$  & $ \textbf{5}$  & $ \textbf{8}$  & $ \textbf{7}$  & $ \textbf{6}$  & $ \textbf{9}$  \\
$\textbf{4}$  & $ \textbf{4}$  & $ \textbf{3}$  & $ \textbf{2}$  & $ \textbf{1}$  & $ \textbf{5}$  & $ \textbf{6}$  & $ \textbf{9}$  & $ \textbf{8}$  & $ \textbf{7}$  \\
$\textbf{5}$  & $ \textbf{5}$  & $ \textbf{5}$  & $ \textbf{5}$  & $ \textbf{5}$  & $ \textbf{1} \oplus \textbf{2} \oplus \textbf{3} \oplus \textbf{4}$  & $ \textbf{7} \oplus \textbf{9}$  & $ \textbf{6} \oplus \textbf{8}$  & $ \textbf{7} \oplus \textbf{9}$  & $ \textbf{6} \oplus \textbf{8}$  \\
$\textbf{6}$  & $ \textbf{6}$  & $ \textbf{8}$  & $ \textbf{8}$  & $ \textbf{6}$  & $ \textbf{7} \oplus \textbf{9}$  & $ \textbf{1} \oplus \textbf{4}$  & $ \textbf{5}$  & $ \textbf{2} \oplus \textbf{3}$  & $ \textbf{5}$  \\
$\textbf{7}$  & $ \textbf{7}$  & $ \textbf{9}$  & $ \textbf{7}$  & $ \textbf{9}$  & $ \textbf{6} \oplus \textbf{8}$  & $ \textbf{5}$  & $ \textbf{1} \oplus \textbf{3}$  & $ \textbf{5}$  & $ \textbf{2} \oplus \textbf{4}$  \\
$\textbf{8}$  & $ \textbf{8}$  & $ \textbf{6}$  & $ \textbf{6}$  & $ \textbf{8}$  & $ \textbf{7} \oplus \textbf{9}$  & $ \textbf{2} \oplus \textbf{3}$  & $ \textbf{5}$  & $ \textbf{1} \oplus \textbf{4}$  & $ \textbf{5}$  \\
$\textbf{9}$  & $ \textbf{9}$  & $ \textbf{7}$  & $ \textbf{9}$  & $ \textbf{7}$  & $ \textbf{6} \oplus \textbf{8}$  & $ \textbf{5}$  & $ \textbf{2} \oplus \textbf{4}$  & $ \textbf{5}$  & $ \textbf{1} \oplus \textbf{3}$  \\
\hline
\end{tabular}
\end{table}

The first and the second entries in Table \ref{SETZ2-5} describe two $N=5$
$\mce{\Rp(Z_2)}$'s.  They describe two different sets of bulk excitations for
$Z_2$-SET  orders.  Those bulk excitations have identical $s_i$ and $d_i$,
but they have different fusion rules $N^{ij}_k$: the 4 $d=1$ particles have a
$Z_2\times Z_2$ fusion rule for the first entry, and they have a $Z_4$ fusion
rule for the second entry (as indicated by F:$Z_2\times Z_2$ or F:$Z_4$ in the
comment column of Table \ref{SETZ2-5}).

\subsubsection{The first entry in Table \ref{SETZ2-5}}

Let us compute the modular extensions of the first entry (\ie
$5^{\zeta_2^1}_{0}$ with $Z_2\times Z_2$ fusion).  Since the total quantum
dimension of the modular extensions is $D^2=16$, the modular extensions must
have rank $N=13$ or less (since quantum dimension $d \geq 1$).

Now we would like to show $N=13$ is not possible.  If a modular extension has
$N=13$, then it must have 12 particles (labeled by $a=1,\cdots,12$) with
quantum dimension $d_a=1$, and one particle (labeled by $x$) with quantum
dimension $d_x=2$, so that $12\times 1^2+2^2=D^2=16$. In this case,
we must have the fusion rule
\begin{align}
 a\otimes x=x,\ \ \ \ x\otimes x= 1\oplus 2\oplus 3\oplus 4.
\end{align}
where $x\otimes x$ is determined by the fusion rule of the $\mce{\Rp(Z_2)}$.
The above determines the fusion matrix $N_x$ defined as $(N_x)_{ij} \equiv
N^{xi}_{j}$.  The largest eigenvalue of $N_x$ should be $2$, the quantum
dimension of $x$. Indeed, we find that the largest eigenvalue of $N_x$ is $2$.
But we also require that $N_x$ can be diagonalized by a unitary matrix (which
happens to be the $S$-matrix).  $N_x$ fails such a test.  So $N$ cannot be 13.

$N$ also cannot be 12.  If $N=12$, then the modular extension will have 10
particles (labeled by $a=1,\cdots,10$) with quantum dimension $d_a=1$, one
particle (labeled by $x$) with quantum dimension $d_x=2$, and one particle
(labeled by $y$) with quantum dimension $d_y=\sqrt 2$.  The fusion of 10
$d_a=1$ particles is described by an abelian group $Z_{10}$ or $Z_2\times Z_5$.
None of them contain $Z_2\times Z_2$ as subgroup. Thus $N=12$ is incompatible
with the $Z_2\times Z_2$ fusion of the first four $d_a=1$ particles.

We searched the modular extensions with $N$ up to 11.  We find four $N=9$
modular extensions (see Table \ref{mextZ2b}), and thus the first entry
corresponds to valid $Z_2$-SET states.

In fact one of the $Z_2$-SET states is the $Z_2$ gauge theory with a $Z_2$
global symmetry, where the $Z_2$ symmetry action exchange the $Z_2$-charge $e$
and the $Z_2$-vortex $m$.  The degenerate $e$ and $m$ give rise to the
$(d,s)=(2,0)$ particle (the fifth particle in the table).  The bound state of
$e$ and $m$ is a fermion $f$.  It may carry the $Z_2$-charge 0 or 1, which
correspond to the third and the fourth particle with $(d,s)=(1,1/2)$ in the
table. 

However, from the discussion in the last few sections, we know that a
$\mce{\Rp(Z_2)}$ always has 2 modular extensions, corresponding to the 2
bosonic $Z_2$-SPT states in 2+1D.  This seems contradictory with the above result
that the $Z_2$-SET state, $5^{\zeta_2^1}_{0}$ with $Z_2\times Z_2$ fusion, has
four different modular extensions. 

In fact, there is no  contradiction.  Here, we only use $(N^{ij}_k,s_i)$ to
label different entries. However, a $\mce{\cE}$ is fully characterized by
$(N^{ij}_k,s_i)$ plus the $F$-tensors and the $R$-tensors. 
To see this point, we note that the Ising-like UMTC $N^B_c=3^B_{m/2}$,
$m=1,3,\cdots,15$ (with central charge $c=m/2$) has three particles: $1$, $f$
with $(d_f,s_f)=(1,1/2)$, and $ \sigma $ with $(d_\sigma ,s_\sigma
)=(\sqrt{2},m/16)$. Its $R$-tensor is given by\cite{K062}
\begin{align}
 R^{ff}_1 &=-1, &
 R^{ \sigma f }_\sigma  &= 
 R^{ f \sigma }_\sigma  = -\ii^m,  
\\
 R^{ \sigma \sigma }_1 &= (-1)^{\frac{m^2-1}{8}} \ee^{-\ii \frac{ \pi }{8} m},  &
 R^{ \sigma \sigma }_f &= (-1)^{\frac{m^2-1}{8}} \ee^{\ii \frac{3 \pi }{8} m},
\nonumber 
\end{align}
and some components of the $F$-tensor are given by
\begin{align}
 F^{f \sigma \sigma; \sigma   }_{ f;1 }=
 F^{ \sigma \sigma f; \sigma   }_{ f;1 }=1.
\end{align}
The values of $R^{ \sigma f }_\sigma$ and $R^{ f \sigma }_\sigma$ are not gauge
invariant. But if we fix the values of the $F$-tensor to be the ones given
above, this will fix the gauge, and we can treat $R^{ \sigma f }_\sigma$ and
$R^{ f \sigma }_\sigma$ as if they are gauge invariant quantities.

If we stack $N^B_c=3^B_{m/2}$ and  $N^B_c=3^B_{m'/2}$ together, the induced
UMTC $ 3^B_{m/2}\boxtimes 3^B_{m'/2}$ contains particles
$\textbf{1}=(1,1)$, $\textbf{2}=(f,f')$, $\textbf{3}=(f,1)$,
$\textbf{4}=(1,f')$, $\textbf{5}=( \sigma , \sigma' )$.  Those 5 particles are
closed under the fusion, and correspond to the 5 particles in $\mce{\Rp(Z_2)}$
$5^{\zeta_2^1}_{m+m'}$.  We note that some components of the $R$-tensor of $
3^B_{m/2}\boxtimes 3^B_{m'/2}$ are given by
\begin{align}
 R^{(f,1),( \sigma , \sigma')}_{( \sigma , \sigma')}
&= R^{( \sigma , \sigma'), (f,1)}_{( \sigma , \sigma')} =-\ii^m,
\nonumber\\
 R^{(1,f'),( \sigma , \sigma')}_{( \sigma , \sigma')}
&= R^{( \sigma , \sigma'), (1,f')}_{( \sigma , \sigma')} =-\ii^{m'}.
\end{align}

Taking $(m,m')=(-1,1)$ and $(1,-1)$, it is clear the $ 3^B_{-1/2}\boxtimes
3^B_{ 1/2}$ and $ 3^B_{ 1/2}\boxtimes 3^B_{-1/2}$ give rise to two different
$R$-tensors that have identical $(N^{ij}_k,s_i)$.  So the first entry in
Table \ref{SETZ2-5} (\ie $5^{\zeta_2^1}_{0}$ with $Z_2\times Z_2$ fusion) split
into two different entries if we include the $R$-tensors. Each give rise to two
modular extensions, and this is why we got four modular extensions.  In
Table \ref{mextZ2b}, the first two modular extensions have the same
$(N^{ij}_k,s_i)$, $F$-tensor and $R$-tensors when restricted to the first 5
particles.  
The  second pair of modular extensions also have the same $(N^{ij}_k,s_i)$,
$F$-tensor and $R$-tensor when restricted to the first 5 particles, but their
$R$-tensor is different from that of the first pair.
However, note that under the exchange of the two fermions, the $R$-tensor of
the first pair becomes that of the second pair.

%In the above, we have mentioned that the first entry in Table \ref{SETZ2-5}
%(\ie $5^{\zeta_2^1}_{0}$ with $Z_2\times Z_2$ fusion) split into two different
%entries if we include the $R$-tensors.  This is a point of view based on the
%numerical data.  This does not imply that the first entry in Table
%\ref{SETZ2-5} corresponds to two $\mce{\Rp(Z_2)}$'s.  In fact, the two entries
%from the splitting describe the same $\mce{\Rp(Z_2)}$.  The observed splitting
%is an artifact that $(N^{ij}_k,s_i)$ does not fully characterize an
%$\mce{\Rp(Z_2)}$.
%With such an understanding, the first entry is equivalent to the third entry ,
%and  the second entry is equivalent to the fourth entry in Table
%\ref{mextZ2b}.  Thus the four entries of Table \ref{mextZ2b} really
%represent two distinct $Z_2$-SET orders.  

We like to stress that Table \ref{mextZ2b} is obtained using the
ME-equivalence relation, \ie the different entries are different under the
ME-equivalence relation (see Section \ref{clGQL2}).  We see that for each fixed
$\mce{\Rp(Z_2)}$ (\ie for each fixed set of $(N^{ij}_k,s_i)$, $F$-tensor and
$R$-tensor), there are two modular extensions, which agrees with our general result
for modular extensions. However, if we ignore  $F$-tensor and $R$-tensor, then
for each fixed set of $(N^{ij}_k,s_i)$, we get four modular extensions.  This
is because $(N^{ij}_k,s_i)$ is only a partial description of a
$\mce{\Rp(Z_2)}$, and
%in this case $(N^{ij}_k,s_i)$ has a particle-exchange
%that the full description $(N^{ij}_k,s_i,F,R)$ does not have.
as discussed above, in this case there are two ways to assign $F$-tensor and $R$-tensor to
them.
This is why each fixed $(N^{ij}_k,s_i)$ has four modular extensions, while
each fixed $(N^{ij}_k,s_i,F,R)$ has only two  modular extensions.

On the other hand, under the TO-equivalence relation (see
Section \ref{clGQL2}),
the two ways to assign $F$-tensor and $R$-tensor are actually equivalent
(related by exchanging the two fermions), and the first entry in
Table \ref{SETZ2-5} corresponds to only one $\mce{\Rp(Z_2)}$. Thus,
the first entry is equivalent to the third entry, and
the second entry is equivalent to the fourth entry in Table \ref{mextZ2b}.
So the four entries of Table \ref{mextZ2b} in fact represent only two distinct
$Z_2$-SET orders.

One of the two $Z_2$-SET orders have been studied extensively.  It corresponds
to $Z_2$ gauge theory with a $\Z_2$ global symmetry that exchanges the
$Z_2$-gauge-charge $e$ and the $Z_2$-gauge-vortex $m$\cite{W0303,KLW0834}.  

\subsubsection{The second entry in Table \ref{SETZ2-5}}

Next, we compute the modular extensions of the second entry in Table
\ref{SETZ2-5} (\ie $5^{\zeta_2^1}_{0}$ with
$Z_4$ fusion).  Again, we can use the same argument to show that modular
extensions of rank 12 and above do not exist.  We searched the modular
extensions with $N$ up to 11,  and find that there is no modular extensions.
So the second entry is not realizable and does not correspond to any valid
bosonic $Z_2$-SET in 2+1D.  This is indicated by NR in the comment column of
Table \ref{SETZ2-5}.

Naively, the (none existing) state from the second entry is very similar to
that from the first entry.  It is also a $Z_2$ gauge theory with a $Z_2$ global
symmetry that exchange $e$ and $m$. However, for the second entry, the $f$
particles (the third and the fourth particles) are assigned fraction
$Z_2$-charge of $\pm 1/2$. This leads to the $Z_4$ fusion rule.  Our result
implies that such an assignment is not realizable (or is illegal).
It turns out that all the $5^{\zeta_2^1}_{c}$'s with $Z_4$ fusion do not have
modular extensions. They are not realizable, and do not correspond to any 2+1D
bosonic $Z_2$-SET orders.

\subsubsection{The third entry in Table \ref{SETZ2-5}}

Third, let us compute the modular extensions of the third entry in Table
\ref{SETZ2-5} (\ie $5^{\zeta_2^1}_{1}$ with $Z_2\times Z_2$ fusion).  We find
that the entry has four modular extensions.  In fact, the entry corresponds to
two different $\mce{\Rp(Z_2)}$s, each with two modular extensions, as implied
by the two $Z_2$-SPT states.  The two $\mce{\Rp(Z_2)}$s have identical
$(N^{ij}_k,s_i,c)$, but different $F$-tensors and $R$-tensors.  Sometimes two
different $\mce{\cE}$'s (with different $F$-tensors and the $R$-tensors) can
have the same $(N^{ij}_k,s_i)$'s. The third, seventh,\dots, entries of Table
\ref{SETZ2-5} provide such examples.  We like to stress that this is different
from the first entry in Table \ref{SETZ2-5} which corresponds to one
$\mce{\Rp(Z_2)}$.

To see those different $F$-tensors and $R$-tensors, we note that one of the two
$5^{\zeta_2^1}_{1}$ with $Z_2\times Z_2$ fusion has modular extensions given by
$3^B_{1/2}\boxtimes 3^B_{1/2}$ and $3^B_{-3/2}\boxtimes 3^B_{5/2}$.  We find the
$R$-tensor for this first $5^{\zeta_2^1}_{1}$ with $Z_2\times Z_2$ fusion is
given by
\begin{align}
 R^{(f,1),( \sigma , \sigma')}_{( \sigma , \sigma')}
&= R^{( \sigma , \sigma'), (f,1)}_{( \sigma , \sigma')} =-\ii,
\nonumber\\
 R^{(1,f'),( \sigma , \sigma')}_{( \sigma , \sigma')}
&= R^{( \sigma , \sigma'), (1,f')}_{( \sigma , \sigma')} =-\ii.
\end{align}
The second $5^{\zeta_2^1}_{1}$ with $Z_2\times Z_2$ fusion has modular
extensions given by $3^B_{-1/2}\boxtimes 3^B_{3/2}$ and $3^B_{3/2}\boxtimes
3^B_{-1/2}$.  We find the $R$-tensor for the second $5^{\zeta_2^1}_{1}$ with
$Z_2\times Z_2$ fusion is given by
\begin{align}
 R^{(f,1),( \sigma , \sigma')}_{( \sigma , \sigma')}
&= R^{( \sigma , \sigma'), (f,1)}_{( \sigma , \sigma')} =\ii,
\nonumber\\
 R^{(1,f'),( \sigma , \sigma')}_{( \sigma , \sigma')}
&= R^{( \sigma , \sigma'), (1,f')}_{( \sigma , \sigma')} =\ii.
\end{align}
We see that the two $5^{\zeta_2^1}_{1}$'s with $Z_2\times Z_2$ fusion are
really different $\mce{\Rp(Z_2)}$.  Each $5^{\zeta_2^1}_{1}$ has two modular
extensions, and that is why we have four entries in Table \ref{mextZ2c}.

Again, Table \ref{mextZ2c} is obtained using the ME-equivalence relation,
and is not a table of GQLs.  Under the  TO-equivalence relation, the third
entry is equivalent to the fourth entry of Table \ref{mextZ2c}.  So the four
entries in  Table \ref{mextZ2c} actually describe \emph{three} different
$Z_2$-SET orders.  This has a very interesting consequence: \emph{The $Z_2$-SET
state described by the third (or fourth) entry in \ref{mextZ2c}, after stacked
with an $Z_2$-SPT state, still remains in the same phase.} This is an example
of the following general statement made previously: \emph{The GQLs with bulk
excitations described by $\cC$ are in one-to-one correspondence with the
quotient $\mext(\cC)/\aute(\cC)$ plus a central charge $c$.} 
In such an example $\aute(\cC)$ is non-trivial.

It is worth noting here that for the second $5^{\zeta_2^1}_{1}$, two modular
extensions $3^B_{-1/2}\boxtimes 3^B_{3/2}$ and $3^B_{3/2}\boxtimes 3^B_{-1/2}$
are actually equivalent UMTCs. This is an example that different embedings
leads to different modular extensions. For $3^B_{-1/2}\boxtimes 3^B_{3/2}$ the
first fermion in $5^{\zeta_2^1}_{1}$ is embedded into $3^B_{-1/2}$ and the
second fermion is embedded into $3^B_{3/2}$, while for $3^B_{3/2}\boxtimes
3^B_{-1/2}$ the first fermion is embedded into $3^B_{3/2}$ and the second
fermion is embedded into $3^B_{-1/2}$.  The equivalence between
$3^B_{-1/2}\boxtimes 3^B_{3/2}$ and $3^B_{3/2}\boxtimes 3^B_{-1/2}$ that
exchanges both fermions and symmetry twists fails to relate the two
embeddings, as they differ by a non-trivial automorphism of $5^{\zeta_2^1}_{1}$ that exchanges only the two fermions. This is an
example that the $\mathrm{Aut}(\cC)$ action permutes the modular extensions, as discussed in Section \ref{clGQL}.

\def\arraystretch{1.25} \setlength\tabcolsep{3pt}
\begin{table}[t] 
\caption{
The three modular extensions of $\Rp(Z_3)$.
} 
\label{mextZ3} 
\centering
\begin{tabular}{ |c|c|l|l|l| } 
\hline 
$N^{|\Th|}_{c}$ & $D^2$ & $d_1,d_2,\cdots$ & $s_1,s_2,\cdots$ & comment \\
 \hline 
$3^{\zeta_{4}^{1}}_{ 0}$ & $3$ & $1, 1, 1$ & $0, 0, 0$ & $\Rp(Z_3)$ \\
\hline
$9^{ B}_{ 0}$ & $9$ & $1\times 9$ & $0, 0, 0, 0, 0, \frac{1}{3}, \frac{1}{3}, \frac{2}{3}, \frac{2}{3}$ & $Z_3$ gauge\\
$9^{ B}_{ 0}$ & $9$ & $1\times 9$ & $0, 0, 0, \frac{1}{9}, \frac{1}{9}, \frac{4}{9}, \frac{4}{9}, \frac{7}{9}, \frac{7}{9}$ & \\
$9^{ B}_{ 0}$ & $9$ & $1\times 9$ & $0, 0, 0, \frac{2}{9}, \frac{2}{9}, \frac{5}{9}, \frac{5}{9}, \frac{8}{9}, \frac{8}{9}$ & \\
 \hline 
\end{tabular} 
\end{table}

\def\arraystretch{1.25} \setlength\tabcolsep{3pt}
\begin{table}[t] 
\caption{
The six modular extensions of $\Rp(S_3)$.
} 
\label{mextS3} 
\centering
\begin{tabular}{ |c|c|l|l|l| } 
\hline 
$N^{|\Th|}_{c}$ & $D^2$ & $d_1,d_2,\cdots$ & $s_1,s_2,\cdots$ & comment \\
\hline 
$3^{\sqrt{6}}_{ 0}$ & $6$ & $1, 1, 2$ & $0, 0, 0$ & $\Rp(S_3)$ \\
\hline
$8^{ B}_{ 0}$ & $36$ & $1, 1, 2, 2, 2, 2, 3, 3$ & $0, 0, 0, 0, \frac{1}{3}, \frac{2}{3}, 0, \frac{1}{2}$ & $S_3$ gauge\\
$8^{ B}_{ 0}$ & $36$ & $1, 1, 2, 2, 2, 2, 3, 3$ & $0, 0, 0, 0, \frac{1}{3}, \frac{2}{3}, \frac{1}{4}, \frac{3}{4}$ & \\
$8^{ B}_{ 0}$ & $36$ & $1, 1, 2, 2, 2, 2, 3, 3$ & $0, 0, 0, \frac{1}{9}, \frac{4}{9}, \frac{7}{9}, 0, \frac{1}{2}$ & $(B_4,2)$ \\
$8^{ B}_{ 0}$ & $36$ & $1, 1, 2, 2, 2, 2, 3, 3$ & $0, 0, 0, \frac{1}{9}, \frac{4}{9}, \frac{7}{9}, \frac{1}{4}, \frac{3}{4}$ & \\
$8^{ B}_{ 0}$ & $36$ & $1, 1, 2, 2, 2, 2, 3, 3$ & $0, 0, 0, \frac{2}{9}, \frac{5}{9}, \frac{8}{9}, 0, \frac{1}{2}$ & $(B_4,-2)$ \\
$8^{ B}_{ 0}$ & $36$ & $1, 1, 2, 2, 2, 2, 3, 3$ & $0, 0, 0, \frac{2}{9}, \frac{5}{9}, \frac{8}{9}, \frac{1}{4}, \frac{3}{4}$ & \\
 \hline 
\end{tabular} 
\end{table}

\subsection{$Z_3$, $Z_5$, and $S_3$ SPT orders for bosonic systems}

We also find that $\Rp(Z_3)$ has 3 modular extensions (see Table \ref{mextZ3}),
$\Rp(Z_5)$ has 5 modular extensions (see Table \ref{mextZ5}), and $\Rp(S_3)$
has 6 modular extensions (see Table \ref{mextS3}).  They correspond to the 3
$Z_3$-SPT states, the 5 $Z_5$-SPT states and the 6 $S_3$-SPT states
respectively.  These results agree with those from group cohomology
theory\cite{CGL1314}.

We note that for $\Rp(Z_2)$, $\Rp(Z_3)$, and $\Rp(S_3)$, their modular
extensions all correspond to distinct UMTCs.  However, for $\Rp(Z_5)$, its 5
modular extensions only correspond to 3 distinct UMTCs.  $\Rp(Z_5)$ has 5
modular extensions because $\Rp(Z_5)$ can be embedded into the same UMTC in
different ways. The different embeddings correspond to different modular
extensions.

\def\arraystretch{1.25} \setlength\tabcolsep{3pt}
\begin{table}[t] 
\caption{
The 16 modular extensions of $\sRp(Z_2^f)$. 
} 
\label{mextZ2f} 
\centering
\begin{tabular}{ |c|c|l|l|l| } 
\hline 
$N^{|\Th|}_{c}$ & $D^2$ & $d_1,d_2,\cdots$ & $s_1,s_2,\cdots$ & comment \\
 \hline 
$2^{ 0}_{0}$ & $2$ & $1, 1$ & $0, \frac{1}{2}$ & $\sRp(Z_2^f)$ \\
\hline
$4^{ B}_{ 0}$ & $4$ & $1, 1, 1, 1$ & $0, \frac{1}{2}, 0, 0$ & $Z_2$ gauge\\
\hline
$4^{ B}_{ 1}$ & $4$ & $1, 1, 1, 1$ & $0, \frac{1}{2}, \frac{1}{8}, \frac{1}{8}$ & F:$Z_4$ \\
$4^{ B}_{ 2}$ & $4$ & $1, 1, 1, 1$ & $0, \frac{1}{2}, \frac{1}{4}, \frac{1}{4}$ & F:$Z_2\times Z_2$ \\
$4^{ B}_{ 3}$ & $4$ & $1, 1, 1, 1$ & $0, \frac{1}{2}, \frac{3}{8}, \frac{3}{8}$ & F:$Z_4$ \\
$4^{ B}_{ 4}$ & $4$ & $1, 1, 1, 1$ & $0, \frac{1}{2}, \frac{1}{2}, \frac{1}{2}$ & F:$Z_2\times Z_2$ \\
$4^{ B}_{-3}$ & $4$ & $1, 1, 1, 1$ & $0, \frac{1}{2}, \frac{5}{8}, \frac{5}{8}$ & F:$Z_4$ \\
$4^{ B}_{-2}$ & $4$ & $1, 1, 1, 1$ & $0, \frac{1}{2}, \frac{3}{4}, \frac{3}{4}$ & F:$Z_2\times Z_2$ \\
$4^{ B}_{-1}$ & $4$ & $1, 1, 1, 1$ & $0, \frac{1}{2}, \frac{7}{8}, \frac{7}{8}$ & F:$Z_4$ \\
$3^{ B}_{ 1/2}$ & $4$ & $1, 1,\zeta_{2}^{1}$ & $0, \frac{1}{2}, \frac{1}{16}$ & $p+\ii p$ SC\\
$3^{ B}_{ 3/2}$ & $4$ & $1, 1,\zeta_{2}^{1}$ & $0, \frac{1}{2}, \frac{3}{16}$ & \\
$3^{ B}_{ 5/2}$ & $4$ & $1, 1,\zeta_{2}^{1}$ & $0, \frac{1}{2}, \frac{5}{16}$ & \\
$3^{ B}_{ 7/2}$ & $4$ & $1, 1,\zeta_{2}^{1}$ & $0, \frac{1}{2}, \frac{7}{16}$ & \\
$3^{ B}_{-7/2}$ & $4$ & $1, 1,\zeta_{2}^{1}$ & $0, \frac{1}{2}, \frac{9}{16}$ & \\
$3^{ B}_{-5/2}$ & $4$ & $1, 1,\zeta_{2}^{1}$ & $0, \frac{1}{2}, \frac{11}{16}$ & \\
$3^{ B}_{-3/2}$ & $4$ & $1, 1,\zeta_{2}^{1}$ & $0, \frac{1}{2}, \frac{13}{16}$ & \\
$3^{ B}_{-1/2}$ & $4$ & $1, 1,\zeta_{2}^{1}$ & $0, \frac{1}{2}, \frac{15}{16}$ & \\
 \hline 
\end{tabular} 
\end{table}

\def\arraystretch{1.25} \setlength\tabcolsep{3pt}
\begin{table*}[t] 
\caption{
The five modular extensions of $\Rp(Z_5)$.
} 
\label{mextZ5} 
\centering
\begin{tabular}{ |c|c|l|l|l| } 
\hline 
$N^{|\Th|}_{c}$ & $D^2$ & $d_1,d_2,\cdots$ & $s_1,s_2,\cdots$ & comment \\
 \hline 
$5^{\sqrt{5}}_{ 0}$ & $5$ & $1\times 5$ & $0, 0, 0, 0, 0$ & \\
\hline
$25^{ B}_{ 0}$ & $25$ & $1\times 25$ & $0, 0, 0, 0, 0, 0, 0, 0, 0, \frac{1}{5}, \frac{1}{5}, \frac{1}{5}, \frac{1}{5}, \frac{2}{5}, \frac{2}{5}, \frac{2}{5}, \frac{2}{5}, \frac{3}{5}, \frac{3}{5}, \frac{3}{5}, \frac{3}{5}, \frac{4}{5}, \frac{4}{5}, \frac{4}{5}, \frac{4}{5}$ & $5^{ B}_{ 0}\boxtimes 5^{ B}_{ 0}$\\
$25^{ B}_{ 0}$ & $25$ & $1\times 25$ & $0, 0, 0, 0, 0, \frac{1}{25}, \frac{1}{25}, \frac{4}{25}, \frac{4}{25}, \frac{6}{25}, \frac{6}{25}, \frac{9}{25}, \frac{9}{25}, \frac{11}{25}, \frac{11}{25}, \frac{14}{25}, \frac{14}{25}, \frac{16}{25}, \frac{16}{25}, \frac{19}{25}, \frac{19}{25}, \frac{21}{25}, \frac{21}{25}, \frac{24}{25}, \frac{24}{25}$ & \\
$25^{ B}_{ 0}$ & $25$ & $1\times 25$ & $0, 0, 0, 0, 0, \frac{1}{25}, \frac{1}{25}, \frac{4}{25}, \frac{4}{25}, \frac{6}{25}, \frac{6}{25}, \frac{9}{25}, \frac{9}{25}, \frac{11}{25}, \frac{11}{25}, \frac{14}{25}, \frac{14}{25}, \frac{16}{25}, \frac{16}{25}, \frac{19}{25}, \frac{19}{25}, \frac{21}{25}, \frac{21}{25}, \frac{24}{25}, \frac{24}{25}$ & \\
$25^{ B}_{ 0}$ & $25$ & $1\times 25$ & $0, 0, 0, 0, 0, \frac{2}{25}, \frac{2}{25}, \frac{3}{25}, \frac{3}{25}, \frac{7}{25}, \frac{7}{25}, \frac{8}{25}, \frac{8}{25}, \frac{12}{25}, \frac{12}{25}, \frac{13}{25}, \frac{13}{25}, \frac{17}{25}, \frac{17}{25}, \frac{18}{25}, \frac{18}{25}, \frac{22}{25}, \frac{22}{25}, \frac{23}{25}, \frac{23}{25}$ & \\
$25^{ B}_{ 0}$ & $25$ & $1\times 25$ & $0, 0, 0, 0, 0, \frac{2}{25}, \frac{2}{25}, \frac{3}{25}, \frac{3}{25}, \frac{7}{25}, \frac{7}{25}, \frac{8}{25}, \frac{8}{25}, \frac{12}{25}, \frac{12}{25}, \frac{13}{25}, \frac{13}{25}, \frac{17}{25}, \frac{17}{25}, \frac{18}{25}, \frac{18}{25}, \frac{22}{25}, \frac{22}{25}, \frac{23}{25}, \frac{23}{25}$ & \\
\hline
\end{tabular} 
\end{table*}

\def\arraystretch{1.25} \setlength\tabcolsep{3pt}
\begin{table*}[t] 
\caption{
All the 8 modular extensions of $\sRp(Z_4^f)$.
} 
\label{mextZ4f} 
\centering
\begin{tabular}{ |c|c|l|l|l| } 
\hline 
$N^{|\Th|}_{c}$ & $D^2$ & $d_1,d_2,\cdots$ & $s_1,s_2,\cdots$ & comment \\
\hline 
$4^{ 0}_{0}$ & $4$ & $1, 1, 1, 1$ & $0, 0, \frac{1}{2}, \frac{1}{2}$ & 
$\sRp(Z_4^f)$ \\
\hline
$16^{ B}_{ 0}$ & $16$ & $1, 1, 1, 1, 1, 1, 1, 1, 1, 1, 1, 1, 1, 1, 1, 1$ & $0, 0, \frac{1}{2}, \frac{1}{2}, 0, 0, 0, 0, 0, 0, \frac{1}{4}, \frac{1}{4}, \frac{1}{2}, \frac{1}{2}, \frac{3}{4}, \frac{3}{4}$ & \\
\hline 
$16^{ B}_{ 1}$ & $16$ & $1, 1, 1, 1, 1, 1, 1, 1, 1, 1, 1, 1, 1, 1, 1, 1$ & $0, 0, \frac{1}{2}, \frac{1}{2}, \frac{1}{32}, \frac{1}{32}, \frac{1}{8}, \frac{1}{8}, \frac{1}{8}, \frac{1}{8}, \frac{9}{32}, \frac{9}{32}, \frac{17}{32}, \frac{17}{32}, \frac{25}{32}, \frac{25}{32}$ & \\
$16^{ B}_{ 2}$ & $16$ & $1, 1, 1, 1, 1, 1, 1, 1, 1, 1, 1, 1, 1, 1, 1, 1$ & $0, 0, \frac{1}{2}, \frac{1}{2}, \frac{1}{16}, \frac{1}{16}, \frac{1}{4}, \frac{1}{4}, \frac{1}{4}, \frac{1}{4}, \frac{5}{16}, \frac{5}{16}, \frac{9}{16}, \frac{9}{16}, \frac{13}{16}, \frac{13}{16}$ & $8^{ B}_{ 1}\boxtimes 2^{ B}_{ 1}$\\
$16^{ B}_{ 3}$ & $16$ & $1, 1, 1, 1, 1, 1, 1, 1, 1, 1, 1, 1, 1, 1, 1, 1$ & $0, 0, \frac{1}{2}, \frac{1}{2}, \frac{3}{32}, \frac{3}{32}, \frac{11}{32}, \frac{11}{32}, \frac{3}{8}, \frac{3}{8}, \frac{3}{8}, \frac{3}{8}, \frac{19}{32}, \frac{19}{32}, \frac{27}{32}, \frac{27}{32}$ & \\
$16^{ B}_{ 4}$ & $16$ & $1, 1, 1, 1, 1, 1, 1, 1, 1, 1, 1, 1, 1, 1, 1, 1$ & $0, 0, \frac{1}{2}, \frac{1}{2}, \frac{1}{8}, \frac{1}{8}, \frac{3}{8}, \frac{3}{8}, \frac{1}{2}, \frac{1}{2}, \frac{1}{2}, \frac{1}{2}, \frac{5}{8}, \frac{5}{8}, \frac{7}{8}, \frac{7}{8}$ & $4^{ B}_{ 3}\boxtimes 4^{ B}_{ 1}$\\
$16^{ B}_{-3}$ & $16$ & $1, 1, 1, 1, 1, 1, 1, 1, 1, 1, 1, 1, 1, 1, 1, 1$ & $0, 0, \frac{1}{2}, \frac{1}{2}, \frac{5}{32}, \frac{5}{32}, \frac{13}{32}, \frac{13}{32}, \frac{5}{8}, \frac{5}{8}, \frac{5}{8}, \frac{5}{8}, \frac{21}{32}, \frac{21}{32}, \frac{29}{32}, \frac{29}{32}$ & \\
$16^{ B}_{-2}$ & $16$ & $1, 1, 1, 1, 1, 1, 1, 1, 1, 1, 1, 1, 1, 1, 1, 1$ & $0, 0, \frac{1}{2}, \frac{1}{2}, \frac{3}{16}, \frac{3}{16}, \frac{7}{16}, \frac{7}{16}, \frac{11}{16}, \frac{11}{16}, \frac{3}{4}, \frac{3}{4}, \frac{3}{4}, \frac{3}{4}, \frac{15}{16}, \frac{15}{16}$ & $8^{ B}_{-1}\boxtimes 2^{ B}_{-1}$\\
$16^{ B}_{-1}$ & $16$ & $1, 1, 1, 1, 1, 1, 1, 1, 1, 1, 1, 1, 1, 1, 1, 1$ & $0, 0, \frac{1}{2}, \frac{1}{2}, \frac{7}{32}, \frac{7}{32}, \frac{15}{32}, \frac{15}{32}, \frac{23}{32}, \frac{23}{32}, \frac{7}{8}, \frac{7}{8}, \frac{7}{8}, \frac{7}{8}, \frac{31}{32}, \frac{31}{32}$ & \\
\hline
\end{tabular}
\end{table*}

\def\arraystretch{1.25} \setlength\tabcolsep{3pt}
\begin{table*}[t] 
\caption{
The two $c=0$ modular extensions of $\sRp(Z_8^f)$ imply that
the $Z_8^f$ fermionic SPT phases are described by $\Z_2$.
All other  modular extensions only appear for integer $c$ and are all abelian (two modular extensions
for each integer $c$).
} 
\label{mextZ8f} 
\centering
\begin{tabular}{ |c|c|l|p{5.3in}|l| } 
\hline 
$N^{|\Th|}_{c}$ & $D^2$ & $d_1,d_2,\cdots$ & $s_1,s_2,\cdots$ & comment \\
\hline 
$8^{ 0}_{0}$ & $8$ & $1\times 8$ & $0, \frac{1}{2}, 0, \frac{1}{2}, 0, \frac{1}{2}, 0, \frac{1}{2}$ & \\
\hline
$64^{ B}_{ 0}$ & $64$ & $1\times 64$ & $0, \frac{1}{2}, 0, \frac{1}{2}, 0, \frac{1}{2}, 0, \frac{1}{2}, 0, 0, 0, 0, 0, 0, 0, 0, 0, 0, 0, 0, 0, 0, 0, 0, \frac{1}{8}, \frac{1}{8}, \frac{1}{8}, \frac{1}{8}, \frac{1}{4}, \frac{1}{4}, \frac{1}{4}, \frac{1}{4}, \frac{1}{4}, \frac{1}{4}, \frac{1}{4}, \frac{1}{4}, \frac{3}{8}, \frac{3}{8}, \frac{3}{8}, \frac{3}{8}$, $\frac{1}{2}, \frac{1}{2}, \frac{1}{2}, \frac{1}{2}, \frac{1}{2}, \frac{1}{2}, \frac{1}{2}, \frac{1}{2}, \frac{5}{8}, \frac{5}{8}, \frac{5}{8}, \frac{5}{8}, \frac{3}{4}, \frac{3}{4}, \frac{3}{4}, \frac{3}{4}, \frac{3}{4}, \frac{3}{4}, \frac{3}{4}, \frac{3}{4}, \frac{7}{8}, \frac{7}{8}, \frac{7}{8}, \frac{7}{8}$& \\
$64^{ B}_{ 0}$ & $64$ & $1\times 64$ & $0, \frac{1}{2}, 0, \frac{1}{2}, 0, \frac{1}{2}, 0, \frac{1}{2}, 0, 0, 0, 0, 0, 0, 0, 0, 0, 0, 0, 0, \frac{1}{16}, \frac{1}{16}, \frac{1}{16}, \frac{1}{16}, \frac{3}{16}, \frac{3}{16}, \frac{3}{16}, \frac{3}{16}, \frac{1}{4}, \frac{1}{4}, \frac{1}{4}, \frac{1}{4}, \frac{5}{16}, \frac{5}{16}, \frac{5}{16}, \frac{5}{16}$, $\frac{7}{16}, \frac{7}{16}, \frac{7}{16}, \frac{7}{16}, \frac{1}{2}, \frac{1}{2}, \frac{1}{2}, \frac{1}{2}, \frac{9}{16}, \frac{9}{16}, \frac{9}{16}, \frac{9}{16}, \frac{11}{16}, \frac{11}{16}, \frac{11}{16}, \frac{11}{16}, \frac{3}{4}, \frac{3}{4}, \frac{3}{4}, \frac{3}{4}, \frac{13}{16}, \frac{13}{16}, \frac{13}{16}, \frac{13}{16}, \frac{15}{16}, \frac{15}{16}, \frac{15}{16}, \frac{15}{16}$ & \\
\hline
\end{tabular}
\end{table*}

%\subsection{$Z_2\times Z_2^f$ SPT orders for fermionic systems}

\subsection{Invertible fermionic topological orders}

We find that $\sRp(Z_2^f)$ has 16 modular extensions (see Table \ref{mextZ2f})
which correspond to invertible fermionic topological orders in 2+1D.  One might
thought that the invertible fermionic topological orders are classified by
$\Z_{16}$.  But in fact, the invertible fermionic topological orders are
classified by $\Z$, obtained by stacking the $c=1/2$ $p+\ii p$ states.  The
discrepancy is due to the fact that the modular extensions cannot see the $c=8$
$E_8$ states. The 16 modular extensions exactly correspond to the invertible
fermionic topological orders modulo the $E_8$ states.  

We also find that the modular extensions with $c=$ even have a
$Z_2\times Z_2$ fusion rule, while the modular extensions with $c=$ odd have a
$Z_4$ fusion rule (indicated by F:$Z_2\times Z_2$ or F:$Z_4$ in the comment
column of Table).

The $Z_2^f$-SPT states for fermions is given by the modular extensions with
zero central charge.  We see that there is only one modular extension with
central charge $c=0$.  Thus there is no non-trivial 2+1D fermionic SPT states
with $Z_2^f$ symmetry.  In general, the  modular extensions of $\sRp(G^f)$ with
zero central charge correspond to the fermionic SPT states in 2+1D with
symmetry $G^f$.

To calculate the $Z_2\times Z_2^f$ SPT orders for fermionic systems, we first
compute the modular extensions for $\sRp(Z_2\times Z_2^f)$.  We note that
$\sRp(Z_2\times Z_2^f)=\sRp(Z_2^f\times \tilde Z_2^f)$.  Thus, the  modular
extensions for $\sRp(Z_2\times Z_2^f)$ is the  modular extensions of
$\sRp(Z_2^f\times \tilde Z_2^f)$.  Some of the modular extensions
of $\sRp(Z_2^f\times \tilde Z_2^f)$ are given by the modular extensions of
$\sRp(Z_2^f)$ stacked (under $\boxtimes$) with the modular extensions of
$\sRp(\tilde Z_2^f)$.  Some of the modular extensions of $\sRp(Z_2\times
Z_2^f)$ are given by the modular extensions for $\Rp(Z_2)$ stacked (under
$\boxtimes$) with the modular extensions of $\sRp(Z_2^f)$.

The above mathematical statements correspond to the following physical picture:
Some fermionic GQLs with $Z_2\times Z_2^f$ symmetry can be viewed as bosonic
GQLs with $Z_2$ symmetry stacked with fermionic GQLs with $Z_2^f$ symmetry.
Also some fermionic GQLs with $Z_2^f\times \tilde Z_2^f$ symmetry can be viewed
as fermionic GQLs with $Z_2^f$ symmetry stacked with fermionic GQLs with
$\tilde Z_2^f$ symmetry.

Using \eqn{Hgcnd}, we find that the modular extensions for $Z_2\times Z_2^f$
symmetry must have ranks $7, 9, 10, 12, 16$.  By direct search for those ranks,
we find that the modular extensions of $\sRp(Z_2\times Z_2^f)$ are given by
Tables \ref{mextZ2Z2f}, \ref{mextZ2Z2f12}, \ref{mextZ2Z2f12b} and
\ref{mextZ2Z2f16}.  The $N=9$ modular extensions of $\sRp(Z_2\times Z_2^f)$ in
Table \ref{mextZ2Z2f} are given by the stacking of the $N=3$ modular extensions
of $\sRp(Z_2^f)$ and the $N=3$ modular extensions of $\sRp(\tilde Z_2^f)$.  The
$N=16$ modular extensions of $\sRp(Z_2\times Z_2^f)$ in Table \ref{mextZ2Z2f16}
are given by the stacking of the $N=4$ modular extensions of $\sRp(Z_2^f)$ and
the $N=4$ modular extensions of $\sRp(\tilde Z_2^f)$.  There are also 64 $N=12$
modular extensions of $\sRp(Z_2\times Z_2^f)$ given by the stacking of the
$N=4$ ($N=3$) modular extensions of $\sRp(Z_2^f)$ and the $N=3$ ($N=4$) modular
extensions of $\sRp(\tilde Z_2^f)$.

Many of the modular extensions have non-trivial topological orders since the
central charge $c$ is non-zero. There are eight modular extensions for each central
charge $c=0,1/2,1,3/2,\dots,15/2$, and in total $8\times 16=128$ modular
extensions.  Those eight with $c=0$  correspond
to the $Z_2\times Z_2^f$ fermionic SPT states.  Those are all the $Z_2\times
Z_2^f$ fermionic SPT states\cite{GL1369}.

\def\arraystretch{1.25} \setlength\tabcolsep{3pt}
\begin{table*}[t] 
\caption{
All the 32 modular extensions of $\sRp(Z_2\times Z_2^f)$ with $N = 9$.
} 
\label{mextZ2Z2f} 
\centering
\begin{tabular}{ |c|c|l|l|l| } 
\hline 
$N^{|\Th|}_{c}$ & $D^2$ & $d_1,d_2,\cdots$ & $s_1,s_2,\cdots$ & comment \\
\hline 
$4^{ 0}_{0}$ & $4$ & $1, 1, 1, 1$ & $0, 0, \frac{1}{2}, \frac{1}{2}$ & $\sRp(Z_2\times Z_2^f)$ \\
\hline
$9^{ B}_{ 0}$ & $16$ & $1, 1, 1, 1,\zeta_{2}^{1},\zeta_{2}^{1},\zeta_{2}^{1},\zeta_{2}^{1}, 2$ & $0, 0, \frac{1}{2}, \frac{1}{2}, \frac{1}{16}, \frac{7}{16}, \frac{9}{16}, \frac{15}{16}, 0$ & $3^{ B}_{-1/2}\boxtimes 3^{ B}_{ 1/2}$\\
$9^{ B}_{ 0}$ & $16$ & $1, 1, 1, 1,\zeta_{2}^{1},\zeta_{2}^{1},\zeta_{2}^{1},\zeta_{2}^{1}, 2$ & $0, 0, \frac{1}{2}, \frac{1}{2}, \frac{1}{16}, \frac{7}{16}, \frac{9}{16}, \frac{15}{16}, 0$ & $3^{ B}_{-1/2}\boxtimes 3^{ B}_{ 1/2}$\\
$9^{ B}_{ 0}$ & $16$ & $1, 1, 1, 1,\zeta_{2}^{1},\zeta_{2}^{1},\zeta_{2}^{1},\zeta_{2}^{1}, 2$ & $0, 0, \frac{1}{2}, \frac{1}{2}, \frac{3}{16}, \frac{5}{16}, \frac{11}{16}, \frac{13}{16}, 0$ & $3^{ B}_{-3/2}\boxtimes 3^{ B}_{ 3/2}$\\
$9^{ B}_{ 0}$ & $16$ & $1, 1, 1, 1,\zeta_{2}^{1},\zeta_{2}^{1},\zeta_{2}^{1},\zeta_{2}^{1}, 2$ & $0, 0, \frac{1}{2}, \frac{1}{2}, \frac{3}{16}, \frac{5}{16}, \frac{11}{16}, \frac{13}{16}, 0$ & $3^{ B}_{-3/2}\boxtimes 3^{ B}_{ 3/2}$\\
\hline 
$9^{ B}_{ 1}$ & $16$ & $1, 1, 1, 1,\zeta_{2}^{1},\zeta_{2}^{1},\zeta_{2}^{1},\zeta_{2}^{1}, 2$ & $0, 0, \frac{1}{2}, \frac{1}{2}, \frac{1}{16}, \frac{1}{16}, \frac{9}{16}, \frac{9}{16}, \frac{1}{8}$ & $3^{ B}_{ 1/2}\boxtimes 3^{ B}_{ 1/2}$\\
$9^{ B}_{ 1}$ & $16$ & $1, 1, 1, 1,\zeta_{2}^{1},\zeta_{2}^{1},\zeta_{2}^{1},\zeta_{2}^{1}, 2$ & $0, 0, \frac{1}{2}, \frac{1}{2}, \frac{3}{16}, \frac{7}{16}, \frac{11}{16}, \frac{15}{16}, \frac{1}{8}$ & $3^{ B}_{-1/2}\boxtimes 3^{ B}_{ 3/2}$\\
$9^{ B}_{ 1}$ & $16$ & $1, 1, 1, 1,\zeta_{2}^{1},\zeta_{2}^{1},\zeta_{2}^{1},\zeta_{2}^{1}, 2$ & $0, 0, \frac{1}{2}, \frac{1}{2}, \frac{3}{16}, \frac{7}{16}, \frac{11}{16}, \frac{15}{16}, \frac{1}{8}$ & $3^{ B}_{-1/2}\boxtimes 3^{ B}_{ 3/2}$\\
$9^{ B}_{ 1}$ & $16$ & $1, 1, 1, 1,\zeta_{2}^{1},\zeta_{2}^{1},\zeta_{2}^{1},\zeta_{2}^{1}, 2$ & $0, 0, \frac{1}{2}, \frac{1}{2}, \frac{5}{16}, \frac{5}{16}, \frac{13}{16}, \frac{13}{16}, \frac{1}{8}$ & $3^{ B}_{-3/2}\boxtimes 3^{ B}_{ 5/2}$\\
$9^{ B}_{ 2}$ & $16$ & $1, 1, 1, 1,\zeta_{2}^{1},\zeta_{2}^{1},\zeta_{2}^{1},\zeta_{2}^{1}, 2$ & $0, 0, \frac{1}{2}, \frac{1}{2}, \frac{1}{16}, \frac{3}{16}, \frac{9}{16}, \frac{11}{16}, \frac{1}{4}$ & $3^{ B}_{ 3/2}\boxtimes 3^{ B}_{ 1/2}$\\
$9^{ B}_{ 2}$ & $16$ & $1, 1, 1, 1,\zeta_{2}^{1},\zeta_{2}^{1},\zeta_{2}^{1},\zeta_{2}^{1}, 2$ & $0, 0, \frac{1}{2}, \frac{1}{2}, \frac{1}{16}, \frac{3}{16}, \frac{9}{16}, \frac{11}{16}, \frac{1}{4}$ & $3^{ B}_{ 3/2}\boxtimes 3^{ B}_{ 1/2}$\\
$9^{ B}_{ 2}$ & $16$ & $1, 1, 1, 1,\zeta_{2}^{1},\zeta_{2}^{1},\zeta_{2}^{1},\zeta_{2}^{1}, 2$ & $0, 0, \frac{1}{2}, \frac{1}{2}, \frac{5}{16}, \frac{7}{16}, \frac{13}{16}, \frac{15}{16}, \frac{1}{4}$ & $3^{ B}_{-1/2}\boxtimes 3^{ B}_{ 5/2}$\\
$9^{ B}_{ 2}$ & $16$ & $1, 1, 1, 1,\zeta_{2}^{1},\zeta_{2}^{1},\zeta_{2}^{1},\zeta_{2}^{1}, 2$ & $0, 0, \frac{1}{2}, \frac{1}{2}, \frac{5}{16}, \frac{7}{16}, \frac{13}{16}, \frac{15}{16}, \frac{1}{4}$ & $3^{ B}_{-1/2}\boxtimes 3^{ B}_{ 5/2}$\\
$9^{ B}_{ 3}$ & $16$ & $1, 1, 1, 1,\zeta_{2}^{1},\zeta_{2}^{1},\zeta_{2}^{1},\zeta_{2}^{1}, 2$ & $0, 0, \frac{1}{2}, \frac{1}{2}, \frac{1}{16}, \frac{5}{16}, \frac{9}{16}, \frac{13}{16}, \frac{3}{8}$ & $3^{ B}_{ 5/2}\boxtimes 3^{ B}_{ 1/2}$\\
$9^{ B}_{ 3}$ & $16$ & $1, 1, 1, 1,\zeta_{2}^{1},\zeta_{2}^{1},\zeta_{2}^{1},\zeta_{2}^{1}, 2$ & $0, 0, \frac{1}{2}, \frac{1}{2}, \frac{1}{16}, \frac{5}{16}, \frac{9}{16}, \frac{13}{16}, \frac{3}{8}$ & $3^{ B}_{ 5/2}\boxtimes 3^{ B}_{ 1/2}$\\
$9^{ B}_{ 3}$ & $16$ & $1, 1, 1, 1,\zeta_{2}^{1},\zeta_{2}^{1},\zeta_{2}^{1},\zeta_{2}^{1}, 2$ & $0, 0, \frac{1}{2}, \frac{1}{2}, \frac{3}{16}, \frac{3}{16}, \frac{11}{16}, \frac{11}{16}, \frac{3}{8}$ & $3^{ B}_{ 3/2}\boxtimes 3^{ B}_{ 3/2}$\\
$9^{ B}_{ 3}$ & $16$ & $1, 1, 1, 1,\zeta_{2}^{1},\zeta_{2}^{1},\zeta_{2}^{1},\zeta_{2}^{1}, 2$ & $0, 0, \frac{1}{2}, \frac{1}{2}, \frac{7}{16}, \frac{7}{16}, \frac{15}{16}, \frac{15}{16}, \frac{3}{8}$ & $3^{ B}_{-1/2}\boxtimes 3^{ B}_{ 7/2}$\\
$9^{ B}_{ 4}$ & $16$ & $1, 1, 1, 1,\zeta_{2}^{1},\zeta_{2}^{1},\zeta_{2}^{1},\zeta_{2}^{1}, 2$ & $0, 0, \frac{1}{2}, \frac{1}{2}, \frac{1}{16}, \frac{7}{16}, \frac{9}{16}, \frac{15}{16}, \frac{1}{2}$ & $3^{ B}_{ 7/2}\boxtimes 3^{ B}_{ 1/2}$\\
$9^{ B}_{ 4}$ & $16$ & $1, 1, 1, 1,\zeta_{2}^{1},\zeta_{2}^{1},\zeta_{2}^{1},\zeta_{2}^{1}, 2$ & $0, 0, \frac{1}{2}, \frac{1}{2}, \frac{1}{16}, \frac{7}{16}, \frac{9}{16}, \frac{15}{16}, \frac{1}{2}$ & $3^{ B}_{ 7/2}\boxtimes 3^{ B}_{ 1/2}$\\
$9^{ B}_{ 4}$ & $16$ & $1, 1, 1, 1,\zeta_{2}^{1},\zeta_{2}^{1},\zeta_{2}^{1},\zeta_{2}^{1}, 2$ & $0, 0, \frac{1}{2}, \frac{1}{2}, \frac{3}{16}, \frac{5}{16}, \frac{11}{16}, \frac{13}{16}, \frac{1}{2}$ & $3^{ B}_{ 5/2}\boxtimes 3^{ B}_{ 3/2}$\\
$9^{ B}_{ 4}$ & $16$ & $1, 1, 1, 1,\zeta_{2}^{1},\zeta_{2}^{1},\zeta_{2}^{1},\zeta_{2}^{1}, 2$ & $0, 0, \frac{1}{2}, \frac{1}{2}, \frac{3}{16}, \frac{5}{16}, \frac{11}{16}, \frac{13}{16}, \frac{1}{2}$ & $3^{ B}_{ 5/2}\boxtimes 3^{ B}_{ 3/2}$\\
$9^{ B}_{-3}$ & $16$ & $1, 1, 1, 1,\zeta_{2}^{1},\zeta_{2}^{1},\zeta_{2}^{1},\zeta_{2}^{1}, 2$ & $0, 0, \frac{1}{2}, \frac{1}{2}, \frac{1}{16}, \frac{1}{16}, \frac{9}{16}, \frac{9}{16}, \frac{5}{8}$ & $3^{ B}_{-7/2}\boxtimes 3^{ B}_{ 1/2}$\\
$9^{ B}_{-3}$ & $16$ & $1, 1, 1, 1,\zeta_{2}^{1},\zeta_{2}^{1},\zeta_{2}^{1},\zeta_{2}^{1}, 2$ & $0, 0, \frac{1}{2}, \frac{1}{2}, \frac{3}{16}, \frac{7}{16}, \frac{11}{16}, \frac{15}{16}, \frac{5}{8}$ & $3^{ B}_{ 7/2}\boxtimes 3^{ B}_{ 3/2}$\\
$9^{ B}_{-3}$ & $16$ & $1, 1, 1, 1,\zeta_{2}^{1},\zeta_{2}^{1},\zeta_{2}^{1},\zeta_{2}^{1}, 2$ & $0, 0, \frac{1}{2}, \frac{1}{2}, \frac{3}{16}, \frac{7}{16}, \frac{11}{16}, \frac{15}{16}, \frac{5}{8}$ & $3^{ B}_{ 7/2}\boxtimes 3^{ B}_{ 3/2}$\\
$9^{ B}_{-3}$ & $16$ & $1, 1, 1, 1,\zeta_{2}^{1},\zeta_{2}^{1},\zeta_{2}^{1},\zeta_{2}^{1}, 2$ & $0, 0, \frac{1}{2}, \frac{1}{2}, \frac{5}{16}, \frac{5}{16}, \frac{13}{16}, \frac{13}{16}, \frac{5}{8}$ & $3^{ B}_{ 5/2}\boxtimes 3^{ B}_{ 5/2}$\\
$9^{ B}_{-2}$ & $16$ & $1, 1, 1, 1,\zeta_{2}^{1},\zeta_{2}^{1},\zeta_{2}^{1},\zeta_{2}^{1}, 2$ & $0, 0, \frac{1}{2}, \frac{1}{2}, \frac{1}{16}, \frac{3}{16}, \frac{9}{16}, \frac{11}{16}, \frac{3}{4}$ & $3^{ B}_{-5/2}\boxtimes 3^{ B}_{ 1/2}$\\
$9^{ B}_{-2}$ & $16$ & $1, 1, 1, 1,\zeta_{2}^{1},\zeta_{2}^{1},\zeta_{2}^{1},\zeta_{2}^{1}, 2$ & $0, 0, \frac{1}{2}, \frac{1}{2}, \frac{1}{16}, \frac{3}{16}, \frac{9}{16}, \frac{11}{16}, \frac{3}{4}$ & $3^{ B}_{-5/2}\boxtimes 3^{ B}_{ 1/2}$\\
$9^{ B}_{-2}$ & $16$ & $1, 1, 1, 1,\zeta_{2}^{1},\zeta_{2}^{1},\zeta_{2}^{1},\zeta_{2}^{1}, 2$ & $0, 0, \frac{1}{2}, \frac{1}{2}, \frac{5}{16}, \frac{7}{16}, \frac{13}{16}, \frac{15}{16}, \frac{3}{4}$ & $3^{ B}_{ 7/2}\boxtimes 3^{ B}_{ 5/2}$\\
$9^{ B}_{-2}$ & $16$ & $1, 1, 1, 1,\zeta_{2}^{1},\zeta_{2}^{1},\zeta_{2}^{1},\zeta_{2}^{1}, 2$ & $0, 0, \frac{1}{2}, \frac{1}{2}, \frac{5}{16}, \frac{7}{16}, \frac{13}{16}, \frac{15}{16}, \frac{3}{4}$ & $3^{ B}_{ 7/2}\boxtimes 3^{ B}_{ 5/2}$\\
$9^{ B}_{-1}$ & $16$ & $1, 1, 1, 1,\zeta_{2}^{1},\zeta_{2}^{1},\zeta_{2}^{1},\zeta_{2}^{1}, 2$ & $0, 0, \frac{1}{2}, \frac{1}{2}, \frac{1}{16}, \frac{5}{16}, \frac{9}{16}, \frac{13}{16}, \frac{7}{8}$ & $3^{ B}_{-3/2}\boxtimes 3^{ B}_{ 1/2}$\\
$9^{ B}_{-1}$ & $16$ & $1, 1, 1, 1,\zeta_{2}^{1},\zeta_{2}^{1},\zeta_{2}^{1},\zeta_{2}^{1}, 2$ & $0, 0, \frac{1}{2}, \frac{1}{2}, \frac{1}{16}, \frac{5}{16}, \frac{9}{16}, \frac{13}{16}, \frac{7}{8}$ & $3^{ B}_{-3/2}\boxtimes 3^{ B}_{ 1/2}$\\
$9^{ B}_{-1}$ & $16$ & $1, 1, 1, 1,\zeta_{2}^{1},\zeta_{2}^{1},\zeta_{2}^{1},\zeta_{2}^{1}, 2$ & $0, 0, \frac{1}{2}, \frac{1}{2}, \frac{3}{16}, \frac{3}{16}, \frac{11}{16}, \frac{11}{16}, \frac{7}{8}$ & $3^{ B}_{-5/2}\boxtimes 3^{ B}_{ 3/2}$\\
$9^{ B}_{-1}$ & $16$ & $1, 1, 1, 1,\zeta_{2}^{1},\zeta_{2}^{1},\zeta_{2}^{1},\zeta_{2}^{1}, 2$ & $0, 0, \frac{1}{2}, \frac{1}{2}, \frac{7}{16}, \frac{7}{16}, \frac{15}{16}, \frac{15}{16}, \frac{7}{8}$ & $3^{ B}_{ 7/2}\boxtimes 3^{ B}_{ 7/2}$\\
\hline
\end{tabular}
\end{table*}

\subsection{$Z_{2n}^f$ SPT orders for fermionic systems}

We also find the modular extensions for $\sRp(Z_4^f)$, $\sRp(Z_6^f)$, and
$\sRp(Z_8^f)$ (see Tables \ref{mextZ4f}, \ref{mextZ6f}, and \ref{mextZ8f}).
Again, many of them has non-trivial topological orders since the central charge
$c$ is non-zero.  

For $Z_4^f$ group, only one of them have $c=0$.  So there is no non-trivial
$Z_4^f$ fermionic SPT states.  For $Z_6^f$ group, only three of them have
$c=0$.  So, the $Z_6^f$ fermionic SPT states are described by $\Z_3$.  For
$Z_8^f$ group, only two of them have $c=0$.  So, the $Z_8^f$ fermionic SPT
states are described by $\Z_2$.  Those results are consistent with the results
in \Ref{KTT1429,CJWang}.  However, the calculation present here is more
complete.

\def\arraystretch{1.25} \setlength\tabcolsep{3pt}
\begin{table*}[t] 
\caption{
The first 32 modular extensions of $\sRp(Z_2\times Z_2^f)$ with $N =12$.
} 
\label{mextZ2Z2f12} 
\centering
\begin{tabular}{ |c|c|l|l|l| } 
\hline 
$N^{|\Th|}_{c}$ & $D^2$ & $d_1,d_2,\cdots$ & $s_1,s_2,\cdots$ & comment \\
\hline 
$4^{ 0}_{0}$ & $4$ & $1, 1, 1, 1$ & $0, 0, \frac{1}{2}, \frac{1}{2}$ & $\sRp(Z_2\times Z_2^f)$ \\
\hline
$12^{ B}_{ 1/2}$ & $16$ & $1, 1, 1, 1, 1, 1, 1, 1,\zeta_{2}^{1},\zeta_{2}^{1},\zeta_{2}^{1},\zeta_{2}^{1}$ & $0, 0, \frac{1}{2}, \frac{1}{2}, 0, 0, \frac{1}{2}, \frac{1}{2}, \frac{1}{16}, \frac{1}{16}, \frac{1}{16}, \frac{9}{16}$ & $4^{ B}_{ 0}\boxtimes 3^{ B}_{ 1/2}$\\
$12^{ B}_{ 1/2}$ & $16$ & $1, 1, 1, 1, 1, 1, 1, 1,\zeta_{2}^{1},\zeta_{2}^{1},\zeta_{2}^{1},\zeta_{2}^{1}$ & $0, 0, \frac{1}{2}, \frac{1}{2}, 0, 0, \frac{1}{2}, \frac{1}{2}, \frac{1}{16}, \frac{1}{16}, \frac{1}{16}, \frac{9}{16}$ & $4^{ B}_{ 0}\boxtimes 3^{ B}_{ 1/2}$\\
$12^{ B}_{ 1/2}$ & $16$ & $1, 1, 1, 1, 1, 1, 1, 1,\zeta_{2}^{1},\zeta_{2}^{1},\zeta_{2}^{1},\zeta_{2}^{1}$ & $0, 0, \frac{1}{2}, \frac{1}{2}, \frac{1}{8}, \frac{1}{8}, \frac{5}{8}, \frac{5}{8}, \frac{1}{16}, \frac{1}{16}, \frac{7}{16}, \frac{15}{16}$ & $4^{ B}_{-3}\boxtimes 3^{ B}_{ 7/2}$\\
$12^{ B}_{ 1/2}$ & $16$ & $1, 1, 1, 1, 1, 1, 1, 1,\zeta_{2}^{1},\zeta_{2}^{1},\zeta_{2}^{1},\zeta_{2}^{1}$ & $0, 0, \frac{1}{2}, \frac{1}{2}, \frac{1}{8}, \frac{1}{8}, \frac{5}{8}, \frac{5}{8}, \frac{1}{16}, \frac{1}{16}, \frac{7}{16}, \frac{15}{16}$ & $4^{ B}_{-3}\boxtimes 3^{ B}_{ 7/2}$\\
$12^{ B}_{ 1/2}$ & $16$ & $1, 1, 1, 1, 1, 1, 1, 1,\zeta_{2}^{1},\zeta_{2}^{1},\zeta_{2}^{1},\zeta_{2}^{1}$ & $0, 0, \frac{1}{2}, \frac{1}{2}, \frac{1}{4}, \frac{1}{4}, \frac{3}{4}, \frac{3}{4}, \frac{1}{16}, \frac{1}{16}, \frac{5}{16}, \frac{13}{16}$ & $6^{ B}_{-1/2}\boxtimes 2^{ B}_{ 1}$\\
$12^{ B}_{ 1/2}$ & $16$ & $1, 1, 1, 1, 1, 1, 1, 1,\zeta_{2}^{1},\zeta_{2}^{1},\zeta_{2}^{1},\zeta_{2}^{1}$ & $0, 0, \frac{1}{2}, \frac{1}{2}, \frac{1}{4}, \frac{1}{4}, \frac{3}{4}, \frac{3}{4}, \frac{1}{16}, \frac{1}{16}, \frac{5}{16}, \frac{13}{16}$ & $6^{ B}_{-1/2}\boxtimes 2^{ B}_{ 1}$\\
$12^{ B}_{ 1/2}$ & $16$ & $1, 1, 1, 1, 1, 1, 1, 1,\zeta_{2}^{1},\zeta_{2}^{1},\zeta_{2}^{1},\zeta_{2}^{1}$ & $0, 0, \frac{1}{2}, \frac{1}{2}, \frac{3}{8}, \frac{3}{8}, \frac{7}{8}, \frac{7}{8}, \frac{1}{16}, \frac{1}{16}, \frac{3}{16}, \frac{11}{16}$ & $4^{ B}_{-1}\boxtimes 3^{ B}_{ 3/2}$\\
$12^{ B}_{ 1/2}$ & $16$ & $1, 1, 1, 1, 1, 1, 1, 1,\zeta_{2}^{1},\zeta_{2}^{1},\zeta_{2}^{1},\zeta_{2}^{1}$ & $0, 0, \frac{1}{2}, \frac{1}{2}, \frac{3}{8}, \frac{3}{8}, \frac{7}{8}, \frac{7}{8}, \frac{1}{16}, \frac{1}{16}, \frac{3}{16}, \frac{11}{16}$ & $4^{ B}_{-1}\boxtimes 3^{ B}_{ 3/2}$\\
$12^{ B}_{ 3/2}$ & $16$ & $1, 1, 1, 1, 1, 1, 1, 1,\zeta_{2}^{1},\zeta_{2}^{1},\zeta_{2}^{1},\zeta_{2}^{1}$ & $0, 0, \frac{1}{2}, \frac{1}{2}, 0, 0, \frac{1}{2}, \frac{1}{2}, \frac{3}{16}, \frac{3}{16}, \frac{3}{16}, \frac{11}{16}$ & $4^{ B}_{ 0}\boxtimes 3^{ B}_{ 3/2}$\\
$12^{ B}_{ 3/2}$ & $16$ & $1, 1, 1, 1, 1, 1, 1, 1,\zeta_{2}^{1},\zeta_{2}^{1},\zeta_{2}^{1},\zeta_{2}^{1}$ & $0, 0, \frac{1}{2}, \frac{1}{2}, 0, 0, \frac{1}{2}, \frac{1}{2}, \frac{3}{16}, \frac{3}{16}, \frac{3}{16}, \frac{11}{16}$ & $4^{ B}_{ 0}\boxtimes 3^{ B}_{ 3/2}$\\
$12^{ B}_{ 3/2}$ & $16$ & $1, 1, 1, 1, 1, 1, 1, 1,\zeta_{2}^{1},\zeta_{2}^{1},\zeta_{2}^{1},\zeta_{2}^{1}$ & $0, 0, \frac{1}{2}, \frac{1}{2}, \frac{1}{8}, \frac{1}{8}, \frac{5}{8}, \frac{5}{8}, \frac{1}{16}, \frac{3}{16}, \frac{3}{16}, \frac{9}{16}$ & $4^{ B}_{ 1}\boxtimes 3^{ B}_{ 1/2}$\\
$12^{ B}_{ 3/2}$ & $16$ & $1, 1, 1, 1, 1, 1, 1, 1,\zeta_{2}^{1},\zeta_{2}^{1},\zeta_{2}^{1},\zeta_{2}^{1}$ & $0, 0, \frac{1}{2}, \frac{1}{2}, \frac{1}{8}, \frac{1}{8}, \frac{5}{8}, \frac{5}{8}, \frac{1}{16}, \frac{3}{16}, \frac{3}{16}, \frac{9}{16}$ & $4^{ B}_{ 1}\boxtimes 3^{ B}_{ 1/2}$\\
$12^{ B}_{ 3/2}$ & $16$ & $1, 1, 1, 1, 1, 1, 1, 1,\zeta_{2}^{1},\zeta_{2}^{1},\zeta_{2}^{1},\zeta_{2}^{1}$ & $0, 0, \frac{1}{2}, \frac{1}{2}, \frac{1}{4}, \frac{1}{4}, \frac{3}{4}, \frac{3}{4}, \frac{3}{16}, \frac{3}{16}, \frac{7}{16}, \frac{15}{16}$ & $6^{ B}_{ 1/2}\boxtimes 2^{ B}_{ 1}$\\
$12^{ B}_{ 3/2}$ & $16$ & $1, 1, 1, 1, 1, 1, 1, 1,\zeta_{2}^{1},\zeta_{2}^{1},\zeta_{2}^{1},\zeta_{2}^{1}$ & $0, 0, \frac{1}{2}, \frac{1}{2}, \frac{1}{4}, \frac{1}{4}, \frac{3}{4}, \frac{3}{4}, \frac{3}{16}, \frac{3}{16}, \frac{7}{16}, \frac{15}{16}$ & $6^{ B}_{ 1/2}\boxtimes 2^{ B}_{ 1}$\\
$12^{ B}_{ 3/2}$ & $16$ & $1, 1, 1, 1, 1, 1, 1, 1,\zeta_{2}^{1},\zeta_{2}^{1},\zeta_{2}^{1},\zeta_{2}^{1}$ & $0, 0, \frac{1}{2}, \frac{1}{2}, \frac{3}{8}, \frac{3}{8}, \frac{7}{8}, \frac{7}{8}, \frac{3}{16}, \frac{3}{16}, \frac{5}{16}, \frac{13}{16}$ & $4^{ B}_{-1}\boxtimes 3^{ B}_{ 5/2}$\\
$12^{ B}_{ 3/2}$ & $16$ & $1, 1, 1, 1, 1, 1, 1, 1,\zeta_{2}^{1},\zeta_{2}^{1},\zeta_{2}^{1},\zeta_{2}^{1}$ & $0, 0, \frac{1}{2}, \frac{1}{2}, \frac{3}{8}, \frac{3}{8}, \frac{7}{8}, \frac{7}{8}, \frac{3}{16}, \frac{3}{16}, \frac{5}{16}, \frac{13}{16}$ & $4^{ B}_{-1}\boxtimes 3^{ B}_{ 5/2}$\\
$12^{ B}_{ 5/2}$ & $16$ & $1, 1, 1, 1, 1, 1, 1, 1,\zeta_{2}^{1},\zeta_{2}^{1},\zeta_{2}^{1},\zeta_{2}^{1}$ & $0, 0, \frac{1}{2}, \frac{1}{2}, 0, 0, \frac{1}{2}, \frac{1}{2}, \frac{5}{16}, \frac{5}{16}, \frac{5}{16}, \frac{13}{16}$ & $4^{ B}_{ 0}\boxtimes 3^{ B}_{ 5/2}$\\
$12^{ B}_{ 5/2}$ & $16$ & $1, 1, 1, 1, 1, 1, 1, 1,\zeta_{2}^{1},\zeta_{2}^{1},\zeta_{2}^{1},\zeta_{2}^{1}$ & $0, 0, \frac{1}{2}, \frac{1}{2}, 0, 0, \frac{1}{2}, \frac{1}{2}, \frac{5}{16}, \frac{5}{16}, \frac{5}{16}, \frac{13}{16}$ & $4^{ B}_{ 0}\boxtimes 3^{ B}_{ 5/2}$\\
$12^{ B}_{ 5/2}$ & $16$ & $1, 1, 1, 1, 1, 1, 1, 1,\zeta_{2}^{1},\zeta_{2}^{1},\zeta_{2}^{1},\zeta_{2}^{1}$ & $0, 0, \frac{1}{2}, \frac{1}{2}, \frac{1}{8}, \frac{1}{8}, \frac{5}{8}, \frac{5}{8}, \frac{3}{16}, \frac{5}{16}, \frac{5}{16}, \frac{11}{16}$ & $4^{ B}_{ 1}\boxtimes 3^{ B}_{ 3/2}$\\
$12^{ B}_{ 5/2}$ & $16$ & $1, 1, 1, 1, 1, 1, 1, 1,\zeta_{2}^{1},\zeta_{2}^{1},\zeta_{2}^{1},\zeta_{2}^{1}$ & $0, 0, \frac{1}{2}, \frac{1}{2}, \frac{1}{8}, \frac{1}{8}, \frac{5}{8}, \frac{5}{8}, \frac{3}{16}, \frac{5}{16}, \frac{5}{16}, \frac{11}{16}$ & $4^{ B}_{ 1}\boxtimes 3^{ B}_{ 3/2}$\\
$12^{ B}_{ 5/2}$ & $16$ & $1, 1, 1, 1, 1, 1, 1, 1,\zeta_{2}^{1},\zeta_{2}^{1},\zeta_{2}^{1},\zeta_{2}^{1}$ & $0, 0, \frac{1}{2}, \frac{1}{2}, \frac{1}{4}, \frac{1}{4}, \frac{3}{4}, \frac{3}{4}, \frac{1}{16}, \frac{5}{16}, \frac{5}{16}, \frac{9}{16}$ & $6^{ B}_{ 3/2}\boxtimes 2^{ B}_{ 1}$\\
$12^{ B}_{ 5/2}$ & $16$ & $1, 1, 1, 1, 1, 1, 1, 1,\zeta_{2}^{1},\zeta_{2}^{1},\zeta_{2}^{1},\zeta_{2}^{1}$ & $0, 0, \frac{1}{2}, \frac{1}{2}, \frac{1}{4}, \frac{1}{4}, \frac{3}{4}, \frac{3}{4}, \frac{1}{16}, \frac{5}{16}, \frac{5}{16}, \frac{9}{16}$ & $6^{ B}_{ 3/2}\boxtimes 2^{ B}_{ 1}$\\
$12^{ B}_{ 5/2}$ & $16$ & $1, 1, 1, 1, 1, 1, 1, 1,\zeta_{2}^{1},\zeta_{2}^{1},\zeta_{2}^{1},\zeta_{2}^{1}$ & $0, 0, \frac{1}{2}, \frac{1}{2}, \frac{3}{8}, \frac{3}{8}, \frac{7}{8}, \frac{7}{8}, \frac{5}{16}, \frac{5}{16}, \frac{7}{16}, \frac{15}{16}$ & $4^{ B}_{-1}\boxtimes 3^{ B}_{ 7/2}$\\
$12^{ B}_{ 5/2}$ & $16$ & $1, 1, 1, 1, 1, 1, 1, 1,\zeta_{2}^{1},\zeta_{2}^{1},\zeta_{2}^{1},\zeta_{2}^{1}$ & $0, 0, \frac{1}{2}, \frac{1}{2}, \frac{3}{8}, \frac{3}{8}, \frac{7}{8}, \frac{7}{8}, \frac{5}{16}, \frac{5}{16}, \frac{7}{16}, \frac{15}{16}$ & $4^{ B}_{-1}\boxtimes 3^{ B}_{ 7/2}$\\
$12^{ B}_{ 7/2}$ & $16$ & $1, 1, 1, 1, 1, 1, 1, 1,\zeta_{2}^{1},\zeta_{2}^{1},\zeta_{2}^{1},\zeta_{2}^{1}$ & $0, 0, \frac{1}{2}, \frac{1}{2}, 0, 0, \frac{1}{2}, \frac{1}{2}, \frac{7}{16}, \frac{7}{16}, \frac{7}{16}, \frac{15}{16}$ & $4^{ B}_{ 0}\boxtimes 3^{ B}_{ 7/2}$\\
$12^{ B}_{ 7/2}$ & $16$ & $1, 1, 1, 1, 1, 1, 1, 1,\zeta_{2}^{1},\zeta_{2}^{1},\zeta_{2}^{1},\zeta_{2}^{1}$ & $0, 0, \frac{1}{2}, \frac{1}{2}, 0, 0, \frac{1}{2}, \frac{1}{2}, \frac{7}{16}, \frac{7}{16}, \frac{7}{16}, \frac{15}{16}$ & $4^{ B}_{ 0}\boxtimes 3^{ B}_{ 7/2}$\\
$12^{ B}_{ 7/2}$ & $16$ & $1, 1, 1, 1, 1, 1, 1, 1,\zeta_{2}^{1},\zeta_{2}^{1},\zeta_{2}^{1},\zeta_{2}^{1}$ & $0, 0, \frac{1}{2}, \frac{1}{2}, \frac{1}{8}, \frac{1}{8}, \frac{5}{8}, \frac{5}{8}, \frac{5}{16}, \frac{7}{16}, \frac{7}{16}, \frac{13}{16}$ & $4^{ B}_{ 1}\boxtimes 3^{ B}_{ 5/2}$\\
$12^{ B}_{ 7/2}$ & $16$ & $1, 1, 1, 1, 1, 1, 1, 1,\zeta_{2}^{1},\zeta_{2}^{1},\zeta_{2}^{1},\zeta_{2}^{1}$ & $0, 0, \frac{1}{2}, \frac{1}{2}, \frac{1}{8}, \frac{1}{8}, \frac{5}{8}, \frac{5}{8}, \frac{5}{16}, \frac{7}{16}, \frac{7}{16}, \frac{13}{16}$ & $4^{ B}_{ 1}\boxtimes 3^{ B}_{ 5/2}$\\
$12^{ B}_{ 7/2}$ & $16$ & $1, 1, 1, 1, 1, 1, 1, 1,\zeta_{2}^{1},\zeta_{2}^{1},\zeta_{2}^{1},\zeta_{2}^{1}$ & $0, 0, \frac{1}{2}, \frac{1}{2}, \frac{1}{4}, \frac{1}{4}, \frac{3}{4}, \frac{3}{4}, \frac{3}{16}, \frac{7}{16}, \frac{7}{16}, \frac{11}{16}$ & $6^{ B}_{ 5/2}\boxtimes 2^{ B}_{ 1}$\\
$12^{ B}_{ 7/2}$ & $16$ & $1, 1, 1, 1, 1, 1, 1, 1,\zeta_{2}^{1},\zeta_{2}^{1},\zeta_{2}^{1},\zeta_{2}^{1}$ & $0, 0, \frac{1}{2}, \frac{1}{2}, \frac{1}{4}, \frac{1}{4}, \frac{3}{4}, \frac{3}{4}, \frac{3}{16}, \frac{7}{16}, \frac{7}{16}, \frac{11}{16}$ & $6^{ B}_{ 5/2}\boxtimes 2^{ B}_{ 1}$\\
$12^{ B}_{ 7/2}$ & $16$ & $1, 1, 1, 1, 1, 1, 1, 1,\zeta_{2}^{1},\zeta_{2}^{1},\zeta_{2}^{1},\zeta_{2}^{1}$ & $0, 0, \frac{1}{2}, \frac{1}{2}, \frac{3}{8}, \frac{3}{8}, \frac{7}{8}, \frac{7}{8}, \frac{1}{16}, \frac{7}{16}, \frac{7}{16}, \frac{9}{16}$ & $4^{ B}_{ 3}\boxtimes 3^{ B}_{ 1/2}$\\
$12^{ B}_{ 7/2}$ & $16$ & $1, 1, 1, 1, 1, 1, 1, 1,\zeta_{2}^{1},\zeta_{2}^{1},\zeta_{2}^{1},\zeta_{2}^{1}$ & $0, 0, \frac{1}{2}, \frac{1}{2}, \frac{3}{8}, \frac{3}{8}, \frac{7}{8}, \frac{7}{8}, \frac{1}{16}, \frac{7}{16}, \frac{7}{16}, \frac{9}{16}$ & $4^{ B}_{ 3}\boxtimes 3^{ B}_{ 1/2}$\\
\hline
\end{tabular}
\end{table*}

\section{Summary}
\label{sum}

GQLs contain both topologically ordered states and SPT states.  In this paper,
we present a theory that classify GQLs in 2+1D for bosonic/fermionic systems
with symmetry.

We propose that the possible non-abelian statistics (or sets of bulk
quasiparticles excitations) in 2+1D GQLs are classified by 
$\mce{\cE}$, where $\cE=\Rp(G)$ or $\sRp(G^f)$ describing the symmetry in
bosonic or fermionic systems.  However, $\mce{\cE}$'s fail to
classify GQLs, since different GQL phases can have identical non-abelian
statistics, which correspond to identical $\mce{\cE}$.  

To fix this problem, we introduce the notion of modular extensions for a
$\mce{\cE}$.  We propose to use the triple $(\cC,\cM,c)$ to
classify 2+1D GQLs with symmetry $G$ (for boson) or $G^f$ (for fermion).  Here
$\cC$ is a $\mce{\cE}$ with $\cE=\Rp(G)$ or $\sRp(G^f)$,  $\cM$ is a modular
extension of $\cC$ and $c$ is the chiral central charge of the edge state.  We
show that the modular extensions of a $\mce{\cE}$ has a one-to-one
correspondence  with the modular extensions of  $\cE$.  So the number of the
modular extensions is solely determined by the symmetry $\cE$.  Also, the $c=0$
modular extensions of a $\cE$ ($\cE=\Rp(G)$ or $\sRp(G^f)$) classify the 2+1D
SPT states for bosons or fermions with symmetry $G$ or $G^f$.

Although the above result has a nice mathematical structure, it is hard to
implement numerically to produce a table of GQLs.  To fix this problem, we
propose a different description of 2+1D GQLs.  We propose to use the data $(
\tilde N^{ab}_c,\tilde s_a; N^{ij}_k,s_i; \cN^{IJ}_K,\cS_I;c)$,  up to some
permutations of the indices, to describe 2+1D GQLs with symmetry $G$
(for boson) or $G^f$ (for fermion), with a restriction that the symmetry group
$G$ can be fully characterized by the fusion ring of its irreducible
representations (for example, for simple groups or abelian groups).  Here the
data $(\tilde N^{ab}_c,\tilde s_a)$ describe the symmetry and the data
$(N^{ij}_k,s_i)$ describes fusion and the spins of the bulk particles in the
GQL.  The modular extensions are obtained by ``gauging'' the symmetry $G$ or
$G^f$.  The data $(\cN^{IJ}_K,\cS_I)$  describes fusion and the spins of the
bulk particles in the ``gauged'' theory.  Last, $c$ is the chiral central
charge of the edge state.

In this paper (see Appendix \ref{cnds}) and in \Ref{W150605768}, we list the
necessary and the sufficient conditions on the data $(\tilde N^{ab}_c,\tilde
s_a; N^{ij}_k,s_i; \cN^{IJ}_K,\cS_I;c)$, which allow us to obtain a
list of GQLs.  However, in this paper, we did not give the list
of GQLs directly.  We first give a list of $(\tilde N^{ab}_c,\tilde s_a;
N^{ij}_k,s_i)$, which is an imperfect list of $\mce{\cE}$'s.  We then compute
the modular extensions $(\cN^{IJ}_K,\cS_I;c)$ for each entry $(\tilde
N^{ab}_c,\tilde s_a; N^{ij}_k,s_i)$, which allows us to obtain a perfect list
of GQLs (for certain symmetry groups).  As a special case, we calculated the
bosonic/fermionic SPT states for some groups in 2+1D.

In \Ref{LW160205936}, we will give a more mathematical description of our
theory.  Certainly we hope to generalize the above framework to higher
dimensions.  We also hope to develop more efficient numerical codes to obtain
bigger tables of GQLs.

\bigskip

\noindent {\bf Acknowledgement}: 
We like to thank Pavel Etingof, Dmitri Nikshych, Chenjie Wang, and Zhenghan
Wang for many helpful discussions.  This research is supported by NSF Grant No.
DMR-1506475, and NSFC 11274192. It is also supported by the John Templeton
Foundation No.  39901.  Research at Perimeter Institute is supported by the
Government of Canada through Industry Canada and by the Province of Ontario
through the Ministry of Research.  LK is supported by the Center of
Mathematical Sciences and Applications at Harvard University.

\def\arraystretch{1.25} \setlength\tabcolsep{3pt}
\begin{table*}[t] 
\caption{
The second 32 modular extensions of $\sRp(Z_2\times Z_2^f)$ with $N =12$.
} 
\label{mextZ2Z2f12b} 
\centering
\begin{tabular}{ |c|c|l|l|l| } 
\hline 
$N^{|\Th|}_{c}$ & $D^2$ & $d_1,d_2,\cdots$ & $s_1,s_2,\cdots$ & comment \\
\hline 
$4^{ 0}_{0}$ & $4$ & $1, 1, 1, 1$ & $0, 0, \frac{1}{2}, \frac{1}{2}$ & $\sRp(Z_2\times Z_2^f)$ \\
\hline
$12^{ B}_{-7/2}$ & $16$ & $1, 1, 1, 1, 1, 1, 1, 1,\zeta_{2}^{1},\zeta_{2}^{1},\zeta_{2}^{1},\zeta_{2}^{1}$ & $0, 0, \frac{1}{2}, \frac{1}{2}, 0, 0, \frac{1}{2}, \frac{1}{2}, \frac{1}{16}, \frac{9}{16}, \frac{9}{16}, \frac{9}{16}$ & $4^{ B}_{ 4}\boxtimes 3^{ B}_{ 1/2}$\\
$12^{ B}_{-7/2}$ & $16$ & $1, 1, 1, 1, 1, 1, 1, 1,\zeta_{2}^{1},\zeta_{2}^{1},\zeta_{2}^{1},\zeta_{2}^{1}$ & $0, 0, \frac{1}{2}, \frac{1}{2}, 0, 0, \frac{1}{2}, \frac{1}{2}, \frac{1}{16}, \frac{9}{16}, \frac{9}{16}, \frac{9}{16}$ & $4^{ B}_{ 4}\boxtimes 3^{ B}_{ 1/2}$\\
$12^{ B}_{-7/2}$ & $16$ & $1, 1, 1, 1, 1, 1, 1, 1,\zeta_{2}^{1},\zeta_{2}^{1},\zeta_{2}^{1},\zeta_{2}^{1}$ & $0, 0, \frac{1}{2}, \frac{1}{2}, \frac{1}{8}, \frac{1}{8}, \frac{5}{8}, \frac{5}{8}, \frac{7}{16}, \frac{9}{16}, \frac{9}{16}, \frac{15}{16}$ & $4^{ B}_{ 1}\boxtimes 3^{ B}_{ 7/2}$\\
$12^{ B}_{-7/2}$ & $16$ & $1, 1, 1, 1, 1, 1, 1, 1,\zeta_{2}^{1},\zeta_{2}^{1},\zeta_{2}^{1},\zeta_{2}^{1}$ & $0, 0, \frac{1}{2}, \frac{1}{2}, \frac{1}{8}, \frac{1}{8}, \frac{5}{8}, \frac{5}{8}, \frac{7}{16}, \frac{9}{16}, \frac{9}{16}, \frac{15}{16}$ & $4^{ B}_{ 1}\boxtimes 3^{ B}_{ 7/2}$\\
$12^{ B}_{-7/2}$ & $16$ & $1, 1, 1, 1, 1, 1, 1, 1,\zeta_{2}^{1},\zeta_{2}^{1},\zeta_{2}^{1},\zeta_{2}^{1}$ & $0, 0, \frac{1}{2}, \frac{1}{2}, \frac{1}{4}, \frac{1}{4}, \frac{3}{4}, \frac{3}{4}, \frac{5}{16}, \frac{9}{16}, \frac{9}{16}, \frac{13}{16}$ & $6^{ B}_{ 7/2}\boxtimes 2^{ B}_{ 1}$\\
$12^{ B}_{-7/2}$ & $16$ & $1, 1, 1, 1, 1, 1, 1, 1,\zeta_{2}^{1},\zeta_{2}^{1},\zeta_{2}^{1},\zeta_{2}^{1}$ & $0, 0, \frac{1}{2}, \frac{1}{2}, \frac{1}{4}, \frac{1}{4}, \frac{3}{4}, \frac{3}{4}, \frac{5}{16}, \frac{9}{16}, \frac{9}{16}, \frac{13}{16}$ & $6^{ B}_{ 7/2}\boxtimes 2^{ B}_{ 1}$\\
$12^{ B}_{-7/2}$ & $16$ & $1, 1, 1, 1, 1, 1, 1, 1,\zeta_{2}^{1},\zeta_{2}^{1},\zeta_{2}^{1},\zeta_{2}^{1}$ & $0, 0, \frac{1}{2}, \frac{1}{2}, \frac{3}{8}, \frac{3}{8}, \frac{7}{8}, \frac{7}{8}, \frac{3}{16}, \frac{9}{16}, \frac{9}{16}, \frac{11}{16}$ & $4^{ B}_{ 3}\boxtimes 3^{ B}_{ 3/2}$\\
$12^{ B}_{-7/2}$ & $16$ & $1, 1, 1, 1, 1, 1, 1, 1,\zeta_{2}^{1},\zeta_{2}^{1},\zeta_{2}^{1},\zeta_{2}^{1}$ & $0, 0, \frac{1}{2}, \frac{1}{2}, \frac{3}{8}, \frac{3}{8}, \frac{7}{8}, \frac{7}{8}, \frac{3}{16}, \frac{9}{16}, \frac{9}{16}, \frac{11}{16}$ & $4^{ B}_{ 3}\boxtimes 3^{ B}_{ 3/2}$\\
$12^{ B}_{-5/2}$ & $16$ & $1, 1, 1, 1, 1, 1, 1, 1,\zeta_{2}^{1},\zeta_{2}^{1},\zeta_{2}^{1},\zeta_{2}^{1}$ & $0, 0, \frac{1}{2}, \frac{1}{2}, 0, 0, \frac{1}{2}, \frac{1}{2}, \frac{3}{16}, \frac{11}{16}, \frac{11}{16}, \frac{11}{16}$ & $4^{ B}_{ 4}\boxtimes 3^{ B}_{ 3/2}$\\
$12^{ B}_{-5/2}$ & $16$ & $1, 1, 1, 1, 1, 1, 1, 1,\zeta_{2}^{1},\zeta_{2}^{1},\zeta_{2}^{1},\zeta_{2}^{1}$ & $0, 0, \frac{1}{2}, \frac{1}{2}, 0, 0, \frac{1}{2}, \frac{1}{2}, \frac{3}{16}, \frac{11}{16}, \frac{11}{16}, \frac{11}{16}$ & $4^{ B}_{ 4}\boxtimes 3^{ B}_{ 3/2}$\\
$12^{ B}_{-5/2}$ & $16$ & $1, 1, 1, 1, 1, 1, 1, 1,\zeta_{2}^{1},\zeta_{2}^{1},\zeta_{2}^{1},\zeta_{2}^{1}$ & $0, 0, \frac{1}{2}, \frac{1}{2}, \frac{1}{8}, \frac{1}{8}, \frac{5}{8}, \frac{5}{8}, \frac{1}{16}, \frac{9}{16}, \frac{11}{16}, \frac{11}{16}$ & $4^{ B}_{-3}\boxtimes 3^{ B}_{ 1/2}$\\
$12^{ B}_{-5/2}$ & $16$ & $1, 1, 1, 1, 1, 1, 1, 1,\zeta_{2}^{1},\zeta_{2}^{1},\zeta_{2}^{1},\zeta_{2}^{1}$ & $0, 0, \frac{1}{2}, \frac{1}{2}, \frac{1}{8}, \frac{1}{8}, \frac{5}{8}, \frac{5}{8}, \frac{1}{16}, \frac{9}{16}, \frac{11}{16}, \frac{11}{16}$ & $4^{ B}_{-3}\boxtimes 3^{ B}_{ 1/2}$\\
$12^{ B}_{-5/2}$ & $16$ & $1, 1, 1, 1, 1, 1, 1, 1,\zeta_{2}^{1},\zeta_{2}^{1},\zeta_{2}^{1},\zeta_{2}^{1}$ & $0, 0, \frac{1}{2}, \frac{1}{2}, \frac{1}{4}, \frac{1}{4}, \frac{3}{4}, \frac{3}{4}, \frac{7}{16}, \frac{11}{16}, \frac{11}{16}, \frac{15}{16}$ & $6^{ B}_{-7/2}\boxtimes 2^{ B}_{ 1}$\\
$12^{ B}_{-5/2}$ & $16$ & $1, 1, 1, 1, 1, 1, 1, 1,\zeta_{2}^{1},\zeta_{2}^{1},\zeta_{2}^{1},\zeta_{2}^{1}$ & $0, 0, \frac{1}{2}, \frac{1}{2}, \frac{1}{4}, \frac{1}{4}, \frac{3}{4}, \frac{3}{4}, \frac{7}{16}, \frac{11}{16}, \frac{11}{16}, \frac{15}{16}$ & $6^{ B}_{-7/2}\boxtimes 2^{ B}_{ 1}$\\
$12^{ B}_{-5/2}$ & $16$ & $1, 1, 1, 1, 1, 1, 1, 1,\zeta_{2}^{1},\zeta_{2}^{1},\zeta_{2}^{1},\zeta_{2}^{1}$ & $0, 0, \frac{1}{2}, \frac{1}{2}, \frac{3}{8}, \frac{3}{8}, \frac{7}{8}, \frac{7}{8}, \frac{5}{16}, \frac{11}{16}, \frac{11}{16}, \frac{13}{16}$ & $4^{ B}_{ 3}\boxtimes 3^{ B}_{ 5/2}$\\
$12^{ B}_{-5/2}$ & $16$ & $1, 1, 1, 1, 1, 1, 1, 1,\zeta_{2}^{1},\zeta_{2}^{1},\zeta_{2}^{1},\zeta_{2}^{1}$ & $0, 0, \frac{1}{2}, \frac{1}{2}, \frac{3}{8}, \frac{3}{8}, \frac{7}{8}, \frac{7}{8}, \frac{5}{16}, \frac{11}{16}, \frac{11}{16}, \frac{13}{16}$ & $4^{ B}_{ 3}\boxtimes 3^{ B}_{ 5/2}$\\
$12^{ B}_{-3/2}$ & $16$ & $1, 1, 1, 1, 1, 1, 1, 1,\zeta_{2}^{1},\zeta_{2}^{1},\zeta_{2}^{1},\zeta_{2}^{1}$ & $0, 0, \frac{1}{2}, \frac{1}{2}, 0, 0, \frac{1}{2}, \frac{1}{2}, \frac{5}{16}, \frac{13}{16}, \frac{13}{16}, \frac{13}{16}$ & $4^{ B}_{ 4}\boxtimes 3^{ B}_{ 5/2}$\\
$12^{ B}_{-3/2}$ & $16$ & $1, 1, 1, 1, 1, 1, 1, 1,\zeta_{2}^{1},\zeta_{2}^{1},\zeta_{2}^{1},\zeta_{2}^{1}$ & $0, 0, \frac{1}{2}, \frac{1}{2}, 0, 0, \frac{1}{2}, \frac{1}{2}, \frac{5}{16}, \frac{13}{16}, \frac{13}{16}, \frac{13}{16}$ & $4^{ B}_{ 4}\boxtimes 3^{ B}_{ 5/2}$\\
$12^{ B}_{-3/2}$ & $16$ & $1, 1, 1, 1, 1, 1, 1, 1,\zeta_{2}^{1},\zeta_{2}^{1},\zeta_{2}^{1},\zeta_{2}^{1}$ & $0, 0, \frac{1}{2}, \frac{1}{2}, \frac{1}{8}, \frac{1}{8}, \frac{5}{8}, \frac{5}{8}, \frac{3}{16}, \frac{11}{16}, \frac{13}{16}, \frac{13}{16}$ & $4^{ B}_{-3}\boxtimes 3^{ B}_{ 3/2}$\\
$12^{ B}_{-3/2}$ & $16$ & $1, 1, 1, 1, 1, 1, 1, 1,\zeta_{2}^{1},\zeta_{2}^{1},\zeta_{2}^{1},\zeta_{2}^{1}$ & $0, 0, \frac{1}{2}, \frac{1}{2}, \frac{1}{8}, \frac{1}{8}, \frac{5}{8}, \frac{5}{8}, \frac{3}{16}, \frac{11}{16}, \frac{13}{16}, \frac{13}{16}$ & $4^{ B}_{-3}\boxtimes 3^{ B}_{ 3/2}$\\
$12^{ B}_{-3/2}$ & $16$ & $1, 1, 1, 1, 1, 1, 1, 1,\zeta_{2}^{1},\zeta_{2}^{1},\zeta_{2}^{1},\zeta_{2}^{1}$ & $0, 0, \frac{1}{2}, \frac{1}{2}, \frac{1}{4}, \frac{1}{4}, \frac{3}{4}, \frac{3}{4}, \frac{1}{16}, \frac{9}{16}, \frac{13}{16}, \frac{13}{16}$ & $6^{ B}_{-5/2}\boxtimes 2^{ B}_{ 1}$\\
$12^{ B}_{-3/2}$ & $16$ & $1, 1, 1, 1, 1, 1, 1, 1,\zeta_{2}^{1},\zeta_{2}^{1},\zeta_{2}^{1},\zeta_{2}^{1}$ & $0, 0, \frac{1}{2}, \frac{1}{2}, \frac{1}{4}, \frac{1}{4}, \frac{3}{4}, \frac{3}{4}, \frac{1}{16}, \frac{9}{16}, \frac{13}{16}, \frac{13}{16}$ & $6^{ B}_{-5/2}\boxtimes 2^{ B}_{ 1}$\\
$12^{ B}_{-3/2}$ & $16$ & $1, 1, 1, 1, 1, 1, 1, 1,\zeta_{2}^{1},\zeta_{2}^{1},\zeta_{2}^{1},\zeta_{2}^{1}$ & $0, 0, \frac{1}{2}, \frac{1}{2}, \frac{3}{8}, \frac{3}{8}, \frac{7}{8}, \frac{7}{8}, \frac{7}{16}, \frac{13}{16}, \frac{13}{16}, \frac{15}{16}$ & $4^{ B}_{ 3}\boxtimes 3^{ B}_{ 7/2}$\\
$12^{ B}_{-3/2}$ & $16$ & $1, 1, 1, 1, 1, 1, 1, 1,\zeta_{2}^{1},\zeta_{2}^{1},\zeta_{2}^{1},\zeta_{2}^{1}$ & $0, 0, \frac{1}{2}, \frac{1}{2}, \frac{3}{8}, \frac{3}{8}, \frac{7}{8}, \frac{7}{8}, \frac{7}{16}, \frac{13}{16}, \frac{13}{16}, \frac{15}{16}$ & $4^{ B}_{ 3}\boxtimes 3^{ B}_{ 7/2}$\\
$12^{ B}_{-1/2}$ & $16$ & $1, 1, 1, 1, 1, 1, 1, 1,\zeta_{2}^{1},\zeta_{2}^{1},\zeta_{2}^{1},\zeta_{2}^{1}$ & $0, 0, \frac{1}{2}, \frac{1}{2}, 0, 0, \frac{1}{2}, \frac{1}{2}, \frac{7}{16}, \frac{15}{16}, \frac{15}{16}, \frac{15}{16}$ & $4^{ B}_{ 4}\boxtimes 3^{ B}_{ 7/2}$\\
$12^{ B}_{-1/2}$ & $16$ & $1, 1, 1, 1, 1, 1, 1, 1,\zeta_{2}^{1},\zeta_{2}^{1},\zeta_{2}^{1},\zeta_{2}^{1}$ & $0, 0, \frac{1}{2}, \frac{1}{2}, 0, 0, \frac{1}{2}, \frac{1}{2}, \frac{7}{16}, \frac{15}{16}, \frac{15}{16}, \frac{15}{16}$ & $4^{ B}_{ 4}\boxtimes 3^{ B}_{ 7/2}$\\
$12^{ B}_{-1/2}$ & $16$ & $1, 1, 1, 1, 1, 1, 1, 1,\zeta_{2}^{1},\zeta_{2}^{1},\zeta_{2}^{1},\zeta_{2}^{1}$ & $0, 0, \frac{1}{2}, \frac{1}{2}, \frac{1}{8}, \frac{1}{8}, \frac{5}{8}, \frac{5}{8}, \frac{5}{16}, \frac{13}{16}, \frac{15}{16}, \frac{15}{16}$ & $4^{ B}_{-3}\boxtimes 3^{ B}_{ 5/2}$\\
$12^{ B}_{-1/2}$ & $16$ & $1, 1, 1, 1, 1, 1, 1, 1,\zeta_{2}^{1},\zeta_{2}^{1},\zeta_{2}^{1},\zeta_{2}^{1}$ & $0, 0, \frac{1}{2}, \frac{1}{2}, \frac{1}{8}, \frac{1}{8}, \frac{5}{8}, \frac{5}{8}, \frac{5}{16}, \frac{13}{16}, \frac{15}{16}, \frac{15}{16}$ & $4^{ B}_{-3}\boxtimes 3^{ B}_{ 5/2}$\\
$12^{ B}_{-1/2}$ & $16$ & $1, 1, 1, 1, 1, 1, 1, 1,\zeta_{2}^{1},\zeta_{2}^{1},\zeta_{2}^{1},\zeta_{2}^{1}$ & $0, 0, \frac{1}{2}, \frac{1}{2}, \frac{1}{4}, \frac{1}{4}, \frac{3}{4}, \frac{3}{4}, \frac{3}{16}, \frac{11}{16}, \frac{15}{16}, \frac{15}{16}$ & $6^{ B}_{-3/2}\boxtimes 2^{ B}_{ 1}$\\
$12^{ B}_{-1/2}$ & $16$ & $1, 1, 1, 1, 1, 1, 1, 1,\zeta_{2}^{1},\zeta_{2}^{1},\zeta_{2}^{1},\zeta_{2}^{1}$ & $0, 0, \frac{1}{2}, \frac{1}{2}, \frac{1}{4}, \frac{1}{4}, \frac{3}{4}, \frac{3}{4}, \frac{3}{16}, \frac{11}{16}, \frac{15}{16}, \frac{15}{16}$ & $6^{ B}_{-3/2}\boxtimes 2^{ B}_{ 1}$\\
$12^{ B}_{-1/2}$ & $16$ & $1, 1, 1, 1, 1, 1, 1, 1,\zeta_{2}^{1},\zeta_{2}^{1},\zeta_{2}^{1},\zeta_{2}^{1}$ & $0, 0, \frac{1}{2}, \frac{1}{2}, \frac{3}{8}, \frac{3}{8}, \frac{7}{8}, \frac{7}{8}, \frac{1}{16}, \frac{9}{16}, \frac{15}{16}, \frac{15}{16}$ & $4^{ B}_{-1}\boxtimes 3^{ B}_{ 1/2}$\\
$12^{ B}_{-1/2}$ & $16$ & $1, 1, 1, 1, 1, 1, 1, 1,\zeta_{2}^{1},\zeta_{2}^{1},\zeta_{2}^{1},\zeta_{2}^{1}$ & $0, 0, \frac{1}{2}, \frac{1}{2}, \frac{3}{8}, \frac{3}{8}, \frac{7}{8}, \frac{7}{8}, \frac{1}{16}, \frac{9}{16}, \frac{15}{16}, \frac{15}{16}$ & $4^{ B}_{-1}\boxtimes 3^{ B}_{ 1/2}$\\
\hline
\end{tabular}
\end{table*} 

\def\arraystretch{1.25} \setlength\tabcolsep{3pt}
\begin{table*}[t] 
\caption{
All the 32 modular extensions of $\sRp(Z_2\times Z_2^f)$ with $N =16$.
} 
\label{mextZ2Z2f16} 
\centering
\begin{tabular}{ |c|c|l|l|l| } 
\hline 
$N^{|\Th|}_{c}$ & $D^2$ & $d_1,d_2,\cdots$ & $s_1,s_2,\cdots$ & comment \\
\hline 
$4^{ 0}_{0}$ & $4$ & $1, 1, 1, 1$ & $0, 0, \frac{1}{2}, \frac{1}{2}$ & $\sRp(Z_2\times Z_2^f)$ \\
\hline
$16^{ B}_{ 0}$ & $16$ & $1, 1, 1, 1, 1, 1, 1, 1, 1, 1, 1, 1, 1, 1, 1, 1$ & $0, 0, \frac{1}{2}, \frac{1}{2}, 0, 0, 0, 0, 0, 0, 0, 0, \frac{1}{2}, \frac{1}{2}, \frac{1}{2}, \frac{1}{2}$ & $4^{ B}_{ 0}\boxtimes 4^{ B}_{ 0}$\\
$16^{ B}_{ 0}$ & $16$ & $1, 1, 1, 1, 1, 1, 1, 1, 1, 1, 1, 1, 1, 1, 1, 1$ & $0, 0, \frac{1}{2}, \frac{1}{2}, 0, 0, 0, 0, \frac{1}{8}, \frac{1}{8}, \frac{3}{8}, \frac{3}{8}, \frac{5}{8}, \frac{5}{8}, \frac{7}{8}, \frac{7}{8}$ & $4^{ B}_{-1}\boxtimes 4^{ B}_{ 1}$\\
$16^{ B}_{ 0}$ & $16$ & $1, 1, 1, 1, 1, 1, 1, 1, 1, 1, 1, 1, 1, 1, 1, 1$ & $0, 0, \frac{1}{2}, \frac{1}{2}, 0, 0, 0, 0, \frac{1}{8}, \frac{1}{8}, \frac{3}{8}, \frac{3}{8}, \frac{5}{8}, \frac{5}{8}, \frac{7}{8}, \frac{7}{8}$ & $4^{ B}_{-1}\boxtimes 4^{ B}_{ 1}$\\
$16^{ B}_{ 0}$ & $16$ & $1, 1, 1, 1, 1, 1, 1, 1, 1, 1, 1, 1, 1, 1, 1, 1$ & $0, 0, \frac{1}{2}, \frac{1}{2}, 0, 0, 0, 0, \frac{1}{4}, \frac{1}{4}, \frac{1}{4}, \frac{1}{4}, \frac{3}{4}, \frac{3}{4}, \frac{3}{4}, \frac{3}{4}$ & $8^{ B}_{-1}\boxtimes 2^{ B}_{ 1}$\\
\hline
$16^{ B}_{ 1}$ & $16$ & $1, 1, 1, 1, 1, 1, 1, 1, 1, 1, 1, 1, 1, 1, 1, 1$ & $0, 0, \frac{1}{2}, \frac{1}{2}, 0, 0, \frac{1}{8}, \frac{1}{8}, \frac{1}{8}, \frac{1}{8}, \frac{1}{8}, \frac{1}{8}, \frac{1}{2}, \frac{1}{2}, \frac{5}{8}, \frac{5}{8}$ & $4^{ B}_{ 1}\boxtimes 4^{ B}_{ 0}$\\
$16^{ B}_{ 1}$ & $16$ & $1, 1, 1, 1, 1, 1, 1, 1, 1, 1, 1, 1, 1, 1, 1, 1$ & $0, 0, \frac{1}{2}, \frac{1}{2}, 0, 0, \frac{1}{8}, \frac{1}{8}, \frac{1}{8}, \frac{1}{8}, \frac{1}{8}, \frac{1}{8}, \frac{1}{2}, \frac{1}{2}, \frac{5}{8}, \frac{5}{8}$ & $4^{ B}_{ 1}\boxtimes 4^{ B}_{ 0}$\\
$16^{ B}_{ 1}$ & $16$ & $1, 1, 1, 1, 1, 1, 1, 1, 1, 1, 1, 1, 1, 1, 1, 1$ & $0, 0, \frac{1}{2}, \frac{1}{2}, \frac{1}{8}, \frac{1}{8}, \frac{1}{8}, \frac{1}{8}, \frac{1}{4}, \frac{1}{4}, \frac{3}{8}, \frac{3}{8}, \frac{3}{4}, \frac{3}{4}, \frac{7}{8}, \frac{7}{8}$ & $8^{ B}_{ 0}\boxtimes 2^{ B}_{ 1}$\\
$16^{ B}_{ 1}$ & $16$ & $1, 1, 1, 1, 1, 1, 1, 1, 1, 1, 1, 1, 1, 1, 1, 1$ & $0, 0, \frac{1}{2}, \frac{1}{2}, \frac{1}{8}, \frac{1}{8}, \frac{1}{8}, \frac{1}{8}, \frac{1}{4}, \frac{1}{4}, \frac{3}{8}, \frac{3}{8}, \frac{3}{4}, \frac{3}{4}, \frac{7}{8}, \frac{7}{8}$ & $8^{ B}_{ 0}\boxtimes 2^{ B}_{ 1}$\\
$16^{ B}_{ 2}$ & $16$ & $1, 1, 1, 1, 1, 1, 1, 1, 1, 1, 1, 1, 1, 1, 1, 1$ & $0, 0, \frac{1}{2}, \frac{1}{2}, 0, 0, \frac{1}{4}, \frac{1}{4}, \frac{1}{4}, \frac{1}{4}, \frac{1}{4}, \frac{1}{4}, \frac{1}{2}, \frac{1}{2}, \frac{3}{4}, \frac{3}{4}$ & $8^{ B}_{ 1}\boxtimes 2^{ B}_{ 1}$\\
$16^{ B}_{ 2}$ & $16$ & $1, 1, 1, 1, 1, 1, 1, 1, 1, 1, 1, 1, 1, 1, 1, 1$ & $0, 0, \frac{1}{2}, \frac{1}{2}, 0, 0, \frac{1}{4}, \frac{1}{4}, \frac{1}{4}, \frac{1}{4}, \frac{1}{4}, \frac{1}{4}, \frac{1}{2}, \frac{1}{2}, \frac{3}{4}, \frac{3}{4}$ & $8^{ B}_{ 1}\boxtimes 2^{ B}_{ 1}$\\
$16^{ B}_{ 2}$ & $16$ & $1, 1, 1, 1, 1, 1, 1, 1, 1, 1, 1, 1, 1, 1, 1, 1$ & $0, 0, \frac{1}{2}, \frac{1}{2}, \frac{1}{8}, \frac{1}{8}, \frac{1}{8}, \frac{1}{8}, \frac{1}{4}, \frac{1}{4}, \frac{1}{4}, \frac{1}{4}, \frac{5}{8}, \frac{5}{8}, \frac{5}{8}, \frac{5}{8}$ & $4^{ B}_{ 1}\boxtimes 4^{ B}_{ 1}$\\
$16^{ B}_{ 2}$ & $16$ & $1, 1, 1, 1, 1, 1, 1, 1, 1, 1, 1, 1, 1, 1, 1, 1$ & $0, 0, \frac{1}{2}, \frac{1}{2}, \frac{1}{4}, \frac{1}{4}, \frac{1}{4}, \frac{1}{4}, \frac{3}{8}, \frac{3}{8}, \frac{3}{8}, \frac{3}{8}, \frac{7}{8}, \frac{7}{8}, \frac{7}{8}, \frac{7}{8}$ & $4^{ B}_{-1}\boxtimes 4^{ B}_{ 3}$\\
$16^{ B}_{ 3}$ & $16$ & $1, 1, 1, 1, 1, 1, 1, 1, 1, 1, 1, 1, 1, 1, 1, 1$ & $0, 0, \frac{1}{2}, \frac{1}{2}, 0, 0, \frac{3}{8}, \frac{3}{8}, \frac{3}{8}, \frac{3}{8}, \frac{3}{8}, \frac{3}{8}, \frac{1}{2}, \frac{1}{2}, \frac{7}{8}, \frac{7}{8}$ & $4^{ B}_{ 3}\boxtimes 4^{ B}_{ 0}$\\
$16^{ B}_{ 3}$ & $16$ & $1, 1, 1, 1, 1, 1, 1, 1, 1, 1, 1, 1, 1, 1, 1, 1$ & $0, 0, \frac{1}{2}, \frac{1}{2}, 0, 0, \frac{3}{8}, \frac{3}{8}, \frac{3}{8}, \frac{3}{8}, \frac{3}{8}, \frac{3}{8}, \frac{1}{2}, \frac{1}{2}, \frac{7}{8}, \frac{7}{8}$ & $4^{ B}_{ 3}\boxtimes 4^{ B}_{ 0}$\\
$16^{ B}_{ 3}$ & $16$ & $1, 1, 1, 1, 1, 1, 1, 1, 1, 1, 1, 1, 1, 1, 1, 1$ & $0, 0, \frac{1}{2}, \frac{1}{2}, \frac{1}{8}, \frac{1}{8}, \frac{1}{4}, \frac{1}{4}, \frac{3}{8}, \frac{3}{8}, \frac{3}{8}, \frac{3}{8}, \frac{5}{8}, \frac{5}{8}, \frac{3}{4}, \frac{3}{4}$ & $8^{ B}_{ 2}\boxtimes 2^{ B}_{ 1}$\\
$16^{ B}_{ 3}$ & $16$ & $1, 1, 1, 1, 1, 1, 1, 1, 1, 1, 1, 1, 1, 1, 1, 1$ & $0, 0, \frac{1}{2}, \frac{1}{2}, \frac{1}{8}, \frac{1}{8}, \frac{1}{4}, \frac{1}{4}, \frac{3}{8}, \frac{3}{8}, \frac{3}{8}, \frac{3}{8}, \frac{5}{8}, \frac{5}{8}, \frac{3}{4}, \frac{3}{4}$ & $8^{ B}_{ 2}\boxtimes 2^{ B}_{ 1}$\\
$16^{ B}_{ 4}$ & $16$ & $1, 1, 1, 1, 1, 1, 1, 1, 1, 1, 1, 1, 1, 1, 1, 1$ & $0, 0, \frac{1}{2}, \frac{1}{2}, 0, 0, 0, 0, \frac{1}{2}, \frac{1}{2}, \frac{1}{2}, \frac{1}{2}, \frac{1}{2}, \frac{1}{2}, \frac{1}{2}, \frac{1}{2}$ & $4^{ B}_{ 4}\boxtimes 4^{ B}_{ 0}$\\
$16^{ B}_{ 4}$ & $16$ & $1, 1, 1, 1, 1, 1, 1, 1, 1, 1, 1, 1, 1, 1, 1, 1$ & $0, 0, \frac{1}{2}, \frac{1}{2}, \frac{1}{8}, \frac{1}{8}, \frac{3}{8}, \frac{3}{8}, \frac{1}{2}, \frac{1}{2}, \frac{1}{2}, \frac{1}{2}, \frac{5}{8}, \frac{5}{8}, \frac{7}{8}, \frac{7}{8}$ & $4^{ B}_{ 3}\boxtimes 4^{ B}_{ 1}$\\
$16^{ B}_{ 4}$ & $16$ & $1, 1, 1, 1, 1, 1, 1, 1, 1, 1, 1, 1, 1, 1, 1, 1$ & $0, 0, \frac{1}{2}, \frac{1}{2}, \frac{1}{8}, \frac{1}{8}, \frac{3}{8}, \frac{3}{8}, \frac{1}{2}, \frac{1}{2}, \frac{1}{2}, \frac{1}{2}, \frac{5}{8}, \frac{5}{8}, \frac{7}{8}, \frac{7}{8}$ & $4^{ B}_{ 3}\boxtimes 4^{ B}_{ 1}$\\
$16^{ B}_{ 4}$ & $16$ & $1, 1, 1, 1, 1, 1, 1, 1, 1, 1, 1, 1, 1, 1, 1, 1$ & $0, 0, \frac{1}{2}, \frac{1}{2}, \frac{1}{4}, \frac{1}{4}, \frac{1}{4}, \frac{1}{4}, \frac{1}{2}, \frac{1}{2}, \frac{1}{2}, \frac{1}{2}, \frac{3}{4}, \frac{3}{4}, \frac{3}{4}, \frac{3}{4}$ & $8^{ B}_{ 3}\boxtimes 2^{ B}_{ 1}$\\
$16^{ B}_{-3}$ & $16$ & $1, 1, 1, 1, 1, 1, 1, 1, 1, 1, 1, 1, 1, 1, 1, 1$ & $0, 0, \frac{1}{2}, \frac{1}{2}, 0, 0, \frac{1}{8}, \frac{1}{8}, \frac{1}{2}, \frac{1}{2}, \frac{5}{8}, \frac{5}{8}, \frac{5}{8}, \frac{5}{8}, \frac{5}{8}, \frac{5}{8}$ & $4^{ B}_{-3}\boxtimes 4^{ B}_{ 0}$\\
$16^{ B}_{-3}$ & $16$ & $1, 1, 1, 1, 1, 1, 1, 1, 1, 1, 1, 1, 1, 1, 1, 1$ & $0, 0, \frac{1}{2}, \frac{1}{2}, 0, 0, \frac{1}{8}, \frac{1}{8}, \frac{1}{2}, \frac{1}{2}, \frac{5}{8}, \frac{5}{8}, \frac{5}{8}, \frac{5}{8}, \frac{5}{8}, \frac{5}{8}$ & $4^{ B}_{-3}\boxtimes 4^{ B}_{ 0}$\\
$16^{ B}_{-3}$ & $16$ & $1, 1, 1, 1, 1, 1, 1, 1, 1, 1, 1, 1, 1, 1, 1, 1$ & $0, 0, \frac{1}{2}, \frac{1}{2}, \frac{1}{4}, \frac{1}{4}, \frac{3}{8}, \frac{3}{8}, \frac{5}{8}, \frac{5}{8}, \frac{5}{8}, \frac{5}{8}, \frac{3}{4}, \frac{3}{4}, \frac{7}{8}, \frac{7}{8}$ & $8^{ B}_{ 4}\boxtimes 2^{ B}_{ 1}$\\
$16^{ B}_{-3}$ & $16$ & $1, 1, 1, 1, 1, 1, 1, 1, 1, 1, 1, 1, 1, 1, 1, 1$ & $0, 0, \frac{1}{2}, \frac{1}{2}, \frac{1}{4}, \frac{1}{4}, \frac{3}{8}, \frac{3}{8}, \frac{5}{8}, \frac{5}{8}, \frac{5}{8}, \frac{5}{8}, \frac{3}{4}, \frac{3}{4}, \frac{7}{8}, \frac{7}{8}$ & $8^{ B}_{ 4}\boxtimes 2^{ B}_{ 1}$\\
$16^{ B}_{-2}$ & $16$ & $1, 1, 1, 1, 1, 1, 1, 1, 1, 1, 1, 1, 1, 1, 1, 1$ & $0, 0, \frac{1}{2}, \frac{1}{2}, 0, 0, \frac{1}{4}, \frac{1}{4}, \frac{1}{2}, \frac{1}{2}, \frac{3}{4}, \frac{3}{4}, \frac{3}{4}, \frac{3}{4}, \frac{3}{4}, \frac{3}{4}$ & $8^{ B}_{-3}\boxtimes 2^{ B}_{ 1}$\\
$16^{ B}_{-2}$ & $16$ & $1, 1, 1, 1, 1, 1, 1, 1, 1, 1, 1, 1, 1, 1, 1, 1$ & $0, 0, \frac{1}{2}, \frac{1}{2}, 0, 0, \frac{1}{4}, \frac{1}{4}, \frac{1}{2}, \frac{1}{2}, \frac{3}{4}, \frac{3}{4}, \frac{3}{4}, \frac{3}{4}, \frac{3}{4}, \frac{3}{4}$ & $8^{ B}_{-3}\boxtimes 2^{ B}_{ 1}$\\
$16^{ B}_{-2}$ & $16$ & $1, 1, 1, 1, 1, 1, 1, 1, 1, 1, 1, 1, 1, 1, 1, 1$ & $0, 0, \frac{1}{2}, \frac{1}{2}, \frac{1}{8}, \frac{1}{8}, \frac{1}{8}, \frac{1}{8}, \frac{5}{8}, \frac{5}{8}, \frac{5}{8}, \frac{5}{8}, \frac{3}{4}, \frac{3}{4}, \frac{3}{4}, \frac{3}{4}$ & $4^{ B}_{-3}\boxtimes 4^{ B}_{ 1}$\\
$16^{ B}_{-2}$ & $16$ & $1, 1, 1, 1, 1, 1, 1, 1, 1, 1, 1, 1, 1, 1, 1, 1$ & $0, 0, \frac{1}{2}, \frac{1}{2}, \frac{3}{8}, \frac{3}{8}, \frac{3}{8}, \frac{3}{8}, \frac{3}{4}, \frac{3}{4}, \frac{3}{4}, \frac{3}{4}, \frac{7}{8}, \frac{7}{8}, \frac{7}{8}, \frac{7}{8}$ & $4^{ B}_{ 3}\boxtimes 4^{ B}_{ 3}$\\
$16^{ B}_{-1}$ & $16$ & $1, 1, 1, 1, 1, 1, 1, 1, 1, 1, 1, 1, 1, 1, 1, 1$ & $0, 0, \frac{1}{2}, \frac{1}{2}, 0, 0, \frac{3}{8}, \frac{3}{8}, \frac{1}{2}, \frac{1}{2}, \frac{7}{8}, \frac{7}{8}, \frac{7}{8}, \frac{7}{8}, \frac{7}{8}, \frac{7}{8}$ & $4^{ B}_{-1}\boxtimes 4^{ B}_{ 0}$\\
$16^{ B}_{-1}$ & $16$ & $1, 1, 1, 1, 1, 1, 1, 1, 1, 1, 1, 1, 1, 1, 1, 1$ & $0, 0, \frac{1}{2}, \frac{1}{2}, 0, 0, \frac{3}{8}, \frac{3}{8}, \frac{1}{2}, \frac{1}{2}, \frac{7}{8}, \frac{7}{8}, \frac{7}{8}, \frac{7}{8}, \frac{7}{8}, \frac{7}{8}$ & $4^{ B}_{-1}\boxtimes 4^{ B}_{ 0}$\\
$16^{ B}_{-1}$ & $16$ & $1, 1, 1, 1, 1, 1, 1, 1, 1, 1, 1, 1, 1, 1, 1, 1$ & $0, 0, \frac{1}{2}, \frac{1}{2}, \frac{1}{8}, \frac{1}{8}, \frac{1}{4}, \frac{1}{4}, \frac{5}{8}, \frac{5}{8}, \frac{3}{4}, \frac{3}{4}, \frac{7}{8}, \frac{7}{8}, \frac{7}{8}, \frac{7}{8}$ & $8^{ B}_{-2}\boxtimes 2^{ B}_{ 1}$\\
$16^{ B}_{-1}$ & $16$ & $1, 1, 1, 1, 1, 1, 1, 1, 1, 1, 1, 1, 1, 1, 1, 1$ & $0, 0, \frac{1}{2}, \frac{1}{2}, \frac{1}{8}, \frac{1}{8}, \frac{1}{4}, \frac{1}{4}, \frac{5}{8}, \frac{5}{8}, \frac{3}{4}, \frac{3}{4}, \frac{7}{8}, \frac{7}{8}, \frac{7}{8}, \frac{7}{8}$ & $8^{ B}_{-2}\boxtimes 2^{ B}_{ 1}$\\
\hline
\end{tabular}
\end{table*}

\def\arraystretch{1.25} \setlength\tabcolsep{3pt}
\begin{table*}[t] 
\caption{
All the modular extensions of $\sRp(Z_6^f)=\sRp(Z_3\times Z_2^f)$.
} 
\label{mextZ6f} 
\centering
\begin{tabular}{ |c|c|l|l| } 
\hline 
$N^{|\Th|}_{c}$ & $D^2$ & $d_1,d_2,\cdots$ & $s_1,s_2,\cdots$ \\
\hline 
$6^{ 0}_{0}$ & $6$ & $1, 1, 1, 1, 1, 1$ & $0, \frac{1}{2}, 0, \frac{1}{2}, 0, \frac{1}{2}$  \hfill $\sRp(Z_6^f)$\\
\hline
$36^{ B}_{ 0}$ & $36$ & $1\times 36$ & $0, \frac{1}{2}, 0, \frac{1}{2}, 0, \frac{1}{2}, 0, 0, 0, 0, 0, 0, 0, 0, 0, 0, 0, 0, \frac{1}{6}, \frac{1}{6}, \frac{1}{3}, \frac{1}{3}, \frac{1}{3}, \frac{1}{3}, \frac{1}{3}, \frac{1}{3}, \frac{1}{2}, \frac{1}{2}, \frac{2}{3}, \frac{2}{3}, \frac{2}{3}, \frac{2}{3}, \frac{2}{3}, \frac{2}{3}, \frac{5}{6}, \frac{5}{6}$  \\
$36^{ B}_{ 0}$ & $36$ & $1\times 36$ & $0, \frac{1}{2}, 0, \frac{1}{2}, 0, \frac{1}{2}, 0, 0, 0, 0, 0, 0, \frac{1}{18}, \frac{1}{18}, \frac{2}{9}, \frac{2}{9}, \frac{2}{9}, \frac{2}{9}, \frac{2}{9}, \frac{2}{9}, \frac{7}{18}, \frac{7}{18}, \frac{5}{9}, \frac{5}{9}, \frac{5}{9}, \frac{5}{9}, \frac{5}{9}, \frac{5}{9}, \frac{13}{18}, \frac{13}{18}, \frac{8}{9}, \frac{8}{9}, \frac{8}{9}, \frac{8}{9}, \frac{8}{9}, \frac{8}{9}$  \\
$36^{ B}_{ 0}$ & $36$ & $1\times 36$ & $0, \frac{1}{2}, 0, \frac{1}{2}, 0, \frac{1}{2}, 0, 0, 0, 0, 0, 0, \frac{1}{9}, \frac{1}{9}, \frac{1}{9}, \frac{1}{9}, \frac{1}{9}, \frac{1}{9}, \frac{5}{18}, \frac{5}{18}, \frac{4}{9}, \frac{4}{9}, \frac{4}{9}, \frac{4}{9}, \frac{4}{9}, \frac{4}{9}, \frac{11}{18}, \frac{11}{18}, \frac{7}{9}, \frac{7}{9}, \frac{7}{9}, \frac{7}{9}, \frac{7}{9}, \frac{7}{9}, \frac{17}{18}, \frac{17}{18}$  \\
\hline
$36^{ B}_{ 1}$ & $36$ & $1\times 36$ &\tiny $0, \frac{1}{2}, 0, \frac{1}{2}, 0, \frac{1}{2}, 0, 0, \frac{1}{8}, \frac{1}{8}, \frac{1}{8}, \frac{1}{8}, \frac{1}{8}, \frac{1}{8}, \frac{1}{8}, \frac{1}{8}, \frac{1}{8}, \frac{1}{8}, \frac{1}{6}, \frac{1}{6}, \frac{1}{3}, \frac{1}{3}, \frac{11}{24}, \frac{11}{24}, \frac{11}{24}, \frac{11}{24}, \frac{1}{2}, \frac{1}{2}, \frac{2}{3}, \frac{2}{3}, \frac{19}{24}, \frac{19}{24}, \frac{19}{24}, \frac{19}{24}, \frac{5}{6}, \frac{5}{6}$  \\
$36^{ B}_{ 1}$ & $36$ & $1\times 36$ &\tiny $0, \frac{1}{2}, 0, \frac{1}{2}, 0, \frac{1}{2}, \frac{1}{72}, \frac{1}{72}, \frac{1}{72}, \frac{1}{72}, \frac{1}{18}, \frac{1}{18}, \frac{1}{8}, \frac{1}{8}, \frac{1}{8}, \frac{1}{8}, \frac{1}{8}, \frac{1}{8}, \frac{2}{9}, \frac{2}{9}, \frac{25}{72}, \frac{25}{72}, \frac{25}{72}, \frac{25}{72}, \frac{7}{18}, \frac{7}{18}, \frac{5}{9}, \frac{5}{9}, \frac{49}{72}, \frac{49}{72}, \frac{49}{72}, \frac{49}{72}, \frac{13}{18}, \frac{13}{18}, \frac{8}{9}, \frac{8}{9}$  \\
$36^{ B}_{ 1}$ & $36$ & $1\times 36$ &\tiny $0, \frac{1}{2}, 0, \frac{1}{2}, 0, \frac{1}{2}, \frac{1}{9}, \frac{1}{9}, \frac{1}{8}, \frac{1}{8}, \frac{1}{8}, \frac{1}{8}, \frac{1}{8}, \frac{1}{8}, \frac{17}{72}, \frac{17}{72}, \frac{17}{72}, \frac{17}{72}, \frac{5}{18}, \frac{5}{18}, \frac{4}{9}, \frac{4}{9}, \frac{41}{72}, \frac{41}{72}, \frac{41}{72}, \frac{41}{72}, \frac{11}{18}, \frac{11}{18}, \frac{7}{9}, \frac{7}{9}, \frac{65}{72}, \frac{65}{72}, \frac{65}{72}, \frac{65}{72}, \frac{17}{18}, \frac{17}{18}$  \\
$36^{ B}_{ 2}$ & $36$ & $1\times 36$ &\tiny $0, \frac{1}{2}, 0, \frac{1}{2}, 0, \frac{1}{2}, 0, 0, \frac{1}{6}, \frac{1}{6}, \frac{1}{4}, \frac{1}{4}, \frac{1}{4}, \frac{1}{4}, \frac{1}{4}, \frac{1}{4}, \frac{1}{4}, \frac{1}{4}, \frac{1}{4}, \frac{1}{4}, \frac{1}{3}, \frac{1}{3}, \frac{1}{2}, \frac{1}{2}, \frac{7}{12}, \frac{7}{12}, \frac{7}{12}, \frac{7}{12}, \frac{2}{3}, \frac{2}{3}, \frac{5}{6}, \frac{5}{6}, \frac{11}{12}, \frac{11}{12}, \frac{11}{12}, \frac{11}{12}$  \\
$36^{ B}_{ 2}$ & $36$ & $1\times 36$ &\tiny $0, \frac{1}{2}, 0, \frac{1}{2}, 0, \frac{1}{2}, \frac{1}{36}, \frac{1}{36}, \frac{1}{36}, \frac{1}{36}, \frac{1}{9}, \frac{1}{9}, \frac{1}{4}, \frac{1}{4}, \frac{1}{4}, \frac{1}{4}, \frac{1}{4}, \frac{1}{4}, \frac{5}{18}, \frac{5}{18}, \frac{13}{36}, \frac{13}{36}, \frac{13}{36}, \frac{13}{36}, \frac{4}{9}, \frac{4}{9}, \frac{11}{18}, \frac{11}{18}, \frac{25}{36}, \frac{25}{36}, \frac{25}{36}, \frac{25}{36}, \frac{7}{9}, \frac{7}{9}, \frac{17}{18}, \frac{17}{18}$  \\
$36^{ B}_{ 2}$ & $36$ & $1\times 36$ &\tiny $0, \frac{1}{2}, 0, \frac{1}{2}, 0, \frac{1}{2}, \frac{1}{18}, \frac{1}{18}, \frac{5}{36}, \frac{5}{36}, \frac{5}{36}, \frac{5}{36}, \frac{2}{9}, \frac{2}{9}, \frac{1}{4}, \frac{1}{4}, \frac{1}{4}, \frac{1}{4}, \frac{1}{4}, \frac{1}{4}, \frac{7}{18}, \frac{7}{18}, \frac{17}{36}, \frac{17}{36}, \frac{17}{36}, \frac{17}{36}, \frac{5}{9}, \frac{5}{9}, \frac{13}{18}, \frac{13}{18}, \frac{29}{36}, \frac{29}{36}, \frac{29}{36}, \frac{29}{36}, \frac{8}{9}, \frac{8}{9}$  \\
$36^{ B}_{ 3}$ & $36$ & $1\times 36$ &\tiny $0, \frac{1}{2}, 0, \frac{1}{2}, 0, \frac{1}{2}, 0, 0, \frac{1}{24}, \frac{1}{24}, \frac{1}{24}, \frac{1}{24}, \frac{1}{6}, \frac{1}{6}, \frac{1}{3}, \frac{1}{3}, \frac{3}{8}, \frac{3}{8}, \frac{3}{8}, \frac{3}{8}, \frac{3}{8}, \frac{3}{8}, \frac{3}{8}, \frac{3}{8}, \frac{3}{8}, \frac{3}{8}, \frac{1}{2}, \frac{1}{2}, \frac{2}{3}, \frac{2}{3}, \frac{17}{24}, \frac{17}{24}, \frac{17}{24}, \frac{17}{24}, \frac{5}{6}, \frac{5}{6}$  \\
$36^{ B}_{ 3}$ & $36$ & $1\times 36$ &\tiny $0, \frac{1}{2}, 0, \frac{1}{2}, 0, \frac{1}{2}, \frac{1}{18}, \frac{1}{18}, \frac{2}{9}, \frac{2}{9}, \frac{19}{72}, \frac{19}{72}, \frac{19}{72}, \frac{19}{72}, \frac{3}{8}, \frac{3}{8}, \frac{3}{8}, \frac{3}{8}, \frac{3}{8}, \frac{3}{8}, \frac{7}{18}, \frac{7}{18}, \frac{5}{9}, \frac{5}{9}, \frac{43}{72}, \frac{43}{72}, \frac{43}{72}, \frac{43}{72}, \frac{13}{18}, \frac{13}{18}, \frac{8}{9}, \frac{8}{9}, \frac{67}{72}, \frac{67}{72}, \frac{67}{72}, \frac{67}{72}$  \\
$36^{ B}_{ 3}$ & $36$ & $1\times 36$ &\tiny $0, \frac{1}{2}, 0, \frac{1}{2}, 0, \frac{1}{2}, \frac{1}{9}, \frac{1}{9}, \frac{11}{72}, \frac{11}{72}, \frac{11}{72}, \frac{11}{72}, \frac{5}{18}, \frac{5}{18}, \frac{3}{8}, \frac{3}{8}, \frac{3}{8}, \frac{3}{8}, \frac{3}{8}, \frac{3}{8}, \frac{4}{9}, \frac{4}{9}, \frac{35}{72}, \frac{35}{72}, \frac{35}{72}, \frac{35}{72}, \frac{11}{18}, \frac{11}{18}, \frac{7}{9}, \frac{7}{9}, \frac{59}{72}, \frac{59}{72}, \frac{59}{72}, \frac{59}{72}, \frac{17}{18}, \frac{17}{18}$  \\
$36^{ B}_{ 4}$ & $36$ & $1\times 36$ &\tiny $0, \frac{1}{2}, 0, \frac{1}{2}, 0, \frac{1}{2}, 0, 0, \frac{1}{6}, \frac{1}{6}, \frac{1}{6}, \frac{1}{6}, \frac{1}{6}, \frac{1}{6}, \frac{1}{3}, \frac{1}{3}, \frac{1}{2}, \frac{1}{2}, \frac{1}{2}, \frac{1}{2}, \frac{1}{2}, \frac{1}{2}, \frac{1}{2}, \frac{1}{2}, \frac{1}{2}, \frac{1}{2}, \frac{1}{2}, \frac{1}{2}, \frac{2}{3}, \frac{2}{3}, \frac{5}{6}, \frac{5}{6}, \frac{5}{6}, \frac{5}{6}, \frac{5}{6}, \frac{5}{6}$  \\
$36^{ B}_{ 4}$ & $36$ & $1\times 36$ &\tiny $0, \frac{1}{2}, 0, \frac{1}{2}, 0, \frac{1}{2}, \frac{1}{18}, \frac{1}{18}, \frac{1}{18}, \frac{1}{18}, \frac{1}{18}, \frac{1}{18}, \frac{2}{9}, \frac{2}{9}, \frac{7}{18}, \frac{7}{18}, \frac{7}{18}, \frac{7}{18}, \frac{7}{18}, \frac{7}{18}, \frac{1}{2}, \frac{1}{2}, \frac{1}{2}, \frac{1}{2}, \frac{1}{2}, \frac{1}{2}, \frac{5}{9}, \frac{5}{9}, \frac{13}{18}, \frac{13}{18}, \frac{13}{18}, \frac{13}{18}, \frac{13}{18}, \frac{13}{18}, \frac{8}{9}, \frac{8}{9}$  \\
$36^{ B}_{ 4}$ & $36$ & $1\times 36$ &\tiny $0, \frac{1}{2}, 0, \frac{1}{2}, 0, \frac{1}{2}, \frac{1}{9}, \frac{1}{9}, \frac{5}{18}, \frac{5}{18}, \frac{5}{18}, \frac{5}{18}, \frac{5}{18}, \frac{5}{18}, \frac{4}{9}, \frac{4}{9}, \frac{1}{2}, \frac{1}{2}, \frac{1}{2}, \frac{1}{2}, \frac{1}{2}, \frac{1}{2}, \frac{11}{18}, \frac{11}{18}, \frac{11}{18}, \frac{11}{18}, \frac{11}{18}, \frac{11}{18}, \frac{7}{9}, \frac{7}{9}, \frac{17}{18}, \frac{17}{18}, \frac{17}{18}, \frac{17}{18}, \frac{17}{18}, \frac{17}{18}$  \\
$36^{ B}_{-3}$ & $36$ & $1\times 36$ &\tiny $0, \frac{1}{2}, 0, \frac{1}{2}, 0, \frac{1}{2}, 0, 0, \frac{1}{6}, \frac{1}{6}, \frac{7}{24}, \frac{7}{24}, \frac{7}{24}, \frac{7}{24}, \frac{1}{3}, \frac{1}{3}, \frac{1}{2}, \frac{1}{2}, \frac{5}{8}, \frac{5}{8}, \frac{5}{8}, \frac{5}{8}, \frac{5}{8}, \frac{5}{8}, \frac{5}{8}, \frac{5}{8}, \frac{5}{8}, \frac{5}{8}, \frac{2}{3}, \frac{2}{3}, \frac{5}{6}, \frac{5}{6}, \frac{23}{24}, \frac{23}{24}, \frac{23}{24}, \frac{23}{24}$  \\
$36^{ B}_{-3}$ & $36$ & $1\times 36$ &\tiny $0, \frac{1}{2}, 0, \frac{1}{2}, 0, \frac{1}{2}, \frac{1}{18}, \frac{1}{18}, \frac{13}{72}, \frac{13}{72}, \frac{13}{72}, \frac{13}{72}, \frac{2}{9}, \frac{2}{9}, \frac{7}{18}, \frac{7}{18}, \frac{37}{72}, \frac{37}{72}, \frac{37}{72}, \frac{37}{72}, \frac{5}{9}, \frac{5}{9}, \frac{5}{8}, \frac{5}{8}, \frac{5}{8}, \frac{5}{8}, \frac{5}{8}, \frac{5}{8}, \frac{13}{18}, \frac{13}{18}, \frac{61}{72}, \frac{61}{72}, \frac{61}{72}, \frac{61}{72}, \frac{8}{9}, \frac{8}{9}$  \\
$36^{ B}_{-3}$ & $36$ & $1\times 36$ &\tiny $0, \frac{1}{2}, 0, \frac{1}{2}, 0, \frac{1}{2}, \frac{5}{72}, \frac{5}{72}, \frac{5}{72}, \frac{5}{72}, \frac{1}{9}, \frac{1}{9}, \frac{5}{18}, \frac{5}{18}, \frac{29}{72}, \frac{29}{72}, \frac{29}{72}, \frac{29}{72}, \frac{4}{9}, \frac{4}{9}, \frac{11}{18}, \frac{11}{18}, \frac{5}{8}, \frac{5}{8}, \frac{5}{8}, \frac{5}{8}, \frac{5}{8}, \frac{5}{8}, \frac{53}{72}, \frac{53}{72}, \frac{53}{72}, \frac{53}{72}, \frac{7}{9}, \frac{7}{9}, \frac{17}{18}, \frac{17}{18}$  \\
$36^{ B}_{-2}$ & $36$ & $1\times 36$ &\tiny $0, \frac{1}{2}, 0, \frac{1}{2}, 0, \frac{1}{2}, 0, 0, \frac{1}{12}, \frac{1}{12}, \frac{1}{12}, \frac{1}{12}, \frac{1}{6}, \frac{1}{6}, \frac{1}{3}, \frac{1}{3}, \frac{5}{12}, \frac{5}{12}, \frac{5}{12}, \frac{5}{12}, \frac{1}{2}, \frac{1}{2}, \frac{2}{3}, \frac{2}{3}, \frac{3}{4}, \frac{3}{4}, \frac{3}{4}, \frac{3}{4}, \frac{3}{4}, \frac{3}{4}, \frac{3}{4}, \frac{3}{4}, \frac{3}{4}, \frac{3}{4}, \frac{5}{6}, \frac{5}{6}$  \\
$36^{ B}_{-2}$ & $36$ & $1\times 36$ &\tiny $0, \frac{1}{2}, 0, \frac{1}{2}, 0, \frac{1}{2}, \frac{1}{18}, \frac{1}{18}, \frac{2}{9}, \frac{2}{9}, \frac{11}{36}, \frac{11}{36}, \frac{11}{36}, \frac{11}{36}, \frac{7}{18}, \frac{7}{18}, \frac{5}{9}, \frac{5}{9}, \frac{23}{36}, \frac{23}{36}, \frac{23}{36}, \frac{23}{36}, \frac{13}{18}, \frac{13}{18}, \frac{3}{4}, \frac{3}{4}, \frac{3}{4}, \frac{3}{4}, \frac{3}{4}, \frac{3}{4}, \frac{8}{9}, \frac{8}{9}, \frac{35}{36}, \frac{35}{36}, \frac{35}{36}, \frac{35}{36}$  \\
$36^{ B}_{-2}$ & $36$ & $1\times 36$ &\tiny $0, \frac{1}{2}, 0, \frac{1}{2}, 0, \frac{1}{2}, \frac{1}{9}, \frac{1}{9}, \frac{7}{36}, \frac{7}{36}, \frac{7}{36}, \frac{7}{36}, \frac{5}{18}, \frac{5}{18}, \frac{4}{9}, \frac{4}{9}, \frac{19}{36}, \frac{19}{36}, \frac{19}{36}, \frac{19}{36}, \frac{11}{18}, \frac{11}{18}, \frac{3}{4}, \frac{3}{4}, \frac{3}{4}, \frac{3}{4}, \frac{3}{4}, \frac{3}{4}, \frac{7}{9}, \frac{7}{9}, \frac{31}{36}, \frac{31}{36}, \frac{31}{36}, \frac{31}{36}, \frac{17}{18}, \frac{17}{18}$  \\
$36^{ B}_{-1}$ & $36$ & $1\times 36$ &\tiny $0, \frac{1}{2}, 0, \frac{1}{2}, 0, \frac{1}{2}, 0, 0, \frac{1}{6}, \frac{1}{6}, \frac{5}{24}, \frac{5}{24}, \frac{5}{24}, \frac{5}{24}, \frac{1}{3}, \frac{1}{3}, \frac{1}{2}, \frac{1}{2}, \frac{13}{24}, \frac{13}{24}, \frac{13}{24}, \frac{13}{24}, \frac{2}{3}, \frac{2}{3}, \frac{5}{6}, \frac{5}{6}, \frac{7}{8}, \frac{7}{8}, \frac{7}{8}, \frac{7}{8}, \frac{7}{8}, \frac{7}{8}, \frac{7}{8}, \frac{7}{8}, \frac{7}{8}, \frac{7}{8}$  \\
$36^{ B}_{-1}$ & $36$ & $1\times 36$ &\tiny $0, \frac{1}{2}, 0, \frac{1}{2}, 0, \frac{1}{2}, \frac{1}{18}, \frac{1}{18}, \frac{7}{72}, \frac{7}{72}, \frac{7}{72}, \frac{7}{72}, \frac{2}{9}, \frac{2}{9}, \frac{7}{18}, \frac{7}{18}, \frac{31}{72}, \frac{31}{72}, \frac{31}{72}, \frac{31}{72}, \frac{5}{9}, \frac{5}{9}, \frac{13}{18}, \frac{13}{18}, \frac{55}{72}, \frac{55}{72}, \frac{55}{72}, \frac{55}{72}, \frac{7}{8}, \frac{7}{8}, \frac{7}{8}, \frac{7}{8}, \frac{7}{8}, \frac{7}{8}, \frac{8}{9}, \frac{8}{9}$  \\
$36^{ B}_{-1}$ & $36$ & $1\times 36$ &\tiny $0, \frac{1}{2}, 0, \frac{1}{2}, 0, \frac{1}{2}, \frac{1}{9}, \frac{1}{9}, \frac{5}{18}, \frac{5}{18}, \frac{23}{72}, \frac{23}{72}, \frac{23}{72}, \frac{23}{72}, \frac{4}{9}, \frac{4}{9}, \frac{11}{18}, \frac{11}{18}, \frac{47}{72}, \frac{47}{72}, \frac{47}{72}, \frac{47}{72}, \frac{7}{9}, \frac{7}{9}, \frac{7}{8}, \frac{7}{8}, \frac{7}{8}, \frac{7}{8}, \frac{7}{8}, \frac{7}{8}, \frac{17}{18}, \frac{17}{18}, \frac{71}{72}, \frac{71}{72}, \frac{71}{72}, \frac{71}{72}$  \\
\hline
$27^{ B}_{ 1/2}$ & $36$ &\tiny $1\times 18 , \zeta_2^1\times 9$ & $0, \frac{1}{2}, 0, \frac{1}{2}, 0, \frac{1}{2}, 0, 0, \frac{1}{6}, \frac{1}{6}, \frac{1}{3}, \frac{1}{3}, \frac{1}{2}, \frac{1}{2}, \frac{2}{3}, \frac{2}{3}, \frac{5}{6}, \frac{5}{6}, \frac{1}{16}, \frac{1}{16}, \frac{1}{16}, \frac{1}{16}, \frac{1}{16}, \frac{19}{48}, \frac{19}{48}, \frac{35}{48}, \frac{35}{48}$  \\
$27^{ B}_{ 1/2}$ & $36$ &\tiny $1\times 18 , \zeta_2^1\times 9$ & $0, \frac{1}{2}, 0, \frac{1}{2}, 0, \frac{1}{2}, \frac{1}{18}, \frac{1}{18}, \frac{2}{9}, \frac{2}{9}, \frac{7}{18}, \frac{7}{18}, \frac{5}{9}, \frac{5}{9}, \frac{13}{18}, \frac{13}{18}, \frac{8}{9}, \frac{8}{9}, \frac{1}{16}, \frac{1}{16}, \frac{1}{16}, \frac{41}{144}, \frac{41}{144}, \frac{89}{144}, \frac{89}{144}, \frac{137}{144}, \frac{137}{144}$  \\
$27^{ B}_{ 1/2}$ & $36$ &\tiny $1\times 18 , \zeta_2^1\times 9$ & $0, \frac{1}{2}, 0, \frac{1}{2}, 0, \frac{1}{2}, \frac{1}{9}, \frac{1}{9}, \frac{5}{18}, \frac{5}{18}, \frac{4}{9}, \frac{4}{9}, \frac{11}{18}, \frac{11}{18}, \frac{7}{9}, \frac{7}{9}, \frac{17}{18}, \frac{17}{18}, \frac{1}{16}, \frac{1}{16}, \frac{1}{16}, \frac{25}{144}, \frac{25}{144}, \frac{73}{144}, \frac{73}{144}, \frac{121}{144}, \frac{121}{144}$  \\
$27^{ B}_{ 3/2}$ & $36$ &\tiny $1\times 18 , \zeta_2^1\times 9$ & $0, \frac{1}{2}, 0, \frac{1}{2}, 0, \frac{1}{2}, 0, 0, \frac{1}{6}, \frac{1}{6}, \frac{1}{3}, \frac{1}{3}, \frac{1}{2}, \frac{1}{2}, \frac{2}{3}, \frac{2}{3}, \frac{5}{6}, \frac{5}{6}, \frac{3}{16}, \frac{3}{16}, \frac{3}{16}, \frac{3}{16}, \frac{3}{16}, \frac{25}{48}, \frac{25}{48}, \frac{41}{48}, \frac{41}{48}$  \\
$27^{ B}_{ 3/2}$ & $36$ &\tiny $1\times 18 , \zeta_2^1\times 9$ & $0, \frac{1}{2}, 0, \frac{1}{2}, 0, \frac{1}{2}, \frac{1}{18}, \frac{1}{18}, \frac{2}{9}, \frac{2}{9}, \frac{7}{18}, \frac{7}{18}, \frac{5}{9}, \frac{5}{9}, \frac{13}{18}, \frac{13}{18}, \frac{8}{9}, \frac{8}{9}, \frac{11}{144}, \frac{11}{144}, \frac{3}{16}, \frac{3}{16}, \frac{3}{16}, \frac{59}{144}, \frac{59}{144}, \frac{107}{144}, \frac{107}{144}$  \\
$27^{ B}_{ 3/2}$ & $36$ &\tiny $1\times 18 , \zeta_2^1\times 9$ & $0, \frac{1}{2}, 0, \frac{1}{2}, 0, \frac{1}{2}, \frac{1}{9}, \frac{1}{9}, \frac{5}{18}, \frac{5}{18}, \frac{4}{9}, \frac{4}{9}, \frac{11}{18}, \frac{11}{18}, \frac{7}{9}, \frac{7}{9}, \frac{17}{18}, \frac{17}{18}, \frac{3}{16}, \frac{3}{16}, \frac{3}{16}, \frac{43}{144}, \frac{43}{144}, \frac{91}{144}, \frac{91}{144}, \frac{139}{144}, \frac{139}{144}$  \\
$27^{ B}_{ 5/2}$ & $36$ &\tiny $1\times 18 , \zeta_2^1\times 9$ & $0, \frac{1}{2}, 0, \frac{1}{2}, 0, \frac{1}{2}, 0, 0, \frac{1}{6}, \frac{1}{6}, \frac{1}{3}, \frac{1}{3}, \frac{1}{2}, \frac{1}{2}, \frac{2}{3}, \frac{2}{3}, \frac{5}{6}, \frac{5}{6}, \frac{5}{16}, \frac{5}{16}, \frac{5}{16}, \frac{5}{16}, \frac{5}{16}, \frac{31}{48}, \frac{31}{48}, \frac{47}{48}, \frac{47}{48}$  \\
$27^{ B}_{ 5/2}$ & $36$ &\tiny $1\times 18 , \zeta_2^1\times 9$ & $0, \frac{1}{2}, 0, \frac{1}{2}, 0, \frac{1}{2}, \frac{1}{18}, \frac{1}{18}, \frac{2}{9}, \frac{2}{9}, \frac{7}{18}, \frac{7}{18}, \frac{5}{9}, \frac{5}{9}, \frac{13}{18}, \frac{13}{18}, \frac{8}{9}, \frac{8}{9}, \frac{29}{144}, \frac{29}{144}, \frac{5}{16}, \frac{5}{16}, \frac{5}{16}, \frac{77}{144}, \frac{77}{144}, \frac{125}{144}, \frac{125}{144}$  \\
$27^{ B}_{ 5/2}$ & $36$ &\tiny $1\times 18 , \zeta_2^1\times 9$ & $0, \frac{1}{2}, 0, \frac{1}{2}, 0, \frac{1}{2}, \frac{1}{9}, \frac{1}{9}, \frac{5}{18}, \frac{5}{18}, \frac{4}{9}, \frac{4}{9}, \frac{11}{18}, \frac{11}{18}, \frac{7}{9}, \frac{7}{9}, \frac{17}{18}, \frac{17}{18}, \frac{13}{144}, \frac{13}{144}, \frac{5}{16}, \frac{5}{16}, \frac{5}{16}, \frac{61}{144}, \frac{61}{144}, \frac{109}{144}, \frac{109}{144}$  \\
$27^{ B}_{ 7/2}$ & $36$ &\tiny $1\times 18 , \zeta_2^1\times 9$ & $0, \frac{1}{2}, 0, \frac{1}{2}, 0, \frac{1}{2}, 0, 0, \frac{1}{6}, \frac{1}{6}, \frac{1}{3}, \frac{1}{3}, \frac{1}{2}, \frac{1}{2}, \frac{2}{3}, \frac{2}{3}, \frac{5}{6}, \frac{5}{6}, \frac{5}{48}, \frac{5}{48}, \frac{7}{16}, \frac{7}{16}, \frac{7}{16}, \frac{7}{16}, \frac{7}{16}, \frac{37}{48}, \frac{37}{48}$  \\
$27^{ B}_{ 7/2}$ & $36$ &\tiny $1\times 18 , \zeta_2^1\times 9$ & $0, \frac{1}{2}, 0, \frac{1}{2}, 0, \frac{1}{2}, \frac{1}{18}, \frac{1}{18}, \frac{2}{9}, \frac{2}{9}, \frac{7}{18}, \frac{7}{18}, \frac{5}{9}, \frac{5}{9}, \frac{13}{18}, \frac{13}{18}, \frac{8}{9}, \frac{8}{9}, \frac{47}{144}, \frac{47}{144}, \frac{7}{16}, \frac{7}{16}, \frac{7}{16}, \frac{95}{144}, \frac{95}{144}, \frac{143}{144}, \frac{143}{144}$  \\
$27^{ B}_{ 7/2}$ & $36$ &\tiny $1\times 18 , \zeta_2^1\times 9$ & $0, \frac{1}{2}, 0, \frac{1}{2}, 0, \frac{1}{2}, \frac{1}{9}, \frac{1}{9}, \frac{5}{18}, \frac{5}{18}, \frac{4}{9}, \frac{4}{9}, \frac{11}{18}, \frac{11}{18}, \frac{7}{9}, \frac{7}{9}, \frac{17}{18}, \frac{17}{18}, \frac{31}{144}, \frac{31}{144}, \frac{7}{16}, \frac{7}{16}, \frac{7}{16}, \frac{79}{144}, \frac{79}{144}, \frac{127}{144}, \frac{127}{144}$  \\
$27^{ B}_{-7/2}$ & $36$ &\tiny $1\times 18 , \zeta_2^1\times 9$ & $0, \frac{1}{2}, 0, \frac{1}{2}, 0, \frac{1}{2}, 0, 0, \frac{1}{6}, \frac{1}{6}, \frac{1}{3}, \frac{1}{3}, \frac{1}{2}, \frac{1}{2}, \frac{2}{3}, \frac{2}{3}, \frac{5}{6}, \frac{5}{6}, \frac{11}{48}, \frac{11}{48}, \frac{9}{16}, \frac{9}{16}, \frac{9}{16}, \frac{9}{16}, \frac{9}{16}, \frac{43}{48}, \frac{43}{48}$  \\
$27^{ B}_{-7/2}$ & $36$ &\tiny $1\times 18 , \zeta_2^1\times 9$ & $0, \frac{1}{2}, 0, \frac{1}{2}, 0, \frac{1}{2}, \frac{1}{18}, \frac{1}{18}, \frac{2}{9}, \frac{2}{9}, \frac{7}{18}, \frac{7}{18}, \frac{5}{9}, \frac{5}{9}, \frac{13}{18}, \frac{13}{18}, \frac{8}{9}, \frac{8}{9}, \frac{17}{144}, \frac{17}{144}, \frac{65}{144}, \frac{65}{144}, \frac{9}{16}, \frac{9}{16}, \frac{9}{16}, \frac{113}{144}, \frac{113}{144}$  \\
$27^{ B}_{-7/2}$ & $36$ &\tiny $1\times 18 , \zeta_2^1\times 9$ & $0, \frac{1}{2}, 0, \frac{1}{2}, 0, \frac{1}{2}, \frac{1}{9}, \frac{1}{9}, \frac{5}{18}, \frac{5}{18}, \frac{4}{9}, \frac{4}{9}, \frac{11}{18}, \frac{11}{18}, \frac{7}{9}, \frac{7}{9}, \frac{17}{18}, \frac{17}{18}, \frac{1}{144}, \frac{1}{144}, \frac{49}{144}, \frac{49}{144}, \frac{9}{16}, \frac{9}{16}, \frac{9}{16}, \frac{97}{144}, \frac{97}{144}$  \\
$27^{ B}_{-5/2}$ & $36$ &\tiny $1\times 18 , \zeta_2^1\times 9$ & $0, \frac{1}{2}, 0, \frac{1}{2}, 0, \frac{1}{2}, 0, 0, \frac{1}{6}, \frac{1}{6}, \frac{1}{3}, \frac{1}{3}, \frac{1}{2}, \frac{1}{2}, \frac{2}{3}, \frac{2}{3}, \frac{5}{6}, \frac{5}{6}, \frac{1}{48}, \frac{1}{48}, \frac{17}{48}, \frac{17}{48}, \frac{11}{16}, \frac{11}{16}, \frac{11}{16}, \frac{11}{16}, \frac{11}{16}$  \\
$27^{ B}_{-5/2}$ & $36$ &\tiny $1\times 18 , \zeta_2^1\times 9$ & $0, \frac{1}{2}, 0, \frac{1}{2}, 0, \frac{1}{2}, \frac{1}{18}, \frac{1}{18}, \frac{2}{9}, \frac{2}{9}, \frac{7}{18}, \frac{7}{18}, \frac{5}{9}, \frac{5}{9}, \frac{13}{18}, \frac{13}{18}, \frac{8}{9}, \frac{8}{9}, \frac{35}{144}, \frac{35}{144}, \frac{83}{144}, \frac{83}{144}, \frac{11}{16}, \frac{11}{16}, \frac{11}{16}, \frac{131}{144}, \frac{131}{144}$  \\
$27^{ B}_{-5/2}$ & $36$ &\tiny $1\times 18 , \zeta_2^1\times 9$ & $0, \frac{1}{2}, 0, \frac{1}{2}, 0, \frac{1}{2}, \frac{1}{9}, \frac{1}{9}, \frac{5}{18}, \frac{5}{18}, \frac{4}{9}, \frac{4}{9}, \frac{11}{18}, \frac{11}{18}, \frac{7}{9}, \frac{7}{9}, \frac{17}{18}, \frac{17}{18}, \frac{19}{144}, \frac{19}{144}, \frac{67}{144}, \frac{67}{144}, \frac{11}{16}, \frac{11}{16}, \frac{11}{16}, \frac{115}{144}, \frac{115}{144}$  \\
$27^{ B}_{-3/2}$ & $36$ &\tiny $1\times 18 , \zeta_2^1\times 9$ & $0, \frac{1}{2}, 0, \frac{1}{2}, 0, \frac{1}{2}, 0, 0, \frac{1}{6}, \frac{1}{6}, \frac{1}{3}, \frac{1}{3}, \frac{1}{2}, \frac{1}{2}, \frac{2}{3}, \frac{2}{3}, \frac{5}{6}, \frac{5}{6}, \frac{7}{48}, \frac{7}{48}, \frac{23}{48}, \frac{23}{48}, \frac{13}{16}, \frac{13}{16}, \frac{13}{16}, \frac{13}{16}, \frac{13}{16}$  \\
$27^{ B}_{-3/2}$ & $36$ &\tiny $1\times 18 , \zeta_2^1\times 9$ & $0, \frac{1}{2}, 0, \frac{1}{2}, 0, \frac{1}{2}, \frac{1}{18}, \frac{1}{18}, \frac{2}{9}, \frac{2}{9}, \frac{7}{18}, \frac{7}{18}, \frac{5}{9}, \frac{5}{9}, \frac{13}{18}, \frac{13}{18}, \frac{8}{9}, \frac{8}{9}, \frac{5}{144}, \frac{5}{144}, \frac{53}{144}, \frac{53}{144}, \frac{101}{144}, \frac{101}{144}, \frac{13}{16}, \frac{13}{16}, \frac{13}{16}$  \\
$27^{ B}_{-3/2}$ & $36$ &\tiny $1\times 18 , \zeta_2^1\times 9$ & $0, \frac{1}{2}, 0, \frac{1}{2}, 0, \frac{1}{2}, \frac{1}{9}, \frac{1}{9}, \frac{5}{18}, \frac{5}{18}, \frac{4}{9}, \frac{4}{9}, \frac{11}{18}, \frac{11}{18}, \frac{7}{9}, \frac{7}{9}, \frac{17}{18}, \frac{17}{18}, \frac{37}{144}, \frac{37}{144}, \frac{85}{144}, \frac{85}{144}, \frac{13}{16}, \frac{13}{16}, \frac{13}{16}, \frac{133}{144}, \frac{133}{144}$  \\
$27^{ B}_{-1/2}$ & $36$ &\tiny $1\times 18 , \zeta_2^1\times 9$ & $0, \frac{1}{2}, 0, \frac{1}{2}, 0, \frac{1}{2}, 0, 0, \frac{1}{6}, \frac{1}{6}, \frac{1}{3}, \frac{1}{3}, \frac{1}{2}, \frac{1}{2}, \frac{2}{3}, \frac{2}{3}, \frac{5}{6}, \frac{5}{6}, \frac{13}{48}, \frac{13}{48}, \frac{29}{48}, \frac{29}{48}, \frac{15}{16}, \frac{15}{16}, \frac{15}{16}, \frac{15}{16}, \frac{15}{16}$  \\
$27^{ B}_{-1/2}$ & $36$ &\tiny $1\times 18 , \zeta_2^1\times 9$ & $0, \frac{1}{2}, 0, \frac{1}{2}, 0, \frac{1}{2}, \frac{1}{18}, \frac{1}{18}, \frac{2}{9}, \frac{2}{9}, \frac{7}{18}, \frac{7}{18}, \frac{5}{9}, \frac{5}{9}, \frac{13}{18}, \frac{13}{18}, \frac{8}{9}, \frac{8}{9}, \frac{23}{144}, \frac{23}{144}, \frac{71}{144}, \frac{71}{144}, \frac{119}{144}, \frac{119}{144}, \frac{15}{16}, \frac{15}{16}, \frac{15}{16}$  \\
$27^{ B}_{-1/2}$ & $36$ &\tiny $1\times 18 , \zeta_2^1\times 9$ & $0, \frac{1}{2}, 0, \frac{1}{2}, 0, \frac{1}{2}, \frac{1}{9}, \frac{1}{9}, \frac{5}{18}, \frac{5}{18}, \frac{4}{9}, \frac{4}{9}, \frac{11}{18}, \frac{11}{18}, \frac{7}{9}, \frac{7}{9}, \frac{17}{18}, \frac{17}{18}, \frac{7}{144}, \frac{7}{144}, \frac{55}{144}, \frac{55}{144}, \frac{103}{144}, \frac{103}{144}, \frac{15}{16}, \frac{15}{16}, \frac{15}{16}$  \\
\hline
\end{tabular}
\end{table*} 

\vfill
\clearpage
%\newpage
%\eject

\appendix

\section{Tables for the solutions of $(\tilde N^{ab}_c,\tilde s_a;
N^{ij}_k,s_i)$ -- imperfect tables for $\mce{\cE}$}

\label{SETtbl}

\def\arraystretch{1.25} \setlength\tabcolsep{3pt}
\begin{table}[t] 
\caption{
$Z_2$-SET orders (or $\mce{\Rp(Z_2)}$) for bosonic systems labeled in terms of sets of topological excitations.
The list contains all topological orders with 
$N=3,4$ and $D^2\leq 100$.
All the topologically orders in this list are anomaly free
(\ie have modular extensions), and are realizable by 2+1D bosonic systems.
We use $N^{|\Th|}_{c}$ to label $\mce{\cE}$'s, where
$\Theta ={D}^{-1}\sum_{i}\ee^{2\pi\ii s_i} d_i^2=
|\Th|\ee^{2\pi \ii c/8}$ and $D^2=\sum_id_i^2$.
} 
\label{SETZ2-34} 
\centering
\begin{tabular}{ |c|c|l|l|l| } 
\hline 
$N^{|\Th|}_{c}$ & $D^2$ & $d_1,d_2,\cdots$ & $s_1,s_2,\cdots$ & comment \\
\hline 
$2^{\zeta_{2}^{1}}_{ 0}$ & $2$ & $1, 1$ & $0, 0$ & $\cE=\text{Rep}(Z_2)$\\
\hline
$3^{\zeta_{2}^{1}}_{ 2}$ & $6$ & $1, 1, 2$ & $0, 0, \frac{1}{3}$ & 
SB:$K=$\tiny $\begin{pmatrix}
 2 & -1 \\
 -1 & 2 \\
\end{pmatrix}
$
\\
$3^{\zeta_{2}^{1}}_{-2}$ & $6$ & $1, 1, 2$ & $0, 0, \frac{2}{3}$ & SB:$K=$\tiny $\begin{pmatrix}
 -2 & 1 \\
 1 & -2 \\
\end{pmatrix}
$
\\
\hline
$4^{\zeta_{2}^{1}}_{ 1}$ & $4$ & $1, 1, 1, 1$ & $0, 0, \frac{1}{4}, \frac{1}{4}$ & $2^B_1\boxtimes \text{Rep}(Z_2)$\\
 $4^{\zeta_{2}^{1}}_{ 1}$ & $4$ & $1, 1, 1, 1$ & $0, 0, \frac{1}{4}, \frac{1}{4}$ & $2^B_1\boxtimes^t \text{Rep}(Z_2)$ \\
$4^{\zeta_{2}^{1}}_{-1}$ & $4$ & $1, 1, 1, 1$ & $0, 0, \frac{3}{4}, \frac{3}{4}$ & $2^B_{-1}\boxtimes \text{Rep}(Z_2)$\\
 $4^{\zeta_{2}^{1}}_{-1}$ & $4$ & $1, 1, 1, 1$ & $0, 0, \frac{3}{4}, \frac{3}{4}$ & $2^B_{-1}\boxtimes^t \text{Rep}(Z_2)$ \\
$4^{\zeta_{2}^{1}}_{ 14/5}$ & $7.2360$ & $1, 1,\zeta_{3}^{1},\zeta_{3}^{1}$ & $0, 0, \frac{2}{5}, \frac{2}{5}$ & $2^B_{ 14/5}\boxtimes \text{Rep}(Z_2)$ \\
$4^{\zeta_{2}^{1}}_{-14/5}$ & $7.2360$ & $1, 1,\zeta_{3}^{1},\zeta_{3}^{1}$ & $0, 0, \frac{3}{5}, \frac{3}{5}$ & $2^B_{-14/5}\boxtimes \text{Rep}(Z_2)$ \\
$4^{\zeta_{2}^{1}}_{ 0}$ & $10$ & $1, 1, 2, 2$ & $0, 0, \frac{1}{5}, \frac{4}{5}$ & 
SB:$K=$\tiny $\begin{pmatrix}
 2 & -3 \\
 -3 & 2 \\
\end{pmatrix}
$
\\
$4^{\zeta_{2}^{1}}_{ 4}$ & $10$ & $1, 1, 2, 2$ & $0, 0, \frac{2}{5}, \frac{3}{5}$ & SB:$K=$\tiny $
\begin{pmatrix}
2 & 1 & 0 & 0 \\ 
1 & 2 & 0 & 1 \\ 
0 & 0 & 2 & 1\\
0 & 1 & 1 & 2\\
\end{pmatrix}
$\\
 \hline 
\end{tabular} 
\end{table}

\def\arraystretch{1.25} \setlength\tabcolsep{3pt}
\begin{table}[t] 
\caption{
$Z_2$-SET orders for bosonic systems labeled in terms of sets of topological excitations.
The list contains all topological orders with 
$N=5$ and $D^2\leq 100$.
} 
\label{SETZ2-5} 
\centering
\begin{tabular}{ |c|c|l|l|l| } 
\hline 
$N^{|\Th|}_{c}$ & $D^2$ & $d_1,d_2,\cdots$ & $s_1,s_2,\cdots$ & comment \\
\hline 
$2^{\zeta_{2}^{1}}_{ 0}$ & $2$ & $1, 1$ & $0, 0$ & $\cE=\Rp(Z_2)$\\
\hline
$5^{\zeta_{2}^{1}}_{ 0}$ & $8$ & $1, 1, 1, 1, 2$ & $0, 0, \frac{1}{2}, \frac{1}{2}, 0$ & SB:$4^{ B}_{ 0}$ F:$Z_2\times Z_2$ \\
$5^{\zeta_{2}^{1}}_{ 0}$ & $8$ & $1, 1, 1, 1, 2$ & $0, 0, \frac{1}{2}, \frac{1}{2}, 0$ & SB:$4^{ B}_{ 0}$ F:$Z_4$ NR \\
$5^{\zeta_{2}^{1}}_{ 1}$ & $8$ & $1, 1, 1, 1, 2$ & $0, 0, \frac{1}{2}, \frac{1}{2}, \frac{1}{8}$ & SB:$4^{ B}_{ 1}$  F:$Z_2\times Z_2$\\
$5^{\zeta_{2}^{1}}_{ 1}$ & $8$ & $1, 1, 1, 1, 2$ & $0, 0, \frac{1}{2}, \frac{1}{2}, \frac{1}{8}$ & SB:$4^{ B}_{ 1}$ F:$Z_4$ NR \\
$5^{\zeta_{2}^{1}}_{ 2}$ & $8$ & $1, 1, 1, 1, 2$ & $0, 0, \frac{1}{2}, \frac{1}{2}, \frac{1}{4}$ & SB:$4^{ B}_{ 2}$  F:$Z_2\times Z_2$\\
$5^{\zeta_{2}^{1}}_{ 2}$ & $8$ & $1, 1, 1, 1, 2$ & $0, 0, \frac{1}{2}, \frac{1}{2}, \frac{1}{4}$ & SB:$4^{ B}_{ 2}$ F:$Z_4$ NR \\
$5^{\zeta_{2}^{1}}_{ 3}$ & $8$ & $1, 1, 1, 1, 2$ & $0, 0, \frac{1}{2}, \frac{1}{2}, \frac{3}{8}$ & SB:$4^{ B}_{ 3}$  F:$Z_2\times Z_2$\\
$5^{\zeta_{2}^{1}}_{ 3}$ & $8$ & $1, 1, 1, 1, 2$ & $0, 0, \frac{1}{2}, \frac{1}{2}, \frac{3}{8}$ & SB:$4^{ B}_{ 3}$ F:$Z_4$ NR \\
$5^{\zeta_{2}^{1}}_{ 4}$ & $8$ & $1, 1, 1, 1, 2$ & $0, 0, \frac{1}{2}, \frac{1}{2}, \frac{1}{2}$ & SB:$4^{ B}_{ 4}$ 
\tiny $
\begin{pmatrix}
2 & 1 & 1 & 1 \\ 
1 & 2 & 0 & 0 \\ 
1 & 0 & 2 & 0\\
1 & 0 & 0 & 2\\
\end{pmatrix}
$
\\
$5^{\zeta_{2}^{1}}_{ 4}$ & $8$ & $1, 1, 1, 1, 2$ & $0, 0, \frac{1}{2}, \frac{1}{2}, \frac{1}{2}$ & SB:$4^{ B}_{ 4}$ F:$Z_4$ NR \\
$5^{\zeta_{2}^{1}}_{-3}$ & $8$ & $1, 1, 1, 1, 2$ & $0, 0, \frac{1}{2}, \frac{1}{2}, \frac{5}{8}$ & SB:$4^{ B}_{-3}$  F:$Z_2\times Z_2$\\
$5^{\zeta_{2}^{1}}_{-3}$ & $8$ & $1, 1, 1, 1, 2$ & $0, 0, \frac{1}{2}, \frac{1}{2}, \frac{5}{8}$ & SB:$4^{ B}_{-3}$ F:$Z_4$ NR \\
$5^{\zeta_{2}^{1}}_{-2}$ & $8$ & $1, 1, 1, 1, 2$ & $0, 0, \frac{1}{2}, \frac{1}{2}, \frac{3}{4}$ & SB:$4^{ B}_{-2}$  F:$Z_2\times Z_2$\\
$5^{\zeta_{2}^{1}}_{-2}$ & $8$ & $1, 1, 1, 1, 2$ & $0, 0, \frac{1}{2}, \frac{1}{2}, \frac{3}{4}$ & SB:$4^{ B}_{-2}$ F:$Z_4$ NR \\
$5^{\zeta_{2}^{1}}_{-1}$ & $8$ & $1, 1, 1, 1, 2$ & $0, 0, \frac{1}{2}, \frac{1}{2}, \frac{7}{8}$ & SB:$4^{ B}_{-1}$  F:$Z_2\times Z_2$\\
$5^{\zeta_{2}^{1}}_{-1}$ & $8$ & $1, 1, 1, 1, 2$ & $0, 0, \frac{1}{2}, \frac{1}{2}, \frac{7}{8}$ & SB:$4^{ B}_{-1}$ F:$Z_4$ NR \\
$5^{\zeta_{2}^{1}}_{ 2}$ & $14$ & $1, 1, 2, 2, 2$ & $0, 0, \frac{1}{7}, \frac{2}{7}, \frac{4}{7}$ & SB:$7^{ B}_{ 2}$ \\
$5^{\zeta_{2}^{1}}_{-2}$ & $14$ & $1, 1, 2, 2, 2$ & $0, 0, \frac{3}{7}, \frac{5}{7}, \frac{6}{7}$ & SB:$7^{ B}_{-2}$ \\
$5^{\zeta_{2}^{1}}_{ 12/5}$ & $26.180$ & $1, 1,\zeta_{8}^{2},\zeta_{8}^{2},\zeta_{8}^{4}$ & $0, 0, \frac{1}{5}, \frac{1}{5}, \frac{3}{5}$ & SB:$4^{ B}_{ 12/5}$ \\
$5^{\zeta_{2}^{1}}_{-12/5}$ & $26.180$ & $1, 1,\zeta_{8}^{2},\zeta_{8}^{2},\zeta_{8}^{4}$ & $0, 0, \frac{4}{5}, \frac{4}{5}, \frac{2}{5}$ & SB:$4^{ B}_{-12/5}$ \\
 \hline 
\end{tabular} 
\end{table}

In this appendix, we list $\mce{\cE}$'s for various symmetry $\cE$,  which can
also be viewed as the list of 2+1D SET orders (up to invertible ones) with
symmetry $\cE$.  Those lists are created using a naive calculation, by
checking the necessary conditions on the data $(\tilde N^{ab}_c,\tilde s_a;
N^{ij}_k,s_i)$ (for details, see Appendix \ref{cnds}).  So those lists should
contain all $\mce{\cE}$'s (\ie all SET orders).  But since the conditions are
only known to be necessary, the lists may contain fake entries that do not
correspond to any $\mce{\cE}$ (or any SET order).  In other words, some entries
in the lists have no modular extensions and those entries do not correspond any
real 2+1D SET order. % On the other hand, some other entries in the lists have modular extensions and they correspond to existing 2+1D SET orders.  

The entries with known decomposition $N^B_c\boxtimes \Rp(G)$ or $N^B_c\boxtimes
\sRp(G^f)$, or with given $K$-matrix in the comment column all  correspond to
existing 2+1D SET orders.  (The topological orders described by $N^B_c$ are
given by the tables in \Ref{W150605768}.) Other entries may or may not
correspond to existing 2+1D SET orders, which need to be determined by checking
the existence of modular extensions.

Even for the entries that have modular extensions, some times they may
correspond to more than one $\mce{\cE}$'s.  This is because $(\tilde
N^{ab}_c,\tilde s_a;N^{ij}_k,s_i)$ cannot distinguish all different
$\mce{\cE}$'s.

\subsection{$Z_2$-SET orders}

Tables \ref{SETZ2-34}, \ref{SETZ2-5}, and \ref{SETZ2-6}
list the $Z_2$-SET
orders  (up to invertible ones) for 2+1D bosonic systems.  
For bosonic systems the central charge is determined up to $8$ by the bulk
excitations.  The $3^{\zeta_{2}^{1}}_{2}$ states and the two
$4^{\zeta_{2}^{1}}_{ 1}$ states in Table \ref{SETZ2-34} are discussed in the
main text.

All the  $Z_2$-SET orders in Table \ref{SETZ2-34} are
realizable.  Some of the them are realized as $N^B_c \boxtimes \Rp(Z_2)$, as
indicated in the comment column.  Here $N^B_c$ describes a neutral bosonic
topological order (which was denoted as $N^B_c$ in \Ref{W150605768}) with rank
$N$ and central charge $c$, which does not transform under the $Z_2$ symmetry.
For example $2^B_1$ is the $\nu=1/2$ bosonic Laughlin state, and $2^B_{14/5}$
is the bosonic Fibonacci state\cite{W150605768}.  $\Rp(Z_2)$ describes a
product state with $Z_2$ symmetry of $Z_2$ charged bosons.  $N^B_c \boxtimes
\Rp(Z_2)$ is simply the stacking of the neutral bosonic topological order
$N^B_c$ with the $Z_2$ symmetric product state.

We also introduced $N^B_c \boxtimes^t \Rp(Z_2)$ which describe a state
similar to $N^B_c \boxtimes \Rp(Z_2)$, except here the bosons that form the
topological order $N^B_c$ also carries a $Z_2$ charge.
The $3^{\zeta_2^1}_2$ state can be realized by double-layer FQH state with
$K$-matrix $\bpm 2&-1\\ -1&2\epm$, which is discussed in the main text.

Since we did not use the condition of the existence of modular extensions when
we calculate the tables, some the entries in the tables may not by realizable
by any 2+1D bosonic systems. We use NR in the comment column
to indicate such entries (see Table \ref{SETZ2-5}).

\def\arraystretch{1.25} \setlength\tabcolsep{3pt}
\begin{table*}[t] 
\caption{
$Z_2$-SET orders for bosonic systems labeled in terms of sets of topological excitations.
The list contains all topological orders with $N=6$ $D^2\leq 50$.
} 
\label{SETZ2-6} 
\centering
\begin{tabular}{ |c|c|l|l|l| } 
\hline 
$N^{|\Th|}_{c}$ & $D^2$ & $d_1,d_2,\cdots$ & $s_1,s_2,\cdots$ & comment \\
\hline 
$2^{\zeta_{2}^{1}}_{ 0}$ & $2$ & $1, 1$ & $0, 0$ & $\cE=\Rp(Z_2)$\\
\hline
$6^{\zeta_{2}^{1}}_{ 2}$ & $6$ & $1, 1, 1, 1, 1, 1$ & $0, 0, \frac{1}{3}, \frac{1}{3}, \frac{1}{3}, \frac{1}{3}$ & $3^{ B}_{ 2}\boxtimes \Rp(Z_2)$\\
$6^{\zeta_{2}^{1}}_{-2}$ & $6$ & $1, 1, 1, 1, 1, 1$ & $0, 0, \frac{2}{3}, \frac{2}{3}, \frac{2}{3}, \frac{2}{3}$ & $3^{ B}_{-2}\boxtimes \Rp(Z_2)$\\
$6^{\zeta_{2}^{1}}_{ 1/2}$ & $8$ & $1, 1, 1, 1,\zeta_{2}^{1},\zeta_{2}^{1}$ & $0, 0, \frac{1}{2}, \frac{1}{2}, \frac{1}{16}, \frac{1}{16}$ & $3^{ B}_{ 1/2}\boxtimes \Rp(Z_2)$\\
$6^{\zeta_{2}^{1}}_{ 1/2}$ & $8$ & $1, 1, 1, 1,\zeta_{2}^{1},\zeta_{2}^{1}$ & $0, 0, \frac{1}{2}, \frac{1}{2}, \frac{1}{16}, \frac{1}{16}$ & SB:$3^{ B}_{ 1/2}$ \\
$6^{\zeta_{2}^{1}}_{ 3/2}$ & $8$ & $1, 1, 1, 1,\zeta_{2}^{1},\zeta_{2}^{1}$ & $0, 0, \frac{1}{2}, \frac{1}{2}, \frac{3}{16}, \frac{3}{16}$ & $3^{ B}_{ 3/2}\boxtimes \Rp(Z_2)$\\
$6^{\zeta_{2}^{1}}_{ 3/2}$ & $8$ & $1, 1, 1, 1,\zeta_{2}^{1},\zeta_{2}^{1}$ & $0, 0, \frac{1}{2}, \frac{1}{2}, \frac{3}{16}, \frac{3}{16}$ & SB:$3^{ B}_{ 3/2}$ \\
$6^{\zeta_{2}^{1}}_{ 5/2}$ & $8$ & $1, 1, 1, 1,\zeta_{2}^{1},\zeta_{2}^{1}$ & $0, 0, \frac{1}{2}, \frac{1}{2}, \frac{5}{16}, \frac{5}{16}$ & $3^{ B}_{ 5/2}\boxtimes \Rp(Z_2)$\\
$6^{\zeta_{2}^{1}}_{ 5/2}$ & $8$ & $1, 1, 1, 1,\zeta_{2}^{1},\zeta_{2}^{1}$ & $0, 0, \frac{1}{2}, \frac{1}{2}, \frac{5}{16}, \frac{5}{16}$ & SB:$3^{ B}_{ 5/2}$ \\
$6^{\zeta_{2}^{1}}_{ 7/2}$ & $8$ & $1, 1, 1, 1,\zeta_{2}^{1},\zeta_{2}^{1}$ & $0, 0, \frac{1}{2}, \frac{1}{2}, \frac{7}{16}, \frac{7}{16}$ & $3^{ B}_{ 7/2}\boxtimes \Rp(Z_2)$\\
$6^{\zeta_{2}^{1}}_{ 7/2}$ & $8$ & $1, 1, 1, 1,\zeta_{2}^{1},\zeta_{2}^{1}$ & $0, 0, \frac{1}{2}, \frac{1}{2}, \frac{7}{16}, \frac{7}{16}$ & SB:$3^{ B}_{ 7/2}$ \\
$6^{\zeta_{2}^{1}}_{-7/2}$ & $8$ & $1, 1, 1, 1,\zeta_{2}^{1},\zeta_{2}^{1}$ & $0, 0, \frac{1}{2}, \frac{1}{2}, \frac{9}{16}, \frac{9}{16}$ & $3^{ B}_{-7/2}\boxtimes \Rp(Z_2)$\\
$6^{\zeta_{2}^{1}}_{-7/2}$ & $8$ & $1, 1, 1, 1,\zeta_{2}^{1},\zeta_{2}^{1}$ & $0, 0, \frac{1}{2}, \frac{1}{2}, \frac{9}{16}, \frac{9}{16}$ & SB:$3^{ B}_{-7/2}$ \\
$6^{\zeta_{2}^{1}}_{-5/2}$ & $8$ & $1, 1, 1, 1,\zeta_{2}^{1},\zeta_{2}^{1}$ & $0, 0, \frac{1}{2}, \frac{1}{2}, \frac{11}{16}, \frac{11}{16}$ & $3^{ B}_{-5/2}\boxtimes \Rp(Z_2)$\\
$6^{\zeta_{2}^{1}}_{-5/2}$ & $8$ & $1, 1, 1, 1,\zeta_{2}^{1},\zeta_{2}^{1}$ & $0, 0, \frac{1}{2}, \frac{1}{2}, \frac{11}{16}, \frac{11}{16}$ & SB:$3^{ B}_{-5/2}$ \\
$6^{\zeta_{2}^{1}}_{-3/2}$ & $8$ & $1, 1, 1, 1,\zeta_{2}^{1},\zeta_{2}^{1}$ & $0, 0, \frac{1}{2}, \frac{1}{2}, \frac{13}{16}, \frac{13}{16}$ & $3^{ B}_{-3/2}\boxtimes \Rp(Z_2)$\\
$6^{\zeta_{2}^{1}}_{-3/2}$ & $8$ & $1, 1, 1, 1,\zeta_{2}^{1},\zeta_{2}^{1}$ & $0, 0, \frac{1}{2}, \frac{1}{2}, \frac{13}{16}, \frac{13}{16}$ & SB:$3^{ B}_{-3/2}$ \\
$6^{\zeta_{2}^{1}}_{-1/2}$ & $8$ & $1, 1, 1, 1,\zeta_{2}^{1},\zeta_{2}^{1}$ & $0, 0, \frac{1}{2}, \frac{1}{2}, \frac{15}{16}, \frac{15}{16}$ & $3^{ B}_{-1/2}\boxtimes \Rp(Z_2)$\\
$6^{\zeta_{2}^{1}}_{-1/2}$ & $8$ & $1, 1, 1, 1,\zeta_{2}^{1},\zeta_{2}^{1}$ & $0, 0, \frac{1}{2}, \frac{1}{2}, \frac{15}{16}, \frac{15}{16}$ & SB:$3^{ B}_{-1/2}$ \\
$6^{\zeta_{2}^{1}}_{ 1}$ & $12$ & $1, 1, 1, 1, 2, 2$ & $0, 0, \frac{3}{4}, \frac{3}{4}, \frac{1}{12}, \frac{1}{3}$ & SB:$6^{ B}_{ 1}$ \\
$6^{\zeta_{2}^{1}}_{ 3}$ & $12$ & $1, 1, 1, 1, 2, 2$ & $0, 0, \frac{1}{4}, \frac{1}{4}, \frac{1}{3}, \frac{7}{12}$ & SB:$6^{ B}_{ 3}$ \\
$6^{\zeta_{2}^{1}}_{-3}$ & $12$ & $1, 1, 1, 1, 2, 2$ & $0, 0, \frac{3}{4}, \frac{3}{4}, \frac{5}{12}, \frac{2}{3}$ & SB:$6^{ B}_{-3}$ \\
$6^{\zeta_{2}^{1}}_{-1}$ & $12$ & $1, 1, 1, 1, 2, 2$ & $0, 0, \frac{1}{4}, \frac{1}{4}, \frac{2}{3}, \frac{11}{12}$ & SB:$6^{ B}_{-1}$ \\
$6^{\zeta_{2}^{1}}_{ 0}$ & $18$ & $1, 1, 2, 2, 2, 2$ & $0, 0, 0, 0, \frac{1}{3}, \frac{2}{3}$ & SB:$9^{ B}_{ 0}$ \\
$6^{\zeta_{2}^{1}}_{ 0}$ & $18$ & $1, 1, 2, 2, 2, 2$ & $0, 0, 0, \frac{1}{9}, \frac{4}{9}, \frac{7}{9}$ & SB:$9^{ B}_{ 0}$ \\
$6^{\zeta_{2}^{1}}_{ 0}$ & $18$ & $1, 1, 2, 2, 2, 2$ & $0, 0, 0, \frac{2}{9}, \frac{5}{9}, \frac{8}{9}$ & SB:$9^{ B}_{ 0}$ \\
$6^{\zeta_{2}^{1}}_{ 4}$ & $18$ & $1, 1, 2, 2, 2, 2$ & $0, 0, \frac{1}{3}, \frac{1}{3}, \frac{2}{3}, \frac{2}{3}$ & SB:$9^{ B}_{ 4}$ \\
$6^{\zeta_{2}^{1}}_{ 8/7}$ & $18.591$ & $1, 1,\zeta_{5}^{1},\zeta_{5}^{1},\zeta_{5}^{2},\zeta_{5}^{2}$ & $0, 0, \frac{6}{7}, \frac{6}{7}, \frac{2}{7}, \frac{2}{7}$ & $3^{ B}_{ 8/7}\boxtimes \Rp(Z_2)$\\
$6^{\zeta_{2}^{1}}_{-8/7}$ & $18.591$ & $1, 1,\zeta_{5}^{1},\zeta_{5}^{1},\zeta_{5}^{2},\zeta_{5}^{2}$ & $0, 0, \frac{1}{7}, \frac{1}{7}, \frac{5}{7}, \frac{5}{7}$ & $3^{ B}_{-8/7}\boxtimes \Rp(Z_2)$\\
$6^{\zeta_{2}^{1}}_{ 4/5}$ & $21.708$ & $1, 1,\zeta_{3}^{1},\zeta_{3}^{1}, 2,\zeta_{8}^{4}$ & $0, 0, \frac{2}{5}, \frac{2}{5}, \frac{2}{3}, \frac{1}{15}$ & $2^{ B}_{ 14/5}\boxtimes 3^{\zeta_{2}^{1}}_{-2}$\\
$6^{\zeta_{2}^{1}}_{ 16/5}$ & $21.708$ & $1, 1,\zeta_{3}^{1},\zeta_{3}^{1}, 2,\zeta_{8}^{4}$ & $0, 0, \frac{3}{5}, \frac{3}{5}, \frac{2}{3}, \frac{4}{15}$ & $2^{ B}_{-14/5}\boxtimes 3^{\zeta_{2}^{1}}_{-2}$\\
$6^{\zeta_{2}^{1}}_{-16/5}$ & $21.708$ & $1, 1,\zeta_{3}^{1},\zeta_{3}^{1}, 2,\zeta_{8}^{4}$ & $0, 0, \frac{2}{5}, \frac{2}{5}, \frac{1}{3}, \frac{11}{15}$ & $2^{ B}_{ 14/5}\boxtimes 3^{\zeta_{2}^{1}}_{ 2}$\\
$6^{\zeta_{2}^{1}}_{-4/5}$ & $21.708$ & $1, 1,\zeta_{3}^{1},\zeta_{3}^{1}, 2,\zeta_{8}^{4}$ & $0, 0, \frac{3}{5}, \frac{3}{5}, \frac{1}{3}, \frac{14}{15}$ & $2^{ B}_{-14/5}\boxtimes 3^{\zeta_{2}^{1}}_{ 2}$\\
 \hline 
\end{tabular} 
\end{table*}

\subsection{$Z_3$-SET orders}

Table \ref{SETZ3} lists the $Z_3$-SET orders  (up
to invertible ones) for 2+1D bosonic systems.  

The $Z_3$-SET state $4^{\zeta_{4}^{1}}_{ 4}$ in the table
becomes the $K=\tiny \begin{pmatrix}
2 & 1 & 1 & 1 \\ 
1 & 2 & 0 & 0 \\ 
1 & 0 & 2 & 0\\
1 & 0 & 0 & 2\\
\end{pmatrix}
$
4-layer FQH state after we break the $Z_3$-symmetry.  We can add the
$Z_3$-symmetry back to obtain the  $Z_3$-SET state.  The $Z_3$-symmetry is the
cyclic permutation of the second, the third, and the fourth layers.
 
Without the symmetry,
the $K=\tiny \begin{pmatrix}
2 & 1 & 1 & 1 \\ 
1 & 2 & 0 & 0 \\ 
1 & 0 & 2 & 0\\
1 & 0 & 0 & 2\\
\end{pmatrix}
$
state has four types of particles, a trivial boson and three non-trivial
fermions.  With the symmetry, the  three fermions become degenerate and is
combined into the $d=3$ particle (the fourth particle) for the
$4^{\zeta_{4}^{1}}_{ 4}$ state.
The first three particles for the
$4^{\zeta_{4}^{1}}_{ 4}$ state all come from the trivial boson.
They carry different $Z_3$ charges: $0,1,2$, in the presence
of the symmetry.

\def\arraystretch{1.25} \setlength\tabcolsep{3pt}
\begin{table*}[t] 
\caption{
$Z_3$-SET orders for bosonic systems labeled in terms of sets of topological excitations.
The list contains all topological orders with 
$N=4,5,6$ $D^2\leq 100$,
$N=7$ $D^2\leq 60$,
$N=8$ $D^2\leq 40$, and
$N=9$ $D^2\leq 28$.
} 
\label{SETZ3} 
\centering
\begin{tabular}{ |c|c|l|l|l| } 
\hline 
$N^{|\Th|}_{c}$ & $D^2$ & $d_1,d_2,\cdots$ & $s_1,s_2,\cdots$ & comment \\
\hline 
$3^{\zeta_{4}^{1}}_{ 0}$ & $3$ & $1, 1, 1$ & $0, 0, 0$ & $\cE=\Rp(Z_3)$\\
\hline
$4^{\zeta_{4}^{1}}_{ 4}$ & $12$ & $1, 1, 1, 3$ & $0, 0, 0, \frac{1}{2}$ & 
SB:$K=$\tiny $\begin{pmatrix}
2 & 1 & 1 & 1 \\ 
1 & 2 & 0 & 0 \\ 
1 & 0 & 2 & 0\\
1 & 0 & 0 & 2\\
\end{pmatrix}
$
\\
\hline
$6^{\zeta_{4}^{1}}_{ 1}$ & $6$ & $1, 1, 1, 1, 1, 1$ & $0, 0, 0, \frac{1}{4}, \frac{1}{4}, \frac{1}{4}$ & $2^B_1\boxtimes \text{Rep}(Z_3)$ \\
$6^{\zeta_{4}^{1}}_{-1}$ & $6$ & $1, 1, 1, 1, 1, 1$ & $0, 0, 0, \frac{3}{4}, \frac{3}{4}, \frac{3}{4}$ & $2^B_{-1}\boxtimes \text{Rep}(Z_3)$ \\
$6^{\zeta_{4}^{1}}_{ 14/5}$ & $10.854$ & $1, 1, 1,\zeta_{3}^{1},\zeta_{3}^{1},\zeta_{3}^{1}$ & $0, 0, 0, \frac{2}{5}, \frac{2}{5}, \frac{2}{5}$ &  $2^B_{ 14/5}\boxtimes \text{Rep}(Z_3)$\\
$6^{\zeta_{4}^{1}}_{-14/5}$ & $10.854$ & $1, 1, 1,\zeta_{3}^{1},\zeta_{3}^{1},\zeta_{3}^{1}$ & $0, 0, 0, \frac{3}{5}, \frac{3}{5}, \frac{3}{5}$ &  $2^B_{-14/5}\boxtimes \text{Rep}(Z_3)$\\
\hline
$8^{\zeta_{4}^{1}}_{ 3}$ & $24$ & $1, 1, 1, 1, 1, 1, 3, 3$ & $0, 0, 0, \frac{3}{4}, \frac{3}{4}, \frac{3}{4}, \frac{1}{4}, \frac{1}{2}$ & $2^{ B}_{-1}\boxtimes 4^{\zeta_{4}^{1}}_{ 4}$\\
$8^{\zeta_{4}^{1}}_{-3}$ & $24$ & $1, 1, 1, 1, 1, 1, 3, 3$ & $0, 0, 0, \frac{1}{4}, \frac{1}{4}, \frac{1}{4}, \frac{1}{2}, \frac{3}{4}$ & $2^{ B}_{ 1}\boxtimes 4^{\zeta_{4}^{1}}_{ 4}$\\
$8^{\zeta_{4}^{1}}_{ 6/5}$ & $43.416$ & $1, 1, 1,\zeta_{3}^{1},\zeta_{3}^{1},\zeta_{3}^{1}, 3,\frac{3+\sqrt{45}}{2}$ & $0, 0, 0, \frac{3}{5}, \frac{3}{5}, \frac{3}{5}, \frac{1}{2}, \frac{1}{10}$ & $2^{ B}_{-14/5}\boxtimes 4^{\zeta_{4}^{1}}_{ 4}$\\
$8^{\zeta_{4}^{1}}_{-6/5}$ & $43.416$ & $1, 1, 1,\zeta_{3}^{1},\zeta_{3}^{1},\zeta_{3}^{1}, 3,\frac{3+\sqrt{45}}{2}$ & $0, 0, 0, \frac{2}{5}, \frac{2}{5}, \frac{2}{5}, \frac{1}{2}, \frac{9}{10}$ & $2^{ B}_{ 14/5}\boxtimes 4^{\zeta_{4}^{1}}_{ 4}$\\
 \hline 
$9^{\zeta_{4}^{1}}_{ 2}$ & $9$ & $1, 1, 1, 1, 1, 1, 1, 1, 1$ & $0, 0, 0, \frac{1}{3}, \frac{1}{3}, \frac{1}{3}, \frac{1}{3}, \frac{1}{3}, \frac{1}{3}$ &  SB:$3^{ B}_{ 2}$ F:$Z_9$\\
$9^{\zeta_{4}^{1}}_{ 2}$ & $9$ & $1, 1, 1, 1, 1, 1, 1, 1, 1$ & $0, 0, 0, \frac{1}{3}, \frac{1}{3}, \frac{1}{3}, \frac{1}{3}, \frac{1}{3}, \frac{1}{3}$ & $3^{ B}_{ 2}\boxtimes \Rp(Z_3)$ F:$Z_3\times Z_3$\\
$9^{\zeta_{4}^{1}}_{-2}$ & $9$ & $1, 1, 1, 1, 1, 1, 1, 1, 1$ & $0, 0, 0, \frac{2}{3}, \frac{2}{3}, \frac{2}{3}, \frac{2}{3}, \frac{2}{3}, \frac{2}{3}$ & SB:$3^{ B}_{-2}$ F:$Z_9$\\
$9^{\zeta_{4}^{1}}_{-2}$ & $9$ & $1, 1, 1, 1, 1, 1, 1, 1, 1$ & $0, 0, 0, \frac{2}{3}, \frac{2}{3}, \frac{2}{3}, \frac{2}{3}, \frac{2}{3}, \frac{2}{3}$ & $3^{ B}_{-2}\boxtimes \Rp(Z_3)$ F:$Z_3\times Z_3$\\
$9^{\zeta_{4}^{1}}_{ 1/2}$ & $12$ & $1, 1, 1, 1, 1, 1,\zeta_{2}^{1},\zeta_{2}^{1},\zeta_{2}^{1}$ & $0, 0, 0, \frac{1}{2}, \frac{1}{2}, \frac{1}{2}, \frac{1}{16}, \frac{1}{16}, \frac{1}{16}$ & $3^{ B}_{ 1/2}\boxtimes \Rp(Z_3)$\\
$9^{\zeta_{4}^{1}}_{ 3/2}$ & $12$ & $1, 1, 1, 1, 1, 1,\zeta_{2}^{1},\zeta_{2}^{1},\zeta_{2}^{1}$ & $0, 0, 0, \frac{1}{2}, \frac{1}{2}, \frac{1}{2}, \frac{3}{16}, \frac{3}{16}, \frac{3}{16}$ & $3^{ B}_{ 3/2}\boxtimes \Rp(Z_3)$\\
$9^{\zeta_{4}^{1}}_{ 5/2}$ & $12$ & $1, 1, 1, 1, 1, 1,\zeta_{2}^{1},\zeta_{2}^{1},\zeta_{2}^{1}$ & $0, 0, 0, \frac{1}{2}, \frac{1}{2}, \frac{1}{2}, \frac{5}{16}, \frac{5}{16}, \frac{5}{16}$ & $3^{ B}_{ 5/2}\boxtimes \Rp(Z_3)$\\
$9^{\zeta_{4}^{1}}_{ 7/2}$ & $12$ & $1, 1, 1, 1, 1, 1,\zeta_{2}^{1},\zeta_{2}^{1},\zeta_{2}^{1}$ & $0, 0, 0, \frac{1}{2}, \frac{1}{2}, \frac{1}{2}, \frac{7}{16}, \frac{7}{16}, \frac{7}{16}$ & $3^{ B}_{ 7/2}\boxtimes \Rp(Z_3)$\\
$9^{\zeta_{4}^{1}}_{-7/2}$ & $12$ & $1, 1, 1, 1, 1, 1,\zeta_{2}^{1},\zeta_{2}^{1},\zeta_{2}^{1}$ & $0, 0, 0, \frac{1}{2}, \frac{1}{2}, \frac{1}{2}, \frac{9}{16}, \frac{9}{16}, \frac{9}{16}$ & $3^{ B}_{-7/2}\boxtimes \Rp(Z_3)$\\
$9^{\zeta_{4}^{1}}_{-5/2}$ & $12$ & $1, 1, 1, 1, 1, 1,\zeta_{2}^{1},\zeta_{2}^{1},\zeta_{2}^{1}$ & $0, 0, 0, \frac{1}{2}, \frac{1}{2}, \frac{1}{2}, \frac{11}{16}, \frac{11}{16}, \frac{11}{16}$ & $3^{ B}_{-5/2}\boxtimes \Rp(Z_3)$\\
$9^{\zeta_{4}^{1}}_{-3/2}$ & $12$ & $1, 1, 1, 1, 1, 1,\zeta_{2}^{1},\zeta_{2}^{1},\zeta_{2}^{1}$ & $0, 0, 0, \frac{1}{2}, \frac{1}{2}, \frac{1}{2}, \frac{13}{16}, \frac{13}{16}, \frac{13}{16}$ & $3^{ B}_{-3/2}\boxtimes \Rp(Z_3)$\\
$9^{\zeta_{4}^{1}}_{-1/2}$ & $12$ & $1, 1, 1, 1, 1, 1,\zeta_{2}^{1},\zeta_{2}^{1},\zeta_{2}^{1}$ & $0, 0, 0, \frac{1}{2}, \frac{1}{2}, \frac{1}{2}, \frac{15}{16}, \frac{15}{16}, \frac{15}{16}$ & $3^{ B}_{-1/2}\boxtimes \Rp(Z_3)$\\
$9^{\zeta_{4}^{1}}_{ 8/7}$ & $27.887$ & $1, 1, 1,\zeta_{5}^{1},\zeta_{5}^{1},\zeta_{5}^{1},\zeta_{5}^{2},\zeta_{5}^{2},\zeta_{5}^{2}$ & $0, 0, 0, \frac{6}{7}, \frac{6}{7}, \frac{6}{7}, \frac{2}{7}, \frac{2}{7}, \frac{2}{7}$ & $3^{ B}_{ 8/7}\boxtimes \Rp(Z_3)$\\
$9^{\zeta_{4}^{1}}_{-8/7}$ & $27.887$ & $1, 1, 1,\zeta_{5}^{1},\zeta_{5}^{1},\zeta_{5}^{1},\zeta_{5}^{2},\zeta_{5}^{2},\zeta_{5}^{2}$ & $0, 0, 0, \frac{1}{7}, \frac{1}{7}, \frac{1}{7}, \frac{5}{7}, \frac{5}{7}, \frac{5}{7}$ & $3^{ B}_{-8/7}\boxtimes \Rp(Z_3)$\\
 \hline 
\end{tabular} 
\end{table*}

\subsection{$S_3$-SET orders}

\def\arraystretch{1.25} \setlength\tabcolsep{3pt}
\begin{table}[t] 
\caption{
The fusion rules for the three $5^{\sqrt{6}}_{ 4}$ entries in Table \ref{SETS3-567}.
The three entries have identical $(d_i,s_i)$ but different fusions rules.
$\textbf{1}$, $a$, $b$ are the three irreducible representations
of $S_3$ with dimension 1, 1, 2.
} 
\label{frS3} 
\centering
\begin{tabular}{ |c|ccccc|}
 \hline 
 $s_i$ & $0$ & $ 0$ & $ 0$ & $ \frac{1}{2}$ & $ \frac{1}{2}$\\
 $d_i$ & $1$ & $ 1$ & $ 2$ & $ 3$ & $ 3$\\
\hline
 $5^{\sqrt{6}}_{ 4}$ & $\textbf{1}$  & $a$  & $b$  & $\sigma$  & $\tau$ \\
\hline
$\textbf{1}$  & $ \textbf{1}$  & $ a$  & $ b$  & $ \sigma$  & $ \tau$  \\
$a$  & $ a$  & $ \textbf{1}$  & $ b$  & $ \tau$  & $ \sigma$  \\
$b$  & $ b$  & $ b$  & $ \textbf{1} \oplus a \oplus b$  & $ \sigma \oplus \tau$  & $ \sigma \oplus \tau$  \\
$\sigma$  & $ \sigma$  & $ \tau$  & $ \sigma \oplus \tau$  & $ \textbf{1} \oplus b \oplus 2\sigma$  & $ a \oplus b \oplus 2\tau$  \\
$\tau$  & $ \tau$  & $ \sigma$  & $ \sigma \oplus \tau$  & $ a \oplus b \oplus 2\tau$  & $ \textbf{1} \oplus b \oplus 2\sigma$  \\
\hline
\end{tabular}
\\[3mm]
\begin{tabular}{ |c|ccccc|}
 \hline 
 $s_i$ & $0$ & $ 0$ & $ 0$ & $ \frac{1}{2}$ & $ \frac{1}{2}$\\
 $d_i$ & $1$ & $ 1$ & $ 2$ & $ 3$ & $ 3$\\
\hline
 $5^{\sqrt{6}}_{ 4}$ & $\textbf{1}$  & $a$  & $b$  & $\sigma$  & $\tau$ \\
\hline
$\textbf{1}$  & $ \textbf{1}$  & $ a$  & $ b$  & $ \sigma$  & $ \tau$  \\
$a$  & $ a$  & $ \textbf{1}$  & $ b$  & $ \tau$  & $ \sigma$  \\
$b$  & $ b$  & $ b$  & $ \textbf{1} \oplus a \oplus b$  & $ \sigma \oplus \tau$  & $ \sigma \oplus \tau$  \\
$\sigma$  & $ \sigma$  & $ \tau$  & $ \sigma \oplus \tau$  & $ \textbf{1} \oplus b \oplus \sigma \oplus \tau$  & $ a \oplus b \oplus \sigma \oplus \tau$  \\
$\tau$  & $ \tau$  & $ \sigma$  & $ \sigma \oplus \tau$  & $ a \oplus b \oplus \sigma \oplus \tau$  & $ \textbf{1} \oplus b \oplus \sigma \oplus \tau$  \\
\hline
\end{tabular}
\\[3mm]
\begin{tabular}{ |c|ccccc|}
 \hline 
 $s_i$ & $0$ & $ 0$ & $ 0$ & $ \frac{1}{2}$ & $ \frac{1}{2}$\\
 $d_i$ & $1$ & $ 1$ & $ 2$ & $ 3$ & $ 3$\\
\hline
 $5^{\sqrt{6}}_{ 4}$ & $\textbf{1}$  & $a$  & $b$  & $\sigma$  & $\tau$ \\
\hline
$\textbf{1}$  & $ \textbf{1}$  & $ a$  & $ b$  & $ \sigma$  & $ \tau$  \\
$a$  & $ a$  & $ \textbf{1}$  & $ b$  & $ \tau$  & $ \sigma$  \\
$b$  & $ b$  & $ b$  & $ \textbf{1} \oplus a \oplus b$  & $ \sigma \oplus \tau$  & $ \sigma \oplus \tau$  \\
$\sigma$  & $ \sigma$  & $ \tau$  & $ \sigma \oplus \tau$  & $ a \oplus b \oplus \sigma \oplus \tau$  & $ \textbf{1} \oplus b \oplus \sigma \oplus \tau$  \\
$\tau$  & $ \tau$  & $ \sigma$  & $ \sigma \oplus \tau$  & $ \textbf{1} \oplus b \oplus \sigma \oplus \tau$  & $ a \oplus b \oplus \sigma \oplus \tau$  \\
\hline
\end{tabular}
\end{table}

Tables \ref{SETS3-567} and \ref{SETS3-9}
list the $S_3$-SET orders  (up to invertible ones) for
2+1D bosonic systems.  

Table \ref{SETS3-567} has three $5^{\sqrt{6}}_{ 4}$ entries that have
identical $(d_i,s_i)$.  But the three entries have different fusion rules (see
Table \ref{frS3}).  If we break the symmetry, the three entries all reduce to
the $K=\tiny \begin{pmatrix}
2 & 1 & 1 & 1 \\ 
1 & 2 & 0 & 0 \\ 
1 & 0 & 2 & 0\\
1 & 0 & 0 & 2\\
\end{pmatrix}
$
4-layer state.  So we expect the $S_3$ symmetry is the permutation symmetry of
the second, the third, and the fourth layers. 

The second
$5^{\sqrt{6}}_{ 4}$ entry can be realized
by the $K=\tiny \begin{pmatrix}
2 & 1 & 1 & 1 \\ 
1 & 2 & 0 & 0 \\ 
1 & 0 & 2 & 0\\
1 & 0 & 0 & 2\\
\end{pmatrix}
$
4-layer state.
The two $d=3$ fermions are the direct-sum of the three degenerate
fermions in the $K=\tiny \begin{pmatrix}
2 & 1 & 1 & 1 \\ 
1 & 2 & 0 & 0 \\ 
1 & 0 & 2 & 0\\
1 & 0 & 0 & 2\\
\end{pmatrix}
$ state.
They carry the following $S_3$ representations
\begin{align}
\label{S3ch}
 \sigma \to \textbf{1}\oplus b,\ \ \ \
 \tau \to a\oplus b.
\end{align}
It is strange that two different irreducible representations are degenerate in
energy. But this can happen for topological excitations in the presence of
symmetry.

Such an assignment of the $S_3$-representations (or $S_3$ ``charges'')
is consistent with the fusion rule (see the second table in Table \ref{frS3}).
For example
\begin{align}
& \ \ \ \
\sigma \otimes \sigma \to \textbf{1} \oplus  2 b \oplus  b\otimes b
=\textbf{1} \oplus  2 b \oplus   (\textbf{1} \oplus  a \oplus b)  
\nonumber\\
& \to  \textbf{1} \oplus  b \oplus   \sigma \oplus \tau
\end{align}
This is why we say that
the second
$5^{\sqrt{6}}_{ 4}$ entry can be realized
by the $K=\tiny \begin{pmatrix}
2 & 1 & 1 & 1 \\ 
1 & 2 & 0 & 0 \\ 
1 & 0 & 2 & 0\\
1 & 0 & 0 & 2\\
\end{pmatrix}
$
state.

However, the $S_3$-charge assignment \eqn{S3ch} does not work for the first and
the third $5^{\sqrt{6}}_{ 4}$ entries (\ie inconsistent with fusion rules in
the first and the third tables in Table \ref{frS3}).  In fact, none of the
$S_3$-charge assignment works.  This mean that the $d=3$ fermions in the first
and the third $5^{\sqrt{6}}_{ 4}$ entries must carry fractionalized
$S_3$-charges or fractionalized $S_3$-representations.  It is not clear if such
fractionalized $S_3$-charges are realizable or not, since we cannot calculate
the modular extensions for those entries (due to the limitation of computer
power).

\def\arraystretch{1.25} \setlength\tabcolsep{3pt}
\begin{table*}[t] 
\caption{
$S_3$-SET orders for bosonic systems labeled in terms of sets of topological excitations.
The list contains all topological orders with 
$N=4,5,6$ $D^2\leq 100$,
$N=7$ $D^2\leq 60$, and
$N=8$ $D^2\leq 40$. (In fact, we fail to find any bosonic $S_3$-SET orders with $N=4,7,8$.)
} 
\label{SETS3-567} 
\centering
\begin{tabular}{ |c|c|l|l|l| } 
\hline 
$N^{|\Th|}_{c}$ & $D^2$ & $d_1,d_2,\cdots$ & $s_1,s_2,\cdots$ & comment \\
\hline 
$3^{\sqrt{6}}_{ 0}$ & $6$ & $1, 1, 2$ & $0, 0, 0$ & $\cE=\Rp(S_3)$\\
\hline
$5^{\sqrt{6}}_{ 4}$ & $24$ & $1, 1, 2, 3, 3$ & $0, 0, 0, \frac{1}{2}, \frac{1}{2}$ & 
SB:$4^{ B}_{ 4}$ 
\\
$5^{\sqrt{6}}_{ 4}$ & $24$ & $1, 1, 2, 3, 3$ & $0, 0, 0, \frac{1}{2}, \frac{1}{2}$ & SB:$4^{ B}_{ 4}$ 
\tiny $
\begin{pmatrix}
2 & 1 & 1 & 1 \\ 
1 & 2 & 0 & 0 \\ 
1 & 0 & 2 & 0\\
1 & 0 & 0 & 2\\
\end{pmatrix}
$
\\
$5^{\sqrt{6}}_{ 4}$ & $24$ & $1, 1, 2, 3, 3$ & $0, 0, 0, \frac{1}{2}, \frac{1}{2}$ & SB:$4^{ B}_{ 4}$ \\
\hline
$6^{\sqrt{6}}_{ 1}$ & $12$ & $1, 1, 2, 1, 1, 2$ & $0, 0, 0, \frac{1}{4}, \frac{1}{4}, \frac{1}{4}$ & $2^{ B}_{ 1}\boxtimes \Rp(S_3)$\\
$6^{\sqrt{6}}_{ 1}$ & $12$ & $1, 1, 2, 1, 1, 2$ & $0, 0, 0, \frac{1}{4}, \frac{1}{4}, \frac{1}{4}$ & SB:$2^{ B}_{ 1}$ \\
$6^{\sqrt{6}}_{-1}$ & $12$ & $1, 1, 2, 1, 1, 2$ & $0, 0, 0, \frac{3}{4}, \frac{3}{4}, \frac{3}{4}$ & $2^{ B}_{-1}\boxtimes \Rp(S_3)$\\
$6^{\sqrt{6}}_{-1}$ & $12$ & $1, 1, 2, 1, 1, 2$ & $0, 0, 0, \frac{3}{4}, \frac{3}{4}, \frac{3}{4}$ & SB:$2^{ B}_{-1}$ \\
$6^{\sqrt{6}}_{ 2}$ & $18$ & $1, 1, 2, 2, 2, 2$ & $0, 0, 0, \frac{1}{3}, \frac{1}{3}, \frac{1}{3}$ & SB:$3^{ B}_{ 2}$ \\
$6^{\sqrt{6}}_{ 2}$ & $18$ & $1, 1, 2, 2, 2, 2$ & $0, 0, 0, \frac{1}{3}, \frac{1}{3}, \frac{1}{3}$ & SB:$3^{ B}_{ 2}$ \\
$6^{\sqrt{6}}_{-2}$ & $18$ & $1, 1, 2, 2, 2, 2$ & $0, 0, 0, \frac{2}{3}, \frac{2}{3}, \frac{2}{3}$ & SB:$3^{ B}_{-2}$ \\
$6^{\sqrt{6}}_{-2}$ & $18$ & $1, 1, 2, 2, 2, 2$ & $0, 0, 0, \frac{2}{3}, \frac{2}{3}, \frac{2}{3}$ & SB:$3^{ B}_{-2}$ \\
$6^{\sqrt{6}}_{ 14/5}$ & $21.708$ & $1, 1, 2,\zeta_{3}^{1},\zeta_{3}^{1},\zeta_{8}^{4}$ & $0, 0, 0, \frac{2}{5}, \frac{2}{5}, \frac{2}{5}$ & $2^{ B}_{ 14/5}\boxtimes \Rp(S_3)$\\
$6^{\sqrt{6}}_{-14/5}$ & $21.708$ & $1, 1, 2,\zeta_{3}^{1},\zeta_{3}^{1},\zeta_{8}^{4}$ & $0, 0, 0, \frac{3}{5}, \frac{3}{5}, \frac{3}{5}$ & $2^{ B}_{-14/5}\boxtimes \Rp(S_3)$\\
 \hline 
\end{tabular} 
\end{table*}

\def\arraystretch{1.25} \setlength\tabcolsep{3pt}
\begin{table*}[t] 
\caption{
$S_3$-SET orders for bosonic systems labeled in terms of sets of topological excitations.
The list contains all topological orders with 
$N=9$ $D^2\leq 30$.
} 
\label{SETS3-9} 
\centering
\begin{tabular}{ |c|c|l|l|l| } 
\hline 
$N^{|\Th|}_{c}$ & $D^2$ & $d_1,d_2,\cdots$ & $s_1,s_2,\cdots$ & comment \\
\hline 
$3^{\sqrt{6}}_{ 0}$ & $6$ & $1, 1, 2$ & $0, 0, 0$ & $\cE=\Rp(S_3)$\\
\hline
$9^{\sqrt{6}}_{ 2}$ & $18$ & $1, 1, 2, 1, 1, 1, 1, 2, 2$ & $0, 0, 0, \frac{1}{3}, \frac{1}{3}, \frac{1}{3}, \frac{1}{3}, \frac{1}{3}, \frac{1}{3}$ & $3^{ B}_{ 2}\boxtimes \Rp(S_3)$\\
$9^{\sqrt{6}}_{-2}$ & $18$ & $1, 1, 2, 1, 1, 1, 1, 2, 2$ & $0, 0, 0, \frac{2}{3}, \frac{2}{3}, \frac{2}{3}, \frac{2}{3}, \frac{2}{3}, \frac{2}{3}$ & $3^{ B}_{-2}\boxtimes \Rp(S_3)$\\
\hline
$9^{\sqrt{6}}_{ 0}$ & $24$ & $1, 1, 2, 1, 1, 2, 2, 2, 2$ & $0, 0, 0, \frac{1}{2}, \frac{1}{2}, 0, 0, 0, \frac{1}{2}$ & SB:$4^{ B}_{ 0}$ \\
$9^{\sqrt{6}}_{ 0}$ & $24$ & $1, 1, 2, 1, 1, 2, 2, 2, 2$ & $0, 0, 0, \frac{1}{2}, \frac{1}{2}, 0, 0, 0, \frac{1}{2}$ & SB:$4^{ B}_{ 0}$ \\
$9^{\sqrt{6}}_{ 1}$ & $24$ & $1, 1, 2, 1, 1, 2, 2, 2, 2$ & $0, 0, 0, \frac{1}{2}, \frac{1}{2}, \frac{1}{8}, \frac{1}{8}, \frac{1}{8}, \frac{1}{2}$ & SB:$4^{ B}_{ 1}$ \\
$9^{\sqrt{6}}_{ 1}$ & $24$ & $1, 1, 2, 1, 1, 2, 2, 2, 2$ & $0, 0, 0, \frac{1}{2}, \frac{1}{2}, \frac{1}{8}, \frac{1}{8}, \frac{1}{8}, \frac{1}{2}$ & SB:$4^{ B}_{ 1}$ \\
$9^{\sqrt{6}}_{ 2}$ & $24$ & $1, 1, 2, 1, 1, 2, 2, 2, 2$ & $0, 0, 0, \frac{1}{2}, \frac{1}{2}, \frac{1}{4}, \frac{1}{4}, \frac{1}{4}, \frac{1}{2}$ & SB:$4^{ B}_{ 2}$ \\
$9^{\sqrt{6}}_{ 2}$ & $24$ & $1, 1, 2, 1, 1, 2, 2, 2, 2$ & $0, 0, 0, \frac{1}{2}, \frac{1}{2}, \frac{1}{4}, \frac{1}{4}, \frac{1}{4}, \frac{1}{2}$ & SB:$4^{ B}_{ 2}$ \\
$9^{\sqrt{6}}_{ 3}$ & $24$ & $1, 1, 2, 1, 1, 2, 2, 2, 2$ & $0, 0, 0, \frac{1}{2}, \frac{1}{2}, \frac{3}{8}, \frac{3}{8}, \frac{3}{8}, \frac{1}{2}$ & SB:$4^{ B}_{ 3}$ \\
$9^{\sqrt{6}}_{ 3}$ & $24$ & $1, 1, 2, 1, 1, 2, 2, 2, 2$ & $0, 0, 0, \frac{1}{2}, \frac{1}{2}, \frac{3}{8}, \frac{3}{8}, \frac{3}{8}, \frac{1}{2}$ & SB:$4^{ B}_{ 3}$ \\
$9^{\sqrt{6}}_{ 4}$ & $24$ & $1, 1, 2, 1, 1, 2, 2, 2, 2$ & $0, 0, 0, \frac{1}{2}, \frac{1}{2}, \frac{1}{2}, \frac{1}{2}, \frac{1}{2}, \frac{1}{2}$ & SB:$4^{ B}_{ 4}$ \\
$9^{\sqrt{6}}_{ 4}$ & $24$ & $1, 1, 2, 1, 1, 2, 2, 2, 2$ & $0, 0, 0, \frac{1}{2}, \frac{1}{2}, \frac{1}{2}, \frac{1}{2}, \frac{1}{2}, \frac{1}{2}$ & SB:$4^{ B}_{ 4}$ \\
$9^{\sqrt{6}}_{-3}$ & $24$ & $1, 1, 2, 1, 1, 2, 2, 2, 2$ & $0, 0, 0, \frac{1}{2}, \frac{1}{2}, \frac{1}{2}, \frac{5}{8}, \frac{5}{8}, \frac{5}{8}$ & SB:$4^{ B}_{-3}$ \\
$9^{\sqrt{6}}_{-3}$ & $24$ & $1, 1, 2, 1, 1, 2, 2, 2, 2$ & $0, 0, 0, \frac{1}{2}, \frac{1}{2}, \frac{1}{2}, \frac{5}{8}, \frac{5}{8}, \frac{5}{8}$ & SB:$4^{ B}_{-3}$ \\
$9^{\sqrt{6}}_{-2}$ & $24$ & $1, 1, 2, 1, 1, 2, 2, 2, 2$ & $0, 0, 0, \frac{1}{2}, \frac{1}{2}, \frac{1}{2}, \frac{3}{4}, \frac{3}{4}, \frac{3}{4}$ & SB:$4^{ B}_{-2}$ \\
$9^{\sqrt{6}}_{-2}$ & $24$ & $1, 1, 2, 1, 1, 2, 2, 2, 2$ & $0, 0, 0, \frac{1}{2}, \frac{1}{2}, \frac{1}{2}, \frac{3}{4}, \frac{3}{4}, \frac{3}{4}$ & SB:$4^{ B}_{-2}$ \\
$9^{\sqrt{6}}_{-1}$ & $24$ & $1, 1, 2, 1, 1, 2, 2, 2, 2$ & $0, 0, 0, \frac{1}{2}, \frac{1}{2}, \frac{1}{2}, \frac{7}{8}, \frac{7}{8}, \frac{7}{8}$ & SB:$4^{ B}_{-1}$ \\
$9^{\sqrt{6}}_{-1}$ & $24$ & $1, 1, 2, 1, 1, 2, 2, 2, 2$ & $0, 0, 0, \frac{1}{2}, \frac{1}{2}, \frac{1}{2}, \frac{7}{8}, \frac{7}{8}, \frac{7}{8}$ & SB:$4^{ B}_{-1}$ \\
\hline
$9^{\sqrt{6}}_{ 5/2}$ & $24$ & $1, 1, 2, 1, 1,\zeta_{2}^{1},\zeta_{2}^{1}, 2,\sqrt{8}$ & $0, 0, 0, \frac{1}{2}, \frac{1}{2}, \frac{5}{16}, \frac{5}{16}, \frac{1}{2}, \frac{5}{16}$ & $3^{ B}_{ 5/2}\boxtimes \Rp(S_3)$\\
$9^{\sqrt{6}}_{ 5/2}$ & $24$ & $1, 1, 2, 1, 1,\zeta_{2}^{1},\zeta_{2}^{1}, 2,\sqrt{8}$ & $0, 0, 0, \frac{1}{2}, \frac{1}{2}, \frac{5}{16}, \frac{5}{16}, \frac{1}{2}, \frac{5}{16}$ & SB:$3^{ B}_{ 5/2}$ \\
$9^{\sqrt{6}}_{ 1/2}$ & $24$ & $1, 1, 2, 1, 1,\zeta_{2}^{1},\zeta_{2}^{1}, 2,\sqrt{8}$ & $0, 0, 0, \frac{1}{2}, \frac{1}{2}, \frac{1}{16}, \frac{1}{16}, \frac{1}{2}, \frac{1}{16}$ & $3^{ B}_{ 1/2}\boxtimes \Rp(S_3)$\\
$9^{\sqrt{6}}_{ 1/2}$ & $24$ & $1, 1, 2, 1, 1,\zeta_{2}^{1},\zeta_{2}^{1}, 2,\sqrt{8}$ & $0, 0, 0, \frac{1}{2}, \frac{1}{2}, \frac{1}{16}, \frac{1}{16}, \frac{1}{2}, \frac{1}{16}$ & SB:$3^{ B}_{ 1/2}$ \\
$9^{\sqrt{6}}_{ 3/2}$ & $24$ & $1, 1, 2, 1, 1,\zeta_{2}^{1},\zeta_{2}^{1}, 2,\sqrt{8}$ & $0, 0, 0, \frac{1}{2}, \frac{1}{2}, \frac{3}{16}, \frac{3}{16}, \frac{1}{2}, \frac{3}{16}$ & $3^{ B}_{ 3/2}\boxtimes \Rp(S_3)$\\
$9^{\sqrt{6}}_{ 3/2}$ & $24$ & $1, 1, 2, 1, 1,\zeta_{2}^{1},\zeta_{2}^{1}, 2,\sqrt{8}$ & $0, 0, 0, \frac{1}{2}, \frac{1}{2}, \frac{3}{16}, \frac{3}{16}, \frac{1}{2}, \frac{3}{16}$ & SB:$3^{ B}_{ 3/2}$ \\
$9^{\sqrt{6}}_{ 7/2}$ & $24$ & $1, 1, 2, 1, 1,\zeta_{2}^{1},\zeta_{2}^{1}, 2,\sqrt{8}$ & $0, 0, 0, \frac{1}{2}, \frac{1}{2}, \frac{7}{16}, \frac{7}{16}, \frac{1}{2}, \frac{7}{16}$ & $3^{ B}_{ 7/2}\boxtimes \Rp(S_3)$\\
$9^{\sqrt{6}}_{ 7/2}$ & $24$ & $1, 1, 2, 1, 1,\zeta_{2}^{1},\zeta_{2}^{1}, 2,\sqrt{8}$ & $0, 0, 0, \frac{1}{2}, \frac{1}{2}, \frac{7}{16}, \frac{7}{16}, \frac{1}{2}, \frac{7}{16}$ & SB:$3^{ B}_{ 7/2}$ \\
$9^{\sqrt{6}}_{-7/2}$ & $24$ & $1, 1, 2, 1, 1,\zeta_{2}^{1},\zeta_{2}^{1}, 2,\sqrt{8}$ & $0, 0, 0, \frac{1}{2}, \frac{1}{2}, \frac{9}{16}, \frac{9}{16}, \frac{1}{2}, \frac{9}{16}$ & $3^{ B}_{-7/2}\boxtimes \Rp(S_3)$\\
$9^{\sqrt{6}}_{-7/2}$ & $24$ & $1, 1, 2, 1, 1,\zeta_{2}^{1},\zeta_{2}^{1}, 2,\sqrt{8}$ & $0, 0, 0, \frac{1}{2}, \frac{1}{2}, \frac{9}{16}, \frac{9}{16}, \frac{1}{2}, \frac{9}{16}$ & SB:$3^{ B}_{-7/2}$ \\
$9^{\sqrt{6}}_{-5/2}$ & $24$ & $1, 1, 2, 1, 1,\zeta_{2}^{1},\zeta_{2}^{1}, 2,\sqrt{8}$ & $0, 0, 0, \frac{1}{2}, \frac{1}{2}, \frac{11}{16}, \frac{11}{16}, \frac{1}{2}, \frac{11}{16}$ & $3^{ B}_{-5/2}\boxtimes \Rp(S_3)$\\
$9^{\sqrt{6}}_{-5/2}$ & $24$ & $1, 1, 2, 1, 1,\zeta_{2}^{1},\zeta_{2}^{1}, 2,\sqrt{8}$ & $0, 0, 0, \frac{1}{2}, \frac{1}{2}, \frac{11}{16}, \frac{11}{16}, \frac{1}{2}, \frac{11}{16}$ & SB:$3^{ B}_{-5/2}$ \\
$9^{\sqrt{6}}_{-3/2}$ & $24$ & $1, 1, 2, 1, 1,\zeta_{2}^{1},\zeta_{2}^{1}, 2,\sqrt{8}$ & $0, 0, 0, \frac{1}{2}, \frac{1}{2}, \frac{13}{16}, \frac{13}{16}, \frac{1}{2}, \frac{13}{16}$ & $3^{ B}_{-3/2}\boxtimes \Rp(S_3)$\\
$9^{\sqrt{6}}_{-3/2}$ & $24$ & $1, 1, 2, 1, 1,\zeta_{2}^{1},\zeta_{2}^{1}, 2,\sqrt{8}$ & $0, 0, 0, \frac{1}{2}, \frac{1}{2}, \frac{13}{16}, \frac{13}{16}, \frac{1}{2}, \frac{13}{16}$ & SB:$3^{ B}_{-3/2}$ \\
$9^{\sqrt{6}}_{-1/2}$ & $24$ & $1, 1, 2, 1, 1,\zeta_{2}^{1},\zeta_{2}^{1}, 2,\sqrt{8}$ & $0, 0, 0, \frac{1}{2}, \frac{1}{2}, \frac{15}{16}, \frac{15}{16}, \frac{1}{2}, \frac{15}{16}$ & $3^{ B}_{-1/2}\boxtimes \Rp(S_3)$\\
$9^{\sqrt{6}}_{-1/2}$ & $24$ & $1, 1, 2, 1, 1,\zeta_{2}^{1},\zeta_{2}^{1}, 2,\sqrt{8}$ & $0, 0, 0, \frac{1}{2}, \frac{1}{2}, \frac{15}{16}, \frac{15}{16}, \frac{1}{2}, \frac{15}{16}$ & SB:$3^{ B}_{-1/2}$ \\
\hline
$9^{\sqrt{6}}_{ 0}$ & $30$ & $1, 1, 2, 2, 2, 2, 2, 2, 2$ & $0, 0, 0, \frac{1}{5}, \frac{1}{5}, \frac{1}{5}, \frac{4}{5}, \frac{4}{5}, \frac{4}{5}$ & SB:$5^{ B}_{ 0}$ \\
$9^{\sqrt{6}}_{ 4}$ & $30$ & $1, 1, 2, 2, 2, 2, 2, 2, 2$ & $0, 0, 0, \frac{2}{5}, \frac{2}{5}, \frac{2}{5}, \frac{3}{5}, \frac{3}{5}, \frac{3}{5}$ & SB:$5^{ B}_{ 4}$ \\
\hline
$9^{\sqrt{6}}_{ 8/7}$ & $55.775$ & $1, 1, 2,\zeta_{5}^{1},\zeta_{5}^{1},\zeta_{5}^{2},\zeta_{5}^{2},2\zeta_5^1,\zeta_{12}^{6}$ & $0, 0, 0, \frac{6}{7}, \frac{6}{7}, \frac{2}{7}, \frac{2}{7}, \frac{6}{7}, \frac{2}{7}$ & $3^{ B}_{ 8/7}\boxtimes \Rp(S_3)$\\
$9^{\sqrt{6}}_{-8/7}$ & $55.775$ & $1, 1, 2,\zeta_{5}^{1},\zeta_{5}^{1},\zeta_{5}^{2},\zeta_{5}^{2},2\zeta_5^1,\zeta_{12}^{6}$ & $0, 0, 0, \frac{1}{7}, \frac{1}{7}, \frac{5}{7}, \frac{5}{7}, \frac{1}{7}, \frac{5}{7}$ & $3^{ B}_{-8/7}\boxtimes \Rp(S_3)$\\
 \hline 
\end{tabular} 
\end{table*}

\subsection{$Z_2 \times Z_2$-SET orders}

Tables \ref{SETZ2Z2-56}, \ref{SETZ2Z2-7}, and \ref{SETZ2Z2-8} list the $Z_2
\times Z_2$-SET orders  (up to invertible ones) for
2+1D bosonic systems.  

\def\arraystretch{1.25} \setlength\tabcolsep{3pt}
\begin{table*}[t] 
\caption{
$Z_2\times Z_2$-SET orders for bosonic systems labeled in terms of sets of topological excitations.
The list contains all topological orders with 
$N=5$ $D^2\leq 100$ and
$N=6$ $D^2\leq 200$.
} 
\label{SETZ2Z2-56} 
\centering
\begin{tabular}{ |c|c|l|l|l| } 
\hline 
$N^{|\Th|}_{c}$ & $D^2$ & $d_1,d_2,\cdots$ & $s_1,s_2,\cdots$ & comment \\
\hline 
$4^{ 2}_{ 0}$ & $4$ & $1, 1, 1, 1$ & $0, 0, 0, 0$ & $\cE=\Rp(Z_2\times Z_2)$\\
\hline
$5^{ 2}_{ 1}$ & $8$ & $1, 1, 1, 1, 2$ & $0, 0, 0, 0, \frac{1}{4}$ & SB:$2^{ B}_{ 1}$ \\
$5^{ 2}_{-1}$ & $8$ & $1, 1, 1, 1, 2$ & $0, 0, 0, 0, \frac{3}{4}$ & SB:$2^{ B}_{-1}$ \\
$5^{ 2}_{ 14/5}$ & $14.472$ & $1, 1, 1, 1,\zeta_{8}^{4}$ & $0, 0, 0, 0, \frac{2}{5}$ & SB:$2^{ B}_{ 14/5}$ \\
$5^{ 2}_{-14/5}$ & $14.472$ & $1, 1, 1, 1,\zeta_{8}^{4}$ & $0, 0, 0, 0, \frac{3}{5}$ & SB:$2^{ B}_{-14/5}$ \\
 \hline 
$6^{ 2}_{ 2}$ & $12$ & $1, 1, 1, 1, 2, 2$ & $0, 0, 0, 0, \frac{1}{3}, \frac{1}{3}$ & SB:$3^{ B}_{ 2}$ \\
$6^{ 2}_{ 2}$ & $12$ & $1, 1, 1, 1, 2, 2$ & $0, 0, 0, 0, \frac{1}{3}, \frac{1}{3}$ & SB:$3^{ B}_{ 2}$ \\
$6^{ 2}_{ 2}$ & $12$ & $1, 1, 1, 1, 2, 2$ & $0, 0, 0, 0, \frac{1}{3}, \frac{1}{3}$ & SB:$3^{ B}_{ 2}$ \\
$6^{ 2}_{ 2}$ & $12$ & $1, 1, 1, 1, 2, 2$ & $0, 0, 0, 0, \frac{1}{3}, \frac{1}{3}$ & SB:$3^{ B}_{ 2}$ \\
$6^{ 2}_{-2}$ & $12$ & $1, 1, 1, 1, 2, 2$ & $0, 0, 0, 0, \frac{2}{3}, \frac{2}{3}$ & SB:$3^{ B}_{-2}$ \\
$6^{ 2}_{-2}$ & $12$ & $1, 1, 1, 1, 2, 2$ & $0, 0, 0, 0, \frac{2}{3}, \frac{2}{3}$ & SB:$3^{ B}_{-2}$ \\
$6^{ 2}_{-2}$ & $12$ & $1, 1, 1, 1, 2, 2$ & $0, 0, 0, 0, \frac{2}{3}, \frac{2}{3}$ & SB:$3^{ B}_{-2}$ \\
$6^{ 2}_{-2}$ & $12$ & $1, 1, 1, 1, 2, 2$ & $0, 0, 0, 0, \frac{2}{3}, \frac{2}{3}$ & SB:$3^{ B}_{-2}$ \\
$6^{ 2}_{ 1/2}$ & $16$ & $1, 1, 1, 1, 2,\sqrt{8}$ & $0, 0, 0, 0, \frac{1}{2}, \frac{1}{16}$ & SB:$3^{ B}_{ 1/2}$ \\
$6^{ 2}_{ 3/2}$ & $16$ & $1, 1, 1, 1, 2,\sqrt{8}$ & $0, 0, 0, 0, \frac{1}{2}, \frac{3}{16}$ & SB:$3^{ B}_{ 3/2}$ \\
$6^{ 2}_{ 5/2}$ & $16$ & $1, 1, 1, 1, 2,\sqrt{8}$ & $0, 0, 0, 0, \frac{1}{2}, \frac{5}{16}$ & SB:$3^{ B}_{ 5/2}$ \\
$6^{ 2}_{ 7/2}$ & $16$ & $1, 1, 1, 1, 2,\sqrt{8}$ & $0, 0, 0, 0, \frac{1}{2}, \frac{7}{16}$ & SB:$3^{ B}_{ 7/2}$ \\
$6^{ 2}_{-7/2}$ & $16$ & $1, 1, 1, 1, 2,\sqrt{8}$ & $0, 0, 0, 0, \frac{1}{2}, \frac{9}{16}$ & SB:$3^{ B}_{-7/2}$ \\
$6^{ 2}_{-5/2}$ & $16$ & $1, 1, 1, 1, 2,\sqrt{8}$ & $0, 0, 0, 0, \frac{1}{2}, \frac{11}{16}$ & SB:$3^{ B}_{-5/2}$ \\
$6^{ 2}_{-3/2}$ & $16$ & $1, 1, 1, 1, 2,\sqrt{8}$ & $0, 0, 0, 0, \frac{1}{2}, \frac{13}{16}$ & SB:$3^{ B}_{-3/2}$ \\
$6^{ 2}_{-1/2}$ & $16$ & $1, 1, 1, 1, 2,\sqrt{8}$ & $0, 0, 0, 0, \frac{1}{2}, \frac{15}{16}$ & SB:$3^{ B}_{-1/2}$ \\
$6^{ 2}_{ 4}$ & $36$ & $1, 1, 1, 1, 4, 4$ & $0, 0, 0, 0, \frac{1}{3}, \frac{2}{3}$ & SB:$9^{ B}_{ 4}$ \\
$6^{ 2}_{ 8/7}$ & $37.183$ & $1, 1, 1, 1,2\zeta_5^1,\zeta_{12}^{6}$ & $0, 0, 0, 0, \frac{6}{7}, \frac{2}{7}$ & SB:$3^{ B}_{ 8/7}$ \\
$6^{ 2}_{-8/7}$ & $37.183$ & $1, 1, 1, 1,2\zeta_5^1,\zeta_{12}^{6}$ & $0, 0, 0, 0, \frac{1}{7}, \frac{5}{7}$ & SB:$3^{ B}_{-8/7}$ \\
 \hline 
\end{tabular} 
\end{table*}

\def\arraystretch{1.25} \setlength\tabcolsep{3pt}
\begin{table*}[t] 
\caption{
$Z_2\times Z_2$-SET orders for bosonic systems labeled in terms of sets of topological excitations.
The list contains all topological orders with 
$N=7$ $D^2\leq 120$.
} 
\label{SETZ2Z2-7} 
\centering
\begin{tabular}{ |c|c|l|l|l| } 
\hline 
$N^{|\Th|}_{c}$ & $D^2$ & $d_1,d_2,\cdots$ & $s_1,s_2,\cdots$ & comment \\
\hline 
$4^{ 2}_{ 0}$ & $4$ & $1, 1, 1, 1$ & $0, 0, 0, 0$ & $\cE=\Rp(Z_2\times Z_2)$\\
\hline 
$7^{ 2}_{ 0}$ & $16$ & $1, 1, 1, 1, 2, 2, 2$ & $0, 0, 0, 0, 0, 0, \frac{1}{2}$ & SB:$4^{ B}_{ 0}$ \\
$7^{ 2}_{ 0}$ & $16$ & $1, 1, 1, 1, 2, 2, 2$ & $0, 0, 0, 0, 0, \frac{1}{4}, \frac{3}{4}$ & SB:$4^{ B}_{ 0}$ \\
$7^{ 2}_{ 1}$ & $16$ & $1, 1, 1, 1, 2, 2, 2$ & $0, 0, 0, 0, \frac{1}{8}, \frac{1}{8}, \frac{1}{2}$ & SB:$4^{ B}_{ 1}$ \\
$7^{ 2}_{ 1}$ & $16$ & $1, 1, 1, 1, 2, 2, 2$ & $0, 0, 0, 0, \frac{1}{8}, \frac{1}{8}, \frac{1}{2}$ & SB:$4^{ B}_{ 1}$ \\
$7^{ 2}_{ 1}$ & $16$ & $1, 1, 1, 1, 2, 2, 2$ & $0, 0, 0, 0, \frac{1}{8}, \frac{1}{8}, \frac{1}{2}$ & SB:$4^{ B}_{ 1}$ \\
$7^{ 2}_{ 1}$ & $16$ & $1, 1, 1, 1, 2, 2, 2$ & $0, 0, 0, 0, \frac{1}{8}, \frac{1}{8}, \frac{1}{2}$ & SB:$4^{ B}_{ 1}$ \\
$7^{ 2}_{ 1}$ & $16$ & $1, 1, 1, 1, 2, 2, 2$ & $0, 0, 0, 0, \frac{1}{8}, \frac{1}{8}, \frac{1}{2}$ & SB:$4^{ B}_{ 1}$ \\
$7^{ 2}_{ 1}$ & $16$ & $1, 1, 1, 1, 2, 2, 2$ & $0, 0, 0, 0, \frac{1}{8}, \frac{1}{8}, \frac{1}{2}$ & SB:$4^{ B}_{ 1}$ \\
$7^{ 2}_{ 1}$ & $16$ & $1, 1, 1, 1, 2, 2, 2$ & $0, 0, 0, 0, \frac{1}{8}, \frac{1}{8}, \frac{1}{2}$ & SB:$4^{ B}_{ 1}$ \\
$7^{ 2}_{ 2}$ & $16$ & $1, 1, 1, 1, 2, 2, 2$ & $0, 0, 0, 0, \frac{1}{4}, \frac{1}{4}, \frac{1}{2}$ & SB:$4^{ B}_{ 2}$ \\
$7^{ 2}_{ 3}$ & $16$ & $1, 1, 1, 1, 2, 2, 2$ & $0, 0, 0, 0, \frac{3}{8}, \frac{3}{8}, \frac{1}{2}$ & SB:$4^{ B}_{ 3}$ \\
$7^{ 2}_{ 3}$ & $16$ & $1, 1, 1, 1, 2, 2, 2$ & $0, 0, 0, 0, \frac{3}{8}, \frac{3}{8}, \frac{1}{2}$ & SB:$4^{ B}_{ 3}$ \\
$7^{ 2}_{ 3}$ & $16$ & $1, 1, 1, 1, 2, 2, 2$ & $0, 0, 0, 0, \frac{3}{8}, \frac{3}{8}, \frac{1}{2}$ & SB:$4^{ B}_{ 3}$ \\
$7^{ 2}_{ 3}$ & $16$ & $1, 1, 1, 1, 2, 2, 2$ & $0, 0, 0, 0, \frac{3}{8}, \frac{3}{8}, \frac{1}{2}$ & SB:$4^{ B}_{ 3}$ \\
$7^{ 2}_{ 3}$ & $16$ & $1, 1, 1, 1, 2, 2, 2$ & $0, 0, 0, 0, \frac{3}{8}, \frac{3}{8}, \frac{1}{2}$ & SB:$4^{ B}_{ 3}$ \\
$7^{ 2}_{ 3}$ & $16$ & $1, 1, 1, 1, 2, 2, 2$ & $0, 0, 0, 0, \frac{3}{8}, \frac{3}{8}, \frac{1}{2}$ & SB:$4^{ B}_{ 3}$ \\
$7^{ 2}_{ 3}$ & $16$ & $1, 1, 1, 1, 2, 2, 2$ & $0, 0, 0, 0, \frac{3}{8}, \frac{3}{8}, \frac{1}{2}$ & SB:$4^{ B}_{ 3}$ \\
$7^{ 2}_{ 4}$ & $16$ & $1, 1, 1, 1, 2, 2, 2$ & $0, 0, 0, 0, \frac{1}{2}, \frac{1}{2}, \frac{1}{2}$ & SB:$4^{ B}_{ 4}$ \\
$7^{ 2}_{-3}$ & $16$ & $1, 1, 1, 1, 2, 2, 2$ & $0, 0, 0, 0, \frac{1}{2}, \frac{5}{8}, \frac{5}{8}$ & SB:$4^{ B}_{-3}$ \\
$7^{ 2}_{-3}$ & $16$ & $1, 1, 1, 1, 2, 2, 2$ & $0, 0, 0, 0, \frac{1}{2}, \frac{5}{8}, \frac{5}{8}$ & SB:$4^{ B}_{-3}$ \\
$7^{ 2}_{-3}$ & $16$ & $1, 1, 1, 1, 2, 2, 2$ & $0, 0, 0, 0, \frac{1}{2}, \frac{5}{8}, \frac{5}{8}$ & SB:$4^{ B}_{-3}$ \\
$7^{ 2}_{-3}$ & $16$ & $1, 1, 1, 1, 2, 2, 2$ & $0, 0, 0, 0, \frac{1}{2}, \frac{5}{8}, \frac{5}{8}$ & SB:$4^{ B}_{-3}$ \\
$7^{ 2}_{-3}$ & $16$ & $1, 1, 1, 1, 2, 2, 2$ & $0, 0, 0, 0, \frac{1}{2}, \frac{5}{8}, \frac{5}{8}$ & SB:$4^{ B}_{-3}$ \\
$7^{ 2}_{-3}$ & $16$ & $1, 1, 1, 1, 2, 2, 2$ & $0, 0, 0, 0, \frac{1}{2}, \frac{5}{8}, \frac{5}{8}$ & SB:$4^{ B}_{-3}$ \\
$7^{ 2}_{-3}$ & $16$ & $1, 1, 1, 1, 2, 2, 2$ & $0, 0, 0, 0, \frac{1}{2}, \frac{5}{8}, \frac{5}{8}$ & SB:$4^{ B}_{-3}$ \\
$7^{ 2}_{-2}$ & $16$ & $1, 1, 1, 1, 2, 2, 2$ & $0, 0, 0, 0, \frac{1}{2}, \frac{3}{4}, \frac{3}{4}$ & SB:$4^{ B}_{-2}$ \\
$7^{ 2}_{-1}$ & $16$ & $1, 1, 1, 1, 2, 2, 2$ & $0, 0, 0, 0, \frac{1}{2}, \frac{7}{8}, \frac{7}{8}$ & SB:$4^{ B}_{-1}$ \\
$7^{ 2}_{-1}$ & $16$ & $1, 1, 1, 1, 2, 2, 2$ & $0, 0, 0, 0, \frac{1}{2}, \frac{7}{8}, \frac{7}{8}$ & SB:$4^{ B}_{-1}$ \\
$7^{ 2}_{-1}$ & $16$ & $1, 1, 1, 1, 2, 2, 2$ & $0, 0, 0, 0, \frac{1}{2}, \frac{7}{8}, \frac{7}{8}$ & SB:$4^{ B}_{-1}$ \\
$7^{ 2}_{-1}$ & $16$ & $1, 1, 1, 1, 2, 2, 2$ & $0, 0, 0, 0, \frac{1}{2}, \frac{7}{8}, \frac{7}{8}$ & SB:$4^{ B}_{-1}$ \\
$7^{ 2}_{-1}$ & $16$ & $1, 1, 1, 1, 2, 2, 2$ & $0, 0, 0, 0, \frac{1}{2}, \frac{7}{8}, \frac{7}{8}$ & SB:$4^{ B}_{-1}$ \\
$7^{ 2}_{-1}$ & $16$ & $1, 1, 1, 1, 2, 2, 2$ & $0, 0, 0, 0, \frac{1}{2}, \frac{7}{8}, \frac{7}{8}$ & SB:$4^{ B}_{-1}$ \\
$7^{ 2}_{-1}$ & $16$ & $1, 1, 1, 1, 2, 2, 2$ & $0, 0, 0, 0, \frac{1}{2}, \frac{7}{8}, \frac{7}{8}$ & SB:$4^{ B}_{-1}$ \\
$7^{ 2}_{ 9/5}$ & $28.944$ & $1, 1, 1, 1, 2,\zeta_{8}^{4},\zeta_{8}^{4}$ & $0, 0, 0, 0, \frac{3}{4}, \frac{3}{20}, \frac{2}{5}$ & SB:$4^{ B}_{ 9/5}$ \\
$7^{ 2}_{ 19/5}$ & $28.944$ & $1, 1, 1, 1, 2,\zeta_{8}^{4},\zeta_{8}^{4}$ & $0, 0, 0, 0, \frac{1}{4}, \frac{2}{5}, \frac{13}{20}$ & SB:$4^{ B}_{ 19/5}$ \\
$7^{ 2}_{-19/5}$ & $28.944$ & $1, 1, 1, 1, 2,\zeta_{8}^{4},\zeta_{8}^{4}$ & $0, 0, 0, 0, \frac{3}{4}, \frac{7}{20}, \frac{3}{5}$ & SB:$4^{ B}_{-19/5}$ \\
$7^{ 2}_{-9/5}$ & $28.944$ & $1, 1, 1, 1, 2,\zeta_{8}^{4},\zeta_{8}^{4}$ & $0, 0, 0, 0, \frac{1}{4}, \frac{3}{5}, \frac{17}{20}$ & SB:$4^{ B}_{-9/5}$ \\
$7^{ 2}_{ 0}$ & $52.360$ & $1, 1, 1, 1,\zeta_{8}^{4},\zeta_{8}^{4},3+\sqrt{5}$ & $0, 0, 0, 0, \frac{2}{5}, \frac{3}{5}, 0$ & SB:$4^{ B}_{ 0}$ \\
$7^{ 2}_{ 12/5}$ & $52.360$ & $1, 1, 1, 1,\zeta_{8}^{4},\zeta_{8}^{4},3+\sqrt{5}$ & $0, 0, 0, 0, \frac{3}{5}, \frac{3}{5}, \frac{1}{5}$ & SB:$4^{ B}_{ 12/5}$ \\
$7^{ 2}_{-12/5}$ & $52.360$ & $1, 1, 1, 1,\zeta_{8}^{4},\zeta_{8}^{4},3+\sqrt{5}$ & $0, 0, 0, 0, \frac{2}{5}, \frac{2}{5}, \frac{4}{5}$ & SB:$4^{ B}_{-12/5}$ \\
$7^{ 2}_{ 10/3}$ & $76.937$ & $1, 1, 1, 1,2\zeta_7^1,2\zeta_7^2,\zeta_{16}^{8}$ & $0, 0, 0, 0, \frac{1}{3}, \frac{2}{9}, \frac{2}{3}$ & SB:$4^{ B}_{ 10/3}$ \\
$7^{ 2}_{-10/3}$ & $76.937$ & $1, 1, 1, 1,2\zeta_7^1,2\zeta_7^2,\zeta_{16}^{8}$ & $0, 0, 0, 0, \frac{2}{3}, \frac{7}{9}, \frac{1}{3}$ & SB:$4^{ B}_{-10/3}$ \\
 \hline 
\end{tabular} 
\end{table*}

\def\arraystretch{1.25} \setlength\tabcolsep{3pt}
\begin{table*}[t] 
\caption{
$Z_2\times Z_2$-SET orders for bosonic systems labeled in terms of sets of topological excitations.
The list contains all topological orders with 
$N=8$ $D^2\leq 60$.
} 
\label{SETZ2Z2-8} 
\centering
\begin{tabular}{ |c|c|l|l|l| } 
\hline 
$N^{|\Th|}_{c}$ & $D^2$ & $d_1,d_2,\cdots$ & $s_1,s_2,\cdots$ & comment \\
\hline 
$4^{ 2}_{ 0}$ & $4$ & $1, 1, 1, 1$ & $0, 0, 0, 0$ & $\cE=\Rp(Z_2\times Z_2)$\\
\hline
$8^{ 2}_{ 1}$ & $8$ & $1, 1, 1, 1, 1, 1, 1, 1$ & $0, 0, 0, 0, \frac{1}{4}, \frac{1}{4}, \frac{1}{4}, \frac{1}{4}$ & $2^{ B}_{ 1}\boxtimes \Rp(Z_2\times Z_2)$\\
$8^{ 2}_{ 1}$ & $8$ & $1, 1, 1, 1, 1, 1, 1, 1$ & $0, 0, 0, 0, \frac{1}{4}, \frac{1}{4}, \frac{1}{4}, \frac{1}{4}$ & SB:$2^{ B}_{ 1}$ \\
$8^{ 2}_{ 1}$ & $8$ & $1, 1, 1, 1, 1, 1, 1, 1$ & $0, 0, 0, 0, \frac{1}{4}, \frac{1}{4}, \frac{1}{4}, \frac{1}{4}$ & SB:$2^{ B}_{ 1}$ \\
$8^{ 2}_{ 1}$ & $8$ & $1, 1, 1, 1, 1, 1, 1, 1$ & $0, 0, 0, 0, \frac{1}{4}, \frac{1}{4}, \frac{1}{4}, \frac{1}{4}$ & SB:$2^{ B}_{ 1}$ \\
$8^{ 2}_{-1}$ & $8$ & $1, 1, 1, 1, 1, 1, 1, 1$ & $0, 0, 0, 0, \frac{3}{4}, \frac{3}{4}, \frac{3}{4}, \frac{3}{4}$ & $2^{ B}_{-1}\boxtimes \Rp(Z_2\times Z_2)$\\
$8^{ 2}_{-1}$ & $8$ & $1, 1, 1, 1, 1, 1, 1, 1$ & $0, 0, 0, 0, \frac{3}{4}, \frac{3}{4}, \frac{3}{4}, \frac{3}{4}$ & SB:$2^{ B}_{-1}$ \\
$8^{ 2}_{-1}$ & $8$ & $1, 1, 1, 1, 1, 1, 1, 1$ & $0, 0, 0, 0, \frac{3}{4}, \frac{3}{4}, \frac{3}{4}, \frac{3}{4}$ & SB:$2^{ B}_{-1}$ \\
$8^{ 2}_{-1}$ & $8$ & $1, 1, 1, 1, 1, 1, 1, 1$ & $0, 0, 0, 0, \frac{3}{4}, \frac{3}{4}, \frac{3}{4}, \frac{3}{4}$ & SB:$2^{ B}_{-1}$ \\
$8^{ 2}_{ 14/5}$ & $14.472$ & $1, 1, 1, 1,\zeta_{3}^{1},\zeta_{3}^{1},\zeta_{3}^{1},\zeta_{3}^{1}$ & $0, 0, 0, 0, \frac{2}{5}, \frac{2}{5}, \frac{2}{5}, \frac{2}{5}$ & $2^{ B}_{ 14/5}\boxtimes \Rp(Z_2\times Z_2)$\\
$8^{ 2}_{-14/5}$ & $14.472$ & $1, 1, 1, 1,\zeta_{3}^{1},\zeta_{3}^{1},\zeta_{3}^{1},\zeta_{3}^{1}$ & $0, 0, 0, 0, \frac{3}{5}, \frac{3}{5}, \frac{3}{5}, \frac{3}{5}$ & $2^{ B}_{-14/5}\boxtimes \Rp(Z_2\times Z_2)$\\
$8^{ 2}_{ 0}$ & $20$ & $1, 1, 1, 1, 2, 2, 2, 2$ & $0, 0, 0, 0, \frac{1}{5}, \frac{1}{5}, \frac{4}{5}, \frac{4}{5}$ & SB:$5^{ B}_{ 0}$ \\
$8^{ 2}_{ 0}$ & $20$ & $1, 1, 1, 1, 2, 2, 2, 2$ & $0, 0, 0, 0, \frac{1}{5}, \frac{1}{5}, \frac{4}{5}, \frac{4}{5}$ & SB:$5^{ B}_{ 0}$ \\
$8^{ 2}_{ 0}$ & $20$ & $1, 1, 1, 1, 2, 2, 2, 2$ & $0, 0, 0, 0, \frac{1}{5}, \frac{1}{5}, \frac{4}{5}, \frac{4}{5}$ & SB:$5^{ B}_{ 0}$ \\
$8^{ 2}_{ 0}$ & $20$ & $1, 1, 1, 1, 2, 2, 2, 2$ & $0, 0, 0, 0, \frac{1}{5}, \frac{1}{5}, \frac{4}{5}, \frac{4}{5}$ & SB:$5^{ B}_{ 0}$ \\
$8^{ 2}_{ 4}$ & $20$ & $1, 1, 1, 1, 2, 2, 2, 2$ & $0, 0, 0, 0, \frac{2}{5}, \frac{2}{5}, \frac{3}{5}, \frac{3}{5}$ & SB:$5^{ B}_{ 4}$ \\
$8^{ 2}_{ 4}$ & $20$ & $1, 1, 1, 1, 2, 2, 2, 2$ & $0, 0, 0, 0, \frac{2}{5}, \frac{2}{5}, \frac{3}{5}, \frac{3}{5}$ & SB:$5^{ B}_{ 4}$ \\
$8^{ 2}_{ 4}$ & $20$ & $1, 1, 1, 1, 2, 2, 2, 2$ & $0, 0, 0, 0, \frac{2}{5}, \frac{2}{5}, \frac{3}{5}, \frac{3}{5}$ & SB:$5^{ B}_{ 4}$ \\
$8^{ 2}_{ 4}$ & $20$ & $1, 1, 1, 1, 2, 2, 2, 2$ & $0, 0, 0, 0, \frac{2}{5}, \frac{2}{5}, \frac{3}{5}, \frac{3}{5}$ & SB:$5^{ B}_{ 4}$ \\
$8^{ 2}_{ 2}$ & $48$ & $1, 1, 1, 1, 2,\sqrt{12},\sqrt{12}, 4$ & $0, 0, 0, 0, 0, \frac{1}{8}, \frac{5}{8}, \frac{1}{3}$ & SB:$5^{ B}_{ 2}$ \\
$8^{ 2}_{ 2}$ & $48$ & $1, 1, 1, 1, 2,\sqrt{12},\sqrt{12}, 4$ & $0, 0, 0, 0, 0, \frac{3}{8}, \frac{7}{8}, \frac{1}{3}$ & SB:$5^{ B}_{ 2}$ \\
$8^{ 2}_{-2}$ & $48$ & $1, 1, 1, 1, 2,\sqrt{12},\sqrt{12}, 4$ & $0, 0, 0, 0, 0, \frac{1}{8}, \frac{5}{8}, \frac{2}{3}$ & SB:$5^{ B}_{-2}$ \\
$8^{ 2}_{-2}$ & $48$ & $1, 1, 1, 1, 2,\sqrt{12},\sqrt{12}, 4$ & $0, 0, 0, 0, 0, \frac{3}{8}, \frac{7}{8}, \frac{2}{3}$ & SB:$5^{ B}_{-2}$ \\
$8^{ 2}_{ 16/11}$ & $138.58$ & $1, 1, 1, 1,2\zeta_9^1,2\zeta_9^2,2\zeta_9^3, \zeta_{20}^{10}$ & $0, 0, 0, 0, \frac{9}{11}, \frac{2}{11}, \frac{1}{11}, \frac{6}{11}$ & SB:$5^{ B}_{ 16/11}$ \\
$8^{ 2}_{-16/11}$ & $138.58$ & $1, 1, 1, 1,2\zeta_9^1,2\zeta_9^2,2\zeta_9^3, \zeta_{20}^{10}$ & $0, 0, 0, 0, \frac{2}{11}, \frac{9}{11}, \frac{10}{11}, \frac{5}{11}$ & SB:$5^{ B}_{-16/11}$ \\
$8^{ 2}_{ 18/7}$ & $141.36$ & $1, 1, 1, 1,\zeta_{12}^{6},\zeta_{12}^{6},2\zeta_{12}^2,2\zeta_{12}^4 $ & $0, 0, 0, 0, \frac{6}{7}, \frac{6}{7}, \frac{1}{7}, \frac{3}{7}$ & SB:$5^{ B}_{ 18/7}$ \\
$8^{ 2}_{-18/7}$ & $141.36$ & $1, 1, 1, 1,\zeta_{12}^{6},\zeta_{12}^{6},2\zeta_{12}^2,2\zeta_{12}^4 $ & $0, 0, 0, 0, \frac{1}{7}, \frac{1}{7}, \frac{6}{7}, \frac{4}{7}$ & SB:$5^{ B}_{-18/7}$ \\
 \hline 
\end{tabular} 
\end{table*}

Table \ref{SETZ2Z2frule} list the fusion rules for some $Z_2\times
Z_2$-SET orders.  We see that the $5_1^2$
state is a $\nu=1/2$ bosonic Laughlin state with $Z_2\times Z_2$ symmetry,
where the only topological excitation carries the projective representation of
$Z_2\times Z_2$.  We also see that the $5_{14/2}^2$ state is a bosonic
Fibonacci state with $Z_2\times Z_2$ symmetry, where the only non-abelian
topological excitation carries the projective representation of $Z_2\times
Z_2$.

\begin{table*}[t] 
\caption{
The fusion rules for some $Z_2\times Z_2$-SET orders.
} 
\label{SETZ2Z2frule} 
\centering
\begin{tabular}{ |c|ccccc|}
 \hline 
 $s_i$ & $0$ & $ 0$ & $ 0$ & $ 0$ & $ \frac{1}{4}$\\
 $d_i$ & $1$ & $ 1$ & $ 1$ & $ 1$ & $ 2$\\
\hline
 $5^{ 2}_{ 1}$ & $\textbf{1}$  & $a$  & $b$  & $c$  & $\phi$ \\
\hline
$\textbf{1}$  & $ \textbf{1}$  & $ a$  & $ b$  & $ c$  & $ \phi$  \\
$a$  & $ a$  & $ \textbf{1}$  & $ c$  & $ b$  & $ \phi$  \\
$b$  & $ b$  & $ c$  & $ \textbf{1}$  & $ a$  & $ \phi$  \\
$c$  & $ c$  & $ b$  & $ a$  & $ \textbf{1}$  & $ \phi$  \\
$\phi$  & $ \phi$  & $ \phi$  & $ \phi$  & $ \phi$  & $ \textbf{1} \oplus a \oplus b \oplus c$  \\
\hline
\end{tabular}
~~~~~~~~~
\begin{tabular}{ |c|ccccc|}
 \hline 
 $s_i$ & $0$ & $ 0$ & $ 0$ & $ 0$ & $ \frac{2}{5}$\\
 $d_i$ & $1$ & $ 1$ & $ 1$ & $ 1$ & $2\zeta_{3}^{1}$\\
\hline
 $5^{ 2}_{ 14/5}$ & $\textbf{1}$  & $a$  & $b$  & $c$  & $\eta$ \\
\hline
$\textbf{1}$  & $ \textbf{1}$  & $ a$  & $ b$  & $ c$  & $ \eta$  \\
$a$  & $ a$  & $ \textbf{1}$  & $ c$  & $ b$  & $ \eta$  \\
$b$  & $ b$  & $ c$  & $ \textbf{1}$  & $ a$  & $ \eta$  \\
$c$  & $ c$  & $ b$  & $ a$  & $ \textbf{1}$  & $ \eta$  \\
$\eta$  & $ \eta$  & $ \eta$  & $ \eta$  & $ \eta$  & $ \textbf{1} \oplus a \oplus b \oplus c \oplus 2\eta$  \\
\hline
\end{tabular}
\end{table*}

\subsection{$Z_2 \times Z_2^f$-SET and $Z_4^f$-SET orders}

Table \ref{SETZ2Z2f} lists the $Z_2 \times Z_2^f$-SET
orders  (up to invertible ones) for 2+1D fermionic systems.  
Table \ref{SETZ4f} lists the $Z_4^f$-SET orders (up to invertible ones) for
2+1D fermionic systems.  
For fermionic systems the central charge is determined up to $c_\text{min}$ by the bulk
excitations, where $c_\text{min}$ is the smallest positive central charge of the
modular extensions of $\sRp(G^f)$, for example, $c_\text{min}=1/2$ for
$Z_2^f,Z_2\times Z_2^f, Z_6^f$, $c_\text{min}=1$ for $Z_4^f,Z_8^f$.

\def\arraystretch{1.25} \setlength\tabcolsep{3pt}
\begin{table*}[t] 
\caption{
$Z_2\times Z_2^f$-SET orders  (up to invertible ones) for fermionic systems.
The list contains all topological orders with 
$N=6$ $D^2\leq 300$,
$N=8$ $D^2\leq 60$, and
$N=10$ $D^2\leq 20$.
} 
\label{SETZ2Z2f} 
\centering
\begin{tabular}{ |c|c|l|l|l| } 
\hline 
$N^{|\Th|}_{c}$ & $D^2$ & $d_1,d_2,\cdots$ & $s_1,s_2,\cdots$ & comment \\
\hline 
$4^{ 0}_{0}({ 2\atop  0})$ & $4$ & $1, 1, 1, 1$ & $0, 0, \frac{1}{2}, \frac{1}{2}$ & $\cE=\sRp(Z_2\times Z_2^f)$\\
\hline
$6^{ 0}_{0}$ & $12$ & $1, 1, 1, 1, 2, 2$ & $0, 0, \frac{1}{2}, \frac{1}{2}, \frac{1}{6}, \frac{2}{3}$ & 
SB:$K=$\tiny $\bpm -1 & -2 \\ -2&-1\epm$\\
$6^{ 0}_{0}$ & $12$ & $1, 1, 1, 1, 2, 2$ & $0, 0, \frac{1}{2}, \frac{1}{2}, \frac{1}{3}, \frac{5}{6}$ & SB:$K=$\tiny $\bpm 1 & 2 \\ 2&1\epm$\\
\hline
$8^{ 0}_{0}({ 0\atop 0})$ & $8$ & $1, 1, 1, 1, 1, 1, 1, 1$ & $0, 0, \frac{1}{2}, \frac{1}{2}, \frac{1}{4}, \frac{1}{4}, \frac{3}{4}, \frac{3}{4}$ & $2^{ B}_{ 1}\boxtimes \sRp(Z_2\times Z_2^f)$\\
$8^{ 0}_{0}({ 0\atop 0})$ & $8$ & $1, 1, 1, 1, 1, 1, 1, 1$ & $0, 0, \frac{1}{2}, \frac{1}{2}, \frac{1}{4}, \frac{1}{4}, \frac{3}{4}, \frac{3}{4}$ & SB:$4^{ F}_{0}({ 0\atop 0})$ \\
$8^{ 0}_{-14/5}({\zeta_{8}^{4}\atop  3/20})$ & $14.472$ & $1, 1, 1, 1,\zeta_{3}^{1},\zeta_{3}^{1},\zeta_{3}^{1},\zeta_{3}^{1}$ & $0, 0, \frac{1}{2}, \frac{1}{2}, \frac{1}{10}, \frac{1}{10}, \frac{3}{5}, \frac{3}{5}$ & $2^{ B}_{-14/5}\boxtimes \sRp(Z_2\times Z_2^f)$\\
$8^{ 0}_{14/5}({\zeta_{8}^{4}\atop -3/20})$ & $14.472$ & $1, 1, 1, 1,\zeta_{3}^{1},\zeta_{3}^{1},\zeta_{3}^{1},\zeta_{3}^{1}$ & $0, 0, \frac{1}{2}, \frac{1}{2}, \frac{2}{5}, \frac{2}{5}, \frac{9}{10}, \frac{9}{10}$ & $2^{ B}_{ 14/5}\boxtimes \sRp(Z_2\times Z_2^f)$\\
$8^{ 0}_{0}({ 2\atop  0})$ & $20$ & $1, 1, 1, 1, 2, 2, 2, 2$ & $0, 0, \frac{1}{2}, \frac{1}{2}, \frac{1}{10}, \frac{2}{5}, \frac{3}{5}, \frac{9}{10}$ & SB:$10^{ F}_{0}({\zeta_{2}^{1}\atop  0})$ \\
$8^{ 0}_{0}({ 2\atop  1/2})$ & $20$ & $1, 1, 1, 1, 2, 2, 2, 2$ & $0, 0, \frac{1}{2}, \frac{1}{2}, \frac{1}{5}, \frac{3}{10}, \frac{7}{10}, \frac{4}{5}$ & SB:$10^{ F}_{0}({\zeta_{2}^{1}\atop  1/2})$ \\
$8^{ 0}_{1/4}({\zeta_{2}^{1}\zeta_{6}^{3} \atop  1/2})$ & $27.313$ & $1, 1, 1, 1,\zeta_{6}^{2},\zeta_{6}^{2},\zeta_{6}^{2},\zeta_{6}^{2}$ & $0, 0, \frac{1}{2}, \frac{1}{2}, \frac{1}{4}, \frac{1}{4}, \frac{3}{4}, \frac{3}{4}$ & SB:$4^{ F}_{1/4}({\zeta_{6}^{3}\atop  1/2})$ \\
$8^{ 0}_{1/4}({\zeta_{2}^{1}\zeta_{6}^{3}\atop  1/2})$ & $27.313$ & $1, 1, 1, 1,\zeta_{6}^{2},\zeta_{6}^{2},\zeta_{6}^{2},\zeta_{6}^{2}$ & $0, 0, \frac{1}{2}, \frac{1}{2}, \frac{1}{4}, \frac{1}{4}, \frac{3}{4}, \frac{3}{4}$ & SB:$4^{ F}_{1/4}({\zeta_{6}^{3}\atop  1/2})$ \\
\hline
$10^{ 0}_{0}({ 4\atop  0})$ & $16$ & $1, 1, 1, 1, 1, 1, 1, 1, 2, 2$ & $0, 0, \frac{1}{2}, \frac{1}{2}, 0, 0, \frac{1}{2}, \frac{1}{2}, 0, \frac{1}{2}$ & SB:$8^{ F}_{0}({\sqrt{8}\atop  0})$ \\
$10^{ 0}_{0}({ 4\atop  0})$ & $16$ & $1, 1, 1, 1, 1, 1, 1, 1, 2, 2$ & $0, 0, \frac{1}{2}, \frac{1}{2}, 0, 0, \frac{1}{2}, \frac{1}{2}, 0, \frac{1}{2}$ & SB:$8^{ F}_{0}({\sqrt{8}\atop  0})$ \\
$10^{ 0}_{0}({\sqrt{8}\atop  1/8})$ & $16$ & $1, 1, 1, 1, 1, 1, 1, 1, 2, 2$ & $0, 0, \frac{1}{2}, \frac{1}{2}, 0, 0, \frac{1}{2}, \frac{1}{2}, \frac{1}{8}, \frac{5}{8}$ & SB:$8^{ F}_{0}({ 2\atop  1/8})$ \\
$10^{ 0}_{0}({\sqrt{8}\atop  1/8})$ & $16$ & $1, 1, 1, 1, 1, 1, 1, 1, 2, 2$ & $0, 0, \frac{1}{2}, \frac{1}{2}, 0, 0, \frac{1}{2}, \frac{1}{2}, \frac{1}{8}, \frac{5}{8}$ & SB:$8^{ F}_{0}({ 2\atop  1/8})$ \\
$10^{ 0}_{0}({ 0\atop 0})$ & $16$ & $1, 1, 1, 1, 1, 1, 1, 1, 2, 2$ & $0, 0, \frac{1}{2}, \frac{1}{2}, 0, 0, \frac{1}{2}, \frac{1}{2}, \frac{1}{4}, \frac{3}{4}$ & SB:$8^{ F}_{0}({ 0\atop 0})$ \\
$10^{ 0}_{0}({ 0\atop 0})$ & $16$ & $1, 1, 1, 1, 1, 1, 1, 1, 2, 2$ & $0, 0, \frac{1}{2}, \frac{1}{2}, 0, 0, \frac{1}{2}, \frac{1}{2}, \frac{1}{4}, \frac{3}{4}$ & SB:$8^{ F}_{0}({ 0\atop 0})$ \\
$10^{ 0}_{0}({\sqrt{8}\atop -1/8})$ & $16$ & $1, 1, 1, 1, 1, 1, 1, 1, 2, 2$ & $0, 0, \frac{1}{2}, \frac{1}{2}, 0, 0, \frac{1}{2}, \frac{1}{2}, \frac{3}{8}, \frac{7}{8}$ & SB:$8^{ F}_{0}({ 2\atop -1/8})$ \\
$10^{ 0}_{0}({\sqrt{8}\atop -1/8})$ & $16$ & $1, 1, 1, 1, 1, 1, 1, 1, 2, 2$ & $0, 0, \frac{1}{2}, \frac{1}{2}, 0, 0, \frac{1}{2}, \frac{1}{2}, \frac{3}{8}, \frac{7}{8}$ & SB:$8^{ F}_{0}({ 2\atop -1/8})$ \\
 \hline 
\end{tabular} 
\end{table*}

\def\arraystretch{1.25} \setlength\tabcolsep{3pt}
\begin{table*}[t] 
\caption{
$Z_4^f$-SET orders for fermionic systems.
The list contains all topological orders with 
$N=6$ $D^2\leq 100$,
$N=8$ $D^2\leq 60$, and
$N=10$ $D^2\leq 20$.
} 
\label{SETZ4f} 
\centering
\begin{tabular}{ |c|c|l|l|l| } 
\hline 
$N^{|\Th|}_{c}$ & $D^2$ & $d_1,d_2,\cdots$ & $s_1,s_2,\cdots$ & comment \\
\hline 
$4^{ 0}_{0}$ & $4$ & $1, 1, 1, 1$ & $0, 0, \frac{1}{2}, \frac{1}{2}$ & $\cE=\sRp(Z_4^f)$\\
\hline
$6^{ 0}_{0}$ & $12$ & $1, 1, 1, 1, 2, 2$ & $0, 0, \frac{1}{2}, \frac{1}{2}, \frac{1}{6}, \frac{2}{3}$ & 
\tiny $K=
-\begin{pmatrix}
1&2 \\
2&1\\ 
\end{pmatrix}
$
\\
$6^{ 0}_{0}$ & $12$ & $1, 1, 1, 1, 2, 2$ & $0, 0, \frac{1}{2}, \frac{1}{2}, \frac{1}{3}, \frac{5}{6}$ & 
\tiny $K=
\begin{pmatrix}
1&2 \\
2&1\\ 
\end{pmatrix}
$
\\
\hline
$8^{ 0}_{0}$ & $8$ & $1, 1, 1, 1, 1, 1, 1, 1$ & $0, 0, \frac{1}{2}, \frac{1}{2}, \frac{1}{4}, \frac{1}{4}, \frac{3}{4}, \frac{3}{4}$ & $2^{ B}_{-1}\boxtimes \sRp(Z_4^f)$\\
$8^{ 0}_{0}$ & $8$ & $1, 1, 1, 1, 1, 1, 1, 1$ & $0, 0, \frac{1}{2}, \frac{1}{2}, \frac{1}{4}, \frac{1}{4}, \frac{3}{4}, \frac{3}{4}$ & $2^{ B}_{ 1}\boxtimes \sRp(Z_4^f)$\\
$8^{ 0}_{-14/5}$ & $14.472$ & $1, 1, 1, 1,\zeta_{3}^{1},\zeta_{3}^{1},\zeta_{3}^{1},\zeta_{3}^{1}$ & $0, 0, \frac{1}{2}, \frac{1}{2}, \frac{1}{10}, \frac{1}{10}, \frac{3}{5}, \frac{3}{5}$ & $2^{ B}_{-14/5}\boxtimes \sRp(Z_4^f)$\\
$8^{ 0}_{14/5}$ & $14.472$ & $1, 1, 1, 1,\zeta_{3}^{1},\zeta_{3}^{1},\zeta_{3}^{1},\zeta_{3}^{1}$ & $0, 0, \frac{1}{2}, \frac{1}{2}, \frac{2}{5}, \frac{2}{5}, \frac{9}{10}, \frac{9}{10}$ & $2^{ B}_{ 14/5}\boxtimes \sRp(Z_4^f)$\\
$8^{ 0}_{0}$ & $20$ & $1, 1, 1, 1, 2, 2, 2, 2$ & $0, 0, \frac{1}{2}, \frac{1}{2}, \frac{1}{10}, \frac{2}{5}, \frac{3}{5}, \frac{9}{10}$ &  SB:$10^{ F}_{0}({\zeta_{2}^{1}\atop  0})$ \\
$8^{ 0}_{0}$ & $20$ & $1, 1, 1, 1, 2, 2, 2, 2$ & $0, 0, \frac{1}{2}, \frac{1}{2}, \frac{1}{5}, \frac{3}{10}, \frac{7}{10}, \frac{4}{5}$ & SB:$10^{ F}_{0}({\zeta_{2}^{1}\atop  1/2})$ \\
\hline
$10^{ 0}_{0}({ 4\atop  0})$ & $16$ & $1, 1, 1, 1, 1, 1, 1, 1, 2, 2$ & $0, 0, \frac{1}{2}, \frac{1}{2}, 0, 0, \frac{1}{2}, \frac{1}{2}, 0, \frac{1}{2}$ & SB:$8^{ F}_{0}({\sqrt{8}\atop  0})$ \\
$10^{ 0}_{0}({ 4\atop  0})$ & $16$ & $1, 1, 1, 1, 1, 1, 1, 1, 2, 2$ & $0, 0, \frac{1}{2}, \frac{1}{2}, 0, 0, \frac{1}{2}, \frac{1}{2}, 0, \frac{1}{2}$ & SB:$8^{ F}_{0}({\sqrt{8}\atop  0})$ \\
$10^{ 0}_{0}({\sqrt{8}\atop  1/8})$ & $16$ & $1, 1, 1, 1, 1, 1, 1, 1, 2, 2$ & $0, 0, \frac{1}{2}, \frac{1}{2}, 0, 0, \frac{1}{2}, \frac{1}{2}, \frac{1}{8}, \frac{5}{8}$ & SB:$8^{ F}_{0}({ 2\atop  1/8})$ \\
$10^{ 0}_{0}({\sqrt{8}\atop  1/8})$ & $16$ & $1, 1, 1, 1, 1, 1, 1, 1, 2, 2$ & $0, 0, \frac{1}{2}, \frac{1}{2}, 0, 0, \frac{1}{2}, \frac{1}{2}, \frac{1}{8}, \frac{5}{8}$ & SB:$8^{ F}_{0}({ 2\atop  1/8})$ \\
$10^{ 0}_{0}({ 0\atop 0})$ & $16$ & $1, 1, 1, 1, 1, 1, 1, 1, 2, 2$ & $0, 0, \frac{1}{2}, \frac{1}{2}, 0, 0, \frac{1}{2}, \frac{1}{2}, \frac{1}{4}, \frac{3}{4}$ & SB:$8^{ F}_{0}({ 0\atop 0})$ \\
$10^{ 0}_{0}({ 0\atop 0})$ & $16$ & $1, 1, 1, 1, 1, 1, 1, 1, 2, 2$ & $0, 0, \frac{1}{2}, \frac{1}{2}, 0, 0, \frac{1}{2}, \frac{1}{2}, \frac{1}{4}, \frac{3}{4}$ & SB:$8^{ F}_{0}({ 0\atop 0})$ \\
$10^{ 0}_{0}({\sqrt{8}\atop -1/8})$ & $16$ & $1, 1, 1, 1, 1, 1, 1, 1, 2, 2$ & $0, 0, \frac{1}{2}, \frac{1}{2}, 0, 0, \frac{1}{2}, \frac{1}{2}, \frac{3}{8}, \frac{7}{8}$ & SB:$8^{ F}_{0}({ 2\atop -1/8})$ \\
$10^{ 0}_{0}({\sqrt{8}\atop -1/8})$ & $16$ & $1, 1, 1, 1, 1, 1, 1, 1, 2, 2$ & $0, 0, \frac{1}{2}, \frac{1}{2}, 0, 0, \frac{1}{2}, \frac{1}{2}, \frac{3}{8}, \frac{7}{8}$ & SB:$8^{ F}_{0}({ 2\atop -1/8})$ \\
 \hline 
\end{tabular} 
\end{table*}

\section{Fusion ring for the modular extensions of $\Rp(G)$ or $\sRp(G^f)$ when
$G$ or $G^f$ is abelian group}

\label{FRgroup}

When the symmetry group $G$ is abelian, the different irreducible
representations, under the fusion, form the same group $G$.  Thus different
irreducible representations can be labeled by the group elements: $(q),\ q\in
G$.  The different symmetry twists are also labeled by the group elements:
$[g],\ g\in G$.  More general symmetry twists may carry some charge. We denote
such charge carrying symmetry twists by $[g, q]$ where $q\in G$.  In fact we can
identify $(q)$ as $[1,q]$.  Those irreducible representations and charged
symmetry twists are particles in the modular extensions of
$\Rp(G)$ or $\sRp(G^f)$.

Since the group is abelian, the symmetry twists do not break the symmetry.
Thus, we have the following fusion rule
\begin{align}
 [1,q] \otimes [g,q'] = [g,qq']
\end{align}
This means that $[g,q']$ and $[g,qq']$ differ by charge $q$.
We also have
\begin{align}
 [g,q]\otimes [g',q']=[gg',qq']
\end{align}
However, the above fusion rule is too restrictive.
Although  $[g,q']$ and $[g,qq']$ differ by charge $q$, we do not know
the net charge of $[g,q']$ when $g\neq 1$.
Thus the more general fusion rule that still preserves charge conservation is
\begin{align}
  [g,q]\otimes [g',q']=[gg',\omega_2(g,g') qq'], \ \ \ \omega_2(g,g') \in G.
\end{align}

From
\begin{align}
&\ \ \ ( [g_1,q_1]\otimes [g_2,q_2])\otimes [g_3,q_3]
\nonumber\\
&=
[g_1g_2g_3, \omega(g_1,g_2)\omega(g_1g_2,g_3)q_1q_2q_3]
\nonumber\\
&= [g_1,q_1]\otimes ([g_2,q_2]\otimes [g_3,q_3])
\nonumber\\
&=
[g_1g_2g_3, \omega(g_1,g_2g_3)\omega(g_2,g_3)q_1q_2q_3]
\end{align}
we see that
\begin{align}
\label{omcnd}
 \omega(g_1,g_2)\omega(g_1g_2,g_3)= \omega(g_1,g_2g_3)\omega(g_2,g_3).
\end{align}
\ie $\omega(g_1,g_2)$ is a group 2-cocycle in $\cH^2(G,G)$.

In the above, we have assumed that the modular extension is abelian (\ie all
the particles in the modular extension have a quantum dimension $1$).  We see
that the fusion rules of abelian  modular extensions are labeled by 2-cocycles
in $\cH^2(G,G)$.

However, sometimes the modular extension can be non-abelian, such as the
modular extension of $\sRp(Z_2^f)$ and $\Rp(Z_2\times Z_2\times Z_2)$.  To
allow such a possibility, we allow $[g,q]$ to be a many-to-one label of the
particle, and define a subgroup $H_g\subset G$:
\begin{align}
H_g = \{ h| [g,q] =[g,hq], h\in G\}.
\end{align} 
The mapping $g\to H_g$ is an important data to describe the fusion.  $H_g$
represents the charge ambiguity of the symmetry twist $[g,q]$.
To get an one-to-one label, we can use
\begin{align}
 [g,qH_g].
\end{align}
Note that, when $g$ is an identity: $g=1$, $H_{g}$ is trivial: $H_1=1$.

The fusion of $[1,q']$ and $[g, qH_g]$ is still given by
\begin{align}
 [1,q']\otimes [g, qH_g] = [g, q'qH_g].
\end{align}
We also have $H_g=H_{g^{-1}}$ and
\begin{align}
 [g, qH_g]\otimes [g^{-1}, q'H_g]= \oplus_{h\in qq' H_g} [1,h]
\end{align}
We see that the quantum dimension of
$[g,qH_g]$ is $d=\sqrt{|H_g|}$.

The fusion rule should satisfy
\begin{align}
&\ \ \ \ [1,q]\otimes  ([g_1, q_1H_{g_1}]\otimes [g_2, q_2H_{g_2}]) 
\\
&= ([1,q]\otimes  [g_1, q_1H_{g_1}])\otimes [g_2, q_2H_{g_2}] 
\nonumber\\
&= [g_1, q_1H_{g_1}]\otimes ([1,q]\otimes[g_2, q_2H_{g_2}]) 
\nonumber
\end{align}
We find that the following ansatz satisfy the above condition
\begin{align}
\label{frG}
&
 [g_1, q_1H_{g_1}]\otimes [g_2, q_2H_{g_2}] 
=
\frac{m^{g_1g_2}}{|(H_{g_1}\vee H_{g_2})\cap H_{g_1g_2}|}
\nonumber\\
&\ \ \ \
\oplus_{q \in \omega(g_1,g_2) q_1 q_2 H_{g_1}\vee H_{g_2}} [g_1g_2, qH_{g_1g_2}]
\end{align}
where $m^{g_1g_2}\in \Z$ and $H_{g_1}\vee H_{g_2}$ is the subgroup generated by $H_{g_1}$ and
$H_{g_2}$.  The above implies that
\begin{align}
\label{Hgcnd}
 \sqrt{|H_{g_1}|}\sqrt{|H_{g_2}|}=m^{g_1g_2}\frac{|H_{g_1}\vee H_{g_2}|}{|(H_{g_1}\vee H_{g_2})\cap H_{g_1g_2}|} \sqrt{|H_{g_1g_2}|}
\end{align}
We see that different fusion rules are labeled by $\omega(g_1,g_2)$
and $H_g$.

It is much easier to find all the $H_g$'s that satisfy \eqn{Hgcnd} and all the
$\omega(g_1,g_2)$ that satisfy \eqn{omcnd}.  From those solutions, we can
directly construct the fusion rule from \eqn{frG}.

\section{Conditions to obtain $\mce{\cE}$'s 
}
\label{cnds}

In our simplified theory, a $\mce{\cE}$ is described by an integer tensor
$N^{ij}_k$ and a mod-1 rational vector $s_i$, where $i,j,k$ run from 1 to $N$ and
$N$ is called the rank of the $\mce{\cE}$.  We may simply denote a $\mce{\cE}$
(the collection of data ($N^{ij}_k,s_i$)) by $\cC$, a particle $i$ in $\cC$ by
$i\in\cC$.  Sometimes it is more convenient to use abstract labels rather than
1 to $N$; we may also abuse $\cC$ as the set of labels (particles).

Not all $(N^{ij}_k, s_i)$ describe a valid $\mce{\cE}$ $\cC$
with modular extensions.  In order to describe a valid $\cC$, $(N^{ij}_k,s_i)$
must satisfy the following
conditions:\cite{W8951,GKh9410089,BK01,RSW0777,Wang10}
\begin{enumerate} 
\item \emph{Fusion ring}:\\
$N^{ij}_k$ for the $\mce{\cE}$ $\cC$ are non-negative integers that satisfy 
\begin{align} 
\label{Ncnd} &
N^{ij}_k=N_k^{ji}, \ \ N_j^{1i}=\delta_{ij}, \ \ \sum_{k=1}^N N_1^{i k}N_1^{
kj}=\delta_{ij},
\\
 & \sum_{m} N_m^{ij}N_l^{m k} = \sum_{n}
 N_l^{in}N_n^{j k} \text{ or } \sum_m N^{ij}_m N_m = N_i N_ j 
\nonumber 
\end{align} 
where the matrix $N_i$ is given by $(N_{i})_{ kj} =
N^{ij}_k$, and the indices $i,j,k$ run from 1 to $N$.  In fact $ N_1^{ij}$ defines a charge conjugation $i\to \bar i$:
\begin{align} 
N_1^{ij}=\delta_{\bar ij}.  
\end{align}
$N^{ij}_k$ satisfying the above conditions define a fusion ring  which
is viewed as the set (of simple objects)
\begin{align}
 \{1,2,\cdots,N\} .
\end{align}

\item \emph{Charge conjugation condition}:\\
\begin{align}
\label{conj}
&\ \ \ \ N^{ij}_k =
 N^{j\bar k}_{\bar i} =
 N^{\bar k i}_{\bar j} 
\nonumber\\
&= N^{\bar i k}_j
= N^{k\bar j}_{i} 
= N^{\bar j\bar i}_{\bar k}. 
\end{align}

\item \emph{Rational condition}:\\
$N^{ij}_k$ and $s_i$ for $\cC$ satisfy\cite{V8821,AM8841,BK01,Em0207007}
\begin{align}
\label{Vcnd}
\sum_r V_{ijkl}^r s_r =0 \text{ mod }1
\end{align} 
where
\begin{align}
\ \ \ \ \ \ 
V_{ijkl}^r &=  
N^{ij}_r N^{kl}_{\bar r}+
N^{il}_r N^{jk}_{\bar r}+
N^{ik}_r N^{jl}_{\bar r}
\nonumber\\
&\ \ \ \
- ( \delta_{ir}+ \delta_{jr}+ \delta_{kr}+ \delta_{lr}) \sum_m N^{ij}_m N^{kl}_{\bar m}
\end{align}

\item \emph{Verlinde fusion characters}:\\
Let the topological $S$-matrix be [see eqn. (223) in \Ref{K062}] 
\begin{align} 
\label{SNsss}
S_{ij}&=\frac{1}{D}\sum_k N^{i j}_k \ee^{2\pi\ii(s_i+s_j-s_k)} d_k ,
\end{align}
where $d_i$ (called quantum dimension) is the largest eigenvalue of the matrix
$N_i$ and ${D=\sqrt{\sum_i d_i^2}}$ (called the total quantum dimension).  
Then
\cite{V8860}:
\begin{align}
\label{Ver}
  \frac{ S_{il} S_{jl}}{ S_{1l} } =\sum_k N^{ij}_k S_{kl}.
\end{align} 

\item
\emph{Weak modularity}: \\
Let the topological $T$-matrix be
\begin{align} 
T_{ij}&=\delta_{ij}\ee^{2\pi\ii s_i}.
\end{align} 
Then [see eqn. (232) in \Ref{K062}]
  \begin{align}
    S^\dag T S&=\Theta T^\dag S^\dag T^\dag,
\nonumber\\
\Theta &={D}^{-1}\sum_{i}\ee^{2\pi\ii s_i} d_i^2=
|\Th|\ee^{2\pi \ii c/8}.
  \end{align}
The parameter $c$ mod 8 is defined via $\Th$, if $|\Th|\neq 0$.

\item 
\emph{Charge conjugation symmetry}:\\
\begin{align}
  S_{ij}=S_{i\bar j}^*,\ s_i=s_{\bar i}, \text{ or } S=S^\dag C,\ \ T=TC,
\end{align}
where the charge conjugation matrix $C$ is given by $C_{ij}=N_1^{ij}=\delta_{i\bar j}$.

\item
\emph{The centralizer describes the symmetry}:\\
Let the centralizer of $\cC$, $\cC_\cC^\text{cen}$, be the subset of the particle labels:
\begin{align}
\cC_\cC^\text{cen}=
\{i\ |\
  S_{ij}=\frac{d_id_j}{D},\  
\forall j \in \cC\}.
\end{align}
Then, $\cC_\cC^\text{cen} = \cE$.

\item
\emph{The second Frobenius-Schur indicator}:\\
Let 
\begin{align}
 \nu_k=D^{-2}\sum_{ij} N^{ij}_k d_id_j \cos(4\pi (s_i-s_j)),
\end{align}
then $\nu_k \in \Z$ if $k=\bar k$\cite{B12044836}.

\item
\emph{Symmetry breaking}:\\
There is a symmetry breaking induced map $\cC \to \cC_0$,
where $\cC_0$ is a UMTC if $\cE=\Rp(G)$
or a $\mce{\sRp(Z_2)}$ if $\cE=\sRp(G^f)$.
See Appendix \ref{SB} for details.
 
\item
\emph{Modular extension}:\\
The $\mce{\cE}$ $\cC$ has modular extensions.

\end{enumerate} 
The above conditions are necessary and sufficient (due to the condition 10) for
$(N^{ij}_k,s_i)$ to describe a $\mce{\cE}$ $\cC$ with modular extensions.

However, when we calculate the tables in Appeandix \ref{SETtbl}, we do not use
the condition 10. So the used conditions are only necessary.  As a result, the
tables may contain fake entries that have no modular extensions.

\section{Symmetry breaking}
\label{SB}

A $\mce{\cE}$ $\cC$ describes a SET with symmetry $\cE$ (up to invertible
GQLs).  If we break the symmetry $\cE$, then the $\mce{\cE}$ will become a UMTC
$\cC_0$ if $\cE=\Rp{G}$ or become a $\mce{Z_2^f}$ $\cC_0$ if $\cE=\sRp{G^f}$.
So there is a natural mapping from $\mce{\cE}$'s to UMTCs or $\mce{Z_2^f}$:
$\cC\to \cC_0$.  Requiring the existance of such map can give us some
additional conditions on $(N^{ij}_k,s_i)$ of $\cC$.

To understand such a map, we note that $\cC$ can be viewed as a subcategory of
$\cC_0$, in the sense that the simple objects in $\cC$ can be viewed
as the simple or composite objects in $\cC_0$:
\begin{align}
\label{iMI}
 i \to \oplus_I M^{iI} I,\ \ \ i \in \cC,\ \ \ I\in \cC_0.
\end{align}
Physical, if we just pretend the symmetry is not there, then every particle in
$\cC$ can also be viewed as a particle in $\cC_0$.  However, a particle in
$\cC$ may be the direct sum of several degenerate particles in $\cC_0$, where
the degeneracy is due to the symmetry, as described by \eqn{iMI}.

In the following, we will obtain some conditions on $M^{iI}$, which will help
us to calculate it.  Let us label the particles in $\cC$ as
$\{i\}=\{1,a,b,\cdots,x,y,\cdots\}$.  Here $a,b,\cdots$ label the
\emph{bosonic} part of $\cE$, and $x,y,\cdots$ label the fermionic part of
$\cE$ (if any) and the rest of non-trival topological excitations.  We have
also used $I$ to label the particles in $\cC_0$.  Clearly, the bosonic part of
$\cE$ are local excitations and are direct sums of $\textbf{1} \in \cC_0$:
\begin{align}
 a \to d_a \textbf{1},\ \ \ \ \text{ or } \ \ \
M^{aI}=d_a\delta_{1I},
\end{align}
(Here $\textbf{1}$ is the trivial particle in $\cC_0$.)
By computing $i\otimes j$ in two different ways, we find that
$M^{iI}$ must also satisfy
\begin{align}
\sum_{IJ} M^{iI} M^{jJ} N^{IJ}_K= \sum_k N^{ij}_k M^{kK}
\end{align}
Taking $K=\textbf{1}$, we obtain
\begin{align}
\sum_{I} M^{iI} M^{j\bar I} = \sum_a N^{ij}_a d_a
\end{align}
Assuming the charge conjugation symmetry: $M^{iI}=M^{\bar i\bar I}$, 
we can rewrite the above as
\begin{align}
\label{MiMj}
\sum_{I} M^{iI} M^{j I} = \sum_a N^{i\bar j}_a d_a,
\end{align}
which implies that
\begin{align}
\sum_{I} (M^{iI})^2  = \sum_a N^{i\bar i}_a d_a.
\end{align}

To obtain more properties of $M^{iI}$ and to solve the above conditions on
$M^{iI}$, let us consider the fusion with $a$ partciles: \begin{align} a\otimes
x = \oplus_y N^{ax}_y y.  \end{align} We define $x$ to be equivalent to $y$ if
there exists $a$ such that $N^{ax}_y\neq 0$.  Let $[x]$ be the equivalent class
of $x$. Clearly $[1]=[a]$.  

First, we like to pointed out that if $i$ and $j$ are equivalent, then
$i$ and $j$ are formed by the same combination of $I$'s, up to
an overall factor, such as
\begin{align}
 i \to I_1\oplus 2 I_2,\ \ \ \ j \to 3I_1\oplus 6I_2.
\end{align}
This is because $a$ particles in $\cC$ is mapped to the direct-sum of identity
in $\cC_0$. Since $i$ and $j$ is related by fusing $a$ or identity in $\cC_0$,
then $i$ and $j$ must be formed by the same combination of $I$'s.

Second, if  $i$ and $j$ are not equivalent, then the $I$'s that enter $i$ do not
overlap with the $I$'s that enter $j$. This is a consequence of
\eqn{MiMj}. The right hand side of \eqn{MiMj} will vanish if
$i$ and $j$ are not equivalent.

Third, the $I$'s that appear in $i$ must have the same quantum dimensions and
spins. This is because those $I$'s must be degenerate. This can only happen if
they have  the same quantum dimensions and spins.

Fourth,  the $I$'s that appears in $i$ must each enter with an equal weight,
such as
\begin{align}
 i \to 2I_1 \oplus 2I_2.
\end{align} 
Again, this is because those $I$'s must be degenerate.  This can only happen if
they can be mapped into each other by symmetry transformations.  Since the
symmetry transformations only permute $I$'s, each $I$ enters with an equal
weight.

Combine the above results, we see that
$M^{iI}$ has the following block structure. We can divide the index $I$
into groups $[I]$, such that there is one-to-one correspondence between
$[i]$ and $[I]$: $[i]\leftrightarrow [I]_{[i]}$, and
\begin{align}
 M^{iI}=0\ \ \ &\text{ if }\ \ \ i\in [i],\ I\notin [I]_{[i]},
\nonumber\\
 M^{iI}=m_i>0\ \ \ &\text{ if }\ \ \ i\in [i],\ I\in [I]_{[i]}.
\end{align}
Therefore, we have
\begin{align}
 m_i^2 n_{[i]} = \sum_a N^{i\bar i}_a d_a,
\end{align}
where $n_{[i]}$ is the size of the set $[I]_{[i]}$.
Since 
\begin{align}
 i = \oplus_{I\in [I]_{[i]}} m_i I,
\end{align}
we have
\begin{align}
 m_im_j n_{[i]} = \sum_a N^{i\bar j}_a d_a,\ \ \ 
i,j \in [i]
\end{align}
In other words, the matrix $\tilde N$ with elements $\tilde N_{ij}=\sum_a
N^{i\bar j}_a d_a$ is block diagonal.  Each block is formed by particles in an
equivalent class $[i]$, and is given by the above expression.  We see that, for
$i,j \in [i]$, $\sum_a N^{i\bar j}_a d_a$ must be a symmetric matrix with a
single non-zero eigenvalue $n_{[i]} \sum_{j\in [i]} m_j^2$ and eigenvector
$(m_j)$.

We also find that
\begin{align}
 d_i =m_in_{[i]} d_I,
\end{align}
or
\begin{align}
 d_I =\frac{m_i d_i}{\sum_a N^{i\bar i}_a d_a}\ \ \ \forall \ \ \ I \in [I]_{[i]}.
\end{align}
Using the fact $s_i=s_j=s_I$, $\forall i,j \in [i],\ I \in [I]_{[i]}$, we can
obtain $(d_I,s_I)$ of $\cC_0$ from $(N^{ij}_k,s_i)$ of $\cC$.  The resulting
$(d_I,s_I)$ must be the quantum dimensions and the spins of a UMTC.
This gives us some extra conditions on $(N^{ij}_k,s_i)$.

\section{Physical and mathematical meaning of
 $\mce{\cE}$ and its modular extensions}

In the main text of the paper, we have explained why $\mce{\cE}$ describes the
bulk particle-like excitations.  We also explained the motivation of  modular
extension via ``gauging'' the symmetry.  In this section, we will discuss a
deeper meaning of  $\mce{\cE}$ and its modular extensions.

We know that $\mce{\cE}$ is a very abstract way to describe the non-abelian
statistics of the excitations.  It is not clear at all that why the excitations
described by $\mce{\cE}$ can be realized by a local lattice model with on-site
symmetry. In physics, we mainly concern about local lattice models and their
properties.  It appears that there is a big gap between the $\mce{\cE}$ studied
in this paper and local lattice models that physicists want to study.  In fact,
the two are closely related.  Here, we will try to explain such a connection
between lattice models and $\mce{\cE}$ (with their modular extensions).

We know that the fusion-braiding properties of particles within a 2-dimensional
open disk can be described by a unitary braided fusion category.  From this
point of view,  a unitary braided fusion category is a \emph{local} theory that
only encode the local properties of the fusion and braiding ({\it i.e.} on an
open disk). We want to promote fusion-braiding properties to be integrable to
any 2-dimensional manifolds because we want those fusion-braiding properties to
be realizable by some local lattice models, which can always be defined on any
2-dimensional manifolds. Therefore, the integrability of fusion-braiding properties to any 2-dimensional manifolds is necessary for the fusion-braiding properties to be realized by some local lattice models.

Now we assume that ``all 2-dimensional manifolds'' are the most powerful
probes. This means that the integrability of the local fusion-braiding
properties to global invariants (on all 2-dimensional manifolds), satisfying
natural physically required properties, is also sufficient for those properties to be realizable by some local lattice models.

The process of integrating the local fusion-braiding properties of particles
(described by a UBFC $\cC$) to give global invariants is defined by the
so-called factorization homology.\cite{lurie, bbj} In order to be free of
framing anomaly, we need a spherical structure, which is guaranteed by the
unitarity of a UBFC. \cite{K062} For general UBFCs, although the global
invariants are well-defined by factorization homology,\cite{bbj} they do not
have nice properties that allow us to give them a natural physical meaning. A
stronger {\it integrability condition} needs to be imposed in order for the
global invariants to have natural physical meanings.

For example, if $\cC$ is assumed to be non-degenerate (\ie UMTC), it was shown
in \Ref{ai-kong-zheng} that factorization homology of a UMTC $\cC$ over a
closed 2-dimensional manifold is given by the category of finite dimensional
Hilbert spaces. If one inserts a finite number of particle-like excitations
$x_1, \cdots, x_r$ on the closed surface, one simply obtain the Hilbert space
$\hom_\cC(\one, x_1\otimes \cdots \otimes x_r)$, which is also the space of
degenerate ground states. This result remains to be true for all closed
2-dimensional manifolds with topological gapped defects and with 2-cells
decorated by different phases.\cite{ai-kong-zheng} This includes the cases that the topological order is defined on any surfaces with boundaries. Therefore, the non-degeneracy is certainly a sufficient integrability condition, which is too strong for the purpose of this work.

%We also like to promote unitary braided fusion category in another direction, so that each of its particles is detectable via double braidings with all particles. In particular,  if there is a particle that fuses and braids just like the vacuum (or the unit) object in the category, then it should be identified with the vacuum.  This motivates us to only consider non-degenerate unitary braided fusion categories, {\it i.e.} the unitary modular tensor categories, where each non-identity particle has a non-trivial double braiding ({\it i.e.} non-trivial mutual statistics) with some particles.  

In this paper, we consider something more complicated -- the fusion-braiding
properties of particles with symmetry.  By ``with symmetry'', we mean to
include local excitations that carry representations of the symmetry group.
Mathematically, this means that the unitary braided fusion category $\cC$ contain a
SFC $\cE$ as its M\"{u}ger center, i.e. a $\mce{\cE}$. We know that either $\cE=\text{Rep}(G)$ or $\cE=\text{sRep}(G^f)$, where $G$ or $G^f$ is the symmetry group. In this case, we must find a proper integrability condition that is weaker than the non-degeneracy of UBFC. 

In order for the factorization homology of $\cC$ on a surface, a unitary category denoted by $\cC_\Sigma$, to have a physical meaning, we suspect that we should be able to interpret its 
object as finite dimensional Hilbert spaces in a natural way. This suggests
that the category $\cC_\Sigma$ should equipped with a natural functor to the
category of finite dimensional Hilbert spaces, which is a factorization
homology $\cM_\Sigma$ of a UMTC $\cM$.\cite{ai-kong-zheng} So we expect that we
should be able to embed $\cC$ into a UMTC $\cM$ such that the embedding
naturally descends to a functor $\cC_\Sigma \to \cM_\Sigma$ on factorization
homologies. An arbitrary UMTC such as the Drinfeld center $Z(\cC)$ of $\cC$ can
not do the job because there is no canonical way to identify $\cC$ in $\cM$
(with a fixed symmetry $\cE$) so that it is unlikely that it can be compatible
with the integration process. So we expect that the condition
$\cen{\cE}{\cM}=\cC$ is a natural integrability condition that replace the
non-degeneracy condition in this case. This flow of thinking leads us to the
concept of the modular extension of $\cC$. It also suggests that the
non-existence of the modular extension of a given $\cC$ means that $\cC$ is
somewhat inconsistent globally or not integrable to all 2-dimensional manifolds
with natural physical meanings. 

\medskip
This can also be viewed from a different point of view. If we require each particle to be non-trivial in some sense, then we must only consider the non-degenerate unitary braided fusion category over SFC $\cE$.  In this case, for particles not in $\cE$, we know they are non-trivial because their non-trivial double braiding (or non-trivial mutual statistics)
with some particles.  But we still have trouble to know why the particles in
$\cE$ are non-trivial?  From their fusion and braiding properties, they just
behave like the identity or a composite of identities.  

To fix this problem, we put our particles on any 2-dimensional manifolds.  In this case, we can find a way to understand the non-trivialness of the particle
in $\cE$.  This require us to twist the symmetry $G$ or $G^f$ on the
2-dimensional manifold.  In other words, we equip the 2-dimensional manifold
with a flat $G$-connection.  Since the particles in $\cE$ all carry irreducible
representations of $G$, as we move the particles along a non-contractile loop,
the flat $G$-connection will induce a $G$ transformation on the particle (or
more precisely, on the hom space of the particles).  This allows us to probe
the particles in $\cE$ and detect their non-trivialness.

Therefore, as we put particles on a 2-dimensional manifold, it is important to
allow any flat $G$-connection on the manifold.  Now we ask, in this case, can a
non-degenerate unitary braided fusion category $\cC$ over a SFC $\cE$ describes
the fusion-braiding properties of particles that are consistent on any
2-dimensional manifolds with any flat $G$-connections?

In this paper, we propose that the answer is no.  We also propose that the
answer is yes iff the $\cC$ over $\cE$ has modular extensions, which are the categorical ways of gauging the symmetry $\cE$. So, non-degenerate unitary braided fusion categories over SFC can describe the consistent local fusion and braiding on an open disk.  Only the ones with modular extensions can describe the consistent fusion and braiding on any
manifolds (with any flat $G$-connections).

\begin{figure}[tb] 
\centerline{ \includegraphics[scale=0.8]{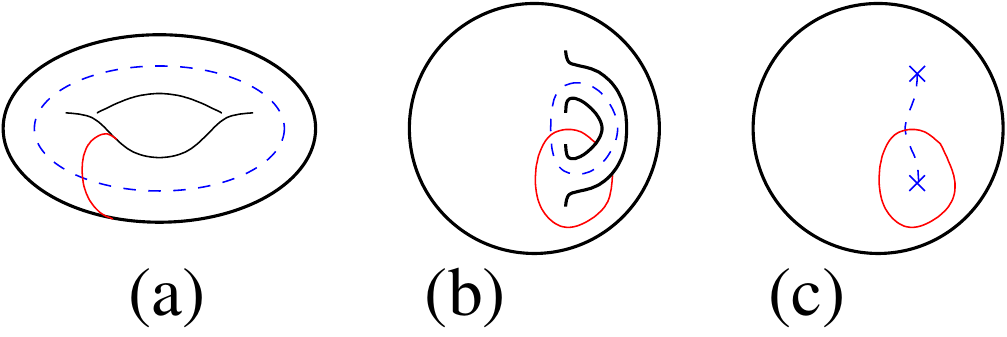} } 
\caption{
(a) A torus with a flat $G$-connection (described by a symmetry twist
along the dashed loop). The thin solid loop is a braiding path.
(b) A handle is deformed into a very thin one.
(c) A very thin handle can be viewed as two defects, and each defect
corresponds to the added particle in the modular extension.
} 
\label{torusbraid} 
\end{figure}

The intuition for the above conjecture is explained in the Fig.
\ref{torusbraid}.  Fig. \ref{torusbraid}(a) describes a braiding of particles
on a torus with flat $G$-connection.  As we deform a handle into a very thin
one, we may view the above braiding on torus as a braiding around the added
particles in the modular extension.  So the consistent fusion and braiding on
any manifolds with any flat $G$-connection must be closely related to the
consistent fusion and braiding on a sphere with the added particles in the
modular extension.  So, the mathematical meaning of the modular extension is to
make the fusion and braiding to be consistent on any manifolds with any flat
$G$-connection.

For a given $\cC$ over $\cE$, there can  be several modular extensions $\cM$.
We believe that those different modular extensions describe the different
structures at the boundary.  This picture leads to the physical conjecture that
the triple $(\cC, \cM, c)$ classify the 2+1D topological/SPT orders
with symmetry $\cE$.

\vfill
\clearpage

\bibliography{../../bib/wencross,../../bib/all,../../bib/publst,./local} 

\end{document}